\documentclass[reqno,11pt]{amsart}
\usepackage[margin=1in]{geometry}
\usepackage{xcolor}

\usepackage{amsmath,amssymb,amsthm, comment, mathtools}
\usepackage{slashed}
\usepackage[english]{babel}
\usepackage[utf8]{inputenc}

\usepackage{hyperref}

\newtheorem{theorem}{Theorem}[section]\newtheorem{corollary}[theorem]{Corollary}
\newtheorem{proposition}[theorem]{Proposition}\newtheorem{lemma}[theorem]{Lemma}
\newtheorem{remark}[theorem]{Remark}

\numberwithin{equation}{section}
\mathtoolsset{showonlyrefs=true}

\def\eps{{\varepsilon}}

\def\bF{\mathbf{F}}\def\bphi{\boldsymbol{\phi}}\def\bpsi{\boldsymbol{\psi}}
\def\bA{\mathbf{A}}
\def\pa{\partial}
\def\swpa{\slashed{\wpa}}
\def\pas{\text{$\pa\mkern -11.0mu$\slash}}
\def\wpa{\widetilde{\pa}}
\def\wpao{\overline{\widetilde{\pa}}}\def\wL{\widetilde{L}}
\def\wuL{\widetilde{\underline{L}}}

\def\Lie{{\mathcal L^{{}^{{\!}}}}}
\newcommand{\onabla}{\widetilde\nabla}
\newcommand{\Wcal}{\mathcal{W}}
\def\rectangle#1#2{\hbox{\vrule\vbox to #2 {\hrule\hbox to #1{\hfil}\vfil\hrule}\vrule}}
\def\Box{\square}

\def\tr{\text{tr}}\def\trs{\,\slash{\!\!\!\!\tr}\,}
\def\Lb{\underline{L}}\def\wh{ \widetilde{h}}\def\th{ \widetilde{h}}
\def\tg{{\widetilde{g}}}
\def\rs{r^{\!*}}\def\us{u^{\!*}}\def\wZ{\widetilde{Z}}

\newcommand{\R}{{\mathbb R}}

\newcommand{\shortminus}{\scalebox{0.75}[1.0]{\( - \)}}

\newcommand*{\bigtwo}[1]{\vcenter{\hbox{\scalebox{1.4}{\ensuremath#1}}}}
\newcommand*{\bigthree}[1]{\vcenter{\hbox{\scalebox{1.6}{\ensuremath#1}}}}

\newcommand{\les}{\lesssim}
\newcommand{\wt}{\widetilde}
\newcommand{\beq}{\begin{equation}}\newcommand{\eq}{\end{equation}}
\newcommand{\whm}{\widehat{m}}\newcommand{\whh}{\widehat{h}}
\newcommand{\tplusr}{\langle \mkern 1.0mu t\!{}_{{}_{\!}}+\!r \mkern -1.0mu \rangle}
\newcommand{\tminusrstar}{\langle \mkern 1.0mu t\!-\!r^*\rangle}
\newcommand{\tminusr}{\langle \mkern 1.0mu t\!-\!r \mkern -1.0mu \rangle}

\DeclareMathOperator{\sgn}{sgn}

\DeclareFontFamily{U}{mathx}{\hyphenchar\font45}
\DeclareFontShape{U}{mathx}{m}{n}{<5> <6> <7> <8> <9> <10>
      <10.95> <12> <14.4> <17.28> <20.74> <24.88> mathx10}{}
\DeclareSymbolFont{mathx}{U}{mathx}{m}{n}
\DeclareFontSubstitution{U}{mathx}{m}{n}
\DeclareMathAccent{\widecheck}{0}{mathx}{"71}
\DeclareMathAccent{\wideparen}{0}{mathx}{"75}

\begin{document}
\title[Stability of Minkowski space for the Einstein-Maxwell-Klein-Gordon system]
{Global stability of Minkowski space for the Einstein-Maxwell-Klein-Gordon system in generalized wave coordinates}

\author{Christopher Kauffman}
\address[C.K.]{Westfälische Wilhelms-Universität Münster, Mathematisches Institut,
Einsteinstrasse 62, 48149 Münster, Bundesrepublik Deutschland}
\email{ckauffma@uni-muenster.de}

\author{Hans Lindblad}
\address[H.L.]{Johns Hopkins University, Department of Mathematics, 3400 N.\@ Charles St., Baltimore, MD 21218, USA}
\email{lindblad@math.jhu.edu}

\date{\today}

\begin{abstract} We prove global existence for Einstein's equations with a charged scalar field for initial conditions sufficiently close to the Minkowski spacetime without matter.
The proof relies on  generalized wave coordinates adapted to the outgoing Schwarzschild light cones and the estimates for the massless Maxwell-Klein-Gordon system, on the background of metrics asymptotically approaching Schwarzschild at null infinity in such coordinates, by Kauffman \cite{Ka18}.
The generalized wave coordinates are obtained from a change of variables, introduced in Lindblad \cite{L17},
to asymptotically Schwarzschild coordinates at null infinity.
The main technical advances are that the change of coordinates
makes critical components of the metric decay faster, 
making the quasilinear wave operator closer to the flat wave operator, and that commuting with modified Lie
derivatives preserves the geometric null structure, improving the error terms.
This improved decay of the metric is essential for proving the estimates in \cite{Ka18},
and will likely be useful in other contexts as well.

\end{abstract}

\maketitle

\vspace{-0.3in}

\tableofcontents

\vspace{-0.5in}

\section{Introduction}

 Einstein's equations state that the Ricci curvature of the space time metric  satisfies
\begin{equation}
R_{\mu\nu}\!=\widecheck{T}_{\mu\nu}, \qquad \text{where}\quad \widecheck{T}_{\mu\nu}=T_{\mu\nu}-g_{\mu\nu} g^{\alpha\beta} T_{\alpha\beta}/2,
\end{equation}
and $T_{\mu\nu}$ is the energy momentum tensor of matter that satisfies the divergence free condition $$\nabla^{{}_{\,}\mu} T_{\mu\nu}=0.$$
Einstein's equations in harmonic or wave coordinates
are a system of nonlinear wave equations
\begin{equation}\label{eq:EinsteinWave0}
\Box^{\,g}  g_{\mu\nu} =F_{\mu\nu} (g) [\pa g, \pa g]+\widecheck{T}_{\mu\nu},
\qquad\text{where}\quad
\Box^{\,g}=
g^{\alpha\beta}\pa_\alpha\pa_\beta,
\end{equation}
is the reduced wave operator
for a Lorentzian metric $g_{\alpha\beta}$, that satisfy the wave coordinate condition
\begin{equation}\label{eq:WaveCordinateCond}
g^{\alpha\beta}\Gamma^{\gamma}_{\alpha\beta}
=|g|^{-1/2}\partial_\alpha (|g|^{1/2} g^{\alpha\gamma})=0,
\qquad \text{where} \quad  |g|=|\det{(g)}|,
\end{equation}
that is preserved under the flow of  \eqref{eq:EinsteinWave0}.
Here  $F(g)[\pa g,\pa g]$ is quadratic form in $\pa g$ with coefficients depending on $g$.
Choquet-Bruhat \!\!\cite{ChBr} proved local existence in these coordinates.

 Christodoulou-Klainerman  \!\!\cite{CK93} proved global existence for Einstein vacuum equations $R_{\mu\nu}\!=\widecheck{T}_{\mu\nu}=\!0$ for asymptotically flat
 initial
data:
\vspace{-0.05in}
\begin{equation}\label{eq:asymptoticallyflatdata}
g_{ij}|_{t=0}=(1+M r^{-1})\, \delta_{ij} + o(r^{-1-\gamma}),\quad
\pa_t g_{ij}|_{t=0}=o(r^{-2-\gamma}),\quad r\!=|x|,
\end{equation}
with $\gamma\!>\!1\!/2$,
that are small perturbations of the Minkowski metric $ m_{\alpha\beta}$. Here $M\!>\!0$, by the positive mass theorem.
\cite{CK93}  avoids using coordinates and instead uses equations for the full curvature tensor $R_{\alpha\beta\gamma\delta}$ since it
 was believed the metric in harmonic coordinates would blow up for large times, because this is
 true for wave equations with quadratic nonlinearities without any extra structure.

 However, Lindblad \cite{L92} and Lindblad-Rodnianski \cite{LR03} identified a weak null structure in Einstein's equations and formulated a weak null condition under which they expected that systems of nonlinear wave equations would have global solutions. Then in Lindblad-Rodnianski \cite{LR05,LR10} they proved global existence of solutions to Einstein's equations in wave coordinates for small asymptotically flat initial data with $\gamma>0$ in the case of the energy momentum tensor $T$ of a scalar field.

 The proof in \cite{CK93} involved proving strong decay of various components of the curvature tensor that may not hold in the presence of matter, whereas the proof in \cite{LR05,LR10} only relies on weaker decay for most components of the metric.
Moreover the proof in \cite{CK93} uses null coordinate and vector fields
adapted to the outgoing curved light cones or characteristic surfaces $u=const$, where $u$ solves the eikonal equation
\begin{equation}
g^{\alpha\beta} \pa_\alpha u \,\pa_\beta u=0,
\label{eq:eikonalintro}
\end{equation}
whereas \cite{LR05,LR10} only use vector fields and coordinates adapted to the background Minkowski space.

Due to the simpler and more perturbative approach in
 \cite{LR05,LR10} it was followed by global existence for Einstein's equations in wave coordinates coupled to the energy momentum tensor of Maxwell's equations \cite{Lo}, of Vlasov matter \cite{LT17,FaJoSm} and of Klein-Gordon\cite{IP17,LeMa}.

 Here we will prove global existence for
 Einstein's equations coupled to Maxwell-Klein-Gordon. This however requires more precise asymptotics of the metric. There are strong estimates of Morawetz type for Maxwell-Klein-Gordon on Minkowski background see Lindblad-Sterbenz \cite{LS}, but these estimates are not stable under perturbations as large as the difference between the metric and the Minkowski metric. However, they hold for the Schwarzschild metric.

  In was shown in  \cite{LR10} that solutions with asymptotically flat data approaches the Schwarzschild metric with the same ADM mass, as $r\!\to\!\infty$.
  This was achieved by subtracting off the leading order term in the Schwarzschild metric $m_{\alpha\beta}\!+\delta_{\alpha\beta} M\!/r$ from the solution of the wave equation.
It was further shown in Lindblad \cite{L17}
that the outgoing light cones for the metric converge to the outgoing light cones for the Schwarzschild metric, $u^*\!=t-r^*\!=const$, where $r^*\!=r+M\ln{r}$, i.e. there  is a solution $u$ to \eqref{eq:eikonalintro} such that $u\sim u^*$ at null infinity. Moreover changing to the coordinates $(\wt{t}, \wt{x})$, with $\wt{t} = t$ and $\wt{x} = r^*x{}_{\!}/r$, the precise asymptotics for the metric was given and
  a set of vector fields were constructed analogous to the commutation fields in Minkowski space,  depending only on the ADM mass of initial data.
  If $g$ satisfy the wave coordinate condition then the metric expressed in the new coordinates $\widetilde{g}$ will satisfy a generalized wave coordinate condition,
  see \eqref{eq:wavetildetowavecoord}.

In these coordinates Kauffman \cite{Ka18} proved that the same kind of Morawetz estimates as in \cite{LS} hold for MKG on the background of a metric that is asymptotic to Schwarzschild at null infinity. The purpose of this paper is to couple MKG to Einstein using \cite{Ka18} to estimate the MKG fields.

In order to deal with more difficult matter fields we prove global existence for Einstein's equations  directly in the new coordinates by composing the wave coordinates with the change of variables above. The change of variables can either be thought of as the new metric satisfying a generalized
wave coordinate condition,
\beq\label{def:WCCTrueIntro}
\widetilde{g}^{\beta\gamma}\widetilde{\Gamma}^\alpha_{\beta\gamma} = \Wcal^\alpha(\widetilde{g})
\eq
for some function $\Wcal^\alpha$ as in Appendix B, cf. \cite{Hu,HV20},  or the new metric satisfying the equation obtained by replacing derivatives with covariant derivatives with respect to the change of variables. This will introduce new error terms but they are fast decaying and easily estimated. However, the estimates of the quasilinear part simplifies substantially because the commutators of the vector fields with the wave operator in these coordinates have an improved null structure. 
Moreover, the estimates of the semilinear terms are also simplified substantially by using Lie derivatives with respect to the vector fields instead of vector fields, which as observed in \cite{L17} preserve the geometric weak null structure of the semilinear terms. As a byproduct of our proof we hence have an improved proof of the vacuum case as well, that we think will be useful when dealing with more difficult matter fields. We will first formulate the result for a general matter field satisfying certain assumptions and then afterward we formulate the theorem in the coupled case specifically for MKG.

\subsection{Einstein's equations in generalized wave coordinates}
 Since $m$ is constant  $h\!=\!g\shortminus m$ satisfies
\begin{equation}\label{eq:EinsteinWaveintro}
\Box^{\,g}  h_{\mu\nu}\! =F_{\mu\nu} (g) [\pa h, \pa h]+\!\widecheck{T}_{\!\mu\nu},
\quad\text{where}\quad
F_{\mu\nu} (g) [\pa h, \pa h]\!=P(g)[\pa_\mu h, \pa_\nu h]+Q_{\mu\nu}(g)[\pa h, \pa h],
\end{equation}
where $Q$ is a sum of classical null forms with two  metric contractions and
$P$ has a weak null structure
\beq\label{eq:curvedPdef}
P(g)[k,p_{\,}]=g^{\alpha\beta} k_{\alpha\beta}  g^{\gamma\delta} p_{\gamma\delta}/4
-g^{\alpha\beta} g^{\gamma\delta}k_{\alpha\gamma}  p_{\beta\delta}/2,
\eq
see Section \ref{sec:inhom}.
Here  $F(\cdots)[u_1,u_2]$ stands for functions that are separately linear in  $u_{{}_{\!1}}$ and $u_2$.
With
$H^{\alpha\beta}\!=g^{\alpha\beta}\!-m^{\alpha\beta}\!
=-h^{\alpha\beta}\!+O^{\alpha\beta}(g)[h,h]$,
where $h^{\alpha\beta}\!=m^{\alpha\delta}m^{\beta\gamma} h_{\delta\gamma}$ the wave coordinate condition \eqref{eq:WaveCordinateCond} becomes
\begin{equation}\label{eq:approximatewavecordinateconditionintro}
\pa_\mu\big(H^{\mu\nu}\! - m^{\mu\nu} m_{\alpha\beta} H^{\alpha\beta}\!/2\big)
=W^\nu(g)[H,\pa H]
:= (g^{\mu\nu}g_{\alpha\beta} - m^{\mu\nu}m_{\alpha\beta} )\partial_\mu H^{\alpha\beta}/2 ,
\end{equation}
where the right hand side is decaying faster since it is quadratic.

\subsubsection{The weak null structure, vector fields and Lie derivatives} To see the weak null structure we
introduce a {\it null frame} ${\mathcal N}$ of vectors
tangential to the outgoing light cones plus a vector perpendicular to the cone:
\begin{equation} 
	\underline{L}=\pa_t-\pa_r,
	\quad
	L=\pa_t+\pa_r,
	\quad
	S_1,S_2\in\mathbf{S}^2,
	\quad
	\langle S_i,S_j\rangle =\delta_{ij},
\end{equation}
It is well known
that, for solutions of wave equations, derivatives tangential to the outgoing light
cones $\overline{\pa }\in\mathcal{T}=\{L,S_1,S_2\}$ decay faster.
Since $Q_{\mu\nu}\!=Q_{\mu\nu}(m)$ satisfy the classical null condition
\beq \label{eq:nullcondestintro}
|Q_{\mu\nu}(\pa h,\pa k)|\les |\overline{\pa} h|\,|\pa k|+|\pa h|\, |\overline{\pa} k|.
\eq
Moreover, projecting derivatives onto the frame
\beq
	|\pa_\mu \phi-L_\mu \pa_q \phi|\lesssim |\overline{\pa} \phi|,
\qquad\text{where}\quad \pa_q=(\pa_r-\pa_t)/2,
	\quad L_\mu=m_{\mu\nu} L^\nu,
\eq
we see that the main term $P=P(m)$ has the following weak null structure:
\beq\label{eq:PPnullintro}
\big| P(\pa_\mu h,\pa_\nu k)- L_{\mu}L_\nu P(\pa_q h,\pa_q k)\big|\les |\overline{\pa} h| \,|\pa k|
+|\pa h|\,|\overline{\pa} k|.
\eq
Furthermore, from the wave coordinate condition
\beq\label{eq:WCCIntro}
|L_\nu L_\mu \pa_q H^{\mu\nu}|\lesssim |\overline{\pa} H|+|H|\,|\pa H|,
\eq
and expanding $P$ in a null frame using this it follows that
\beq\label{eq:Pestintro}
|P(\pa_q h,\pa_q h)|\lesssim {\sum}_{T\in\mathcal{T},U\in\mathcal{N}}\,|\pa h_{TU}|^2+ |\overline{\pa} h|\, |\pa h| + |h|\, |\pa h|^2.
\eq
The weak null structure is that the tangential components
of the right hand side of  \eqref{eq:EinsteinWaveintro} are decaying faster, due to \eqref{eq:nullcondestintro}-\eqref{eq:PPnullintro}, and the remaining component $F_{\underline{L}\underline{L}}\!\!=\!F_{\mu\nu} \underline{L}^\mu\underline{L}^\nu$ to highest order does not depend on that component of the solution $\pa_q h_{\underline{L}\underline{L}}$, due to \eqref{eq:Pestintro}.
To highest order there is no mutual interaction between the components, similar to the structure of  $\Box \phi\!=\!(\pa_t\psi_{\!})^2\!$, $\Box\psi\!=\!0$.

In order to estimate the solution we also need higher order equations for vector fields applied to it.
The vector fields ${\Omega}_{ab}\!={x}_a\pa_{b}-{x}_b\pa_a$, where ${x}_a\!=m_{ab}{x}^b\!\!$,
commute with ${\Box}$ and ${S}\!={x}^a\pa_a$ satisfy $[{\Box},{}_{\!}{S}]\!=\!-2{\Box}$.
Applying vector fields to $F_{\mu\nu}(m)[\pa h,\pa k]$ will
produce lower order terms that no longer have the weak null structure, but, as observed in \cite{L17}, if we instead apply Lie derivatives the weak null structure is preserved:
\begin{equation}
 {\mathcal L}_Z\,\big(  F_{\mu\nu}(m)[\pa h,\pa k]\big)
 =F_{\mu\nu}(m)[\pa \widehat{\mathcal L}_Z h,k]
 +F_{\mu\nu}(m)[\pa h,\pa \widehat{\mathcal L}_Z k].
\end{equation}

\subsubsection{Reduction to Schwarzschild at infinity}
We will make two reductions that will put us closer to  Schwarzschild  in the exterior and at null infinity. First we write
$h_{\mu\nu}\!=h^0_{\mu\nu}\!+h^1_{\mu\nu}$ and $H^{\mu\nu}\!\!=H_0^{\mu\nu}\!\!+H_1^{\mu\nu}\!\!$,
where
\beq
h^0_{\mu\nu}=\chi\big(\tfrac{r}{t+1}\big)Mr^{-1}\delta_{\mu\nu},
\quad\text{and}\quad
H_0^{\mu\nu}=-\chi\big(\tfrac{r}{t+1}\big)Mr^{-1}\delta^{\mu\nu},
\eq
and $\chi$ is a cutoff function such that $\chi(s)\!=\!1$ for
$s\!\geq\! 3/4 $ and $\chi(s)\!=\!0$ for $s\!\leq\! 1\!/4$. Then we change variables
\beq
\widetilde{t}=t,\qquad \widetilde{x}=\rs x/r, \qquad \text{where} \quad
\rs=r+{\chi}(\tfrac{r}{t+1}) M\ln{r},
\eq
and with  $x^0=t$, $\widetilde{x}^0=\widetilde{t}$, set
\beq
\widetilde{g}^{ab}=A^a_{\alpha} A^b_{\beta} g^{\alpha\beta}\!\!\!,\qquad
\widetilde{g}_{ab}=A^{\alpha}_{a}A^{\beta}_{b} g_{\alpha\beta},\qquad \wpa_a=A^{\alpha}_{a}\pa_\alpha,\quad\text{where}\quad
A^{a}_{\alpha}\!={\partial \widetilde{x}^{\,a}}/{\partial x^\alpha}\!,
\quad A^{\alpha}_{a}\!={\partial x^\alpha}/{\partial \widetilde{x}^{\,a}}\!.
\eq
Then
$\widetilde{h}^1_{ab}=A^{\alpha}_{a}A^{\beta}_{b} h^1_{\alpha\beta}$
satisfies
\begin{equation}\label{eq:EinsteinGenWaveintro}
\widetilde{\Box}^{\,\widetilde{g}}  \widetilde{h}^1_{cd} =F_{cd} (\widetilde{g}) [\wpa\widetilde{h}^1, \wpa \widetilde{h}^1]+\widetilde{\widecheck{T}}_{cd}+R_{cd}^{\,error},\qquad\text{where}
\quad \widetilde{\Box}^{\,\widetilde{g}}=\widetilde{g}^{ab}\wpa_a\wpa_b,
\end{equation}
and $\widetilde{H}_1^{ab}=A^a_{\alpha} A^b_{\beta} H_1^{\alpha\beta}$ satisfies
\begin{equation}\label{eq:approximategenwavecordinatecponditionintro}
\wpa_c\big(\widetilde{H}_1^{cd}\! - m^{cd} m_{ab}\widetilde{H}_1^{ab}\!/2\big)
=W^d(\widetilde{g})[\widetilde{H}_1,\wpa\widetilde{H}_1]+W^d_{error}.
\end{equation}
The error terms $R_{cd}^{\,error}$ and $W^d_{error}$ include terms coming from the mass error and from the change in coordinates, and generally decay faster. One way to see the structure of these terms is by treating \eqref{eq:EinsteinWaveintro} as a geometric quasilinear equation on the Minkowski spacetime and changing to the modified coordinates. This introduces error terms coming from covariant differentiation with respect to the new coordinates, which may therefore be expressed in terms of the Christoffel symbols $\widehat{\Gamma}$ for the new coordinates. These decay like $r^{-2}\ln r$ and are supported away from the spatial origin, and therefore the error terms which arise are straightforward to handle. We use this framework as it makes the weak null structure of Einstein's equations clear; alternatively, one could derive Einstein's equations directly from the  generalized wave coordinate condition \eqref{def:WCCTrueIntro}, see Appendices \ref{sec:GeneralizedWaveCurvature} and \ref{sec:AppB}.

The choice of coordinates ensures that the initial data for key components of $\wt{g}$ decay faster than $r^{-1}$ in a suitable null decomposition. The harmonic coordinate condition \eqref{eq:approximategenwavecordinatecponditionintro}, via the relation \eqref{eq:WCCIntro}, propagates this improved decay along incoming null characteristics. This improvement is necessary if one wants to uniformly bound the energy of the fields, as the analogous estimate appearing in \cite{LR10} gives slow energy growth even for solutions of the homogeneous geometric wave equation. The cost is that tangential components of $\wt{g}$ will decay slower than those for $g$, but due to the null structure these terms are paired with terms with faster decay. This improved null structure is preserved by Lie differentiation, see Section \ref{sec:higherorderwavecoord}.

\subsubsection{Higher order energies}
Let $\widetilde{Z}$ be any of the vector fields $\widetilde{\Omega}_{ab}$, $\widetilde{S}$, and $\wpa$, and for a multindex $I$ let  $\widetilde{Z}^I$ respectively $\Lie_{\widetilde{Z}}^I $ denote any combination of $|I|$ vector fields respectively Lie derivatives. We will be using energies with an additional exterior weight
\begin{equation} \label{eq:energydefEintro}
	E_N(t)={\sum}_{\vert I \vert \leq N\,\,}
	\int
	|\wpa\Lie_{\widetilde{Z}}^I \widetilde{h}^1 (t,x)|^2 w\, dx, \quad
	\text{where}
	\quad
	w(t,x)
	=\bigg\{
	\begin{aligned}
		&(1+|r^*\!\!-t|)^{1+2\gamma},\,\,\,
		&
		r^*\!>t,
		\\
		&\sim 1
		\quad
		&
		r^*\!\leq t,
	\end{aligned}
\end{equation}
for some $0<\gamma<1$.
These energies remains bounded for the linear homogeneous wave equations
$\widetilde{\Box}^{\widetilde{g}} \phi=0$, and the only small growth is caused by that the inhomogeneous term $F_{\mu\nu}$ in \eqref{eq:EinsteinGenWaveintro} doesn't satisfy the classical null condition. The extra weight in the exterior is meant to catch the extra decay of $h^1$ in the exterior, which in turn using the wave coordinate condition \eqref{eq:approximategenwavecordinatecponditionintro} gives
additional decay for the critical components controlling the geometry of the light cones.

Together with the energy estimate comes for free a
space time norm
\begin{equation}
	S_N(T)
	\!=
	\!\!{\sum}_{\vert I \vert \leq N\,\,}\int_0^T\!\!\!\int 	
	|\wpao
\Lie_{\widetilde{Z}}^I \wt{h}^{1} (t,x)|^2 w'dx dt,
	\quad
	\text{where}
	\quad
w^\prime(t,x)\!=\bigg\{\begin{aligned} &(1\!+2\gamma)(1\!+|r^*\!\!-t|)^{2\gamma}\!\!,\quad\!\! &r^*\!>t,\\
&2\mu (1\!+|r^*\!\!-t|)^{-1-2\mu}\!,\quad &r^*\!\leq t,\end{aligned}
\end{equation}
where $|\wpao \wt{h}|^2=|\swpa \wt{h}|^2 \!+|\wL \wt{h}|^2$ is the norm of the derivatives tangential to the outgoing curved light cones
$r^*\!\!-t=q^*$, where $|\swpa \wt{h}|^2\!=\!\sum_{\,i=1,2}|S_i \wt{h}|^2$ is the norm of the derivatives tangential to the sphere.
\subsubsection{Conditions on matter} The energy momentum tensor has to satisfy some smallness conditions and be compatible with the weak null structure for the metric. In order for it to be compatible with the energy estimate for the metric we assume that for a sufficiently small $\varepsilon$, the following  hold:
\begin{equation}\label{eq:L2Tboundintro}
\lVert\Lie_{\widetilde{Z}}^I\widecheck{T}(t, \cdot) w^{1/2}\rVert_{L^2(\mathbb{R}^3)} \lesssim {\varepsilon^2}{(1+t)^{-1}},\qquad\text{for}\quad |I|\leq N.
\end{equation}
It is also natural to assume that with $q=r-t$ and $q_+=q$ for $q > 0$ or $q_+=0$ for $q < 0$, we have
\beq \label{eq:L1Tboundintro}
 \|\Lie_{\widetilde{Z}}^J\widecheck{T}(t, \cdot) \langle q_+\rangle^\gamma \|_{L^1}
\lesssim \varepsilon^2 ,\qquad |J|\leq N-3.
\eq
We also assume that the following decay estimate holds for sufficiently small $\varepsilon$
\beq
|\Lie_{\widetilde{Z}}^K\widecheck{T}(t,x)|\lesssim \varepsilon^2\langle t+r^* \rangle^{-2}\langle t-r^*\rangle^{-1}\langle (t-r^*)_+\rangle^{-\gamma},\qquad |K|\! \leq \!N\!-\!6.\label{eq:decayTintro}
\eq
In addition we must assume that $T$ is compatible with the weak null structure of Einstein's equation;
for some  $s>0$ and all $T\in\mathcal{T}$, $U\in{\mathcal{N}}$ we assume that
\beq
|\Lie_{\widetilde{Z}}^K\widecheck{T}(t,x)|_{TU}\lesssim \varepsilon^2\langle t +r^*\rangle^{-2-s}\langle t-r^*\rangle^{-1+s}\langle (t-r^*)_+\rangle^{-\gamma},\qquad |K|\! \leq \!N\!-\!6.\label{eq:decayweaknullTintro}
\eq

\subsubsection{The existence theorem} We are now ready to state our first result.
\begin{theorem}\label{thm:EExistence} Suppose that the energy momentum tensor $\widecheck{T}$ satisfies the conditions \eqref{eq:L2Tboundintro}, \eqref{eq:L1Tboundintro},
\eqref{eq:decayTintro} and \eqref{eq:decayweaknullTintro} in a spacetime decaying weakly to Minkowski for some $N\!\geq\! 11$ and $0\!<\!\gamma\!<\!1$. Then there is a $\varepsilon_0\!>\!0$ such that if
$\varepsilon\!<\!\varepsilon_0$ and
\beq
 E_N(0)+ M^2\leq \varepsilon^2,
\eq
then  \eqref{eq:EinsteinWaveintro} has a global solution and there is a constant $C$ such that for all $t>0$
\beq
E_N(t)+S_N(T)\leq C \varepsilon^2(1+t)^{C\varepsilon}.
\eq
Moreover, with $\overline{w}(q^*)=\langle q^*\rangle^{1-\delta}\langle q_+^*\rangle^{\!\gamma}\!$, for any $\delta>0$, we have for $|I|\leq N-7$
\beq
(1+t+|x|) \,|(\pa \widehat{\mathcal L}_{\widetilde{Z}}^I
h^1)_{TU}(t,x)\, \overline{w}(q^*)|\les \varepsilon ,
\eq
\beq
(1+t+|x|) \,|\pa \widehat{\mathcal L}_{\widetilde{Z}}^I
h^1(t,x)\, \overline{w}(q^*)|
\les \varepsilon \Big(1+ \ln\frac{t\!+\!r^*}{\langle t\!-\!r^*\rangle}\Big),
\eq
and
\beq
|\pa \widehat{\mathcal L}_Z^I{H}_1|_{L\mathcal T}
+|\pa \trs\widehat{\mathcal L}_Z^I {H}_1 |\les
\varepsilon(1+t+r)^{-2+2\delta}\langle q_+^*\rangle^{-\gamma}\langle q^*\rangle^{-\delta}.
\eq
\end{theorem}
\begin{remark}
It is in general impossible to fully decouple the analysis of the metric and the fields, as the conditions \eqref{eq:L2Tboundintro}, \eqref{eq:L1Tboundintro}, \eqref{eq:decayTintro} and \eqref{eq:decayweaknullTintro} will in general only hold when $g$ satisfies some weak energy and decay bounds. In Theorem \ref{thm:EFEgenericstability} we more precisely detail the nature of this coupling.
\end{remark}
\begin{remark}
Even in the case of vanishing matter fields these decay estimates are improvements of the decay estimates in \cite{LR10} and the proof here is simplification of the proof in \cite{LR10}. This is due to that the error terms are more efficiently controlled by using coordinates adapted to the outgoing Schwarzschild light cones instead of the Minkowski light cones, and that the geometric structure is better preserved by using Lie derivatives instead of vector fields.
\end{remark}

\subsection{Einstein with Maxwell-Klein-Gordon matter}
As an application of this general theorem we couple Einstein's equations to the massless Maxwell-Klein-Gordon system, which has been studied extensively in the Minkowski spacetime and for which the associated trace-reversed energy momentum tensor $\widecheck{T}$ satisfies the conditions \eqref{eq:L2Tboundintro}, \eqref{eq:L1Tboundintro}, \eqref{eq:decayTintro} and \eqref{eq:decayweaknullTintro} for small initial data. We first define the system and state stability results in perturbed spacetime, as proven in \cite{Ka18}.

\subsubsection{The electromagnetic field and complex scalar field}
Given a real connection one-form $\bA_\alpha$ and the consequent quantities ${\bf F}=d\bA, D_\alpha = \nabla_\alpha + i\bA_\alpha$, where $\nabla_\alpha$ is covariant differentiation, $\{\bphi,{\bf F}\}$ is a solution to the charge-scalar field system if
\begin{align}\label{eq:MKG}
\Box^{\mathbb{C}}_g \bphi &= D^\alpha D_\alpha\bphi = 0, \\
\label{eq:CSFField}\nabla^\beta {\bf F}_{\alpha\beta} &= \mathfrak{I}(\bphi \overline{D_\alpha\bphi}),\\
\label{eq:PotentialStatement}\nabla^\beta(^*{\bf F})_{\alpha\beta} &= 0,
\end{align}
where $\mathfrak{I}, \mathfrak{R}$ denote the imaginary and real parts of a quantity respectively. From a physical perspective it is useful to introduce the current
\[
J_\mu = \mathfrak{I}(\bphi \overline{D_\mu\bphi}).
\]
 The energy momentum tensor $T$ corresponding to a solution is given by
\begin{equation}\label{eq:energymomentumchargedscalarfield}
T_{\alpha\beta}[\bphi, {\bf F}] = \mathfrak{R}\bigtwo(D_\alpha\bphi \overline{D_\beta\bphi} - g_{\alpha\beta}D_\gamma\bphi \overline{D^\gamma\bphi}/2\bigtwo) + {\bf F}_{\alpha\gamma}{\bf F}_\beta^{\,\,\gamma} - \frac14 g_{\alpha\beta}{\bf F}_{\gamma\delta}{\bf F}^{\gamma\delta}.
\end{equation}

There is an ambiguity in the choice of $\bA$, called gauge freedom,
 which does not affect the system on a physical level,
 and which it is therefore not necessary to resolve. There is a choice of gauge, called Lorenz gauge, for which $\bA$ has a clear weak null structure and components
 decay like solutions of the corresponding asymptotic system, as proven by Candy-Kauffman-Lindblad \cite{CKL19}
and He \cite{H21}.

\begin{remark}
A characteristic difficulty for the Maxwell-Klein-Gordon system, even in Minkowski, is the poor decay of the right hand side of \eqref{eq:CSFField}. Heuristically, the presence of an undifferentiated $\bphi$ on the right hand side implies that the right hand side of \eqref{eq:CSFField} satisfies worse decay bounds in the interior than the right hand side of \eqref{eq:EinsteinWaveintro}. Additionally, in order to prove a spacetime energy bound for the current vector close to the light cone, we will need to take full advantage of the null structure of the system. In order to compensate for these poor decay rates, we will need to prove a global energy bound which allows for peeling estimates, in the spirit of a conformal Morawetz-type estimate. This difficulty was not present in \cite{LR10}, as the improved decay of the semilinear terms in the interior means that the argument closes even for slowly growing energy, but for MKG there is no room to spare.
\end{remark}

Another characteristic difficulty for $\bF$, which also has an analogue in the behavior of Einstein's Vacuum Equations, is the decay of $\bF$ in space:
\beq
\label{eq:fielddata0}
{\bf F}_{0i}(0,x) \sim C\omega_i|x|^{-2} + o(|x|^{-2-\gamma}),\qquad \pa \bphi = o(|x|^{-2-\gamma}).
\eq
The value of $C$ scales with the charge $\bold{q}[\bF]$, which depends on $\bphi$, and is in general not 0 even if $\bphi$ is compactly supported. This decay rate causes issues when attempting to establish an energy estimate, as the (fractional) conformal Morawetz estimate we use requires a decay rate of $O(r^{-2-\gamma})$ for some $\gamma>0$. In order to deal with this we construct a model field ${\bF}^0$ with the same charge as ${\bF}$, and run analysis on the remainder quantity ${\bF}^1 = {\bF}-{\bF}^0$, which decays faster as $r\to\infty$.
Then, for a smooth increasing function $\chi_{\bF}(x)$ which is identically 0 for $x\leq 0$ and 1 for $x\geq 1$, we define the 2-forms $\bF^0, \bF^1$ by
\begin{equation}
\bF^0_{0i} = \frac{\omega_iq}{4\pi}\frac{\chi_{\bF}(r^*-t-2)\partial_r(r^*)}{r^{*2}}, \qquad \bF^0_{ij} = \bF^0_{00} =0, \qquad \bF^1 = \bF - \bF^0.
\end{equation}

\subsubsection{The existence theorem for the Einstein-Maxwell-Klein-Gordon system}
For the fields $\{\bphi, \bF\},$ we define the weighted higher order energy
\begin{equation}
	Q_N(T)={\sup}_{0\leq t\leq T}{\sum}_{\vert I \vert \leq N\,\,}
	\int
	|\!D_{\widetilde{Z}}^I\widetilde{D} \bphi(t,x)|^2\!\!
+|\widetilde{Z}^{I} \bF^1 (t,x)|^2\,\,  w\, dx,
\end{equation}
where $\bF^1$ is the charge-free part of $\bF$, and $q$ is the charge of $\bF$, both of which are defined in Section \ref{sec:MKGIVPdefs}. As with $E_N(T)$, we are able to refine our energy bound by first subtracting off the term with the worst decay in space.
Additionally for a 2-form $\mathbf{G}$, we define the null decomposition
\begin{equation}
\alpha_i[\mathbf{G}] = \mathbf{G}[\wL, \widetilde{S}_i], \qquad \underline{\alpha}_i[\mathbf{G}] = \mathbf{G}[\wuL, \widetilde{S}_i], \qquad \sigma[\mathbf{G}] = \frac12\mathbf{G}[\widetilde{S}_1, \widetilde{S}_2], \qquad \rho[\mathbf{G}] = \frac12\mathbf{G}[\wuL, \wL],
\end{equation}
where $\{\wL, \wuL, \widetilde{S}_1. \widetilde{S}_2\}$ is the modified null frame defined in \eqref{def:ModifiedNullFrame}. By antisymmetry this fully characterizes $\mathbf{G}$.

\begin{theorem} \label{thm:EMKGIntro} Einstein's equations in harmonic coordinates \eqref{eq:EinsteinWave0} with energy momentum tensor \eqref{eq:energymomentumchargedscalarfield} coupled with the massless Maxwell-Klein-Gordon system
\eqref{eq:MKG}, \eqref{eq:CSFField}, \eqref{eq:PotentialStatement} have global solutions for sufficiently small and smooth asymptotically flat initial data. Specifically, given a real $\gamma$ satisfying $\frac12 < \gamma < 1$, and an integer $N$ with $N > 11$, there exists an $\varepsilon_0 > 0$ such that for all $\varepsilon < \varepsilon_0$, the coupled system has global solutions in time whenever
\beq
E_N(0) + Q_N(0) + M^2 + q^2\leq \varepsilon^2.
\eq
These solutions satisfy the global in time bounds
\beq
E_N(t)+S_N(t)\leq C_N\varepsilon^2(1\!+t)^{C_N^{\,\prime}\varepsilon}\quad\text{and}\quad  Q_N(t)\leq C_N\varepsilon^2.
\eq
Taking $s < (1+\gamma)/2$, then the bounds \eqref{eq:L2Tboundintro}--\eqref{eq:decayweaknullTintro} hold, and consequently the bounds on $h^1$ in Theorem \ref{thm:EExistence} hold as well. For $|I|\leq N-4$, the field quantities $\bphi$ and $\bF^1$ satisfy the bounds
\begin{subequations}
\begin{align}
\!\!\!\!\!\!|D_{\widetilde{L}} D_{\wt{Z}}^I\bphi| + |\alpha[\Lie_{\wt{Z}}^I \bF^1]|
&\leq C \varepsilon \langle t+r^*\rangle^{-2}\langle t-r^*\rangle^{1/2-s}
\langle (r^*-t)_+\rangle^{s-1/2-\gamma},\\
\!\!\!\!\!\!|\slashed{D}D_{\wt{Z}}^I \bphi| + |\rho[\Lie_{\wt{Z}}^I \bF^1]| + |\sigma[\Lie_{\wt{Z}}^I \bF^1]|&\leq C\varepsilon\langle t+r^*\rangle^{-1-s}\langle t-r^*\rangle^{-1/2}\langle (r^*-t)_+\rangle^{s-1/2-\gamma}, \\
\!\!\!\!\!\!\langle t-r^* \rangle^{-1} |D_{\wt{Z}}^I\bphi| + |D_{\widetilde{\underline{L}}}D_{\wt{Z}}^I \bphi| + |\underline{\alpha}[\Lie_{\wt{Z}}^I \bF^1]|&\leq C \varepsilon\langle t+r^*\rangle^{-1}\langle t-r^*\rangle^{-1/2-s}\langle (r^*-t)_+\rangle^{s-1/2-\gamma}.
\end{align}
\end{subequations}
\end{theorem}


\subsubsection{History of related results}
An early global existence result for the Maxwell-Klein-Gordon system with small initial data in the Lorenz gauge was given by Choquet-Bruhat-Christodoulou \cite{CBC}. However, this method required $H^s$ bounds for components of the potential $\bA$, which was incompatible with a nonzero charge, and furthermore, did not give any decay for solutions of the system. Eardley-Moncrief \cite{EM1}-\cite{EM2} showed global existence for general initial data bounded in $H^s$ in a gauge adapted to the light cone. These existence results were later refined by Klainerman-Machedon \cite{KlMa}, who improved the requirements on the the initial decay. However, these results did not give any insight on the rate of decay in time. Furthermore, the initial bounds were strongly dependent on the choice of gauge.

A study of asymptotic behavior of the system with small initial data was undertaken by Shu \cite{Shu} in the massless case, and in the massive case by Psarrelli \cite{Ps}. In each case, compact support of the initial data is assumed, modulo the presence of a charge. We restrict our discussion to the massless system, as the massive system is more closely modelled by the Klein-Gordon system, and so the treatment of asymptotics is substantially different. The approach in \cite{Shu}, though heuristic, provided a model for more complete global stability results and asymptotics for small initial data by Lindblad-Sterbenz \cite{LS} and Bieri-Miao-Shahshahani \cite{BMS}.

Asymptotics for large initial data were determined more recently by Yang \cite{Y16} and Yang-Yu \cite{YY19}. The first of these papers removed the assumption of smallness for the field $\bF$, and the second removed it for $\bphi$ as well, though in each case the fields were assumed to be in certain weighted Sobolev spaces. Results for the massive case with small initial data (not necessarily compactly supported) were shown for Klainerman-Wang-Yang \cite{K1} and for large $\bF$ by Fang-Wang-Yang \cite{FWY}.

\subsection{Structure of the paper and strategy of the proof}
In Section \ref{sec:Asymptotics} we give a heuristic outline of our approach, in the form of a detailed analysis of the expected asymptotic behavior of the system. This includes first-order corrections to account for the ADM mass of the metric and the total charge of the field, as well as the structure of the metric in the asymptotically Schwarzschild coordinates. Section \ref{sec:EEWaveCoordinatesStructure} gives an outline of the null structure of Einstein's equations. In Section \ref{sec:masssubtractionsection}, we give precise definitions for the ADM mass and charge terms which we subtract off. In Section \ref{sec:Genralizedwavecoord}, we precisely state the change of coordinates that we use, and give the structure of the metric in these coordinates.

In Section \ref{sec:Genralizedwavecoordsection} we derive a more precise form for Einstein's equations and for the wave coordinate condition in asymptotically Schwarzschild coordinates, and determine the structure of the first-order error terms which arise from the change of coordinates. In Section \ref{sec:higherordereinstein} we commute vector fields through the wave equation and wave coordinate condition, to determine higher order equations with the same geometric structure. Sections \ref{sec:EnergyEstimates} and \ref{sec:DecayEstimates} respectively contain the weighted energy and decay estimates for the wave equation, which will be applied to Einstein's equations. In Section \ref{sec:ChrisResults} we restate the main result of Kauffman \cite{Ka18}, which includes existence results for the Maxwell-Klein-Gordon system in a curved spacetime assuming weak energy and decay bounds for the metric.

Sections \ref{sec:StatementandBootstrap}-\ref{sec:ProofMetricEnergy} contain the full proof of Theorem \ref{thm:EMKGIntro}. Additionally, in Section \ref{sec:TheoremGenericProof} we show how this may be extended to the ``black box" result given in Theorem \ref{thm:EExistence} by proving Proposition \ref{prop:theL1bound} and the results of Section \ref{sec:sharpdecayfields} for a given system of field equations. We outline the structure of the proof here:
\begin{enumerate}
\item In Section \ref{sec:StatementandBootstrap} we state the precise version of the theorems and give the overall bootstrap argument that will rely on estimates from the following sections. We start by assuming weak energy bounds that we will improve
    by the end of the argument closing the proof.
\item In Section \ref{sec:ProofMetricDecay} we prove weak decay estimates for the metric which follow almost directly from the assumed energy bounds and the structure of the wave equation.
\item In Section \ref{sec:wavecoordbounds} we prove improved decay estimates and $L^2$  bounds for components of the metric that can be controlled by the wave coordinate condition.
\item In Section \ref{sec:sharpdecayfields} we use the estimates obtained so far for the metric to deduce sharp decay and sharp energy bounds for the fields from the theorem stated in Section \ref{sec:ChrisResults}. The improved energy bounds close the bootstrap argument for the fields.
\item In Section \ref{sec:metricdecaysharp} we use the decay estimates for the metric and fields in the wave equation for the metric to obtain the sharp decay estimates for the metric.
\item In Section \ref{sec:ProofMetricEnergy} we use the sharp decay estimates for the metric and the energy bounds for the fields to finally get back the sharp energy bounds, which closes the bootstrap argument for the metric.
\end{enumerate}

In Appendix \ref{sec:GeneralizedWaveCurvature} we derive an expression for the Ricci curvature for a more generalized wave coordinate condition, which we use in Appendix \ref{sec:AppB} to derive Einstein's equations for generalized wave coordinates. In Appendix \ref{sec:AppC} we outline a method to improve our approximation of the asymptotic behavior of $h$ in the interior.

\section{Asymptotics of  different types
of  terms and error terms}\label{sec:Asymptotics}
\subsection{Asymptotics} In this section we will use asymptotics for solutions of wave equations
to give a heuristic argument for why certain nonlinear effects
are under control.

\subsubsection{General nonlinear wave equations and the weak null condition}
\label{sec:notnull} General wave equations with quadratic nonlinearities
may blow up for large times even for small data
as shown in John \cite{J1,J2} for
 $\Box\,\phi\!=\phi_t^2
 $ or $\Box \,\phi\!=\phi_t
 \triangle_x\phi$.
The {\it null condition} is an algebraic condition on the structure of
the nonlinear terms
that  guarantees small data global existence, e.g.  $\Box
 \,\phi\!=\!\phi_t^2\!-|\nabla_x\phi|^2\!$, see Christodoulou \cite{C86} and
Klainerman \cite{K1}.
The null condition is however not satisfied for  Einstein's equations
in wave coordinates.
 Neither is it satisfied for the
 quasilinear equation
$\label{eq:hormander4}
\Box\, \phi= c^{\,\alpha\beta} \phi\,\,\pa_\alpha\pa_\beta \phi,
$
that resembles the quasilinear terms in Einstein's equations, yet
global existence was proven in Lindblad \cite{L92,L08} and Alinhac \cite{A03} for these equations.
A simple semilinear system that violates the null condition yet trivially
has global solutions is
\beq\label{eq:semilinearmodel}
\Box \,\psi\!=(\pa_t \phi)^2\!,\qquad\text{where}\qquad
\Box\,\phi\!=0.
\eq
The semilinear terms in Einstein's equations resemble this system, see Section \ref{sec:inhom}.
Based on this Lindblad-Rodnianski \cite{LR03} introduced
the {\it weak null condition} and showed that
it was satisfied for Einstein's equations.

\subsubsection{Asymptotics, decay along the light cone and total decay}\label{sec:asymptotics}
 A solution of a linear homogeneous wave equation $\Box\,\phi=0$
 with smooth  initial data decaying like $r^{-1}$,
 decays like $t^{-1}$ and has a radiation field
\begin{equation}\label{eq:radiationfield}
\phi(t,x)\sim \mathcal{F}(r-t,\omega)/r,
\qquad\text{where}\quad | \mathcal{F}(q,\omega)|+\langle q\rangle | \mathcal{F}_q(q,\omega)|\lesssim 1,\qquad
\langle q\rangle= 1+|q|.
\end{equation}
For spherically symmetric compactly supported data this is exact for large $r$.
The same is true if
\beq\label{eq:decayforradiation}
|\Box \,\phi_{{}_{\,}} |+r^{-2}|{\triangle}_\omega \phi|\lesssim
r^{-1} \langle t+r\rangle^{-1-\eps}\langle t-r\rangle^{-1+\eps}\langle (r-t)_+\rangle^{-\eps},\qquad
\varepsilon>0,
\eq
and data decays like $r^{-1}$. This is seen by expressing the wave operator in spherical
coordinates:
\beq\label{eq:sphericalwaveoperator}
\Box\, \phi =-r^{-1}(\pa_t+\pa_r)(\pa_t-\pa_r)(r\phi)+r^{-2}{\triangle}_\omega \phi,
\eq
 and integrating, in the $t\!+_{\!}r$ direction and in the $t\!-\!r$ direction, to obtain
a bound for $r\phi$ and the asymptotics \eqref{eq:radiationfield}. The decay in
\eqref{eq:decayforradiation} separates into decay along the light cone in the $t\!+_{\!}r$ direction and decay
away from the light cone in the $t\!-\!r$ direction,  the sum of the two we call the total decay or
homogeneity. The total decay of the solution $\phi$ is one and since differentiating twice
will decrease the homogeneity by two this is consistent with a total decay of three
for \eqref{eq:decayforradiation}. Note also how the decay in $t\!-\!r$ helped but only up to
integrable decay. In the nonradial case when ${\triangle}_\omega \phi\neq 0$ the  integration along
characteristics can be replaced by the energy integral method in which case the bound for
the tangential Laplacian  is not needed.

\subsubsection{Sources along light cones}\label{sec:sources}
However, general quadratic inhomogeneous terms as in \eqref{eq:semilinearmodel} do not decay enough
along the light cone for \eqref{eq:decayforradiation} to hold.
In fact, by \eqref{eq:radiationfield}
\beq\label{eq:inhomogeneitywithradiationfield}
\phi_t(t,x)^2\sim \mathcal{F}_q(r-t,\omega)^2\!/r^2 ,\qquad
\text{where}\quad
| \mathcal{F}_q(q,\omega)|\lesssim  \langle q\rangle^{-1}.
\eq
The asymptotics for the wave equation with such sources  was studied in Lindblad \cite{L90}:
\begin{equation} \label{eq:source}
-\Box \,\psi=n(r-t,\omega)/r^2.
\end{equation}
The solution
is asymptotically given by a formula
which leads to a log correction in the asymptotics:
\beq \label{eq:approximatesource}
\psi(t,r\omega)
\sim \int_{r-t}^\infty \frac{1}{2r}\ln{\Big|\frac{t+q+r}{t+q-r}\Big|} \, n(q,\omega)\, dq
\sim \ln{|r|}\,\mathcal{F}_{\log}(r-t,\omega) /r+\mathcal{F}(r-t,\omega)/r,\qquad\text{when}\quad r\sim t,
\eq
In fact applying the wave operator \eqref{eq:sphericalwaveoperator} to the first expression in
\eqref{eq:approximatesource}
  shows that \eqref{eq:source} holds up to a nonspherical error bounded by
\eqref{eq:decayforradiation} which hence correspond to a solution with a radiation field without a logarithm, see Proposition 7.2 in \cite{L17},

Only certain combinations of quadratic terms
that satisfy the condition decay better, e.g.
\beq\label{eq:inhomogeneitynullcondwithradiationfield}
\phi_t(t,x)^2-|\nabla_x \phi|^2\sim \mathcal{F}_q(r-t,\omega)^2\!/r^2 -\mathcal{F}_q(r-t,\omega)^2\!/r^2
\sim r^{-3}\langle t-r\rangle^{-1},
\eq
but semilinear terms in Einstein's equations 
 behave like
\eqref{eq:semilinearmodel} and hence will produce logarithms.

\subsubsection{The asymptotics in the wave coordinate condition}
Assuming that we have asymptotics for $H^{\mu\nu}\!$
of the form $\mathcal{H}^{\mu\nu}(t\!-\!r,\omega)/r$ plus a similar term multiplied with a log,
in order for this to
be compatible with the wave coordinate condition
\eqref{eq:approximatewavecordinateconditionintro} we must have the following relation:
\beq
L_\mu\pa_q \big(\mathcal{H}^{\mu\nu}
 -\tfrac{1}{2} m^{\mu\nu} m_{\alpha\beta} \mathcal{H}^{\alpha\beta}\big)=0,
\eq
where $L_0=-1$ and $L_i=\omega_i$, for $i=1,2,3$. In particular
$L_\mu L_\nu \pa_q \mathcal{H}^{\mu\nu} =0$.
Since by the asymptotic flatness condition \eqref{eq:asymptoticallyflatdata}
 $ \mathcal{H}^{\mu\nu}L_\mu L_\nu=-2M$ and
 $ \mathcal{H}_{\log}^{\mu\nu}L_\mu L_\nu=0$ when $t=0$ it follows that
 \beq\label{eq:HLLRadiation}
 \mathcal{H}^{\mu\nu}L_\mu L_\nu  =-2M.
 \eq

\subsubsection{The principal quasilinear part of the wave operator}
If we plug in the expansion into the wave operator we get that, up to lower order terms involving derivatives tangential to the outgoing light cone that decay better,
\beq\label{eq:HLLRadiationWaveOperator}
\Box^{\,g}-\Box=H^{\alpha\beta}\pa_\alpha\pa_\beta
\sim -  4^{-1} r^{-1}\mathcal{H}^{\alpha\beta}L_\alpha L _\beta \pa_q^2
=   2^{-1}r^{-1}M \pa_q^2.
\eq

\subsection{Inhomogeneous error terms}
By Section \ref{sec:sources}  general quadratic inhomogeneous terms, e.g \eqref{eq:inhomogeneitywithradiationfield},
just fail to decay enough
whereas by  Section \ref{sec:asymptotics} terms with any additional decay
along the light cone, e.g. \eqref{eq:inhomogeneitynullcondwithradiationfield},
will not qualitatively affect the asymptotic behaviour of the solution.
Apart from possible logarithmic factors we expect the solution to decay
like linear wave $\phi\sim t^{-1}$ and
$\pa\phi\sim t^{-1}\langle t-r\rangle^{-1}\!$ and hence we expect the
quadratic inhomogeneous terms to decay like $\sim t^{-2}\langle t-r\rangle^{-2}$. In fact, a closer analysis shows that the components of the metric with additional logarithmic factors are not present in the quadratic semilinear nonlinearity for Einstein's equations since the system satisfies the {\it weak null condition}.
Anything with additional decay is to be considered a negligible error term.
The first type of
error term has additional decay everywhere $\sim t^{-2-a}\langle t-r\rangle^{-2}$, $a>0$
and the second type has additional decay along the light cone but the same total decay
$\sim t^{-2-b}\langle t-r\rangle^{-2+b}$, $b>0$.

\subsubsection{Error terms with additional decay everywhere}\label{sec:cubicerrors}
  This is in particular true for cubic or higher order terms, which we will denote
  {\it cubic error terms}.
Moreover changing coordinates, see Section \ref{sec:coordchange}, will produce lower order terms
which additional decaying factors, which we will call
{\it covariant error terms}, see Section \ref{sec:covariant}.
 It is easy to estimate terms with
additional  decay everywhere.

\subsubsection{Error terms with additional decay along the light cone}\label{sec:lightconedecay}
There are two types of inhomogeneous quadratic terms that, although they decay faster along
the light cone, do not have faster total decay.  The first are terms that satisfy the classical
null condition e.g.
\eqref{eq:inhomogeneitynullcondwithradiationfield}.
The second are  terms produced when we subtract off a term to pick up the leading order
spatial decay from the mass in initial data, see Section \ref{sec:exteriorH0}.
Since these are proportional to the mass we will call these {\it mass error terms}.
Although by Section \ref{sec:asymptotics} solutions to the wave equation with
inhomogeneous terms with
additional decay along the light cone will not distort the asymptotics along the light cones,
we will also need additional exterior decay and for this we need the inhomogeneous terms
to decay additionally in the exterior, see Section \ref{exteriordecaymass}.

\subsubsection{Commutator errors with no additional decay}\label{sec:commutatorsflat} The most difficult quadratic error terms show up when one commutes
vector fields $Z$ with the reduced wave operator, in order to obtain wave equations also for product of
vector fields
applied to the solution $Z^{I \!}h_{\mu\nu}$. If we use the vector fields that commutes with
the Minkowski wave operator, or satisfy $[Z,\Box\,]\!=\!c_Z \Box$, then
 the commutator with the reduced wave
operator $\Box^{\,g}$  is roughly a product of vector fields applied to this wave
 operator plus
terms of the form
\beq\label{eq:commutatorsflatintro}
 (Z^J H^{\alpha\beta})\pa_\alpha\pa_\beta Z^K h_{\mu\nu}, \qquad
  \text{for}\quad |J|+|K|\leq |I|,\quad |J|\geq 1.
\eq
Such  terms are  problematic but due to
\eqref{eq:HLLRadiation}-\eqref{eq:HLLRadiationWaveOperator}
the leading behavior is determined by a fixed tensor and as a result
\eqref{eq:commutatorsflatintro}
  are to leading order linear.
  \eqref{eq:HLLRadiationWaveOperator} does however affect the asymptotic behaviour.
  To remedy this we introduce coordinates that asymptotically straighten out the light cones,
  see Section \ref{sec:changingcoordinates}.

\subsection{Subtracting off a term that picks up the mass contribution and
exterior decay}\label{exteriordecaymass}
The Schwarzschild metric in harmonic coordinates
is a solution of Einstein's vacuum equations, \eqref{eq:EinsteinWaveintro} with $\widecheck{T}=0$ also satisfying
the harmonic coordinate condition \eqref{eq:approximatewavecordinateconditionintro}.
Since this metric has an expansion in powers of $r^{-1}$ it follows that the first term in the expansion
$Mr^{-1}\delta_{\mu\nu}$ satisfies   \eqref{eq:EinsteinWaveintro}  with $\widecheck{T}\!=\!0$ up to terms of order $r^{-4}$
as well as \eqref{eq:approximatewavecordinateconditionintro} up to terms of order $r^{-3}\!\!$.
In fact this is the contribution of the next term in the expansion.
However, since we expect the solution to decay like $t^{-1}\!$ this approximation can only be valid for $r\!>\!t_{\!}/2$, say. Therefore we will multiply with a cutoff function $\wt\chi\big(\tfrac{r}{t+1}\big)$, where $\wt\chi(s)=1$ for
$s\geq 3/4 $ and $\wt\chi(s)=0$ for $s\leq 1/4$, and write
$h_{\mu\nu}\!=h^0_{\mu\nu}\!+h^1_{\mu\nu}$ and $H^{\mu\nu}\!\!=H_0^{\mu\nu}\!\!+H_1^{\mu\nu}\!\!$,
where
\beq
h^0_{\mu\nu}=\wt\chi\big(\tfrac{r}{t+1}\big)Mr^{-1}\delta_{\mu\nu},
\quad\text{and}\quad
H_0^{\mu\nu}=-\wt\chi\big(\tfrac{r}{t+1}\big)Mr^{-1}\delta^{\mu\nu}.
\eq
Since $H^{\alpha\beta} \!=\!( (m+h)^{\shortminus 1})^{\alpha\beta}\! - (m^{\shortminus 1})^{\alpha\beta}\! =\! - h^{\alpha\beta} \!+ O^{\alpha\beta}(g)[h,h]$, it follows that $H_{1}^{\alpha\beta}\! \!= \! \shortminus h^{1\alpha\beta} \!+O^{\alpha\beta}(g)[h,h]$.

\subsubsection{Asymptotics and additional exterior decay}
\label{sec:exteriorH0}
Solutions of linear homogeneous wave equations $\Box_{\,}\varphi\!_{\!}=\!0$
 with smooth initial  data decaying like ${}_{\!}r^{-1-_{\!}\gamma}\!$, $\gamma\!>\!0,$
  have radiation fields
\begin{equation}
\varphi(t,x)\sim {\mathcal{F}}(r-t,\omega)/r,
\qquad\text{where}\quad | {\mathcal{F}}(q,\omega)|+\langle q\rangle | {\mathcal{F}}_q(q,\omega)|
\lesssim \langle q_+\rangle^{-\gamma}\!,
\end{equation}
where $q_+=\max\{q,0\}$.
The same is true if only
\beq\label{eq:decayforradiationexterior}
|\Box \,\varphi |+r^{-2}|{\triangle}_\omega \varphi|\lesssim
r^{-1} \tplusr^{-1-\eps}\tminusr^{-1+\eps}\langle (r-t)_+\rangle^{-\gamma},\qquad
\gamma>0,\quad\eps>0.
\eq
as is seen by integrating along characteristics as in Section \ref{sec:asymptotics}.
When $\eps=0$ we do get the logarithm but multiplied by the additional exterior decay
$\langle (r\!-t)_+\rangle^{-\gamma}$
as is seen in  \eqref{eq:approximatesource} with
$|n(q,\omega)|\lesssim \langle q\rangle^{-1} \langle q_+\rangle^{-\gamma}\!$.

\subsubsection{Subtracting off a  term picking up the mass in the wave equation}\label{sec:subtracth0intro}
Since $r^{-1}$ is a fundamental solution of $\triangle$ one can also see directly
that $\Box\, r^{-1}=0$, for $r\neq 0$. The cutoff will introduce an
error $\langle t+r\rangle^{-3}$ but only in the region $(t+1)/4<r<3(t+1)/4$.
We have
 \beq\label{eq:Boxh0intro}
 |\Box \,h^0_{\mu\nu}|\lesssim \tplusr^{-3},\quad \text{when } (t+1)/4<r<3(t+1)/4,\quad  \text{and}\quad \Box\, h^0_{\mu\nu}=0
 ,\quad \text{otherwise}.
 \eq
 As alluded to in Section \ref{sec:commutatorsflat} we want to subtract off $h^0_{\mu\nu}$ in order to pick up the leading behavior of $h_{\mu\nu}$ in the exterior region. Then $h^1_{\mu\nu}\!=h_{\mu\nu}\!-\!h_{\mu\nu}^0$
 will satisfy the same equation as $h_{\mu\nu}$, i.e. $\Box^{\,g} h^1_{\mu\nu}\!=F_{{}_{\!}\mu\nu}\!-\Box^{\,g} h^0_{\mu\nu}$,
 apart from an error which to leading order is of the form \eqref{eq:Boxh0intro}. The new error we introduced
  in this way as well as $F_{{}_{\!}\mu\nu}$
  have additional exterior decay as in \eqref{eq:decayforradiationexterior} with any $\gamma\!\leq\! 1$. By
 the asymptotic flatness condition \eqref{eq:asymptoticallyflatdata} data for $h^1_{\mu\nu}$ are decaying like $r^{-1-\gamma}$ for some $\gamma\!<\!1$ so if we could replace
 $\Box^{\,g}$ with $\Box$ we could conclude from  Section \ref{sec:exteriorH0} that $h^1_{\mu\nu}$
 will have additional exterior decay $\langle (r\!-\!t)_+\rangle^{-\gamma}\!\!$. The above heuristic argument can
 be made into a proof in $L^2$ for $\Box^{\,g}$ by using  energy inequalities with exterior weights.

\subsubsection{Subtracting off a term picking up the mass in the wave coordinate condition}
\label{sec:subtractH0intro}
As pointed out above, $Mr^{-1}\delta^{\mu\nu}$ has to be a solution to the wave coordinate condition
\eqref{eq:approximatewavecordinateconditionintro} up to terms
of order $r^{-3}$. However since all terms in the left are order $r^{-2}$ it has to be an exact solution of
the left hand side. Multiplying with the cutoff function introduces an error as for the wave equation above
\beq\label{eq:WveCH0intro}
\big|\pa_\mu\big(H_0^{\mu\nu} -\tfrac{1}{2} m^{\mu\nu} m_{\alpha\beta} H_0^{\alpha\beta}\big)\big|
\lesssim \langle t+r\rangle^{-2},\quad \text{when } (t+1)/4<r<3(t+1)/4,\quad  \text{and}\quad =0
 ,\quad \text{otherwise}.
 \eq
 Subtracting this from \eqref{eq:approximatewavecordinateconditionintro}  gives a similar equation for
 $H_1^{\alpha\beta}\!\!=H^{\alpha\beta}\!\!-H_0^{\alpha\beta}\!\!$. This equation can then be
 integrated to show that the component \eqref {eq:HLLRadiation} of  $H_1$,
  that determines the bending
 of the light cones and the main component in the commutators, decays faster.
 The proof of this will however require that we first show estimates for all
  components with the
 additional exterior decay mentioned above.

\subsection{The energy momentum tensor of the charged scalar field}

\subsubsection{How the metric enters into the energy momentum tensor}
The energy momentum tensor in its simplest form can be written as
\[
T_{\alpha\beta}[{\bf F}, \bphi] = T^1_{\alpha\beta}(g)[{\bf F},{\bf F}] + T^2_{\alpha\beta}(g)[D\bphi, D\bphi],
\]
where $T^1, T^2$ are quadratic in ${\bf F}$ and $D\bphi$ respectively. We wish to bound $L^\infty$ norms for low derivatives and $L^1$ and $L^2$ norms for higher derivatives. We can write
\[
T\sim (1+O(|H|+|h|))(|\bF|^2+|D\bphi|^2)
\]
with derivatives of $T$ satisfying similar bounds with derivatives of $H, h, \bF, \bphi$. We are most interested with quadratic terms, which appear in the Minkowski case, as well as cubic terms where a large number of derivatives fall on $H$, and for which a slightly different analysis is required. Other cubic terms can be handled similarly to the quadratic terms, with nicer decay, and higher order terms in general behave similarly. For the quadratic terms, the analysis closely mirrors that for Minkowski. On lower derivatives we have, for $T\in \mathcal{T}, U\in\mathcal{N}$, and $s > 1/2$
\begin{subequations}
\begin{align}
|T[{\bf F}, \bphi](t,\cdot)| &\lesssim \varepsilon\langle t+r^* \rangle^{-2}\langle t-r^*\rangle^{-2}\langle (t-r^*)_+\rangle^{1-2s},\label{est:Tall}\\
|T[{\bf F}, \bphi](t,\cdot)_{TU}| &\lesssim \varepsilon \langle t +r^*\rangle^{-2-s}\langle t-r^*\rangle^{-2+s}\langle (t-r^*)_+\rangle^{1-2s}.\label{est:Tnice}
\end{align}
\end{subequations}
Similar bounds hold for $\widecheck{T}$, which are consistent with \eqref{eq:decayTintro} and \eqref{eq:decayweaknullTintro}.

The metric perturbation induces additional considerations in the $L^2$ estimates. We have to consider the case where almost all derivatives fall on the metric $g$, for which we must instead use weighted energy estimates on the metric. The method of this is detailed in \cite{Ka18}.

\subsubsection{Subtracting off the charge contribution in the exterior}\label{subsubsec:Charge}
In the analysis of the Maxwell-Klein Gordon system, one issue that arises is the fact that, even for compactly supported initial data $\bphi$, ${\bf F}$, like $\partial h$, will generally decay like $r^{-2}$, except in the charge free case. This limits the weights we can use when attempting an energy estimate. We resolve this issue by subtracting off a fixed solution of Maxwell's equations, ${\bf F}^0$, which picks up the asymptotic decay up to terms decaying like $o(|x|^{-2-\gamma})$. ${\bf F}^0$ satisfies Maxwell's equations for a current $J^0$ which in the Minkowski spacetime is compactly supported in $u$.
\begin{remark}
In the Minkowski space setting, we are subtracting off an exterior derivative of
\[
\overline{\bf A} \sim \frac{q[\bF]}{4\pi r}\,dt.
\]
As with $h^0$, these terms have additional decay along the light cone, but worse total decay at spatial infinity, see Section \ref{sec:lightconedecay}.
\end{remark}

\subsection{Coordinates adapted to the outgoing characteristic
surfaces of Schwarzschild}\label{sec:changingcoordinates}
In order to unravel the effect of the quasilinear terms one can change to
characteristic coordinates as in \cite{CK93}, but this is not explicit and loses regularity.
Instead we use the asymptotic behavior of the metric to determine the characteristic surfaces asymptotically and use this to construct coordinates.
Due to the wave coordinate condition \eqref{eq:WaveCordinateCond} the outgoing light cones
of a solution with asymptotically flat data \eqref{eq:asymptoticallyflatdata}
 approach those of the Schwarzschild metric with the same mass, which are described by
 the Regge-Wheeler coordinates.

\subsubsection{The outgoing characteristic surfaces}\label{sec:charsurface}
The outgoing light cones or characteristic surfaces are level sets of the solution of the eikonal equation
\beq\label{eq:eikonal}
g^{\alpha\beta} \pa_\alpha u \,\pa_\beta u=0.
\eq
For the Schwarzschild metric there is a solution 
$\us\!\!=\rs\!\!-_{{}_{\!}}t $,
where
$\rs$ is the so-called \emph{tortoise coordinate} of the Schwarzschild metric, which away from the horizon may be approximated by $\rs\!\!=r_{\!}+M\ln(r)$. 
Due to the wave coordinate condition, there is a solution $u$ of \eqref{eq:eikonal}
such $u\!\sim_{\!} \us$, as $r\!>\!t_{\!}/2\to\infty$, see Lindblad \cite{L17}.
In fact $\us\!$ is an approximate solution of \eqref{eq:eikonal} with $g^{\alpha\beta}\!$ replaced by
$m_0^{\alpha\beta}\!\!=m^{\alpha\beta}\!\!+\!H_0^{\alpha\beta}\!$ up to terms of order $\langle t\!+{}_{\!}r\rangle^{-2}$:
  \beq
m_0^{\alpha\beta} \pa_\alpha \us \,\pa_\beta \us
=-(1+M/r)+(1-M/r)(d\rs\!\!/dr)^2=O(M^2 r^{-2}),\quad\text{when}\quad r>t/2.
\eq

  \subsubsection{Asymptotic Schwarzschild coordinates}\label{sec:coordchange}
We therefore make the change of variables
\beq
\widetilde{t}=t,\qquad \widetilde{x}=\rs x/r, \qquad \text{where} \quad
\rs=r+\widetilde{\chi}(\tfrac{r}{t+1}) M\ln{r},
\eq
and $\widetilde{\chi}$ is as in Section \ref{exteriordecaymass}.
 Let
\beq
\widetilde{g}^{ab}=A^a_{\alpha} A^b_{\beta} g^{\alpha\beta},
\qquad\text{where}\quad
A^{a}_{\alpha}=\frac{\partial \widetilde{x}^{\,a}}{\partial x^\alpha},
\quad A^{\alpha}_{a}=\frac{\partial x^\alpha}{\partial \widetilde{x}^{\,a}},
\qquad\wpa_a=A^{\alpha}_{a}\pa_\alpha,
\eq
where $x^0=t$, $\widetilde{x}^0=\widetilde{t}$.
Then if $\widehat{m}^{ab}$ is the Minkowski metric in the $\widetilde{x}$ coordinates,
\beq
\widetilde{m}_0^{ab} \wpa_a \us \,\wpa_b \us
=m_0^{\alpha\beta} \pa_\alpha \us \,\pa_\beta \us,\quad\text{and}\quad
\widehat{m}^{ab} \wpa_a \us \,\wpa_b \us=0.
\eq
We would like to deduce  that as far as  components
determining the characteristic surfaces
\beq
 \widetilde{m}_0^{ab}\sim (1\!+M r^{-1})\widehat{m}^{ab}.
\eq
Since $\widetilde{g}^{\,ab}\!\!=\widetilde{m}_0^{ab}\!\!+\widetilde{H}_{\!1}^{ab}\!\!$,
and we expect the critical
components of $H_{\!1\!}$ to decay faster, we expect the reduced wave operator in the $(t,\widetilde{x})$ coordinates:
 \beq
 \widetilde{\Box}={\Box}^{\,\widetilde{g}}=\widetilde{g}^{\,ab}\wpa_a\wpa_b,
  \eq
  to asymptotically approach the constant coefficient Minkowski wave operator in these coordinates,
 \beq
 \Box^*={\Box}^{\,\widehat{m}}=\widehat{m}^{ab}\wpa_a\wpa_b.
 \eq
One heuristic motivation is to look at the linearized Einstein's Equations, as in Wald \cite{Wa} Chapter 7.5. The $h^0$ terms are
 generated by a metric perturbation term like $\tfrac{\chi}{r}\delta$, and the perturbation for a gauge transformation of the linearized system
  is generated by subtracting off the symmetric part of $2\nabla_\alpha V_\beta$ for an arbitrary vector $V$ (where $\nabla$ is the Levi-Civita connection associated with the
  background (Minkowski) metric).
  Setting $V = \ln r \partial_r$ gives a perturbation with a nice null structure, specifically away from $r=0$ we have
  \beq
  \tfrac{1}{r}\delta_{\alpha\beta} - \partial_\beta V_\alpha-\partial_\alpha V_\beta = -\tfrac{1}{r}m_{\alpha\beta} - \tfrac{2\ln r - 2}{r}(\delta_{\alpha\beta} - \omega_\alpha\omega_\beta)
  \eq
In this gauge, $t-|x|, t+|x|$ are again optical functions.
\subsubsection{Covariant formulation of Einstein's equations}\label{sec:covariant}
Let $\onabla_{\!a}$ be covariant differentiation with respect to the metric
$\widetilde{m}_{ab}\!=m_{\alpha\beta} A_a^{\,\alpha} \! A_b^{\,\beta}$
with Christoffel symbols $\widehat{\Gamma}^{c}_{ab}=\widetilde{m}^{cd}
\big(\wpa_a\widetilde{m}_{bd}
+\wpa_b\widetilde{m}_{ad}-\wpa_d\widetilde{m}_{ab}\big)/2
=O( Mr^{-2} \ln{r})$. Then
$\onabla_{\!\!a} h$ is equal to $\wpa_a h$ plus a correction of the form
$\widehat{\Gamma}\! \cdot\! h $, called the \emph{covariant error terms}.
Moreover Einstein's equations can be written
\beq
\tg^{ab}\onabla_{\!a}\onabla_b \widetilde{h}_{cd}
=\widetilde{F}_{cd}(\tg)[\onabla \widetilde{h},\onabla\widetilde{h}]+\widetilde{\widecheck{T}}_{\!cd},
\eq
and the wave coordinate condition \eqref{eq:approximatewavecordinateconditionintro} become
\beq
\onabla_{\!a}\big(\widetilde{H}^{ac}
-\tfrac{1}{2}\widetilde{m}^{ac}  \widetilde{m}_{bd}\,
 \widetilde{H}^{bd}\big)=\widetilde{W}^c(\tg)[\widetilde{H},\onabla\widetilde{H}].
 \eq

 \subsubsection{Improved commutators in the new coordinates}
  If we use the vector fields $\widetilde{Z}$ that commute with
the Minkowski wave operator
${\Box}^{\widehat{m}}=\widehat{m}^{ab}\wpa_a\wpa_b$,
or satisfy $[\widetilde{Z},{\Box}^{\,\widehat{m}}]
=c_{\widetilde{Z}}{\Box}^{\,\widehat{m}}$
for some constant $c_{\widetilde{Z}}$, then
 the commutator of $\widetilde{Z}^I$ with the reduced wave
operator ${\Box}^{\,\widetilde{g}}=\widetilde{g}^{ab}\wpa_a\wpa_b$
is roughly a product of vector fields applied to this wave
 operator plus
terms of the form
\beq
 (\widetilde{Z}^J \widetilde{g}^{ab})\wpa_a\wpa_b \widetilde{Z}^K h_{\mu\nu}, \qquad
  \text{for}\quad |J|+|K|\leq |I|,\quad |J|\geq 1.
\eq
However in these coordinates
\beq
\widetilde{g}^{ab}= \widetilde{m}_0^{ab}+ \widetilde{H}_1^{ab},
\eq
Here the critical components of
$\widetilde{Z} \widetilde{m}_0^{ab}$, $\widetilde{Z} \widetilde{H}_1^{ab}$ are under control
as mentioned in Sections \ref{sec:coordchange}, \ref{sec:subtractH0intro}.

\section{The geometric structure of Einstein's equations in wave coordinates}\label{sec:EEWaveCoordinatesStructure}
\subsection{The geometric structure}\label{sec:geometricstructure} Generic wave equations with
quadratic nonlinearities do not in general have global solutions for small data,
but some extra cancellation such as the null condition or weak null condition is needed.
In Section \ref{sec:notnull} we explained the need for the null condition, and
what extra cancellation it achieves.
For Einstein's equations the weak null condition can only be seen in a null frame, which we
introduce in Section \ref{sec:nullframe}. The geometric structure of
Einstein's equations in wave coordinates in a null frame really enters
in three different places.
The first cancellation originates from
 the wave coordinate condition in Section \ref{sec:wavecoord} and is
 then used to control the reduced wave operator in Section
 \ref{sec:reducedwave}. The fact that the reduced wave operator
 by itself is under control is then used together with the
 null structure of the reduced system
  with inhomogeneous terms in Section \ref{sec:inhom}.

\subsubsection{The null frame, tangential derivatives and killing vector fields}
\label{sec:nullframe}
By Huygen's principle the solution of the constant coefficient wave equation
emanating from an initial source at the origin propagates along the outgoing light
cone $t=r=|x|$. It is therefore natural to introduce a {\it null frame} ${\mathcal N}$ of vectors
tangential to the outgoing light cones plus a vector perpendicular to the cone:
\begin{equation}
	\underline{L}=\pa_t-\pa_r,
	\quad
	L=\pa_t+\pa_r,
	\quad
	S_1,S_2\in\mathbf{S}^2,
	\quad
	\langle S_i,S_j\rangle =\delta_{ij},
\end{equation}
It is well known
that, for solutions of wave equations, derivatives tangential to the outgoing light
cones $\overline{\pa }\in\mathcal{T}=\{L,S_1,S_2\}$ decay faster.
In fact, derivatives of solutions to the homogeneous wave
equation are also solutions since derivatives commute with the the wave operator,
and therefore decay as much. Moreover, for the generators of the Lorentz transformations and the scaling
\beq\label{eq:MinkowskiVectorFields}
x^i\pa_{x^j}-x^i\pa_{x^i}, \quad x^i\pa_t+t\pa_{x^i},\quad
t\pa_t+x^i\pa_{x^i},
\eq
the commutator with the wave operator is either 0 or a multiple
 of the wave operator. In any case, if $\Box \,\phi=0$, then $\Box Z\phi=0$ if $Z$
 is any of the vector fields \eqref{eq:MinkowskiVectorFields},
 so $Z\phi$ decays like a solution of the wave equation. Since the vector fields span the tangent
 space of the outgoing light cones
 \beq\label{eq:tangentialbyvectorfield}
 |\overline{\pa} \phi{}_{\,}|\lesssim \langle\,t\!+\!|x|\rangle^{-1}{\sum}_{Z} |Z\phi{}_{\,}|,\quad\text{and}\quad
 |{\pa} \phi{}_{\,}|\lesssim \langle\, t\!-\!|x|\rangle^{-1}{\sum}_{Z} |Z\phi{}_{\,}| .
 \eq
Therefore, tangential derivatives decay better and neglecting tangential derivatives
$\overline{\partial}  \phi$  of $\phi$:
\beq\label{eq:transversalderivativeprojectionsec2}
	|\pa_\mu \phi-L_\mu \pa_q \phi|\lesssim |\overline{\pa} \phi|,
\qquad\text{where}\quad \pa_q=(\pa_r-\pa_t)/2,
	\quad L_\mu=m_{\mu\nu} L^\nu.
\eq
In fact we have
\beq
|\pas\, \phi|^2=|\pas_1\phi|^2+|\pas_2\phi|^2+|\pas_3\phi|^2=|S_1\phi|^2+|S_2\phi|^2,
\quad\text{where}\quad
\pas\, \phi=\pa \phi -\omega \pa_r \phi,\quad \omega=x/r.
\eq
\subsubsection{The geometric structure of the wave coordinate condition}
\label{sec:wavecoord}
The wave coordinate condition can be written
$
\pa_\alpha g^{\alpha\delta} -\tfrac{1}{2}g^{\alpha\delta}
  g_{\beta\gamma}\,\pa_\alpha g^{\beta\gamma}=0
  $
from which it follows that $H^{\alpha\beta}=g^{\alpha\beta}-m^{\alpha\beta}$ satisfy
\begin{equation}\label{eq:approximatewavecordinatecpondition}
\pa_\mu\big(H^{\mu\nu} -\tfrac{1}{2} m^{\mu\nu} m_{\alpha\beta} H^{\alpha\beta}\big)
=W^\nu(g)[H,\pa H].
\end{equation}
Expressing the divergence above in a null frame;
\beq
\pa_\mu \widecheck{H}^{\mu \nu}\!=
L_\mu\pa_q \widecheck{H}^{\mu \nu} \!-\Lb_{\,\mu} \pa_s \widecheck{H}^{\mu \nu} \!
+ S_{1\,\mu} \pa_{S_1} \widecheck{H}^{\mu \nu}\!+S_{2\,\mu} \pa_{S_2} \widecheck{H}^{\mu \nu}\!\!,\quad
\text{where}\quad\pa_q\!=(\pa_r\!-\pa_t)/2,\quad \pa_s\!=(\pa_r\!+\pa_t)/2,
\eq
and contracting with $T_\mu\in \mathcal{T}$ respectively $\underline{L}_\nu$
it follows that
\begin{equation}\label{eq:firstwavecoordinateestimate}
|\pa_q H|_{L{\mathcal T}}+|\pa_q \trs H|\lesssim |\overline{\pa} H|+|W|,\quad\text{where}\quad
|W|\lesssim |H|\, |\pa H|.
\end{equation}
Here $H_{UV}=H^{\alpha\beta}U_\alpha V_\beta$, where $U_\alpha=m_{\alpha\beta}U^\beta$ and
\beq
|H|_{L\mathcal T}=|H_{LL}|+|H_{LS_1}|+|H_{LS_2},
|\qquad\text{and}\qquad\trs H=\delta^{AB}H_{AB},\quad A,B\in{\mathcal S}=\{S_1,S_2\}.
\eq

\subsubsection{The geometric null structure of the reduced wave operator}
\label{sec:reducedwave}
By
\eqref{eq:transversalderivativeprojectionsec2}
we have
\beq\label{eq:waveoperatorHLLtangential}
	\big| H^{\alpha \beta} \partial_{\alpha} \partial_{\beta}
\phi-H_{LL}\pa_q^2 \phi \big|
	\lesssim
	\vert  H \vert \vert \overline{\partial} \partial \phi \vert,
\eq
and by \eqref{eq:tangentialbyvectorfield} the tangential derivatives
in the right are
better behaved. Here $H_{LL}$ is controlled by the wave
coordinate condition
\eqref{eq:firstwavecoordinateestimate}. In fact a more detailed
analysis, see Section \ref{sec:approxwavecoord},
shows that at null infinity $H_{{}_{\!}LL{}_{\!}}\!\sim -2M\!/r$,
consistent with
\eqref{eq:HLLRadiation}. In Section \ref{sec:Genralizedwavecoord}
we will  change coordinates to asymptotically remove
$-2M\!/r$ from $H_{{}_{\!}LL{}_{\!}}$ in \eqref{eq:waveoperatorHLLtangential} so that the quasilinear terms will be lower
order, while it will only introduce lower order corrections to
 the wave coordinate condition and the inhomogeneous terms.

\subsubsection{\!\!The geometric structure of the inhomogeneous terms}\label{sec:inhom}
Recall that the inhomogeneous term in Einstein's vacuum equations has the form
\beq
F_{\mu\nu} (g) [\pa h, \pa h]=P(g)[\pa_\mu h, \pa_\nu h]+Q_{\mu\nu}(g)[\pa h, \pa h],
\eq
where $Q$ is a combination of classical null forms and $P$, given by \eqref{eq:curvedPdef},
 has a weak null structure that we will now describe. First,
since $Q_{\mu\nu}\!=Q_{\mu\nu}(m)$ satisfy the classical null condition
\beq 
|Q_{\mu\nu}(\pa h,\pa k)|\les |\overline{\pa} h|\,|\pa k|+|\pa h|\, |\overline{\pa} k|.
\eq
The main term $P=P(m)$ can be further analyzed as follows.
First we note that
by \eqref{eq:transversalderivativeprojectionsec2}
\beq\label{eq:PPnull}
\big| P(\pa_\mu h,\pa_\nu k)- L_{\mu}L_\nu P(\pa_q h,\pa_q k)\big|\les |\overline{\pa} h| \,|\pa k|
+|\pa h|\,|\overline{\pa} k|.
\eq

Expressing $P(h,k)=P_{\mathcal N}(h,k)$ in a null frame we have
\begin{multline}\label{eq:nullframeP}
P_\mathcal{N\,}(h,k)
=-\big({h}_{LL}
{k}_{\underline{L}\underline{L}}
+{h}_{\underline{L}\underline{L}}
{k}_{{L}{L}}\big)/8 -\delta^{CD}\delta^{C^\prime
D^\prime}\big(2{h}_{CC^\prime}{k}_{DD^\prime}-
{h}_{CD} {k}_{C^\prime D^\prime}\big)/4\\
+\delta^{CD}\big(2{h}_{C
L}{k}_{D\underline{L}} +2{h}_{C
\underline{L}}{k}_{D{L}}- {h}_{CD}
{k}_{L\underline{L}}-{h}_{L\underline{L}}
{k}_{CD} \big)/4.
\end{multline}
It follows that
\beq\label{eq:tanPsec3}
\big|P_{\mathcal N}(h,k)-P_{\mathcal S}(h,k)\big|\les \big(|h\,|_{L\mathcal
T}+|\trs h|\big)|k| +|h\,|\big(|k\,|_{L\mathcal
T}+|\trs k_{}|\big).
\eq
where
\beq
P_{\mathcal{S}} (D,E)= -\widehat{D}_{\!AB}\, \widehat{E}^{AB\!\!}/2,\quad A,B\in\mathcal{S},\quad
\text{where} \quad \widehat{D}_{\!AB}=D_{\!AB}-\delta_{AB}\trs D\!/2.
\eq
Hence
\beq\label{eq:PPnullS}
\big| P(\pa_\mu h,\pa_\nu k)- L_{\mu}L_\nu P_{\mathcal S}(\pa_q h,\pa_q k)\big|
\les \big( |\overline{\pa} h|\!+\!|\pa_q h|_{L\mathcal
T}\!+\!|\pa_q\trs_{\!} h| \big)\,|\pa k|
+|\pa h|\,\big(|\overline{\pa} k|\!+\!|\pa_q k|_{L\mathcal
T}\!+\!|\pa_q \trs_{\!} k_{}|\big).
\eq
Also using the wave coordinate condition \eqref{eq:firstwavecoordinateestimate}
and that fact the $H=-h+O(h^2)$ we get
\beq \big| P(\pa_\mu h,\pa_\nu h)- L_{\mu}L_\nu  P_{\mathcal S}(\pa_q h,\pa_q h)\big|
\les \big(|\overline{\pa} h|+|h||\pa h|\big) |\pa h|.
\eq
With respect to the null frame, the Einstein equations become
\beq
	(\Box^{\,g} h)_{TU}\sim 0, \quad T\in \mathcal{T},U\in\mathcal{N}\qquad
	(\Box^{\,g} h)_{\underline{L}\underline{L}}\sim 4 P_{\mathcal{S}}(\pa_q h,\pa_q h),
	\label{eq:simplifiedEinstein}
\eq
since $T^\mu L_\mu\!=\!0$ for $T\!\in\!\mathcal{T}$. Here, $ P_{\mathcal{S}}(\pa_q h,\pa_q h)$ only depends
on tangential components, for which, by the first equation, we
have better control. Hence in a null frame as far as semilinear terms Einstein's
equations look like
\beq
\Box\, \phi=0,\qquad \Box\, \psi=(\pa_t \phi)^2.
\eq

\subsection{Lie derivatives and Commutators}
In order to get estimates for higher derivatives we need to commute the system
 with vector fields $Z$. However, if one instead commutes with Lie derivatives along the vector fields $Z$
 it turns that the geometric structure is preserved also for the lower order terms.
 Note that for a function $\phi$
\[
	\Box^{\,g} Z \phi
	=
	Z \big( \Box^{\,g} \phi \big)
	+
	2 g^{\alpha \beta} \partial_{\alpha} Z^{\mu} \partial_{\beta} \partial_{\mu} \phi - Z(g^{\alpha \beta}) \partial_{\alpha} \partial_{\beta} \phi
	=
	Z \big( \Box^{\,g} \phi \big)
	-
	(\mathcal{L}_Z g^{\alpha \beta}) \partial_{\alpha} \partial_{\beta} \phi,
\]
where the fact that $\partial_{\mu} \partial_{\nu} Z^{\lambda} = 0$ for each $Z$ and
 $\mu, \nu, \lambda = 0,1,2,3$ has been used. Here the Lie derivative applied to a
 $(r,s)$ tensor $K$ is defined by
\begin{equation}
{\mathcal L}_Z K^{\alpha_{\!1}\dots \alpha_r}_{\beta_1\dots \beta_s}\!
=Z K^{\alpha_{\!1}\dots \alpha_r}_{\beta_{1}\dots \beta_s}
-\pa_{\gamma\!} Z^{\alpha_{\!1}} K^{\gamma\dots \alpha_r}_{\!\beta_{\!1}\dots \beta_s}\!-\cdots
-\pa_{\gamma\!} Z^{\alpha_r} K^{\alpha_{\!1}\dots \gamma}_{\!\beta_1\dots \beta_s}\!
+\pa_{\beta_{\!1}\!} Z^{\gamma} K^{\alpha_{\!1}\dots \alpha_r}_{\gamma\dots \beta_s}\!\!+\cdots
+\pa_{\beta_s}\! Z^{\gamma} K^{\alpha_{\!1}\dots \alpha_r}_{\beta_1\dots \gamma}\!\!.
\end{equation}
More generally, for any vector field $Z=Z^\alpha\pa_\alpha$ with linear coefficients
$Z^\alpha=c^\alpha_\beta x^\beta$ we have
\begin{equation} \label{eq:Liepartialcommute}
		{\mathcal L}_Z\pa_{\mu_1}\!\cdots\pa_{\mu_k} K^{\alpha_1\dots \alpha_r}_{\beta_1\dots \beta_s}
		=
		\pa_{\mu_1}\!\cdots\pa_{\mu_k} {\mathcal L}_Z K^{\alpha_1\dots \alpha_r}_{\beta_1\dots\beta_s}.
	\end{equation}
 The procedure simplifies further by commuting with a modified Lie derivative
 $\widehat{\mathcal{L}}$, defined by
\beq\label{eq:modLiedef}
	\widehat{\mathcal L}_Z K^{\alpha_1\dots \alpha_r}_{\beta_1\dots \beta_s}
	=
	{\mathcal L}_Z K^{\alpha_1\dots \alpha_r}_{\beta_1\dots \beta_s}
	+
	\tfrac{r-s}{4}(\pa_\gamma Z^\gamma)K^{\alpha_1\dots \alpha_r}_{\beta_1\dots \beta_s}.
\eq

\subsubsection{Commutators with the wave coordinate condition}
The Lie derivative  commutes
 with exterior differentiation which for any of the vector fields $Z$ leads to
\[
	 \pa_\mu  \widehat{\mathcal L}_Z \widecheck{H}^{\mu\nu}
=\big(\widehat{\mathcal L}_Z
+\tfrac{\pa_\gamma Z^\gamma}{2}\big)\pa_\mu \widecheck{H}^{\mu\nu}.
\]
This applied to \eqref{eq:approximatewavecordinatecpondition} in turn leads to higher order approximate wave coordinate conditions
so the term $\widehat{\mathcal{L}}_Z H^{\Lb \Lb}$ is then controlled
by using \eqref{eq:firstwavecoordinateestimate} and integrating in the $q=r-t$ direction
to control $H^{\Lb \Lb}$ itself.

\subsubsection{Commutators with the reduced curved wave operator}
The modified Lie derivative satisfies
\beq\label{eq:hatLiem}
\widehat{\mathcal{L}}_Z m^{\alpha\beta} = 0,
\eq
 for each of the vector fields $Z$ and moreover
\begin{equation} \label{eq:intromodLiecommsec2}
	\Box^{\,g} \widehat{\mathcal{L}}_Z \phi_{\mu \nu}
	=
	\mathcal{L}_Z \big( \Box^{\,g} \phi_{\mu \nu} \big)
	-
	(\widehat{\mathcal{L}}_Z g^{\alpha \beta}) \partial_{\alpha} \partial_{\beta} \phi_{\mu \nu},
\end{equation}
where in view of \eqref{eq:hatLiem}
$\widehat{\mathcal{L}}_Z g^{\alpha \beta}=\widehat{\mathcal{L}}_Z H^{\alpha \beta}$.
By \eqref{eq:waveoperatorHLLtangential} the commutator can be controlled by
the term $(\widehat{\mathcal{L}}_Z H)_{LL} \partial^2_q \phi_{\mu \nu}$,
plus terms involving one tangential derivative
\[
	\vert (\widehat{\mathcal{L}}_Z H^{\alpha \beta})
\partial_{\alpha} \partial_{\beta} \phi_{\mu \nu} \vert
	\lesssim
	\vert (\widehat{\mathcal{L}}_Z H)_{LL} \vert \vert \partial^2 \phi_{\mu \nu} \vert
	+
	\vert \widehat{\mathcal{L}}_Z H \vert \vert \overline{\partial} \partial
\phi_{\mu \nu} \vert.
\]

\subsubsection{Commutators with the inhomogeneous terms}
When applying the commutation formula \eqref{eq:intromodLiecommsec2} to the
reduced Einstein equations \eqref{eq:EinsteinWaveintro} we have to estimate
Lie derivatives of nonlinear terms:
$\widehat{\mathcal{L}}_Z^I \left( F_{\mu \nu}(g)[\partial h, \partial h] \right)$.
 Let $h_{\alpha\beta}$ and $k_{\alpha\beta}$ be $(0,2)$ tensors and let
 $S_{\mu\nu}(g)[\pa h,\pa k]$ be a $(0,2)$ tensor which is a quadratic form in the $(0,3)$ tensors
 $\pa h$ and $\pa k$ with two contractions with the metric $g$
 (in particular $P(g)[\pa_\mu h,\pa_\nu k]$ or $Q_{\mu\nu}(g)[\pa h,\pa k]$).
 Then $S_{\mu\nu}(g)[\pa h,\pa k]=S_{\mu\nu}[G,G][\pa h,\pa k]$, i.e.
 it is bilinear in the inverse of the metric $G^{\alpha\beta}\!=g^{\alpha\beta}\!$.
 We have
\begin{multline}\label{eq:LieQuad}
 {\widehat{\mathcal L}}_Z\big({}_{\!}  S_{\mu\nu}(g{}_{\!})[\pa h,\pa k]\big)\!\!
 =\!S_{\mu\nu}(g{}_{\!})[\pa \widehat{\mathcal L}_Z h,k]
 +S_{\mu\nu}(g{}_{\!})[\pa h,\pa \widehat{\mathcal L}_Z k]\\
 +  S_{\mu\nu}[\widehat{{\mathcal L}}_Z G,{}_{\!}G][\pa h,\pa k]
 +  S_{\mu\nu}[G,\widehat{{\mathcal L}}_Z G][\pa h,\pa k],
\end{multline}
so the Lie derivative preserves the desirable structure of the nonlinear terms
$P(g)[\pa_\mu h,\pa_\nu h]$ and
$Q_{\mu\nu}(g)[\pa h,\pa k]$.
We remark that the last two terms in \eqref{eq:LieQuad}
are cubic and therefore much easier to control, see Section \ref{sec:cubicerrors}.

\section{Subtracting off terms that picks up the mass and charge contributions}
\label{sec:masssubtractionsection}
In order to bound the solution $h$ to the wave equation in the
weighted energy spaces we will be using,
we need to subtract off an approximate solution to the homogenous wave equation
$h^0$ which picks up
the contribution
 from the initial mass in
\eqref{eq:asymptoticallyflatdata}. Similarly one can subtract off the charge from the electromagnetic field.
These will be explained in further detail when we get in to the specific norms that we will be using, e.g. in Section \ref{sec:StatementandBootstrap}.

\subsection{Subtracting off the mass}
\subsubsection{Subtracting off a term that picks up the mass
contribution from the wave coordinate condition}
\label{sec:approxwavecoord}
Let
\beq\label{eq:m0def}
m_0^{\alpha\beta}=m^{\alpha\beta}+H_0^{\alpha\beta},\qquad\text{where}\qquad
H_0^{\alpha\beta}=-Mr^{-1}\widetilde{\chi}(\tfrac{r}{t+1})
\delta^{\alpha\beta},
\eq
and $\widetilde{\chi}(s)=1$, when $s>3/4$ and $\widetilde{\chi}(s)=0$ when $s<1/4$.
A calculation shows that the approximate wave coordinate condition
\eqref{eq:approximatewavecordinatecpondition} is approximately satisfied by $H_0$:
\beq\label{eq:H0wavecoord}
 W^{\beta,0}_{mass}=-\pa_\alpha \big({H}_0^{\alpha\beta}\! -\tfrac{1}{2}{m}^{\alpha\beta}  {m}_{\mu\nu}\,
 {H}_0^{\mu\nu}\big)
 =-M\delta^{\beta 0}{\chi}_0^{\,\prime\!}\big(\tfrac{r}{t+1}\big) r^{-2},
 \eq
 where ${\chi}_{i}^{\,\prime\!}(s,\cdot)$ stands for functions
 supported when $|s\!-\!1\!/2|\!<\!1\!/4$.
 Hence by \eqref{eq:approximatewavecordinatecpondition}
$H_1^{\alpha\beta}=H^{\alpha\beta}-H_0^{\alpha\beta}$ satisfy an approximate wave coordinate condition
 \beq\label{eq:approxwavecoord}
\pa_\alpha \big(H_1^{\alpha\beta}\! -\tfrac{1}{2} m^{\alpha\beta} m_{\mu\nu} H_1^{\mu\nu}\big)
=W^\beta(g)[H,\pa H]+W^{\beta,0}_{mass}.
\eq
Once we subtracted off $H_0$, the critical components of $H_{1}$ will have better
decay as $r\!\to\!\infty$ and by integrating the equation in the $t\!-\!r$ direction using the null decomposition,
everywhere in the exterior of the
light cone. We can further write
\beq
W^\beta(g)[H,\pa H]=W^\beta(g)[H_1,\pa H_1]
+W^{\beta,1}_{mass}
+W^{\beta,2}_{mass}[H_1]+W^{\beta,3}_{mass}[\pa H_1],
\eq
where
\begin{align}
W^{\beta,1}_{mass}&=W^\beta(g)[H_0,\pa H_0]
=\chi^{\beta,1}\big(\tfrac{r}{1+t},\omega,g\big)M r^{-3}, \\
W^{\beta,2}_{mass}[H_1]&=W^\beta(g)[H_1,\pa H_0]
=\chi_{\mu\nu}^{\beta,2}\big(\tfrac{r}{1+t},\omega,g\big)M r^{-2} H_1^{\mu\nu},\\
W^{\beta,3}_{mass}[\pa H_1]&=W^\beta(g)[H_0,\pa H_1]
=\chi_{\mu\nu}^{\beta\gamma,3}\big(\tfrac{r}{1+t},\omega,g\big)
 Mr^{-1}\pa_\gamma H_1^{\mu\nu} ,
\end{align}
since
\beq
\pa^{\alpha} H^{\mu\nu}_0 = \chi^{\alpha \mu\nu}
\big(\tfrac{r}{1+t},\omega\big)M r^{-1-|\alpha|}.
\eq
Here $\chi^{\alpha \mu\nu}
\big(s,\omega\big)$ and  $\chi_{\mu\nu}^{\beta,i}\big(s,\omega,g\big)$
are bounded functions that vanish for $s\leq 1/4$.

\subsubsection{The asymptotically Schwarzschild wave operator}
We have
\beq
\Box^{\,g}=\Box^{\,m_0}+H_{1}^{\alpha\beta}\pa_\alpha\pa_\beta,
\eq
where we expect the critical coefficient in front of the $\pa_q^2$ term, i.e.
$H_1^{\alpha\beta} L_\alpha L_\beta$ to be small.

Expressing $\Box^{\,m_0}=m_0^{\alpha\beta}\pa_\alpha\pa_\beta$ in spherical coordinates we get
\begin{equation}\label{eq:Boxrm0insphericalcoord}
\Box^{m_0\!}\phi=\bigtwo(\!-\!\pa_t^2\!+\!\triangle_x\!
-\!\tfrac{M\!}{r}{\widetilde{\chi}}\big(\!\tfrac{r}{1+t{}_{\!}}\big)
\big(\pa_t^2\!+\!\triangle_x\big)\!\bigtwo)\phi
=\tfrac{1}{r\!}\bigtwo(\!-\!\pa_t^2\!+\pa_r^2\!-\!\tfrac{M\!}{r}{\widetilde{\chi}}
\big(\!\tfrac{r}{1+t{}_{\!}}\big)
(\pa_t^2\!+\pa_r^2)\!\bigtwo)(r\phi)
+\bigtwo(\!1\!-\!\tfrac{M}{r}{\widetilde{\chi}}\big(\!\tfrac{r}{1+t{}_{\!}}
\big)\!\bigtwo) \!\frac{1}{r^2\!}\triangle_\omega \phi.
\end{equation}

\subsubsection{Subtracting off a term that picks up the mass
contribution from the wave equation}
Let
\beq\label{eq:hup0def}
h^{0}_{\alpha\beta}=Mr^{-1}\widetilde{\chi}(\tfrac{r}{t+1})\delta_{\alpha\beta},
\eq
where $\widetilde{\chi}$ is as in \eqref{eq:m0def}.
We have
\beq\label{eq:h0alphaest}
\pa^{\alpha} h^0_{\mu\nu} = \chi_{\mu\nu}^\alpha
\big(\tfrac{r}{1+t},\omega\big)M r^{-1-|\alpha|},
\eq
where $\chi_{\mu\nu}^\alpha(s,\omega)$ stands for bounded functions that vanish when $s<1/4$.
Using \eqref{eq:Boxrm0insphericalcoord}
we see that
\begin{equation}\label{eq:Boxrm0}
E_{\mu\nu,0}^{\,mass}=\Box^{\,m_0} h^0_{\mu\nu}
=M\chi^\prime_{\mu\nu}\big(\tfrac{r}{1+t},\tfrac{M}{r}\big) r^{-3},
\end{equation}
where $\chi^\prime_{\mu\nu}(s,\cdot\big)$ stands for functions supported when $|s\!-\!1\!/2|\!<\!1\!/4$, in the far interior away from the light cone.
 Moreover
 \beq\label{eq:BoxrH1m0}
 E_{\mu\nu,1}^{\,mass}[H_1]=H_{1}^{\alpha\beta}\pa_\alpha\pa_\beta h^0_{\mu\nu}
 =\chi_{\mu\nu\alpha\beta}
\big(\tfrac{r}{1+t},\omega\big)M r^{-3}H_{1}^{\alpha\beta}.
 \eq
 Hence
\beq\label{eq:thewaveoperatormasserrror}
E_{\mu\nu}^{\,mass}=\Box^{\,g} h^0_{\mu\nu}
=\Box^{\,m_0} h^0_{\mu\nu}+H_{1}^{\alpha\beta}\pa_\alpha\pa_\beta h^0_{\mu\nu}
=E_{\mu\nu,0}^{\,mass}+E_{\mu\nu,1}^{\,mass}[H_1].
\eq
The reason we need to subtract off $h^0$ is that $h$ itself
is not in the weighted energy space we need to use.
The terms \eqref{eq:Boxrm0} decay enough to cover extra exterior weight in $L^{{}_{\!}2}\!$,
since they vanish in the exterior and along the light cone.
The term \eqref{eq:BoxrH1m0}
 is also easy to control since it is linear in $H_1$ and the weight can be
 absorbed in the norm of $h^{{}_{\!}1}\!$.
Since as we will see, in the weighted norms that we will be using
\beq
\| \langle r^*\!\!-t\rangle^{-1}h^1 \|\lesssim \|\pa h^1 \|,
\eq
the term \eqref{eq:thewaveoperatormasserrror} still has an additional decaying factor
of $t^{-2}$.

\subsubsection{Subtracting off a term that picks up the mass contribution from the inhomogeneous term}
Let
\beq\label{eq:Ferror}
F_{\mu\nu}^{\, mass}=F_{\mu\nu} (g) [\pa h, \pa h]-F_{\mu\nu} (g) [\pa h^1\!, \pa h^1]
=F_{\mu\nu,0}^{\,mass}+F_{\mu\nu,1}^{\,mass}[\pa h],
\eq
where
\begin{align}\label{eq:Ferror1}
F_{\mu\nu,0}^{\,mass}&=F_{\mu\nu} (g) [\pa h^0\!, \pa h^0]
=\chi_{\mu\nu}\big(\tfrac{r}{1+t},\omega,\tfrac{M}{r},g\big)M^2 r^{-4},\\
\label{eq:Ferror2}
F_{\mu\nu,1}^{\,mass}[\pa h]&=2F_{\mu\nu} (g) [\pa h^0\!, \pa h^1]
=M^2 r^{-2}\chi_{\mu\nu}^{\alpha\beta\gamma}
\big(\tfrac{r}{1+t},\omega,\tfrac{M}{r},g\big)\pa_\gamma h^{1}_{\alpha\beta}.
\end{align}
Here the terms with linear and quadratic factors in $\pa h^0$ can be estimated
without use of any special geometric structure
only using  the estimate \eqref{eq:h0alphaest}.
The reason we need to subtract off $h^0$ is that $h$ itself
is not in the weighted energy space we use.
The terms \eqref{eq:Ferror} decay enough to cover the exterior weight.
The term \eqref{eq:Ferror2} that is linear in $h^1$ is
 even easy to control since the weight can be absorbed in the norm of $h^{{}_{\!}1}\!$.

\subsubsection{Subtracting off a term that picks up the mass contribution from Einstein's equations
and estimating the mass errors in exterior weighted norms}
Modulo a mass error $R_{\mu\nu}^{mass}=E_{\mu\nu}^{mass}+F_{\mu\nu}^{mass}$ we have
\begin{equation}\label{eq:EinsteinWaveMassErrors}
\Box^{\,g}  h_{\mu\nu}^1
=F_{\mu\nu} (g) [\pa h^1, \pa h^1]+R_{\mu\nu}^{mass}+\widehat{\,T\!}_{\mu\nu}.
\end{equation}
We can write $R_{\mu\nu}^{mass}=R_{\mu\nu,0}^{mass}+R_{\mu\nu,1}^{mass}$, where
$R_{\mu\nu,0}^{mass}=E_{\mu\nu,0}^{\,mass}+F_{\mu\nu,0}^{\,mass}$ and
$R_{\mu\nu,1}^{mass}=E_{\mu\nu,1}^{\,mass}[H_1]+F_{\mu\nu,1}^{\,mass}[\pa h]$.

We will be using a weighted $L^p$ norms, for $p=2,\infty,1$,
\beq
\| \phi(t,\cdot) w_{p,\gamma}\|_{L^p},\quad \text{where} \quad
w_{p\, ,\gamma}=\langle (r^*-t)_+\rangle^{1-1/p\, +\gamma},\quad 0<\gamma <1.
\eq
Using the weighted energy inequality (if $p=2$) we need to show that
\beq
\int_0^\infty (1+t)^{1-2/p}\,
 \| R_{\mu\nu,0}^{mass}(t,\cdot)\, w_{p\, ,\gamma}\|_{L^p}\, dt < C M ,
\eq
which follows from \eqref{eq:thewaveoperatormasserrror} and \eqref{eq:Ferror}. For $p=1,\infty$ this is
used for the decay estimates.

\subsection{Subtracting off the charge} In the Minkowski metric, given a field $\bF$ solving \eqref{eq:MKG}, $\bF$ has a charge $q$ defined by the integral
\[
\bold{q}[\bF] = \int_{\mathbb{R}^3}\mathfrak{I}(-J^0(0,x))\, dx,
\]
where the quantity on the right can be defined in terms of initial data for $\bphi$. This is invariant in time, since the current $J$ is by design divergence free. By an application of the divergence theorem in space;
\[
\lim_{r\to\infty}\int_{\mathbb{S}^2} r^2\omega^i\bF_{0i}(0, r\omega) = \bold{q}[\bF].
\]
Consequently, unless $\bold{q}\!=\!0$,  $\bF$ can not decay uniformly faster than $r^{-2}$, so it is not bounded in our desired energy norms. We instead subtract off a well-defined $\bF^0$ with the same asymptotic decay as $\bF$, such that we can establish a meaningful energy estimate on the difference $\bF-\bF^0$. In the Minkowski spacetime, Lindblad and Sterbenz \cite{LS} use the fact that $\bF = d(r^{-1}\, dt)$ is a solution of Maxwell's equations away from the origin (with vanishing current) to construct a field $\bF^0$ which models the point charge in the exterior and is 0 inside the light cone. The associated current of $\bF^0$ is supported close to the light cone and decays rapidly in time. Kauffman \cite{Ka18} constructed an analogous current for asymptotically Schwarzschild metrics which instead decays rapidly away from the light cone.

We can split the portion of the energy momentum tensor coming from $\bF$ into four parts, recalling the remainder $\bF^1 = \bF - \bF^0$. Using the notation
\[
T[\bF, {\bf G}]_{\alpha\beta} = \bF_{\alpha\gamma}{\bf G}_\beta^{\,\,\gamma} - \frac14 g_{\alpha\beta}\bF_{\gamma\delta}{\bf G}^{\gamma\delta}
\]
we can write
\[
T[\bF, \bF] = T[\bF^0, \bF^0] + 2T[\bF^0, \bF^1] + T[\bF^1, \bF^1],
\]
which has the consequent bounds
\begin{subequations}
\begin{align}
|T[\bF^0, \bF^0]| & \lesssim |q|^2\widetilde\chi(r^*-t)\langle r^*+t\rangle^{-4}, \\
|T[\bF^0, \bF^1]| & \lesssim |q|\widetilde\chi(r^*-t)\langle r^*+t\rangle^{-2}|\bF^1|, \\
|T[\bF^1, \bF^1]|_{\mathcal{U}\mathcal{T}}&\lesssim |\bF^1||\slashed{\bF}^1|,+ |h||\bF^1|^2,\\
|T[\bF^1, \bF^1]| &\lesssim |\bF^1|^2.
\end{align}
\end{subequations}
Here, we use the notation
\[
|{\slashed\bF}^1| = |{\bF}^1_{LS_1}|+|{\bF}^1_{LS_2}|+|{\bF}^1_{\underline{L}S_1}|
+|{\bF}^1_{S_1S_2}|+|{\bF}^1_{\underline{L}L}|;
\]
that is, $|\slashed{\bF}^1|$ does not include the components with the worst decay. We note for now that the component decomposition for $|T[\bF^1, \bF^1]|$ is necessary to establish the bounds \eqref{eq:decayTintro} and \eqref{eq:decayweaknullTintro}.

When we apply Proposition \ref{prop:hormander} we must treat $\bF^0$ and $\bF^1$ separately, using energy bounds when $\bF^1$ appears, and integrating decay bounds for $\bF^0$\!. Importantly, $\bF^1\!$ determines the asymptotic behavior of the metric $g$, as we have worse decay along the light cone.

\section{The structure of the metric in asymptotically Schwarzschild coordinates}
When attempting to prove an energy estimate in wave coordinates, one runs into the complications arising from the fact that $g_{{}_{\!}LL{}_{\!}}$ will asymptotically decay like $-2M\!/r$ along the light cone, an issue which appears even when using the approximation $m_0^{\alpha\beta} = m^{\alpha\beta} + H_0^{\alpha\beta}$. For instance, attempting to repeat the proof of Theorem \ref{thm:EEst} with $\mu = 0$ would give slowly growing energy, as can be seen in \cite{LR10}. A heuristic argument, given by integrating \eqref{eq:firstwavecoordinateestimate} along ingoing null characteristics from the initial data, would imply that this bad decay rate comes solely from corresponding decay for the initial data in space.

In Section \ref{sec:subGenralizedwavecoord} we will give an outline of the change to generalized wave coordinates, with which we may asymptotically remove the first order correction from $H_{{}_{\!}LL{}_{\!}}$ in \eqref{eq:waveoperatorHLLtangential}. In Section \ref{ChangeinCoordBounds} we will determine commutator fields adapted to these new coordinates, and prove certain technical estimates on error terms which appear when changing coordinates. In section \ref{sec:AsymptoticSchwarzschildNullStructure}, we characterize the null structure of $m_0$ in asymptotically Schwarzschild coordinates. Lemma \ref{lem:m0approx} gives a decomposition of the form
 \beq
 \widetilde{m}_0^{ab}= \kappa_{0\,}\widehat{m}^{ab}\!
 +\kappa_{1\,} \slashed{S}^{\,ab}\!+\kappa_{2\,}\omega^a\omega^b\!+\kappa_{3\,} i_+^{\, ab}\!.
 \eq
The quantity $\kappa_{0\,}\widehat{m}^{ab}\! +\kappa_{1\,} \slashed{S}^{\,ab}$ exhibits an improved null structure, and components $\kappa_{2\,}\omega^a\omega^b\!+\kappa_{3\,} i_+^{\, ab}$ decays rapidly along the light cone. In Lemma \ref{lem:m0decomp} we show that Lie derivatives with respect to the new commutator fields preserve this improved structure.
\label{sec:Genralizedwavecoord}\label{section:metricdecomposition}

\subsection{Generalized Wave Coordinates asymptotically adapted to the outgoing light cones}\label{sec:subGenralizedwavecoord}
\subsubsection{The leading behavior of the metric at null infinity and changes of coordinates that
straighten out the characteristic surfaces at null infinity}
As previously pointed out in Section \ref{sec:coordchange}, the metric to leading order approaches
the Schwarzschild metric with the same mass, in the sense that the distance between light cones remains bounded. In order to study the precise asymptotic behavior, Lindblad \cite{L17} straightened out the light cones by a change of coordinates
\beq
\widetilde{t}=t,\qquad \widetilde{x}=\rs x/r, \qquad \text{where} \quad \rs=r+\widetilde{\chi}(\tfrac{r}{t+1}) M\ln{r},
\eq
where $\widetilde{\chi}$ is as in \eqref{eq:m0def}. Therefore,
\beq
\widetilde{x}^a=x^a+M\omega^a \widetilde{\chi} \ln{r}.
\eq
 If
\beq\label{eq:Aexpression}
A^{a}_{\alpha}=\frac{\partial \widetilde{x}^{\,a}}{\partial x^\alpha}
=\delta^{a}_{\alpha}+Mr^{-1}\ln{r} \, \chi^a_\alpha(\tfrac{r}{t+1},\omega)
+Mr^{-1} \widetilde{\chi}^a_\alpha(\tfrac{r}{t+1},\omega),
\eq
then in the new coordinates $m_0$ in \eqref{eq:m0def} satisfies
\beq
\widetilde{m}_0^{ab}=A^{a}_{\alpha} A^{b}_{\beta} m_0^{\alpha\beta}\sim \widehat{m}^{ab},
\eq
where $\widehat{m}^{ab}$ is the inverse of the Minkowski metric in the new coordinates,
\beq
\widehat{m} = -d\wt{t}^2 + {\sum}_i d\wt{x}_i^2.
\eq
In the new coordinates
\begin{equation}\label{eq:tildetotildewaveop}
A^{\alpha}_{ a}A^{\beta}_{ b}\pa_\alpha\pa_\beta
=A^{\alpha}_{ a}A^{\beta}_{ b}\pa_\alpha A^{c}_{\beta} \wpa_{\,c}
=\wpa_a\wpa_{\,b} -\widehat{\Gamma}_{ab}^c\wpa_c,
\eq
where
\beq\label{eq:gammahatdecay}
\widehat{\Gamma}^{c}_{\!ab\!} =-A^{\alpha}_{ a}A^{\beta}_{ b}\pa_\alpha A^{c}_{\beta}
=-A^{\alpha}_{ a} A^{\beta}_{ b}
\pa_\alpha\pa_\beta \widetilde{x}^{\,c}
=Mr^{-2}\ln{r} \, \chi^c_{ab}\big(\tfrac{r}{t+1},\omega,\!\tfrac{M\!}{r},\!
\tfrac{M\ln{r}\!\!}{r}\,\big)
+Mr^{-2} \widetilde{\chi}^c_{ab}\big(\tfrac{r}{t+1},\omega,
\!\tfrac{M\!}{r},\!\tfrac{M\ln{r}\!\!}{r}\,\big),
\eq
are the Christoffel symbols for the metric
$\widetilde{m}_{ab}\!=m_{\alpha\beta} A_a^{\alpha}  A_b^{\beta}$ in the new coordinates:
\begin{equation}\label{eq:christoffelm}
\widehat{\Gamma}^{c}_{ab}=\widetilde{m}^{cd}
\big(\wpa_a\widetilde{m}_{bd}
+\wpa_b\widetilde{m}_{ad}-\wpa_d\widetilde{m}_{ab}\big)/2.
\end{equation}
Here $A_a^{\alpha}$ is the inverse of $A^a_{\alpha}$. In fact
$\widetilde{x}(x)$ is a wave map into $(\mathbf{R}^{1+3},\widetilde{m})$.

\subsubsection{Change of coordinates and the generalized wave coordinate condition}
If $\widetilde{g}^{ab}\!\!=\!A^a_{\alpha}A^b_{\beta} g^{\alpha\beta}$ and $\widetilde{\Gamma}_{\!ab}^c$ are the Christoffel symbols of $\widetilde{g}$ then
since the geometric wave operator is invariant
\beq
\widetilde{g}^{\,ab}\wpa_a\wpa_b -\widetilde{g}^{\,ab}\widetilde{\Gamma}_{ab}^c\wpa_c
=\Box_{\widetilde{g}}=\Box_g =
g^{\alpha\beta}\pa_\alpha\pa_\beta -g^{\alpha\beta}
\Gamma_{\alpha\beta}^\delta\pa_\delta.
\eq
This combined with \eqref{eq:tildetotildewaveop} gives, with $\widehat{\Gamma}_{ab}^c$ given by \eqref{eq:christoffelm},
\begin{equation}\label{eq:wavetildetowavecoord}
\widetilde{g}^{\,ab}\widetilde{\Gamma}_{ab}^c
=\widetilde{g}^{\,ab}\widehat{\Gamma}_{ab}^c
 + g^{\alpha\beta}\Gamma_{\alpha\beta}^\delta A^{c}_{\,\,\delta}.
\end{equation}
If $\Gamma_{\alpha\beta}^\delta$ are the Christoffel symbols of $g$ in harmonic coordinates, the last term vanishes and \eqref{eq:wavetildetowavecoord} is the
generalized wave coordinate condition. Moreover
\beq
\widetilde{g}^{\,ab}\widehat{\Gamma}_{ab}^c
=\Box_g \widetilde{x}^c=g^{\alpha\beta}\pa_\alpha\pa_\beta (\widetilde{x}^{\,c }-x^c)=m_0^{\alpha\beta}\pa_\alpha\pa_\beta( \widetilde{x}^{\,c }-x^c)
+H_1^{\alpha\beta}\pa_\alpha\pa_\beta (\widetilde{x}^{\,c }-x^c)
\eq

Furthermore,
\beq\label{eq:approxtildegm0H1}
\widetilde{g}^{\,ab}=\widetilde{m}_0^{\,ab}
+\widetilde{H}_1^{\,ab}.
\eq
We shall see in Section \ref{subsec:ASmetric}
that
\beq\label{eq:approxtildem0}
\widetilde{m}_0^{\,ab}\sim  (1\!+\!\tfrac{M}{r})\widehat{m}^{\,ab},
\eq
 where
$\widehat{m}^{\,ab}$ is the Minkowski metric. In fact we shall see already a heuristic argument in
 Section \ref{sec:AsymSchwarzWaveOp}, or for the linearized equations at the end of Section \ref{sec:coordchange}.

\subsubsection{The asymptotically Schwarzschild wave operator in the new coordinates}\label{sec:AsymSchwarzWaveOp}
By \eqref{eq:tildetotildewaveop} applied to $m_0$, and recalling \eqref{eq:Boxrm0insphericalcoord},
\beq\label{eq:BoxStar}
{\Box}^{\,\widetilde{m}_0}=\Box^{\,{m}_0}
-\widetilde{m}_0^{ab}\widehat{\Gamma}_{ab}^c\wpa_c.
\eq
Let $\Phi(t,\rs,\omega)=r\phi(t,r\omega)$, where $\rs=r+M \widetilde{\chi}\big(\tfrac{r}{1+t}\big) \ln{r}$.
For $\tfrac{r}{1+t}>\tfrac{3}{4}$ we have $\widetilde{\chi}\big(\tfrac{r}{1+t}\big)=1$ and
\begin{equation*}
r\Box^{\,m_0}\phi- r\big(1\!-\!\tfrac{M}{r}\big)
r^{-2}\!{\triangle}_\omega \phi
=\Big(\!-\pa_t^2\!\!+\pa_r^2-\!\tfrac{M}{r}(\pa_t^2\!\!+\pa_r^2)\Big)
\Phi
= -(1\!+\!\tfrac{M}{r})\Phi_{tt}\!
+(1\!-\!\tfrac{M}{r})\Big(\big(\frac{ d\rs\!\!}{dr}\big)^{\!2}\!\Phi_{\rs\rs}\!
+\frac{ d^{\,2} \rs\!\!\!\!\!}{dr^2\!\!}\,\Phi_{\rs}\!\Big).
\end{equation*}
Hence in the new coordinates
\begin{equation}\label{eq:waveeqradial}
r\Box^{\,m_0}\phi
=(1\!+\!\tfrac{M}{r})\big(-\!\Phi_{tt}\!+\Phi_{\rs\rs}\big)
+r\big(1\!-\!\tfrac{M}{r}\big)
r^{-2}\!\triangle_\omega \Phi
-\big(\!\tfrac{M}{r}\big)^2
\big((1\!+\!\tfrac{M}{r})\Phi_{\rs\rs}\!+(1\!-\!\tfrac{M}{r})\Phi_{\rs}\big),
\quad\text{when}\quad \tfrac{r}{1+t}>\tfrac{3}{4}.
\end{equation}
Hence  $\widetilde{\Box}^{\,\widetilde{m}_0}$
approximately becomes $(1\!+\!\tfrac{M}{r})$ times the
Minkowski wave operator in the new coordinates:
\beq\label{eq:waveoperatorstarsphericalcoord}
\Box^*\phi\equiv \Box^{\,\widehat{m}}\phi
=\frac{1}{\rs}\big(\!-\pa_t^2\!\!+\pa_{\rs}^2\big)(r^* \phi)
+\frac{1}{{\rs}^2}\triangle_\omega\phi .
\eq

\subsection{Estimates on error terms arising from the change in coordinates}\label{ChangeinCoordBounds}
\subsubsection{The modified fields and commutation properties}
 We recall that by $\wt{Z}$ and $\wt{X}$ we mean commutator fields in
\beq\label{eq:ModifiedVectorFields}
{\wpa}_a, \qquad
\widetilde{\Omega}_{ij}=\widetilde{x}^i\wpa_{j}-\widetilde{x}^j\wpa_{i}, \qquad
\widetilde{\Omega}_{i0}=\widetilde{x}^i\wpa_t+t\wpa_{i},\qquad
\widetilde{S}=t\wpa_t+\widetilde{x}^i\wpa_{i},
\eq
where $a$ ranges from 0 to 3 and $i,j$ range from 1 to 3. These fields $\widetilde{Z}$ commute nicely with each other:
\beq
\label{FieldComms}
[\wt{Z}_1, \wt{Z}_2] = {\sum}_{|I|=1}c_{12I}\wt{Z}^I,\qquad [\wt{Z}_1, \wpa_a] = {\sum}_{b=0}^3 c_{1a}^b\wpa_b, \qquad [\wpa_a, \wpa_b] = 0,
\eq
where all $c_{12I}, c_{1a}^b$ are $-1$, $0$, or $1$.

\subsubsection{The change of coordinates}
We now show estimates on error terms arising from the change in coordinates. The expansion \eqref{eq:Aexpression} follows directly from the definition
\begin{subequations}\label{id:Aexpansion}
\begin{align}
\pa_\alpha \widetilde{x}^0 &= \delta_\alpha^0,\label{id:A00expansion} \\
\pa_0 \widetilde{x}^i &= \delta_0^i - M r^{-1}\ln r\,  \omega^i\left(\tfrac{r}{t+1}\right)^2\widetilde{\chi}',\label{id:A0iexpansion}\\
\pa_j\widetilde{x}^i &= \delta_j^i + Mr^{-1}\ln r \left((\delta_j^i - \omega_j\omega^i)\widetilde{\chi} + \omega_j\omega^i\tfrac{r}{t+1}\widetilde{\chi}'\right) + Mr^{-1}\omega_j\omega^i\widetilde{\chi}.\label{id:Aijexpansion},
\end{align}
\end{subequations}
where $\wt\chi' = \wt{\chi}'\left(\tfrac{r}{t+1}\right)$, $\wt\chi = \wt{\chi}\left(\tfrac{r}{t+1}\right)$ are supported when their arguments are in $[1/4, 3/4]$ and $[1/4, \infty)$ respectively. We introduce a lemma which will allow us to handle these terms efficiently.
\begin{lemma}\label{lem:StaPhi}
Take constants $M_0,N_0 \geq 0$. Let $\psi$ be a finite sum of the form
\begin{equation}\label{def:chiexpansion}
\psi = \sum_{\substack{m\leq M_0, n \geq N_0}}\frac{(\ln r)^m}{r^n}\chi_{mn}\Big(\frac{r}{t+1},\frac{M}{r}, \frac{M\ln r}{r}, \frac{t}{r},  \omega\Big),
\end{equation}
where $\chi_{mn}$ are smooth functions supported when $\frac{r}{t+1}\geq \frac14$ such that $\partial_y\chi_{mn}(y, \cdot)$ is supported in the region $y\leq\frac34$. Then, for every $k > 0$, there exists constant $C_k$, $\widetilde{C}_k$ such that the bounds
\beq\label{def:PsiLogBound}
\sum_{|\alpha|= k}|\partial^\alpha\psi| \leq C_k \frac{\langle \ln r\rangle^{M_0}}{(t+r)^{N_0+k}}, \qquad \sum_{|I|\leq k}|Z^I\psi| \leq \widetilde{C}_k \frac{\langle \ln r\rangle^{M_0}}{(t+r)^{N_0}}
\eq
holds globally.
\end{lemma}
\begin{proof}
It suffices to prove the first bound in \eqref{def:PsiLogBound}, as the second bound follows directly from expanding out the fields $Z$. We first note that $t+1 \leq 4r$ in the support of $\chi_{mn}$, so $r > \frac14$ and $r \leq t+r \leq 5r$. Consequently, $\chi_{mn}$ attains its entire range on a compact set, so $\chi_{mn}$ is uniformly bounded.

The result for higher derivatives follows from induction: applying any derivative to $\psi$ introduces a factor of $\frac1r$ and either keeps the logarithmic power the same or decreases it by 1.
\end{proof}

\begin{corollary}\label{cor:Aderivatives}
For all multiindices $\beta$, $I$, the following bound holds on derivatives of $A_\alpha^a$:
\begin{equation}\label{est:AderivativeBound}
|\partial^\beta (A^{a}_{\alpha}-\delta^{a}_{\alpha})| + |\partial^\beta (A_{a}^{\alpha}-\delta_{a}^{\alpha})| \lesssim_\beta
\frac{M\langle\ln r\rangle}{(t+r)^{1+|\beta|}}, \quad |Z^I(A^{a}_{\alpha}-\delta^{a}_{\alpha})| + |Z^I(A_{a}^{\alpha}-\delta_{a}^{\alpha})| \lesssim
\frac{M\langle\ln r\rangle}{t+r},
\end{equation}
where $Z^I$ is a collection of the fields $\{\partial, \Omega, S\}$.
If $\frac{r}{t+1}\leq\frac14$, this is equal to 0.
\end{corollary}
\begin{proof} The inequality for $\partial^\beta(A^{a}_{\alpha} - \delta^a_{\alpha})$ follows from applying Lemma \ref{lem:StaPhi} to the identity \eqref{eq:Aexpression}, noting that the differentiation properties required of $\chi$ terms follow from the formal definition. To show this for $\partial^\beta (A_{a}^{\alpha}-\delta_{a}^{\alpha})$, for the case $|\beta| = 0$ we use the smallness assumption on $M$ implies smoothness for coefficients of $
(I-M)^{-1}.
$
For higher derivatives we differentiate
$ A_a^{\alpha} A_{\alpha}^b = \delta_a^b, $
and prove by induction.
\end{proof}

\begin{corollary}\label{cor:BoundPartial}
If $\psi$ is a finite sum of the form \eqref{def:chiexpansion}, then $\psi$ satisfies
\beq\label{def:PsiLogBoundCurved}
\sum_{|\alpha|= k}|\wpa^\alpha\psi| \leq C_k \frac{\langle \ln r\rangle^{M_0}}{(t+r)^{N_0+k}}, \qquad \sum_{|I|\leq k}|\wt{Z}^I\psi| \leq \widetilde{C}_k \frac{\langle \ln r\rangle^{M_0}}{(t+r)^{N_0}}.
\eq
Consequently,
\begin{equation}\label{est:AderivativeBoundCurved}
|\wpa^\beta (A^{a}_{\alpha}-\delta^{a}_{\alpha})| + |\wpa^\beta (A_{a}^{\alpha}-\delta_{a}^{\alpha})| \lesssim_\beta
\frac{M\langle\ln r\rangle}{(t+r)^{1+|\beta|}}, \qquad |\wt{Z}^I(A^{a}_{\alpha}-\delta^{a}_{\alpha})| + |\wt{Z}^I(A_{a}^{\alpha}-\delta_{a}^{\alpha})| \lesssim
\frac{M\langle\ln r\rangle}{t+r}.
\end{equation}
\end{corollary}
\begin{proof}
The bound \eqref{def:PsiLogBoundCurved} follows from repeated use of $\wpa_a = \delta_a^\alpha\partial_\alpha + (A_a^{\beta} - \delta_a^\beta)\partial_\beta$ as well as Corollary \ref{cor:Aderivatives}. The bound \eqref{est:AderivativeBoundCurved} follows from applying the bound \eqref{def:PsiLogBoundCurved} to the identity \eqref{eq:Aexpression}.
\end{proof}
We can use these to bound components of the inverse metric,
\beq
\widetilde{g}^{ab}=\widetilde{m}^{ab}+\widetilde{H}^{ab}.
\eq
\begin{lemma}\label{lem:ChristoffelBounds} Let $\widehat{\Gamma}$ be the Christoffel symbols of
$\widetilde{m}_{ab}=m_{\alpha\beta}A^\alpha_{a} A^\beta_{b}$. Then
\beq\label{est:tildeGamma}
|\wpa^\beta\widehat{\Gamma}|\les \frac{ M (\ln{(1+r)}+1)}{(1+r+t)^{2+|\beta|}}, \qquad
|\Lie_{\widetilde{X}}^I \widehat{\Gamma}|\les \frac{ M (\ln{(1+r)}+1)}{(1+r+t)^2}
\eq
Moreover if $\widetilde{m}^{ab}=A^a_{\alpha} A^b_{\beta}m^{\alpha\beta}$ is the inverse of the
Minkowski metric, written in the new coordinates then
 \beq\label{est:LieMtilde}
|\wpa^\beta(\widetilde{m}^{ab} - \widehat{m}^{ab})|\les \frac{ M (\ln{(1+r)}+1)}{(1+r+t)^{1+|\beta|}}, \qquad
|\Lie_{\widetilde{X}}^I(\widetilde{m}^{ab} - \widehat{m}^{ab})|\les \frac{ M (\ln{(1+r)}+1)}{1+r+t}
\eq
\end{lemma}
\begin{proof}
First, we consider \eqref{est:tildeGamma}, for which it suffices to replace the Lie derivatives with standard derivatives in each component. We write
\begin{align}
\widehat{\Gamma}^a_{bc} &= \frac12\widetilde{m}^{ad}(\wpa_b \widetilde{m}_{dc} + \wpa_c\widetilde{m}_{bd} - \wpa_d\widetilde{m}_{bc}) \\
&= \frac12(A^a_{\alpha}A^d_{\delta}m^{\alpha\delta})(A_b^{\beta}\partial_\beta (A_d^{\delta} A_c^{\gamma}m_{\delta\gamma}) + A_c^{\gamma}\partial_\beta (A_b^{\beta} A_d^{\delta}m_{\beta\delta}) - A_d^{\delta}\partial_\beta (A_b^{\beta} A_c^{\gamma}m_{\beta\gamma})).
\end{align}
The result follows from Corollary \ref{cor:BoundPartial}.
To prove \eqref{est:LieMtilde}, we use the decomposition
\[
\wpa_c(\widetilde{m}^{ab}-\widehat{m}^{ab}) = \wpa_c(m^{\alpha\beta}(A_\alpha^a (A_\beta^b - \delta_\beta^b) + ( A_\alpha^a - \delta_\alpha^a)\delta_\beta^b)),
\]
and again apply Corollary \ref{cor:BoundPartial}.
\end{proof}

\subsection{The null structure of the metric in asymptotically Schwarzschild coordinates}\label{sec:AsymptoticSchwarzschildNullStructure}
\subsubsection{The null frame in asymptotically Schwarzschild coordinates}
We take the null frame with respect to $\widehat{m}$,
\beq\label{def:ModifiedNullFrame}
\widetilde{L} = \partial_{t^*} + \partial_{r^*}\qquad \widetilde{\underline{L}}
 = \partial_{t^*} - \partial_{r^*} \qquad \widetilde{S}_i = \frac{r}{r^*}S_i.
\eq
Let derivatives tangential to the outgoing light
cones in these coordinates be denoted by
$\overline{\widetilde{\pa}}\in\widetilde{\mathcal{T}}
=\{\widetilde{L},\widetilde{S}_1,\widetilde{S}_2\}$.

All the estimates in Section \ref{sec:nullframe} hold with $\pa$ replaced by $\wpa$, $\overline{\pa}$ by
$\overline{\wpa}$, $\pa_q$ by $\pa_{q^*}$, the frame $\{\underline{L},L,S_1,S_2\}$ replaced by
$\{\widetilde{\underline{L}},\widetilde{L},\widetilde{S}_1,\widetilde{S}_2\}$
and the vector fields \eqref{eq:MinkowskiVectorFields} replaced by  \eqref{eq:ModifiedVectorFields}.

We now estimate the difference between the fields $Z^I$ and $\wt{Z}^I$ applied to an arbitrary function:

\begin{lemma}\label{lemma:starwaveeq} We have
 \begin{align}
\big|_{\,}\pa^\alpha Z^I(\wpa_{\mu}-\pa_\mu )\phi\big| &\les
\frac{M\ln{|1+t+r|}}{1+t+r} {\sum}_{|J|\leq|I|,\, |\beta|\leq
|\alpha|} |\pa \pa^\beta Z^J\phi|.\label{eq:partialstardifference}
 \end{align}
Moreover,
 \begin{align}
\big(1+\frac{M\ln{|1+r|}}{1+|\,q^*|}\big)^{-1}
 &\les \frac{1+|\,q^*|}{1+|\,q|}
 \les \big(1+\frac{M\ln{|1+r|}}{1+|\,q|}\big), \label{eq:qstarq}\\
 \!\!\!\!
 C^{-1}\big(1\!+\!\frac{M\ln{|1\!\!+\!r|}}{1\!+\!|\,q^*|}\big)^{\shortminus k}
 {\sum}_{|I|\leq k}|Z^I\!\phi|
 &\leq\! {\sum}_{|I|\leq k} |{\widetilde{Z}}^I \!\phi|
 \leq C\big(1\!+\!\frac{M\ln{|1\!\!+\!r|}}{1\!+\!|\,q|}\big)^k
 {\sum}_{|I|\leq k}|Z^I\!\phi|\label{eq:ZstarZ}.
 \end{align}
\end{lemma}
\begin{proof}
We first prove \eqref{eq:partialstardifference}. \!We write the left hand side as $\partial^\alpha Z^I((A_\mu^\beta \!- \delta_\mu^\beta)\partial_\beta)$, from which we get
\begin{align}
\big|_{\,}\pa^\alpha Z^I(\wpa_{\mu}-\pa_\mu )\phi\big| &\les {\sum}_{\substack{|I_1|+|I_2|\leq|I| \\ |\alpha_1|+|\alpha_2|\leq|\alpha|}}|\partial^{\alpha_1}Z^{I_1}(A_\mu^\beta - \delta_\mu^\beta)||\partial^{\alpha_2}Z^{I_2}\pa_\beta\phi|.
\end{align}
The result follows from Lemma \ref{lem:StaPhi} and Corollary \ref{cor:Aderivatives}.

To prove the rightmost inequality of \eqref{eq:qstarq} we note that $1+|q^*| \leq 1+|q|+M\wt\chi\ln r$ and divide by $1+|q|$. Taking $1+|q| \leq 1+|q^*|+M\wt\chi\ln r$ and dividing by $1+|q^*|$ gives the left inequality.

We now look at the right hand inequality in \eqref{eq:ZstarZ}. In the region $\frac{r}{t+1}\leq \frac14$ this is trivial. For $\frac{1}{4}\leq \frac{r}{t+1}\leq \frac34$, we don't need to distinguish derivatives. and the inequality follows from expanding out and applying Corollary\ref{cor:Aderivatives}. We now consider the region $\frac34 \leq \frac{r}{t+1}\leq 1$ (the region where it is greater than 1 follows from an almost identical argument). Here, we must take the null decomposition into consideration.
For any set of vector fields $\wt{X}^I$, we can commute fields and take the decompositions $\wt{S} = \frac12(\underline{u}^*\wt{L} + u^*\wt{\underline{L}}), \wt{\Omega}_{0i} = \frac12\omega_i(\underline{u}^*\wt{L} - u^*\wt{\underline{L}}) + \frac{t}{r^*}(r^*\widetilde{\slashed\pa}_i)$, which gives us
\begin{equation}
\widetilde{X}^I = \sum_{j+k+|\alpha|+|\gamma|\leq k}\, \, \prod_{1\leq i<j\leq 3} f^I_{jk\alpha\gamma}(\omega, t^*/r^*)(\underline{u}^*\wt{L})^j
(u^*\wt{\underline{L}})^k(r^*\wt{\slashed\pa}_i)^{\gamma_i}\wpa^\alpha.
\end{equation}
where $\gamma$ is a multiindex with components $\gamma_{ij}$, and $\partial^\alpha$ is a standard multiindex. We can replace $1+u^*$ with $1+u^*$ and $\underline{u}^*$ with $1+\underline{u}^*$, which changes the values of $f$. We can write
\begin{align}
(1+u)\underline{L} &= \Big(1+\frac{M\chi\ln r}{1+u^*}\Big)(1+u^*)\wt{\underline{L}} - (1+u)(\frac{M}{r}\partial_{r^*}),\\
(1+\underline{u})L&=  \Big(1 - \frac{M\ln r}{1+\underline{u}^*}\Big)(1+\underline{u}^*)\widetilde{{L}} +(1+\underline{u})(\frac{M}{r}\partial_{r^*}),\\
(1+u^*)\underline{\widetilde{L}} &= \Big(1-\frac{M\chi\ln r}{1+u}\Big)(1+u)\underline{L} + (1+u^*)(\frac{M}{r}\partial_{r^*}),\\
(1+\underline{u}^*)\wt{L} &= \Big(1 + \frac{M\ln r}{1+\underline{u}}\Big)(1+\underline{u}) L - (1+\underline{u}^*)\frac{M}{r}\partial_{r^*}
\end{align}
Our result follows from iterating this and applying Corollary \ref{cor:Aderivatives} when needed. The right hand inequality, and the exterior case $r^*\geq t$, follow similarly.
\end{proof}

\subsubsection{The asymptotically Schwarzschild metric in the new coordinates}\label{subsec:ASmetric}
We recall that
\beq\widetilde{m}_0^{ab}=(m^{\alpha\beta} - \frac{M\chi}{r}\delta^{\alpha\beta})A_\alpha^{a} A_\beta^{b}\eq
In asymptotically Schwarzschild coordinates, we expect $g_{\wt{L}\wt{L}}$ to decay faster than $r^{-1}$ along the light cone. We first show bounds on components $\wt{m}_0$, but before that we must define certain quantities.
Let $\widehat{\delta}^{ab}$ be the Euclidean metric on the spatial coordinates in this frame, with $\widehat{\delta}^{ab}\! = \! 1$ if $a\! =\! b\! > \! 0$ and $0$ otherwise, and let $\slashed{S}^{ab}\!$ be the tangential part, given by
\[
\slashed{S}^{00} = \slashed{S}^{0i} = \slashed{S}^{i0} = 0, \qquad
\slashed{S}^{ij} = \delta^{ij} - \omega^i\omega^j
\]
where as usual $i,j$ range from $1$ to $3$. In the inverse null decomposition neither $\widehat{m}$ nor $\slashed{S}$ contains a term like $\wt{L}^a\wt{L}^b$. We now decompose $\wt{m}_0$ as a weighted sum of these terms, plus a remainder which decays rapidly along the light cone in all components.
\begin{lemma}\label{lem:m0approx}
We can decompose
 \beq\label{eq:metric0decomposition}
 \widetilde{m}_0^{ab}= \kappa_{0\,}\widehat{m}^{ab}\!
 +\kappa_{1\,} \slashed{S}^{\,ab}\!+\kappa_{2\,}\omega^a\omega^b\!+\kappa_{3\,} i_+^{\, ab}\!,
 \eq
 where
 \beq
 \kappa_0=1+\tfrac{M\widetilde\chi}{r},\quad
 \kappa_1\!=(1{}_{\!}+\tfrac{M\widetilde\chi\ln{r}}{r})^2
 (1-\tfrac{M\widetilde\chi}{r})-(1{}_{\!}+\tfrac{M\widetilde\chi}{r})\!
 =2\tfrac{M\ln{r}\!\!}{r} \chi_1
 \big(\tfrac{r}{t+1},\! \tfrac{M\ln{r}\!\!}{r}, \!\tfrac{M}{r}\big)
 -2\widetilde{\chi}\tfrac{M}{r},
 \eq
 and
 \beq
 \kappa_2=(\tfrac{M\widetilde\chi}{r})^2(1{}_{\!}+\tfrac{M\widetilde\chi}{r}),
 \quad
i_+^{\, ab}=\widetilde{\chi}^{\,\prime}(\tfrac{r}{t+1})
 \chi^{ab}\big(\tfrac{r}{t+1}, \tfrac{M\ln{r}}{r}, \tfrac{M}{r},\omega\big),
 \qquad\text{and} \qquad \kappa_3=\tfrac{M\ln{r}}{r},
 \eq
 for some smooth functions $\chi_1,$ $\chi^{ab}$\!\!. \!We have the following estimates on Lie derivatives of $\kappa_{2\,} \widehat{\delta}^{\,ab}\!$ and $\kappa_{3\,} i_+^{\, ab}\!$:
\begin{align}
{\sum}_{a, b}{\sum}_{\substack{|I|\leq k}}|(\Lie_{\widetilde{X}}^I (\kappa_{2\,} \widehat{\delta}))^{ab}|
&\lesssim \frac{M^2(\ln(1+r))^2}{\langle t+r^*\rangle^2},\\
{\sum}_{a, b}{\sum}_{\substack{|I|\leq k}}|(\Lie_{\widetilde{X}}^I (\kappa_{3\,} i_+))^{ab}|
&\lesssim \frac{M\ln(1+r)}{\langle t+r^*\rangle}.
\end{align}
The support of these are contained in the supports of $\wt\chi$ and $\wt\chi'$ respectively.
\end{lemma}
\begin{proof}
For the bounds on Lie derivatives of $\kappa_{2\,} \widehat{\delta}^{\,ab}\!$ and $\kappa_{3\,} i_+^{\, ab}$ it suffices to prove the bounds on derivatives of each component, for which we can use Lemma \ref{lem:StaPhi} and Corollary \ref{cor:BoundPartial}.

We expand $\widetilde{m}_0^{ab} = m_0^{\alpha\beta}A_{\alpha}^a A_{\beta}^b$ and use \eqref{id:A00expansion}, \eqref{id:A0iexpansion}, and \eqref{id:Aijexpansion} to get
\begin{equation*}
\widetilde{m}_0^{00} = \bigtwo(-1-\tfrac{M\widetilde{\chi}}{r}\bigtwo),\qquad
\widetilde{m}_0^{0i} = -\tfrac{1}{r}\bigtwo(\tfrac{r}{t+1}\bigtwo)^2
\bigtwo(-1-\tfrac{M\widetilde{\chi}}{r}\bigtwo)\omega_i\wt{\chi}'M\ln r.
\end{equation*}
The term $\widetilde{m}_0^{00}$ is equal to the corresponding component of $\kappa_{0\,}\widehat{m}^{ab}\!$, and $\widetilde{m}_0^{0i}$ appears in $\kappa_{3\,} i_+$. The component bounds for this term follow directly from Lemma \ref{lem:StaPhi} and Corollary \ref{cor:BoundPartial}. For $\widetilde{m}_0^{ij}$, the process is slightly longer. Writing $\widetilde{m}_0^{ij} = m_0^{\alpha\beta}A_\alpha^iA_\beta^j$ we see that if $\alpha$ or $\beta$ are 0, then both must be 0 and we can absorb this term in $\kappa_{3\,} i_+$ using Lemma \ref{lem:StaPhi} and Corollary \ref{cor:BoundPartial}. If neither is 0 we write
\begin{equation}
\widetilde{m}_0^{ij} - m_0^{00}A_0^i A_0^j= {\sum}_{k, l=1}^3 \bigtwo(1 - \tfrac{M\widetilde\chi}{r}\bigtwo)\delta^{kl}A_k^i A_l^j,
\end{equation}
Equation \eqref{id:Aijexpansion} is equivalent to
\beq
\pa_j\widetilde{x}^i = \delta_j^i\bigtwo(1 + \tfrac{M\widetilde\chi}{r}\bigtwo) + \bigtwo(\tfrac{(M\ln r - M)\widetilde\chi}{r}\bigtwo)\bigtwo(\delta_j^i - \omega_j\omega^i\bigtwo) + \tfrac{M\ln r}{r}\omega_j\omega^i\tfrac{r}{t+1}\widetilde{\chi}' .
\eq

Terms containing $\chi'$ can be absorbed into $\kappa_{3\,} i_+$ using Lemma \ref{lem:StaPhi} and Corollary \ref{cor:BoundPartial}. The remaining terms are
\beq\label{m0space}
\bigtwo(1-\tfrac{M\widetilde\chi}{r}\bigtwo)\bigtwo(1 + \tfrac{M\widetilde\chi}{r}\bigtwo)^2\delta^{ij} + \bigtwo(1-\tfrac{M\widetilde\chi}{r}\bigtwo)\Big(2\bigtwo(1 + \tfrac{M\widetilde\chi}{r}\bigtwo)\bigtwo(\tfrac{(M\ln r - M)\widetilde\chi}{r}\bigtwo) + \bigtwo(\tfrac{(M\ln r - M)\widetilde\chi}{r}\bigtwo)^2\Big)\bigtwo(\delta^{ij} - \omega^i\omega^j\bigtwo).
\eq
We can expand this, using the identity
\beq
2\bigtwo(1 + \tfrac{M\widetilde\chi}{r}\bigtwo)\bigtwo(\tfrac{(M\ln r - M)\widetilde\chi}{r}\bigtwo) + \bigtwo(\tfrac{(M\ln r - M)\widetilde\chi}{r}\bigtwo)^2= (1{}_{\!}+\tfrac{M\widetilde\chi\ln{r}}{r})^2
 -(1{}_{\!}+\tfrac{M\widetilde\chi}{r})^2.
\eq
For $i, j > 0$ we now that $\slashed{S}^{ij} - \delta^{ij} = -\omega^i\omega^j$ so componentwise \eqref{m0space} is equal to
\beq
 \kappa_{0\,}\widehat{m}^{ij}\!+\kappa_{1\,} \slashed{S}^{\,ij}\!+\kappa_{2\,}\omega^i\omega^j. \tag*{\qedhere}
\eq
\end{proof}

We now show that $\kappa_{0\,}\widehat{m}^{ij}$, $\kappa_{1\,} \slashed{S}^{\,ij}$, $\kappa_{2\,}\omega^i\omega^j$
behave nicely under repeated Lie differentiation. First, we consider $\widehat{m}$.
\begin{equation}
\Lie_{\widetilde{X}}^I\left(\kappa_{0\,}\widehat{m}^{ij}\right)^{ab}
 = \kappa_0^I\widehat{m}^{ab},
\end{equation}
for some bounded function $\kappa_0^I$.
This follows from expanding and applying Lemma \ref{lem:StaPhi} and Corollary \ref{cor:BoundPartial} whenever a derivative falls on $\kappa_0$, or noting that $\Lie_{\wt{X}}\widehat{m}^{-1} = c_{\wt{X}}\widehat{m}^{-1}$ otherwise. Similar reasoning gives us the estimate
\beq
|\widehat{\Lie}_{\widetilde{X}}^I\left(\kappa_{0\,}\widehat{m}\right)^{ab}|= \frac{1}{r}\chi_0(\tfrac{r}{t+1}, \tfrac{M\ln r}{r}, \tfrac{M}{r}),
\eq
where $\chi_0$ is a bounded smooth function which is 0 when $\frac{r}{t+1}<\frac14$.
 In order to deal with the term $\kappa_{1\,} \slashed{S}^{\,ij}$, we use the identity
\begin{equation}
\slashed{S}^{ab} = r^{-2}{\sum}_{i < j}\widetilde\Omega^{a}_{ij}\widetilde\Omega^{b}_{ij}.
\end{equation}
We now expand. If Lie derivatives fall on $r^{-2}\kappa_1$, we can again apply Lemma \ref{lem:StaPhi} and Corollary \ref{cor:BoundPartial}. If derivatives fall on $\widetilde\Omega^{a}_{ij}\widetilde\Omega^{b}_{ij}$, the commutator identities \eqref{FieldComms} and the Leibniz rule gives us the following decomposition, which holds when $\frac{r}{t+1}>1/4$: for $\widetilde{U}, \widetilde{V} \in \{\widetilde{L}, \widetilde{\underline{L}}, \wt{S_1}, \wt{S_2}\}$, we have the pointwise estimates
\begin{subequations}
\begin{align}
|(\Lie_{\widetilde{X}}^I(\wt{\Omega}^a_{ij}
\wt{\Omega}^b_{ij}))_{\widetilde{U}\widetilde{V}}|& \lesssim \langle t+r^* \rangle^2, \\
|(\Lie_{\widetilde{X}}^I(\wt{\Omega}^a_{ij}
\wt{\Omega}^b_{ij}))_{\widetilde{U}\widetilde{L}}|&\lesssim \langle t-r^*\rangle\langle t+r^* \rangle,\\
|(\Lie_{\widetilde{X}}^I(\wt{\Omega}^a_{ij}
\wt{\Omega}^b_{ij}))_{\widetilde{L}\widetilde{L}}|&\lesssim \langle t-r^*\rangle^2.
\end{align}
\end{subequations}
Combining this with the bounds $\wt{X}_{\wt{U}}\lesssim \langle t+r^*\rangle, \wt{X}_{\wt{L}} \lesssim \langle t-r^*\rangle$ gives the following:
\begin{lemma}\label{lem:m0decomp}
Given the inverse metric $\widetilde{m}_0^{ab}$ under the change of coordinates, we can decompose
\begin{equation}
(\Lie_{\widetilde{X}}^I\widetilde{m}_0)^{ab}
= \kappa_0^I\widehat{m}^{ab} + \slashed{S}^{I\,ab} +
(\Lie_{\widetilde{X}}^I R)^{ab},
\end{equation}
where
\begin{subequations}
\begin{align}
\kappa_0^I &\lesssim 1 \\
|\slashed{S}^{I}_{\wt{U}\wt{V}}| &\lesssim \tfrac{M\ln(1 + \langle t + r^*\rangle)}{\langle{t+r^*}\rangle}\chi_1\\
|\slashed{S}^{I}_{\wt{L}\wt{U}}| &\lesssim\tfrac{M\langle t-r^*\rangle\ln(1 + \langle t + r^*\rangle)}{\langle{t+r^*}\rangle^2}\chi_1\\
|(\Lie_{\widetilde{X}}^IR)_{\wt{U}\wt{V}}| &\lesssim \tfrac{M^2(\ln(1+r))^2}{{\langle t+r^*\rangle}^2}\chi_1
+ \tfrac{M\widetilde\ln(1+r)}{\langle t+r^*\rangle}\chi'_2,
\end{align}
\end{subequations}
where $\chi_1$ and $\chi'_2$ are smooth functions functions bounded above by 1 with support respectively in the intervals $\frac{r}{t+1}\in [\tfrac14, \infty)$and $[\tfrac14, \tfrac34]$ respectively.
\end{lemma}

\section{Einstein's equations in asymptotically Schwarzschild coordinates}\label{sec:Genralizedwavecoordsection}
We recall from \eqref{eq:EinsteinGenWaveintro} that Einstein's equations in the new coordinates take the form
\beq\label{eq:EinsteinGenWaveGeneral}
\widetilde{\Box}^{\,\widetilde{g}}  \widetilde{h}^1_{cd} =\wt{F}_{cd} (\widetilde{g}) [\wpa\widetilde{h}^1, \wpa \widetilde{h}^1]+\widetilde{\widecheck{T}}_{cd}+\wt{R}_{cd}^{\,error}.
\eq
 In order to obtain this from \eqref{eq:EinsteinWaveintro}, we perform a change of coordinates and subtract off a term $\widetilde{h}^0$ containing the mass. Each step produces error terms, the \emph{covariant} and \emph{mass} error terms respectively, both from the reduced wave operator and from the quadratic terms appearing in $\wt{F}$. These appear in \eqref{eq:EinsteinGenWaveintro} as $R_{cd}^{error}$. We make the expansion
\begin{equation}
\wt{R}_{cd}^{error}=\widetilde{E}_{cd}^{\, mass}+\widetilde{E}_{cd}^{\,cov}+\widetilde{F}_{cd}^{\, mass}\!+\widetilde{F}_{cd}^{\, cov},
\end{equation}
where $\widetilde{E}$ terms are the mass and covariant error terms coming from the reduced wave operator $\widetilde{\Box}^{\,\widetilde{g}}  \widetilde{h}^1_{cd}$, and the $\widetilde{F}$ terms are the mass and covariant error terms coming from the semilinear error terms in $\wt{F}_{cd} (\widetilde{g}) [\wpa\widetilde{h}^1, \wpa \widetilde{h}^1]$. We determine the exact structure of these error terms in Section \ref{sec:covarianterrorsec5}.

\subsection{Einstein's equations in asymptotically Schwarzschild coordinates}\label{sec:cov}
Let $\onabla$ be covariant differentiation with respect to the metric $\widetilde{m}$:
\begin{equation}\label{eq:covariantderivativedefintion}
\onabla_{\!c}  \widetilde{K}^{a_{\!1}\dots a_r}_{b_1\dots b_s}\!
=\wpa_c \widetilde{K}^{a_{\!1}\dots a_r}_{b_{1}\dots b_s}
+\widehat{\Gamma}_{dc}^{a_{\!1}} \widetilde{K}^{d\dots a_r}_{\! b_{\!1}\dots b_s}\!-\cdots
+\widehat{\Gamma}_{dc}^{a_{r}} \widetilde{K}^{a_{\!1}\dots d}_{\! b_1\dots b_s}\!
-\widehat{\Gamma}_{b_1 c}^{d}\widetilde{K}^{a_{\!1}\dots a_r}_{d\dots b_s}\!\!+\cdots
-\widehat{\Gamma}_{b_s c}^{d}\widetilde{K}^{a_{\!1}\dots a_r}_{b_1\dots d}\!\!.
\end{equation}
Note that here $\onabla$ is simply a tool to systematize
the change of variables and does not depend on $g$, but only on
the Christoffel symbols $\widehat{\Gamma}$ of $\widetilde{m}$, given in
\eqref{eq:christoffelm}.
Since $m$ is the Minkowski metric, the tensor transformation law implies
\begin{equation}
\onabla_{\!c}  \widetilde{K}^{a_{\!1}\dots a_r}_{b_1\dots b_s}\! = A_c^\gamma \!A^{a_1}_{\alpha_1}\!\hdots A^{a_r}_{\alpha_r}\!A_{b_1}^{\beta_1}\!\hdots A_{b_s}^{\beta_s}\partial_\gamma K^{\alpha_{\!1}\dots \alpha_r}_{\beta_1\dots \beta_s}.\!
\end{equation}

Let $\widetilde{H}^{cd}\!
=A^c_{\gamma}A^d_{\delta}H^{\!\gamma\delta}\!\!$.
Since
$\pa_\alpha H^{\alpha\beta}\!=m^{\alpha\alpha'}\! m^{\beta\beta'}\!\pa_\alpha
 H_{\alpha'\!\beta'}\!=A_b^{\beta}\widetilde{m}^{a a'}\! \widetilde{m}^{b b'}\!\onabla_{\! a}
 H_{a'{}_{\!} b'}\!=A_b^{\beta}\onabla_{\!a} \widetilde{H}^{ab}$,
it follows that
\begin{equation}
A^{c}_{\beta}\pa_\alpha H^{\alpha\beta}
=\onabla_{\!a} \widetilde{H}^{ac}.
\end{equation}
Let $\widetilde{h}_{cd}
=A_c^{\gamma}\!A_d^{\delta}h_{\gamma\delta}$.
Since $A_a^{\alpha} \! A_b^{\beta}\! A_c^{\gamma} \! A_d^{\delta}
\pa_\alpha \pa_\beta h_{\gamma\delta}
\!=\onabla_{\!a}\!\onabla_b \widetilde{h}_{cd}$
it follows that
\begin{equation}\label{eq:waveeqcoordchange}
A_c^{\gamma}\!A_d^{\delta}\Box^{g} h_{\gamma\delta}
=\tg^{\,ab}\onabla_{\!a}\onabla_b \widetilde{h}_{cd}.
\end{equation}
Since the covariant derivative satisfies
$  \! A_b^{\beta} \!A_{c}^{\gamma} \! A_d^{\delta}
\pa_\beta h_{\gamma\delta}\!=\!\onabla_{\!b} \widetilde{h}_{cd}$ and since
$F$ consists of two contractions with the (inverse of the) metric
it follows that
\beq\label{eq:inhomcoordchange}
A_c^{\,\gamma}\!A_d^{\,\delta}F_{\!\gamma\delta}(g)[\pa h,\!\pa k]
=\widetilde{F}_{\!cd}(\tg)[\onabla \widetilde{h},\!\onabla\widetilde{k}].
\eq

By \eqref{eq:waveeqcoordchange} and \eqref{eq:inhomcoordchange}
 Einstein's equations \eqref{eq:EinsteinWaveMassErrors} in the new coordinates become
\beq\label{eq:covarianteinstein}
\tg^{\,ab}\onabla_{\!a}\!\onabla_{\!b} \widetilde{h}_{cd}
=\widetilde{F}_{cd}(\tg)[\onabla \widetilde{h},\!\onabla\widetilde{h}]
+\widetilde{\widecheck{T}}_{\,cd}.
\eq
\subsection{The structure of the mass error terms} By \eqref{eq:thewaveoperatormasserrror} respectively \eqref{eq:Ferror} and \eqref{eq:Aexpression}, \eqref{eq:gammahatdecay}, we can write
  \begin{align}\label{eq:covariantwave}
 \tg^{\,ab}\onabla_{\!a}\!\onabla_{\!b} \widetilde{h}_{cd}
 &=\tg^{\,ab}\onabla_{\!a}\!\onabla_{\!b} \widetilde{h}^1_{cd}
 +\widetilde{E}_{cd,0}^{\, mass}\!\!+\widetilde{E}_{cd,1}^{\, mass}[\widetilde{H}_1],\\
\widetilde{F}_{cd}(\tg)[\onabla \widetilde{h},\!\onabla\widetilde{h}]
&= \widetilde{F}_{cd}(\tg)[\onabla \widetilde{h}^1\!,\!\onabla\widetilde{h}^1]
 +\widetilde{F}_{cd,0}^{\, mass}(\tg)\!
+\widetilde{F}_{cd,1}^{\, mass}(\tg)[\onabla\widetilde{h}^1],\label{eq:covariantinhom}
  \end{align}
 where
 \begin{align}
\!\!\!\!\!\!\widetilde{E}_{cd,0}^{\, mass}\!
&=M\chi^{\,\prime}_{cd}\big(\tfrac{r}{1+t},\!
\tfrac{M}{r},\!\tfrac{M\ln{r}\!\!}{r},\omega\big)
r^{-3}\!\!,
\qquad\quad
\widetilde{E}_{cd,1}^{\, mass}[\widetilde{H}_1]= \chi_{cdab}
\big(\tfrac{r}{1+t},\!\tfrac{M}{r},\!\tfrac{M\ln{r}\!\!}{r},\omega\big)
M r^{-3}\widetilde{H}_{1}^{ab}\!,\label{eq:Emass}\\
\!\!\! \!\!\!\widetilde{F}_{cd,0}^{\, mass}\!&=M^2\chi_{cd}\big(\tfrac{r}{1+t},\!\tfrac{M}{r},\!
 \tfrac{M\ln{r}\!\!}{r},\omega,\widetilde{g}\big) r^{-4}\!\!,
\qquad
\widetilde{F}_{cd,1}^{\, mass}[\onabla \widetilde{h}^1]\!
=\! \chi_{cd}^{e a b}
\big(\tfrac{r}{1+t},\!\tfrac{M}{r},\!
\tfrac{M\ln{r}\!\!}{r},\omega,\widetilde{g}\big)M r^{-2}
\onabla_{\!e}
\widetilde{h}^{1}_{ab}\quad\label{eq:Fmass}.
 \end{align}
 Here $\chi^\prime_{cd}(s,\cdot\big)$ stands for functions with
 $|s-1\!/2|\!<\!1\!/4$ in the support and $\chi_{cdab}(s,\cdot)$,  $\chi_{cd}(s,\cdot)$
 and $\chi_{cd}^{e a b}(s,\cdot)$ stands for functions with
 $s\geq \!1\!/4$ in the support. It follows that
 \beq\label{eq:h1covariantEinstein}
 \tg^{\,ab}\onabla_{\!a}\!\onabla_{\!b} \widetilde{h}^1_{cd}
 +\widetilde{E}_{cd,0}^{\, mass}\!\!+\widetilde{E}_{cd,1}^{\, mass}[\widetilde{H}_1]
 =\widetilde{F}_{cd}(\tg)[\onabla \widetilde{h}^1\!,\!\onabla\widetilde{h}^1]
 +\widetilde{F}_{cd,0}^{\, mass}(\tg)\!
+\widetilde{F}_{cd,1}^{\, mass}(\tg)[\onabla\widetilde{h}^1]+\widetilde{\widecheck{T}}_{cd}
 \eq

The approximate wave coordinate condition \eqref{eq:approxwavecoord}
becomes
\beq\label{eq:covariantwavecoord}
\onabla_{\!a}\big(\widetilde{H}_1^{ac}
-\tfrac{1}{2}\widetilde{m}^{ac}  \widetilde{m}_{bd}\,
 \widetilde{H}_1^{bd}\big)=\widetilde{W}^c(\tg)[\widetilde{H},\onabla \widetilde{H}]
 +\widetilde{W}^{c, 0}_{\! mass},
 \eq
where
\beq
\widetilde{W}^c(\tg)[\widetilde{H},\onabla \widetilde{H}]=
\widetilde{W}^c(\tg)[\widetilde{H}_1,\onabla \widetilde{H}_1]
+\widetilde{W}^{c,1}_{mass}
+\widetilde{W}^{c,2}_{mass}[ \widetilde{H}_1]
+\widetilde{W}^{c,3}_{mass}[\onabla  \widetilde{H}_1].
\eq
Here
\begin{align}
\widetilde{W}^{c,0}_{mass}
&=-M{\chi}^{c\,\prime, 0}\big(\tfrac{r}{t+1},\!\tfrac{M}{r},\!\tfrac{M\ln{r}\!\!}{r},\omega\big) r^{-2},
\qquad\qquad
\widetilde{W}^{c,1}_{mass}=\chi^{c,1}\big(\tfrac{r}{1+t},\omega,\tg\big)M r^{-3}, \\
\widetilde{W}^{c,2}_{mass}[ \widetilde{H}_1]&
=\chi_{\mu\nu}^{c,2}\big(\tfrac{r}{1+t},\omega,\tg\big)M r^{-2}  \widetilde{H}_1^{\mu\nu},\qquad\qquad
\widetilde{W}^{c,3}_{mass}[\pa  \widetilde{H}_1]
=\chi_{ab}^{cd,3}\big(\tfrac{r}{1+t},\omega,\tg\big)M r^{-1}\onabla_d  \widetilde{H}_1^{ab} .
\end{align}
with ${\chi}^{c\,\prime}(s,\cdot)$ denoting a function with
 $|s-1\!/2|\!<\!1\!/4$ in the support.

\subsection{The structure of the covariant error terms}\label{sec:cov}
We now show how the wave coordinate condition and the mass-adjusted Einstein equation \eqref{eq:h1covariantEinstein} behave under coordinate transformation.
\begin{lemma}\label{lemma:covariant}
Given a symmetric $(0,2)$-tensor $\widetilde{H}$ and symmetric $(2,0)$-tensors $\widetilde{h}$ and $\widetilde{k}$, the identities
\begin{align}\label{eq:WaveCoordchangecoord}
\onabla_{\!\!a} \widetilde{H}^{ac}
&=\wpa_a \widetilde{H}^{ac}\!
+\widetilde{W}_1^c[\widehat{\Gamma},\widetilde{H}],\\
 \label{eq:ReducedwaveEqchangecoord}
\tg^{\,ab}\onabla_{\!a}\onabla_b \widetilde{h}_{cd}&={\Box}^{\,\tg}\, \widetilde{h}_{cd}
+\widetilde{E}_{cd}^1(\tg)[\widehat{\Gamma},\wpa\widetilde{h}]
+\widetilde{E}_{cd}^2(\tg)[\wpa\widehat{\Gamma},\widetilde{h}]
+\widetilde{E}_{cd}^3(\tg)[\widehat{\Gamma},\widehat{\Gamma},\widetilde{h}],\\
\label{eq:Fchangecoord}
\widetilde{F}_{\!cd}(\tg)[\onabla \widetilde{h},\!\onabla\widetilde{k}]
&=\widetilde{F}_{\!cd}(\tg)[\wpa \widetilde{h},\wpa\widetilde{k}]
+\widetilde{F}_{\!cd}^1(\tg)[\widehat{\Gamma}\!,\widetilde{h},\wpa\widetilde{k}]
+\widetilde{F}_{\!cd}^1(\tg)[\widehat{\Gamma}\!,\widetilde{k},\wpa\widetilde{h}]
+\widetilde{F}_{\!cd}^2(\tg)
[\widehat{\Gamma}\!,\widehat{\Gamma}\!,\widetilde{h},\widetilde{k}].
\end{align}
hold for some $\widetilde{W}_1, \widetilde{E}^1, \widetilde{E}^2, \widetilde{E}^3, \widetilde{F}^1, \widetilde{F}^2$ which are linear in the bracketed arguments.
\end{lemma}
\begin{proof} \eqref{eq:WaveCoordchangecoord} follows from
 $\onabla_{\!a} \widetilde{K}^{ac}\!\!=\wpa_a \widetilde{K}^{ac}\!\!
 +\widehat{\Gamma}_{\!ad}^a \widetilde{K}^{dc}\!\!
+\widehat{\Gamma}_{\!ad}^c \widetilde{K}^{ad}\!\!$.
\!\eqref{eq:ReducedwaveEqchangecoord} follows from
\begin{multline}\label{eq:secondcovariant}
\onabla_{\!a} \onabla_{\!b}\, \widetilde{h}_{cd}
=\wpa_a \big( \wpa_b \widetilde{h}_{cd}
-\widehat{\Gamma}_{bc}^e \widetilde{h}_{ed}\!
-\widehat{\Gamma}_{bd}^e \widetilde{h}_{ce}\!\big)
-\widehat{\Gamma}_{ab}^{b'} \big( \wpa_{b'} \widetilde{h}_{cd}
-\widehat{\Gamma}_{b'c}^e \widetilde{h}_{ed}\!
-\widehat{\Gamma}_{b'd}^e \widetilde{h}_{ce}\!\big)
-\widehat{\Gamma}_{ab}^{c'} \big( \wpa_{b}\widetilde{h}_{c'd}
-\widehat{\Gamma}_{bc'}^e \widetilde{h}_{ed}\!
-\widehat{\Gamma}_{bd}^e \widetilde{h}_{c'e}\!\big)\\
-\widehat{\Gamma}_{ab}^{d'} \big( \wpa_{b} \,\widetilde{h}_{cd'}
-\widehat{\Gamma}_{bc}^e \widetilde{h}_{ed'}\!
-\widehat{\Gamma}_{bd'}^e \widetilde{h}_{ce}\!\big)
=\wpa_a\wpa_b \,\widetilde{h}_{cd}
+\widetilde{E}_{abcd}^1[\widehat{\Gamma},\wpa\widetilde{h}]
+\widetilde{E}_{abcd}^2[\wpa\widehat{\Gamma},\widetilde{h}]
+\widetilde{E}_{abcd}^3[\widehat{\Gamma},\widehat{\Gamma},\widetilde{h}].
\end{multline}
Finally, \!\eqref{eq:Fchangecoord} follows from \eqref{eq:covariantderivativedefintion} and symmetry of $\widetilde{F}$.
\end{proof}
\subsubsection{The covariant error terms}\label{sec:covarianterrorsec5}
Lemma \ref{lemma:covariant} says that we can replace the covariant derivatives by $\wpa$,
up to terms involving at least an extra factor
of $\widehat{\Gamma}$, that by \eqref{eq:gammahatdecay} decays like
\beq
\wpa^\alpha \widehat{\Gamma}^{c}_{\!ab\!}
=M r^{-2-|\alpha|}\ln{r} \, \chi^{c\alpha}_{ab}\big(\tfrac{r}{t+1},\omega,\!\tfrac{M\!}{r},\!
\tfrac{M\ln{r}\!\!}{r}\, ,\widetilde{h}\big)
+M r^{-2-|\alpha|} \widetilde{\chi}^{c\alpha}_{ab}\big(\tfrac{r}{t+1},\omega,\!\tfrac{M\!}{r},\!
\tfrac{M\ln{r}\!\!}{r}\, ,\widetilde{h}\big),\quad
|\alpha|\leq 1,
\eq
 multiplied by a factor of $\widetilde{h}$ instead
of a factor of $\wpa \widetilde{h}$. Since as we will see, in the weighted norms that we will be using
\beq
\| \langle r^*\!-t\rangle^{-1}\widetilde{h}^1 \|\lesssim \|\wpa \widetilde{h}^1 \|,
\eq
this error still gives an extra decay of $\ln{t}/t$. By \eqref{eq:h1covariantEinstein}
\beq\label{eq:newcoordeinstein}
\tg^{\,ab}\wpa_{a}\wpa_b \widetilde{h}^1_{cd}
-\widetilde{E}_{cd}^{\, mass}
 -\widetilde{E}_{cd}^{\,cov}
=\widetilde{F}_{cd}(\tg)[\wpa \widetilde{h}^1\!\!,\wpa\widetilde{h}^1]
+\widetilde{F}_{cd}^{\, mass}\!
+\widetilde{F}_{cd}^{\, cov}+\widetilde{\widecheck{T}}_{\!cd},
\eq
where with notation as in \eqref{eq:Emass}-\eqref{eq:Fmass}
  \beq
 \widetilde{E}_{cd}^{\, mass}=\widetilde{E}_{cd,0}^{\, mass}\!\!
 +\widetilde{E}_{cd,1}^{\, mass}[\widetilde{H}_1],
 ,\qquad\qquad
\widetilde{F}_{cd}^{\, mass}=\widetilde{F}_{cd,0}^{\, mass}(\tg)\!
+\widetilde{F}_{cd,1}^{\, mass}(\tg)[\wpa\widetilde{h}^1],
\eq
and
\beq
\widetilde{E}_{cd}^{\, cov}
=\widetilde{E}_{cd,0}^{\, cov}[h^1]+\widetilde{E}_{cd,1}^{\, cov}[\pa h^1],\qquad\qquad
\widetilde{F}_{cd}^{\, cov}
=\widetilde{F}_{cd,0}^{\, cov}[h^1]+\widetilde{F}_{cd,1}^{\, cov}[\pa h^1],
\eq
are  the covariant error coming from replacing
the covariant derivatives by partial derivatives in \eqref{eq:h1covariantEinstein}: by \eqref{eq:ReducedwaveEqchangecoord} and Lemma \ref{lem:ChristoffelBounds}, we have
\begin{align}
\widetilde{E}_{cd,0}^{\,cov}[ h^1]
&=Mr^{-3}\ln{r} \, \chi^{ab}_{cd,0}\big(\tfrac{r}{t+1},
\omega,\!\tfrac{M\!}{r},\!\tfrac{M\ln{r}\!\!}{r}\, ,h\big) h^1_{ab}
+M r^{-3} \widetilde{\chi}^{ab}_{cd,0}\big(\tfrac{r}{t+1},\omega,\!\tfrac{M\!}{r},\!
\tfrac{M\ln{r}\!\!}{r}\, ,h\big) h^1_{ab},\qquad\label{eq:Ecov1}\\
\widetilde{E}_{cd,1}^{\,cov}[\pa h^1]
&=Mr^{-2}\ln{r} \, \chi^{ab\gamma}_{cd,1}\big(\tfrac{r}{t+1},
\omega,\!\tfrac{M\!}{r},\!\tfrac{M\ln{r}\!\!}{r}\, ,h\big)\wpa_\gamma h^1_{ab}
+M r^{-2} \widetilde{\chi}^{ab\gamma}_{cd,1}\big(\tfrac{r}{t+1},\omega,\!\tfrac{M\!}{r},\!
\tfrac{M\ln{r}\!\!}{r}\, ,h\big)\wpa_\gamma h^1_{ab}\qquad\label{eq:Ecov2}
\end{align}
and by \eqref{eq:Fchangecoord} and Lemma \ref{lem:ChristoffelBounds}, we have
\begin{align}
\widetilde{F}_{cd,0}^{\, cov}[h^1]&=Mr^{-3}\ln{r} \, \chi^{ab}_{cd,0}\big(\tfrac{r}{t+1},
\omega,\!\tfrac{M\!}{r},\!\tfrac{M\ln{r}\!\!}{r}\, ,h^1\big) h^1_{ab}
+M r^{-3} \widetilde{\chi}^{ab}_{cd,0}\big(\tfrac{r}{t+1},\omega,\!\tfrac{M\!}{r},\!
\tfrac{M\ln{r}\!\!}{r}\, ,h\big) h^1_{ab},\qquad\label{eq:Fcov1}\\
\widetilde{F}_{cd,1}^{\, cov}[\pa h^1]&=Mr^{-2}\ln{r} \, \chi^{ab\gamma}_{cd,1}\big(\tfrac{r}{t+1},
\omega,\!\tfrac{M\!}{r},\!\tfrac{M\ln{r}\!\!}{r}\, ,h^1\big)\wpa_\gamma h^1_{ab}
+M r^{-2} \widetilde{\chi}^{ab\gamma}_{cd,1}\big(\tfrac{r}{t+1},\omega,\!\tfrac{M\!}{r},\!
\tfrac{M\ln{r}\!\!}{r}\, ,h\big)\wpa_\gamma h^1_{ab}\qquad\label{eq:Fcov2}.
\end{align}
Moreover by \eqref{eq:covariantwavecoord}
\beq\label{eq:newcoordwavecoord}
\wpa_a\big(\widetilde{H}_1^{ac}
-\tfrac{1}{2}\widetilde{m}^{ac}  \widetilde{m}_{bd}\,
 \widetilde{H}_1^{bd}\big)=\widetilde{W}^c(\tg)[\widetilde{H},\wpa \widetilde{H}]
 +\widetilde{W}^{c,0}_{\!mass}+\widetilde{W}^c_{\!cov},
 \eq
where
\beq
\widetilde{W}^c(\tg)[\widetilde{H},\wpa \widetilde{H}]=
\widetilde{W}^c(\tg)[\widetilde{H}_1,\wpa \widetilde{H}_1]
+\widetilde{W}^{c,1}_{mass}
+\widetilde{W}^{c,2}_{mass}[ \widetilde{H}_1]
+\widetilde{W}^{c,3}_{mass}[\wpa \widetilde{H}_1],
\eq
and
\beq
\widetilde{W}^c_{\!cov}=\widetilde{W}^{c}_{\!cov}(h)[\widetilde{H}_{\!1}]\!
=M r^{-2}\ln{r} \,
\chi^c_{ab}\big(\tfrac{r}{t+\!1},\omega,\!\tfrac{M\!}{r},\!\tfrac{M\ln{r}\!\!}{r} ,h\big)
\widetilde{H}_{1}^{ab}\!\!
+M r^{-2} \widetilde{\chi}^c_{ab}\big(\tfrac{r}{t+\!1},\omega,\!\tfrac{M\!}{r},\!
\tfrac{M\ln{r}\!\!}{r} ,h\big)\widetilde{H}_1^{ab}\!\! .
\eq

\subsubsection{Alternative term picking up the mass}
In view of \eqref{eq:approxtildegm0H1}, \eqref{eq:approxtildem0} and
\eqref{eq:waveoperatorstarsphericalcoord} it may be natural to alternatively
 subtract of a term picking up the mass after having changed coordinates:
\beq
\widetilde{h}^{0\mathcal{I}}_{ab}
=M{\rs}^{-1}\widetilde{\chi}(\tfrac{\rs}{t+1})\delta_{ab},
\eq
where $\widetilde{\chi}$ is as in \eqref{eq:m0def}. The difference
\beq
\wpa^\alpha \widetilde{h}^{0\mathcal{I}}_{ab}-\wpa^\alpha \widetilde{h}^{0}_{ab}
=O(Mr^{-2-|\alpha|}\ln{r}),
\eq
is of the same size as
$\widehat{\Gamma}$ so one could include it in the covariant error terms described above.

\subsection{The structure of the quadratic and higher order terms}
Recall that the inhomogeneous term in Einstein's vacuum equations has the form
\beq
\widetilde{F}_{ab} (\tg) [\wpa \widetilde{h}, \wpa \widetilde{h}]=P(\tg)[\wpa_a \wh, \wpa_b \wh]+Q_{ab}(\tg)[\wpa \wh, \wpa \wh],
\eq
where $Q$ is a combination of classical null forms and $P$, given by \eqref{eq:curvedPdef},
 has a weak null structure. In view of \eqref{eq:approxtildegm0H1}-\eqref{eq:approxtildem0}
we expect that
\beq
\widetilde{F}_{ab} (\tg) [\wpa \widetilde{h}, \wpa \widetilde{h}]
\sim \widetilde{F}_{ab} (\widehat{m}) [\wpa \widetilde{h}, \wpa \widetilde{h}].
\eq
 Since $\widehat{Q}_{ab}\!=Q_{ab}(\widehat{m})$ satisfy the classical null condition with respect to the new coordinates,
\beq \label{eq:nullcondestsec5}
|\widehat{Q}_{ab}(\wpa \wh,\wpa \widetilde{k})|\les |\overline{\wpa} \,\wh|\,|\wpa  \widetilde{k}|+|\wpa h|\, |\overline{\wpa} \,
\widetilde{k}|.
\eq
The main term $\widehat{P}=P(\widehat{m})$ can be further analyzed as follows.
First we note that
by \eqref{eq:transversalderivativeprojectionsec2}
\beq\label{eq:PPnullsec5}
\big| \widehat{P}(\wpa_a \wh,\wpa_b \widetilde{k})- \wL_{a}\wL_b \widehat{P}(\pa_{q^*} \wh,\pa_{q^*} \widetilde{k})\big|
\les |\overline{\wpa}\, \wh| \,|\wpa \widetilde{k}|
+|\wpa \wh|\,|\overline{\wpa} \,\widetilde{k}|.
\eq
where by \eqref{eq:tanPsec3}
\beq\label{eq:tanPsec5}
\big|\widehat{P}(\pa_{q^*}\wh,\pa_{q^*}\widetilde{k})-P_{\mathcal S}(\pa_{q^*}\wh,\pa_{q^*}\widetilde{k})\big|
\les \big(|\pa_{q^*}\wh\,|_{\wL\widetilde{\mathcal
T}}+|\pa_{q^*}\trs \wh|\big)|\pa_{q^*}\widetilde{k}| +|\pa_{q^*}\wh\,|\big(|\pa_{q^*}\widetilde{k}\,|_{\wL\widetilde{\mathcal
T}}+|\pa_{q^*}\trs \widetilde{k}|\big).
\eq
where
\beq\label{eq:tanSPsec5}
P_{\mathcal{S}} (D,E)= -\widehat{D}_{\!AB}\, \widehat{E}^{AB\!\!}/2,\quad A,B\in\mathcal{S},\quad
\text{where} \quad \widehat{D}_{\!AB}=D_{\!AB}-\delta_{AB}\trs D\!/2,\quad
\trs D=\delta^{AB} D_{AB}.
\eq

\section{Vector fields applied to Einstein's equations,
commutators and higher order equations}\label{sec:higherordereinstein}
Modulo mass errors
(that comes from replacing $h$ by $h^1$)
and covariant errors
(that comes from replacing $\nabla_{\!a} $ by
$\wpa_a$ and $h^1$ by $\widetilde{h}^1$)
 \eqref{eq:newcoordeinstein} the reduced wave operator $\widetilde{\Box}=\tg^{\,ab}\wpa_a\wpa_b$ applied to $\widetilde{h}^1$ satisfy
\beq\label{eq:newcoordeinsteinwithRerrors}
\widetilde{\Box}\,\widetilde{h}^1_{cd}
=\widetilde{F}_{cd}(\tg)[\wpa \widetilde{h}^1,\wpa\widetilde{h}^1]
+\widetilde{\widecheck{T}}_{\!cd}
+\widetilde{R}_{cd}^{\,mass}+\widetilde{R}_{cd}^{\,cov},
\eq
where
$\widetilde{R}_{cd}^{\,mass}=\widetilde{E}_{cd}^{\,mass}
+\widetilde{F}_{cd}^{\,mass}$
and $\widetilde{R}_{cd}^{\,cov}=\widetilde{E}_{cd}^{\,cov}+\widetilde{F}_{cd}^{\,cov}$
are given in Section \ref{sec:cov}.
It was explained how to deal with the mass errors in Section
\ref{sec:masssubtractionsection} and covariant errors in Section
\ref{sec:Genralizedwavecoordsection} at lowest order. However, applying vector fields
 using the estimates in Section
\ref{section:metricdecomposition} for curved vector fields in terms of the non curved
gives, apart from additional logarithmic factors, the same estimates for vector fields applied
to these quantities. We can therefore neglect these error terms in the analysis that follows.
Similarly by \eqref{eq:newcoordwavecoord}
\beq\label{eq:wavecoordhattilde}
\wpa_a\big(\widetilde{H}_1^{ac}
-\tfrac{1}{2}\widetilde{m}^{ac}  \widetilde{m}_{bd}\,
 \widetilde{H}_1^{bd}\big)=\widetilde{W}^c(\tg)[\widetilde{H}_1,\wpa \widetilde{H}_1]
 +\widetilde{W}^c_{\!mass}+\widetilde{W}^c_{\!cov}.
\eq
We will now apply vector fields to Einstein's equations \eqref{eq:newcoordeinsteinwithRerrors} to get higher order equations. As we will see the commutators improve if we use Lie derivatives. This gives a wave equation for the differentiated metric $ \widetilde{h}^{1I\!}=\widehat{\mathcal{L}}_{\widetilde{X}}^I \widetilde{h}^{1}$:
\beq\label{eq:newcoordeinsteinwithRerrorshigherorder}
\widetilde{\Box}\, \widetilde{h}^{1I}_{cd}
+\widetilde{R}_{cd}^{\,com\, I}\!
=\!{\sum}_{I^\prime\!+I^{\prime\prime}=I}\! \kappa_{I^{\prime\prime}} \Big(\!{\sum}_{J+K=I^\prime}\!\widetilde{F}_{cd}(\tg)[\wpa \widetilde{h}^{1J}\!\!,\wpa\widetilde{h}^{1K}]\!+\widetilde{R}_{cd,\, 0}^{\,cube\, I^\prime}\!\!\!
+\widetilde{\widecheck{T}}_{\!cd}{\!\!}^{\!I^\prime}\!
+\widetilde{R}_{cd}^{\,mass\,I^\prime}\!\!\!+\widetilde{R}_{cd}^{\,cov\,I^\prime}\!\Big),
\eq
where the remainder terms $\widetilde{R}_{cd}^{\,\bullet\, I}$ are lower order, section \ref{sec:higherordereinstein}

Similarly with $\widecheck{H}_1^{ac}\!=\widetilde{H}_1^{ac}\!
-\tfrac{1}{2}\widetilde{m}^{ac}  \widetilde{m}_{bd}\,
 \widetilde{H}_1^{bd} $
 we get higher order versions of the wave coordinate condition
 \beq\label{eq:Liedivergencetilderepeat}
	 \wpa_c  \widehat{\Lie}_{\widetilde{X}}^I \widecheck{H}_1^{cd}
={\sum}_{|J|\leq |I|} c_J^I\widehat{\Lie}_{\widetilde{X}}^J
\wpa_c \widecheck{H}_1^{cd},
\eq
see section \ref{sec:higherorderwavecoord}.

\subsection{Commutators and Lie derivatives along the modified vector fields}
\subsubsection{Lie derivatives along the modified vector fields}
As for \eqref{eq:Liepartialcommute} we have
\begin{equation} 
		{\mathcal L}_{\widetilde{X}}\wpa_{c_1}\!\cdots\wpa_{c_k}
\widetilde{K}^{a_1\dots a_r}_{b_1\dots b_s}
		=\wpa_{c_1}\!\cdots\wpa_{c_k} \Lie_{\widetilde{X}}
\widetilde{K}^{a_1\dots a_r}_{b_1\dotsb_s}.
	\end{equation}
As in \eqref{eq:modLiedef} we define the modified Lie derivative in the new coordinates by
\beq
	\widehat{\mathcal L}_{\widetilde{X}} \widetilde{K}^{a_1\dots a_r}_{b_1\dots b_s}
	=
	\Lie_{\widetilde{X}} \widetilde{K}^{a_1\dots a_r}_{b_1\dots b_s}
	+
	\tfrac{r-s}{4}(\wpa_e \widetilde{X}^e)\widetilde{K}^{a_1\dots a_r}_{b_1\dots b_s}.
\eq
Then if $\widehat{m}^{ab}={m}^{ab}$ is the Minkowski metric in the $\widetilde{x}$ coordinates
then for the vector fields in \eqref{eq:ModifiedVectorFields} we have
\beq\label{eq:LieMinkowski}
\widehat{\Lie}_{\widetilde{X}}
\widehat{m}^{ab}=0.
\eq

\subsubsection{The commutators of the equations
with Lie derivatives along the modified vector fields}
The Lie derivative  commutes
 with exterior differentiation which for any of the vector fields $\wt{X}$ leads to
\beq\label{eq:Liedivergencetilde}
	 \wpa_c  \widehat{\Lie}_{\widetilde{X}} \widecheck{H}^{cd}
=\big(\widehat{\Lie}_{\widetilde{X}}
+\tfrac{\wpa_c \widetilde{X}^c }{2}\big)\wpa_c \widecheck{H}^{cd}.
\eq
Moreover,
\begin{equation}
[\Lie_{\widetilde{X}}, {\Box}^{\,\tg}]\phi
= -(\Lie_{\widetilde{X}}\widetilde{g}^{\,ab})\widetilde{\pa}_a\widetilde{\pa}_b\phi
+ \widetilde{g}^{\,ab}(\wpa_a\wpa_b \widetilde{X}^c)\widetilde{\pa}_c\phi=
-(\Lie_{\widetilde{X}} \widetilde{g}^{\,ab})
\widetilde{\pa}_a\widetilde{\pa}_b\phi,
\end{equation}
since  $\widetilde{X}^c$ are linear functions. Moreover
\begin{equation} \label{eq:Liecommutatortilde}
	{\Box}^{\,\tg} \widehat{\mathcal{L}}_{\widetilde{X}} \phi_{cd}
	=
	\Lie_{\widetilde{X}} \big({\Box}^{\,\tg} \phi_{cd} \big)
	-
	(\widehat{\mathcal{L}}_{\widetilde{X}} \tg^{\,ab}) \wpa_{a} \wpa_{b} \phi_{cd},
\end{equation}
where by \eqref{eq:approxtildegm0H1}-\eqref{eq:approxtildem0} and \eqref{eq:LieMinkowski}
\beq
\widehat{\mathcal{L}}_{\widetilde{X}} \tg^{\,ab}\sim \widehat{\mathcal{L}}_{\widetilde{X}} \widetilde{H}_1^{\,ab}.
\eq

Finally if $$S_{ab}(\widetilde{g})[\wpa \widetilde{h},\wpa \widetilde{k}]
=S_{ab}[\widetilde{G},\widetilde{G}][\wpa \widetilde{h},\wpa \widetilde{k}]$$
is a quadratic form in the $(0,3)$ forms $\wpa \widetilde{h}$ and
$\wpa \widetilde{k}$ with two contractions with the inverse of the metric
$\widetilde{G}^{ab}=\widetilde{g}^{ab}$, i.e. it is bilinear in $\widetilde{G}$,
we have
\beq\label{eq:LieQuadtilde}
 \Lie_{\!\widetilde{X}}\big(  S_{ab}(\widetilde{g\!})[\wpa \widetilde{h},\wpa \widetilde{k}]\big)\!
 =\!S_{ab}(\widetilde{g})[\wpa \widehat{\Lie}_{\!\widetilde{X}}\widetilde{h},\wpa \widetilde{k}]
 +S_{ab}(\widetilde{g\!})[\wpa \widetilde{h},\wpa \widehat{\Lie}_{\!\widetilde{X}} \widetilde{k}]
 +S_{ab}[ \widehat{\Lie}_{\!\widetilde{X}}\widetilde{G},\!\widetilde{G}]
 [\wpa \widetilde{h},\wpa \widetilde{k}]
 +S_{ab}[\widetilde{G},\widehat{\Lie}_{\!\widetilde{X}}\widetilde{G}]
 [\wpa \widetilde{h},\wpa \widetilde{k}].
\eq
\subsubsection{The cubic error terms}
We remark that the last term in \eqref{eq:LieQuadtilde}
is cubic and therefore much easier to control, see Section \ref{sec:cubicerrors}
and Section \ref{sec:covarianterrorsec5}.

\subsection{Higher order commutators with Einstein's equations}
We will now apply vector fields to Einstein's equations \eqref{eq:newcoordeinsteinwithRerrors} to get higher order equations. As we have seen the commutators improve if we use Lie derivatives and furthermore as will see below they also improve if we first multiply with $\kappa=(1+M\widetilde{\chi}/r)^{-1}$, which just produces lower order terms of the same form. This gives a wave equation for the differentiated metric $ \widetilde{h}^{1I\!}=\widehat{\mathcal{L}}_{\widetilde{X}}^I \widetilde{h}^{1}$:
\beq\label{eq:newcoordeinsteinwithRerrorshigherorder}
\widetilde{\Box}\, \widetilde{h}^{1I}_{cd}
+\widetilde{R}_{cd}^{\,com\, I}\!
=\!{\sum}_{I^\prime\!+I^{\prime\prime}=I}\! \kappa_{I^{\prime\prime}} \Big(\!{\sum}_{J+K=I^\prime}\!\widetilde{F}_{cd}(\tg)[\wpa \widetilde{h}^{1J}\!\!,\wpa\widetilde{h}^{1K}]\!+\widetilde{R}_{cd,\, 0}^{\,cube\, I^\prime}\!\!\!
+\widetilde{\widecheck{T}}_{\!cd}{\!\!}^{\!I^\prime}\!
+\widetilde{R}_{cd}^{\,mass\,I^\prime}\!\!\!+\widetilde{R}_{cd}^{\,cov\,I^\prime}\!\Big),
\eq
where $ \kappa_{I^{\prime\prime}} =\kappa^{-1}\widehat{\mathcal{L}}_{\widetilde{X}}^{I^{\prime\prime}} \kappa $,
$\widetilde{R}_{cd}^{\,mass\,I^\prime}
=\widehat{\mathcal{L}}_{\widetilde{X}}^{I^{\prime}}\widetilde{R}_{cd}^{\,mass}$
and $\widetilde{R}_{cd}^{\,cov\,I^\prime}
=\widehat{\mathcal{L}}_{\widetilde{X}}^{I^{\prime}}\widetilde{R}_{cd}^{\,cov}$.
Here $\kappa\widetilde{R}_{cd}^{\,com\, I}$ is the commutator term of
$\widehat{\mathcal{L}}_{\widetilde{X}}^{I}$ with the wave operator
$\kappa\tg^{ab}\wpa_a\wpa_b$ and $\widetilde{R}_{cd,\, 0}^{\,cube\,I^\prime}$
are the cubic terms that show up when some of the vector fields in
$\widehat{\mathcal{L}}_{\widetilde{X}}^{I^\prime}$ fall on the metric $g$ in the nonlinearity
$\widetilde{F}_{cd}(\tg)[\wpa \widetilde{h}^1,\wpa\widetilde{h}^1]$.
The remainder terms $\widetilde{R}_{cd}^{\,\bullet\, I}$ will be estimated below. Moreover we will show that we can write
\beq
\kappa^{2\!} \widetilde{F}_{\!cd}(\tg)
[\wpa \widetilde{h}^{1J}\!\!,\wpa\widetilde{h}^{1K}]
=\!\wL_{c}\wL_d
\widehat{P}(\pa_{q^*} \wh^{1J}\!\!,\pa_{q^*} \widetilde{h}^{1K})
+\kappa^2\widetilde{R}_{cd}^{\,tan\, JK}+\widetilde{R}_{cd,\, 1}^{\,cube\, JK},
\eq
where $\widehat{P}$ is the constant coefficient quadratic form with the special geometric weak null structure mentioned before and $\widetilde{R}_{cd}^{\,tan\, JK}$
consist of a quadratic form in $\widetilde{\pa} \widetilde{h}^1$, with at least one factor with only tangential derivatives, plus cubic terms $\widetilde{R}_{cd,\, 1}^{\,cube\, JK}$. The cubic terms as well as the mass and covariant terms are lower order and easier to control.

\subsubsection{The leading behavior of the inverse of the metric
in the curved coordinates}
 We have
 \beq
 \widetilde{g}^{\,ab}=\widetilde{m}_0^{ab}+\widetilde{H}_1^{ab}.
 \eq
 We have seen in Lemma \ref{lem:m0approx}
 that
 \beq\label{eq:metric0decomposition}
 \widetilde{m}_0^{ab}= \kappa_{0\,}\widehat{m}^{ab}\!
 +\kappa_{1\,} \slashed{S}^{\,ab}\!+\kappa_{2\,} \widehat{\delta}^{\,ab}\!+\kappa_{3\,} i_+^{\, ab}\!,
 \eq
 where
 \beq
 \kappa_0=1+\tfrac{M\widetilde\chi}{r},\qquad
 \kappa_1\!=(1{}_{\!}+\tfrac{M\widetilde\chi\ln{r}}{r})^2
 (1-\tfrac{M\widetilde\chi}{r})-(1{}_{\!}+\tfrac{M\widetilde\chi}{r})\!
 =2\widetilde{\chi} \tfrac{M\ln{r}\!\!}{r} \chi_1
 \big(\tfrac{r}{t+1},\! \tfrac{M\ln{r}\!\!}{r}, \!\tfrac{M}{r}\big)
 -2\widetilde{\chi}\tfrac{M}{r},
 \eq
 and
 \beq
 \kappa_2=(\tfrac{M\widetilde\chi}{r})^2\!\!,
 \qquad
i_+^{\, ab}=\widetilde{\chi}^{\,\prime}(\tfrac{r}{t+1})
 \chi^{ab}\big(\tfrac{r}{t+1}, \tfrac{M\ln{r}}{r}, \tfrac{M}{r},\omega\big),
 \qquad\text{and} \qquad \kappa_3=\tfrac{M\ln{r}}{r},
 \eq
 for some smooth functions $\chi_1$ and $\chi^{ab}$.
Here $\widehat{m}^{ab}$ is the Minkowski metric in the new coordinates,
$\widehat{\delta}^{ab}=1$, when $a,b=1,2,3$ and $\widehat{\delta}^{0b}=\widehat{\delta}^{a0}=\widehat{\delta}^{00}=0$.
 For $r>t/2$ \eqref{eq:metric0decomposition}
also follows from \eqref{eq:waveeqradial}.

Let $\kappa=1/\kappa_0$.
Lemma \ref{lem:m0decomp} together with \eqref{eq:tangentialbyvectorfield}
gives
\beq
\bigtwo|(\widehat{\mathcal{L}}_{\widetilde{X}}^I
(\kappa\widetilde{m}_0^{\,ab}\!\!-\widehat{m}^{ab})) \wpa_{a} \wpa_{b}\phi\bigtwo|
\lesssim \frac{M\ln{\tplusr}}{\tplusr^2} {\sum}_{|J|\leq 1}
|\wpa \widetilde{Z}^J \phi|.
\eq
Moreover using \eqref{eq:tangentialbyvectorfield} and \eqref{eq:waveoperatorHLLtangential} we get
\beq
\bigtwo|(\widehat{\mathcal{L}}_{\widetilde{X}}^I (\kappa \widetilde{H}_1^{ab}))
\wpa_{a} \wpa_{b}\phi\bigtwo|
\lesssim {\sum}_{|J|\leq |I|}
\Big(
\frac{\big|
\big(\widehat{\mathcal{L}}_{\widetilde{X}}^J
\widetilde{H}_1\big)_{\widetilde{L}\widetilde{L}}
\big|}{\tminusrstar}
+\frac{|\widehat{\mathcal{L}}_{\widetilde{X}}^J\widetilde{H}_1|}{\tplusr}\Big)
 {\sum}_{|K|\leq 1} |\wpa \widetilde{Z}^K \phi|.
\eq
Hence
\beq\label{eq:TheCommutator}
\bigtwo|(\widehat{\mathcal{L}}_{\widetilde{X}}^I (\kappa\tg^{\,ab}\!\!-\widehat{m}^{ab})) \wpa_{a} \wpa_{b}\phi\bigtwo|
\lesssim {\sum}_{|J|\leq |I|}\Big(\frac{M\ln{\tplusr}}{\tplusr^2}
+\frac{\big|\big(\widehat{\mathcal{L}}_{\widetilde{X}}^J
\widetilde{H}_1\big)_{\widetilde{L}\widetilde{L}}
\big|}{\tminusrstar}
+\frac{|\widehat{\mathcal{L}}_{\widetilde{X}}^J\widetilde{H}_1|}{\tplusr}\Big)
 {\sum}_{|K|\leq 1} |\wpa \widetilde{Z}^K \phi|.
\eq

\subsubsection{Higher order commutators with the reduced wave operator}
Since $\widehat{\mathcal{L}}_{\widetilde{X}}  \widehat{m}^{ab}=0$,
it follows that with $\kappa=1/\kappa_0$
\begin{equation} \label{eq:Liecommutatortildekappa}
	\kappa\widetilde{\Box}\, \widehat{\mathcal{L}}_{\widetilde{X}} \phi_{cd}
	=\Lie_{\widetilde{X}} \big( \kappa\widetilde{\Box}\,\phi_{cd} \big)
	-(\widehat{\mathcal{L}}_{\widetilde{X}} (\kappa\tg^{\,ab}\!\!-\widehat{m}^{ab})) \wpa_{a} \wpa_{b} \phi_{cd},
\end{equation}
By repeated use of \eqref{eq:Liecommutatortilde}
(see also \eqref{eq:Liecommutatortildekappa}) we have
\beq\label{eq:Liecommutatortildekapparepeat}
\Lie_{\widetilde{X}}^I \big( \kappa\widetilde{\Box}\, \widetilde{h}_{cd}^1\big)
=	\kappa\widetilde{\Box}\, \widehat{\mathcal{L}}_{\widetilde{X}}^I\widetilde{h}_{cd}^1
	+\kappa \widetilde{R}_{cd}^{\,com\, I}\!\!,
\eq
where
\beq
\widetilde{R}_{cd}^{\,com\, I}={\sum}_{J+K=I, \, |K|<|I|}
\kappa^{-1}\big(\widehat{\mathcal{L}}_{\widetilde{X}}^J (\kappa\tg^{\,ab}\!\!-\widehat{m}^{ab})\big) \wpa_{a} \wpa_{b}
\widehat{\Lie}_{\widetilde{X}}^{K}\widetilde{h}_{cd}^1.
\eq
The commutator term above is the main problematic term to deal with.
We will use the above expression for the curved wave operator with
the energy estimate but for the decay estimates we will use the constant coefficient wave operator $\Box^*=\widehat{m}^{ab}\wpa_a\wpa_b$.  We have
\beq\label{eq:Liecommutatortildekapparepeatstar}
	\Box^* \widehat{\mathcal{L}}_{\widetilde{X}}^I
\widetilde{h}_{cd}^{1}
	=\Lie_{\widetilde{X}}^I \big( \kappa\widetilde{\Box}\, \widetilde{h}^1_{cd}\big)- \kappa\widetilde{R}_{cd}^{\,com\,*\, I}\!,
\eq
where
\beq
\widetilde{R}_{cd}^{\,com\,*\, I}
={\sum}_{J+K=I} \kappa^{-1}
\big(\widehat{\mathcal{L}}_{\widetilde{X}}^J (\kappa\tg^{\,ab}\!\!-\widehat{m}^{ab})\big) \wpa_{a} \wpa_{b}
\widehat{\Lie}^{K}_{\widetilde{X}}\widetilde{h}_{cd}^{1}.
\eq
In conclusion by \eqref{eq:Liecommutatortildekapparepeat} respectively
\eqref{eq:Liecommutatortildekapparepeatstar}
using \eqref{eq:TheCommutator} we have:
\begin{lemma}\label{lem:TheCommutatorLemma}
We have
\beq
\Lie_{\widetilde{X}}^I \big( \kappa\widetilde{\Box}\, \widetilde{h}_{cd}^1\big)
=	\kappa\widetilde{\Box}\, \widehat{\mathcal{L}}_{\widetilde{X}}^I\widetilde{h}_{cd}^1
	+\kappa \widetilde{R}_{cd}^{\,com\, I}\!\!,
\eq
where
\beq
\bigtwo| \widetilde{R}_{cd}^{\,com\, I}\bigtwo|
\lesssim
{\sum}_{\substack{|J|+|K|\leq |I|+1,\,1\leq |J|\leq |I|}}\Big(\frac{M\ln{\tplusr}}{\tplusr^2}
+\frac{\big|\big(\widehat{\mathcal{L}}_{\widetilde{X}}^J
\widetilde{H}_1\big)_{\widetilde{L}\widetilde{L}}
\big|}{\tminusrstar}
+\frac{|\widehat{\mathcal{L}}_{\widetilde{X}}^J\widetilde{H}_1|}{\tplusr}\Big)
\big|\wpa
\widehat{\Lie}^{K}_{\widetilde{X}}\widetilde{h}^{1}\big|.
\eq
Moreover
\beq\label{eq:Liecommutatortildekapparepeatstar}
	\Box^* \widehat{\mathcal{L}}_{\widetilde{X}}^I
\widetilde{h}_{cd}^{1}
	=\Lie_{\widetilde{X}}^I \big( \kappa\widetilde{\Box}\, \widetilde{h}^1_{cd}\big)- \kappa\widetilde{R}_{cd}^{\,com\,*\, I}\!,
\eq
where
\beq
\bigtwo|\widetilde{R}_{cd}^{\,com\,*\, I}\bigtwo|
\lesssim
{\sum}_{|J|+|K|\leq |I|+1,\, |J|\leq |I|}\Big(\frac{M\ln{\tplusr}}{\tplusr^2}
+\frac{\big|\big(\widehat{\mathcal{L}}_{\widetilde{X}}^J
\widetilde{H}_1\big)_{\widetilde{L}\widetilde{L}}
\big|}{\tminusrstar}
+\frac{|\widehat{\mathcal{L}}_{\widetilde{X}}^J\widetilde{H}_1|}{\tplusr}\Big)
\big|\wpa
\widehat{\Lie}^{K}_{\widetilde{X}}\widetilde{h}^{1}\big|.
\eq
\end{lemma}

\subsubsection{Higher order commutators with the inhomogeneous term}
Finally we note that we have some additional structure of the nonlinear term:
$$
 \widetilde{F}_{ab}(\widetilde{g})[\wpa \widetilde{h},\wpa \widetilde{k}]
\!=\!\widetilde{F}_{ab}[\widetilde{G},\widetilde{G}][\wpa \widetilde{h},\wpa \widetilde{k}]
$$
is a quadratic form in the $(0,3)$ forms $\wpa \widetilde{h}$ and
$\wpa \widetilde{k}$ with two contractions with $\kappa$ times the inverse of the metric
$\widetilde{G}^{ab}=\widetilde{g}^{ab}$, i.e. it is bilinear in $\widetilde{G}$. Hence
we have by \eqref{eq:LieQuadtilde}:
\beq\label{eq:LieQuadtildesec7}
 \Lie_{\widetilde{X}}\big( \widetilde{F}_{\!ab}(\widetilde{g})[\wpa \widetilde{h},\!\wpa \widetilde{k}]\big)
 \!=\!  \widetilde{F}_{\!ab}(\widetilde{g})[\wpa \widehat{\Lie}_{\widetilde{X}}\widetilde{h},\!\wpa \widetilde{k}]
 \!+ \! \widetilde{F}_{\!ab}(\widetilde{g})[\wpa \widetilde{h},\!\wpa \widehat{\Lie}_{\widetilde{X}} \widetilde{k}]
\! +\!\widetilde{F}_{\!ab}[ \widehat{\Lie}_{\widetilde{X}}\widetilde{G},\!\widetilde{G}]
 [\wpa \widetilde{h},\!\wpa \widetilde{k}]
\!+\!\widetilde{F}_{\!ab}[\widetilde{G},\!\widehat{\Lie}_{\widetilde{X}}\widetilde{G}]
 [\wpa \widetilde{h},\!\wpa \widetilde{k}].
\eq
  Let
$\widetilde{h}^{1I}=\widehat{\Lie}^{I}_{\widetilde{X}} \widetilde{h}^1$.
By repeated use of \eqref{eq:LieQuadtildesec7} we get modulo cubic error terms
\begin{equation}
\label{eq:Fchangecoordlongrepeat}
\Lie_{\widetilde{X}}^I
\bigtwo(\widetilde{F}_{\!cd}(\tg)
[\wpa \widetilde{h}^1\!\!,\wpa\widetilde{h}^1]\bigtwo)
={\sum}_{J+K=I}\widetilde{F}_{\!cd}(\tg)
[\wpa \widetilde{h}^{1J}\!\!,\wpa\widetilde{h}^{1K}]+\widetilde{R}_{cd,\, 0}^{\,cube\, I}.
\end{equation}
Here with $\widetilde{G}^I=\widehat{\Lie}^{I}_{\widetilde{X}} \widetilde{G}$,
$\widetilde{G}^{ab}=\widetilde{g}^{ab}$,
\begin{equation}
\widetilde{R}_{cd,\, 0}^{\,cube\, I}={\sum}_{J+K+L+M =I,\, |L|+|M|\geq 1}
R_{cdJKLM}^{\, I}
[\widetilde{G}^L\!\!,\widetilde{G}^M][\wpa \widetilde{h}^{1J}\!\!,\wpa\widetilde{h}^{1K}].
\end{equation}
Hence we have proven
\begin{lemma} \label{lem:differentiatedinhomogen} We have
\begin{equation}
\Lie_{\widetilde{X}}^I
\bigtwo(\widetilde{F}_{\!cd}(\tg)
[\wpa \widetilde{h}^1\!\!,\wpa\widetilde{h}^1]\bigtwo)
={\sum}_{J+K=I}\widetilde{F}_{\!cd}(\tg)
[\wpa \widetilde{h}^{1J}\!\!,\wpa\widetilde{h}^{1K}]+\widetilde{R}_{cd,\, 0}^{\,cube\, I},
\end{equation}
where if  $|\widehat{\Lie}^{J}_{\widetilde{X}} \widetilde{H}_1|\lesssim 1$, for $|J|\leq N/2$,
then for $|I|\leq N$
\begin{equation}
\big|\widetilde{R}_{cd,\, 0}^{\,cube\, I}\big|\les
{\sum}_{|J|+|K|+|L| \leq |I|,\,\, |L|\geq 1}
|\widehat{\Lie}^{L}_{\widetilde{X}} \widetilde{H}_1| \,
|\wpa \widetilde{h}^{1J\!}|\,|\wpa\widetilde{h}^{1K\!}|.
\end{equation}
\end{lemma}
We remark that the
cubic errors $R_{cd}^{\,cube\, I}$ are much easier to control,
see Section \ref{sec:cubicerrors}.
The only issue may be when
most derivatives fall on one factor of $\widetilde{G}$ and we have to estimate it in
$L^2$, but we are going to handle worse terms of this form in the commutator errors.

\subsubsection{The fine structure of the differentiated inhomogeneous term}
Assuming the weak estimate $|\widetilde{H}_1|\lesssim 1$ the difference
\beq
\big|\kappa^{2\!}\widetilde{F}_{ab}[\widetilde{G},\widetilde{G}][\wpa \widetilde{h},\wpa \widetilde{k}]
-\widetilde{F}_{ab}[\kappa\widetilde{m}_0,\kappa\widetilde{m}_0][\wpa \widetilde{h},\wpa \widetilde{k}]\big|
\lesssim|\widetilde{H}_1|\,
|\wpa \widetilde{h}|\,|\wpa \widetilde{k}|,
\eq
can be estimated by cubic terms and by \eqref{eq:metric0decomposition}
\beq
\big|\widetilde{F}_{ab}[\kappa\widetilde{m}_0,\kappa\widetilde{m}_0][\wpa \widetilde{h},\wpa \widetilde{k}]
-\widetilde{F}_{ab}[\widehat{m},\widehat{m}][\wpa \widetilde{h},\wpa \widetilde{k}]\big|
\lesssim \frac{M\ln{\tplusr}}{\tplusr}
|\wpa \widetilde{h}|\,|\wpa \widetilde{k}|.
\eq
Moreover by, \eqref{eq:nullcondestsec5} and \eqref{eq:PPnullsec5} with
$\widehat{P}[\wh,\widetilde{k}]=P(\widehat{m})[\wh,\widetilde{k}]$ we have
\beq
\big| \widetilde{F}_{\!ab}(\widehat{m})
[\wpa \widetilde{h},\wpa\widetilde{k}]
- \wL_{a}\wL_b \widehat{P}(\pa_{q^*} \wh,\pa_{q^*} \widetilde{k})\big|
\les |\overline{\wpa}\, \wh| \,|\wpa \widetilde{k}|
+|\wpa\, \wh|\,|\overline{\wpa} \,\widetilde{k}|.
\eq
Hence we conclude that
\beq
\Big|\kappa^{2\!} \widetilde{F}_{\!ab}(\widetilde{g})
[\wpa \widetilde{h},\wpa\widetilde{k}]
- \wL_{a}\wL_b \widehat{P}(\pa_{q^*} \wh,\pa_{q^*} \widetilde{k})\Big|
\les |\overline{\wpa}\, \wh| \,|\wpa \widetilde{k}|
+|\wpa\, \wh|\,|\overline{\wpa} \,\widetilde{k}|
+\Big(|\widetilde{H}_1|\!+\!\frac{M\ln{\tplusr}}{\tplusr}\Big)|\wpa \widetilde{h}|\,|\wpa \widetilde{k}| .
\eq

Hence, also using \eqref{eq:tanPsec5} and \eqref{eq:tanSPsec5} to estimate $\widehat{P}$ we have proven:
\begin{lemma}\label{lem:PPS} Suppose  that $|\widehat{\Lie}^{J}_{\widetilde{X}} \widetilde{H}_1|\lesssim 1$, for $|J|\leq N/2$.
Then for $|I|\leq N$ with $\widetilde{h}^{1I}=\widehat{\Lie}^{I}_{\widetilde{X}} \widetilde{h}^1$
we have
\beq
\kappa^{2\!} \widetilde{F}_{\!cd}(\tg)
[\wpa \widetilde{h}^{1J}\!\!,\wpa\widetilde{h}^{1K}]
=\!\wL_{c}\wL_d
\widehat{P}(\pa_{q^*} \wh^{1J}\!\!,\pa_{q^*} \widetilde{h}^{1K})
+\kappa^2 \widetilde{R}_{cd}^{\,tan\, JK}+\kappa^2\widetilde{R}_{cd,\, 1}^{\,cube\, JK},
\eq
where
\beq
\bigtwo|
\widetilde{R}_{cd}^{\,tan\, JK}
\bigtwo|
\lesssim\! |\overline{\wpa}\, \wh^{1J}| \,|\wpa \widetilde{h}^{1K}|
+|\wpa  \wh^{1J}| \,|\overline{\wpa}\,\widetilde{h}^{1K}|
,\quad \text{and}\quad |\widetilde{R}_{cd,\, 1}^{\,cube\, JK}|\les\Big(|\widetilde{H}_1|\!+\!\frac{M\ln{\tplusr}}{\tplusr}\Big)|\wpa \widetilde{h}^{1J}|\,|\wpa \widetilde{h}^{1K}|
 \eq
  and with $\widehat{P}[\wh,\widetilde{k}]=P(\widehat{m})[\wh,\widetilde{k}]$ we have
\begin{multline}
\big| \widehat{P}(\pa_{q^*} \wh^{1J}\!\!,\pa_{q^*} \widetilde{h}^{1K})\big|\\
\les \big(|\pa_{q^*}\wh^{1J\!}|_{\wL\widetilde{\mathcal
T}}+|\pa_{q^*}\trs \wh^{1J\!}|\big)|\wpa\,\widetilde{h}^{1K\!}| +|\wpa\,\wh^{1J\!}
|\big(|\pa_{q^*}\widetilde{h}^{1K\!}|_{\wL\widetilde{\mathcal
T}}+|\pa_{q^*}\trs \widetilde{h}^{1K\!}|\big)
+|\wpa\, \wh^{1J\!}|_{\widetilde{\mathcal{T}}\widetilde{\mathcal{T}}}|\wpa\,
 \widetilde{h}^{1K\!}|_{\widetilde{\mathcal{T}}\widetilde{\mathcal{T}}}.
\end{multline}
\end{lemma}

\subsubsection{Higher order vector fields applied to the mass and covariant error terms
in Einstein's equations}
We have $\widetilde{E}_{cd}^{\, mass}=\widetilde{E}_{cd,0}^{\, mass}\!\!
 +\widetilde{E}_{cd,1}^{\, mass}[\widetilde{H}_1]$ where it follows from  \eqref{eq:Emass} that
\beq
\big|\Lie_{\widetilde{X}}^I\widetilde{E}_{cd,0}^{\, mass}\big|\lesssim
\frac{M H(r \!<\! t/2)}{\tplusr^3},\quad\text{and}\quad
\big|\Lie_{\widetilde{X}}^I\widetilde{E}_{cd,1}^{\, mass}[\widetilde{H}_1]\big|
\lesssim \frac{M}{\tplusr^3}{\sum}_{|J|\leq |I|}\big|\widetilde{Z}^J
\widetilde{H}_1\big|.\label{eq:Emasserrorsestimate}
\eq
Moreover $\widetilde{E}_{cd}^{\, cov}=\widetilde{E}_{cd,0}^{\, cov}[\widetilde{H}_1]
 +\widetilde{E}_{cd,1}^{\, cov}[\wpa \widetilde{H}_1]$
 where it follows from  \eqref{eq:Ecov1}-\eqref{eq:Ecov2} that
 \beq
\big|\Lie_{\widetilde{X}}^I\widetilde{E}_{cd,0}^{\, cov}[\widetilde{H}_1]\big|\!\lesssim\!
\frac{M\big(\!\ln{\tplusr}\!+\!1\big)}{\tplusr^3}
\!\!\!\!\sum_{|J|\leq |I|}\!\!\!\big|\widetilde{Z}^J\!\widetilde{H}_1\big|,
\quad\text{and}\quad
\big|\Lie_{\widetilde{X}}^I\widetilde{E}_{cd,1}^{\, cov}[\wpa \widetilde{H}_1]\big|\!\lesssim\!
\frac{M\big(\!\ln{\tplusr}\!+\!1\big)}{\tplusr^2}
\!\!\!\!\sum_{|J|\leq |I|}\!\!\!\big|\wpa\widetilde{Z}^J\!\widetilde{H}_1\big|.
\label{eq:Ecoverrorsestimate}
\eq
 We have $\widetilde{F}_{cd}^{\, mass}=\widetilde{F}_{cd,0}^{\, mass}\!\!
 +\widetilde{F}_{cd,1}^{\, mass}[\wpa\widetilde{H}_1]$ where it follows from
 \eqref{eq:Fmass} that
\beq
\big|\Lie_{\widetilde{X}}^I\widetilde{F}_{cd,0}^{\, mass}\big|\lesssim
\frac{M^2 }{\tplusr^4},\quad\text{and}\quad
\big|\Lie_{\widetilde{X}}^I\widetilde{F}_{cd,1}^{\, mass}[\wpa\widetilde{H}_1]\big|
\lesssim \frac{M}{\tplusr^2}{\sum}_{|J|\leq |I|}
\big|\wpa \widetilde{Z}^J
\widetilde{H}_1\big|.\label{eq:Fmasserrorsestimate}
\eq
Moreover $\widetilde{F}_{cd}^{\, cov}=\widetilde{F}_{cd,\,0}^{\, cov}[\widetilde{H}_1]
 +\widetilde{F}_{cd,1}^{\, cov}[\wpa \widetilde{H}_1]$
  where it follows from  \eqref{eq:Fcov1}-\eqref{eq:Fcov2} that
 \beq
\big|\Lie_{\widetilde{X}}^I\widetilde{F}_{cd,\,0}^{\, cov}[\widetilde{H}_1]\big|\!\lesssim\!
\frac{M\big(\!\ln{\tplusr}\!+\!1\big)}{\tplusr^3}
\!\!\!\!\sum_{|J|\leq |I|}\!\!\!\big|\widetilde{Z}^J\!\widetilde{H}_1\big|,
\quad\text{and}\quad
\big|\Lie_{\widetilde{X}}^I\widetilde{F}_{cd,1}^{\, cov}[\wpa \widetilde{H}_1]\big|\!\lesssim\!
\frac{M\big(\!\ln{\tplusr}\!+\!1\big)}{\tplusr^2}
\!\!\!\!\sum_{|J|\leq |I|}\!\!\!\big|\wpa\widetilde{Z}^J\!\widetilde{H}_1\big|.
\label{eq:Fcoverrorsestimate}
\eq

\begin{lemma}\label{lem:errorlemma} Let
$\widetilde{R}_{cd}^{\,mass}\!=\widetilde{E}_{cd}^{\,mass}\!
+\widetilde{F}_{cd}^{\,mass}$
and $\widetilde{R}_{cd}^{\,cov}\!=\widetilde{E}_{cd}^{\,cov}\!
+\widetilde{F}_{cd}^{\,cov}$
be as in \eqref{eq:newcoordeinsteinwithRerrors} and Section \ref{sec:cov}
and let $\widetilde{R}_{cd}^{\,mass\, I}\!=\widehat{\Lie}^{I}_{\widetilde{X}}
\widetilde{R}_{cd}^{\,mass}$ and $\widetilde{R}_{cd}^{\,cov\, I}\!=\widehat{\Lie}^{I}_{\widetilde{X}}
\widetilde{R}_{cd}^{\,cov}$. We have
\begin{equation}
|\widetilde{R}_{cd}^{\,mass\, I}|
+|\widetilde{R}_{cd}^{\,cov\, I}|
\!\lesssim \!\frac{\!M H(r \!<\! 3t/4)\!}{\tplusr^3}+\frac{M^2 }{\!\tplusr^4\!}
+\frac{\!M\big(\!\ln{\tplusr}\!+\!1\big)\!}{\tplusr^3}
\!\!\!\!\sum_{|J|\leq |I|}\!\!\!\big|\widetilde{Z}^J\!\widetilde{H}_1\big|
+\frac{\!M\big(\!\ln{\tplusr}\!+\!1\big)\!}{\tplusr^2}
\!\!\!\!\sum_{|J|\leq |I|}\!\!\!\big|\wpa\widetilde{Z}^J\!\widetilde{H}_1\big|.
\end{equation}
\end{lemma}

\subsection{Higher order commutators with the wave coordinate condition}\label{sec:higherorderwavecoord}
Let $\widecheck{H}_1^{ac}\!=\widetilde{H}_1^{ac}\!
-\tfrac{1}{2}\widetilde{m}^{ac}  \widetilde{m}_{bd}\,
 \widetilde{H}_1^{bd} $.
By repeated use of \eqref{eq:Liedivergencetilde} we have
 \beq\label{eq:Liedivergencetilderepeat}
	 \wpa_c  \widehat{\Lie}_{\widetilde{X}}^I \widecheck{H}_1^{cd}
={\sum}_{|J|\leq |I|} c_J^I\widehat{\Lie}_{\widetilde{X}}^J
\wpa_c \widecheck{H}_1^{cd}.
\eq
We claim that for any $\widetilde{H}^{ab}$
\begin{equation}\label{eq:firstwavecoordinateestimate2}
|\pa_{q^*} \widetilde{H}|_{\widetilde{L}\widetilde{\mathcal{T}}}+|\pa_{q^*} \trs \widetilde{H}|
\lesssim |\overline{\wpa} \widetilde{H}|+|\operatorname{div}\widehat{F}|,\quad\text{if}\quad
\operatorname{div}\widecheck{H}^b\!=\wpa_a \widecheck{H}^{a b}\!\!,\quad \text{and}\quad
\widecheck{H}^{ac}\!=\widetilde{H}^{ac}\!
-\tfrac{1}{2}\widetilde{m}^{ac}  \widetilde{m}_{bd}\,
 \widetilde{H}^{bd}\!\! .
\end{equation}
In fact, expressing the divergence in a null frame;
\beq\label{eq:divergenceinnullframetilde}
\wpa_a \widecheck{H}^{a b}\!=
\widetilde{L}_a\pa_{q^*\!} \widecheck{H}^{a b} \! -\widetilde{\Lb}_{a} \pa_{s^* \!}\widecheck{H}^{a b}\!
 + \widetilde{S}_{1 a} \pa_{\widetilde{S}_1}  \!\widecheck{H}^{a b}\!
 + \widetilde{S}_{2 a} \pa_{\widetilde{S}_2}  \!\widecheck{H}^{a b}\!\!,\quad
\text{where}\quad\pa_{q^*}\!=(\pa_{r^*}\!-\pa_t)/2,\quad \pa_{s^*}\!=(\pa_{r^*}\!+\pa_t)/2,
\eq
and contracting with $\widetilde{T}_b\in \widetilde{\mathcal{T}}$ respectively
$\widetilde{\underline{L}}_b$ proves \eqref{eq:firstwavecoordinateestimate2}.
Applying this to \eqref{eq:Liedivergencetilderepeat} gives
\beq
 |\pa_{q^*} \widehat{\mathcal L}_{\widetilde{Z}}^I\widetilde{H}_{\!1}|_{\widetilde{L}\widetilde{\mathcal T}}\!
+|\pa_{q^*}  \trs{}_{\!}\widehat{\mathcal L}_{\widetilde{Z}}^I \widetilde{H}_{\!1} |
\!\les\! |\overline{\wpa} \widehat{\mathcal L}_{\widetilde{Z}}^I \widetilde{H}_{\!1}|
+\!\!\!\! \sum_{|J|\leq |I|} \!\!\!\! |  \widehat{\mathcal L}_{\widetilde{Z}}^{J}\!
\operatorname{div}{}_{\!} \widecheck{H}_{\!1\!}|,
\quad \operatorname{div}{}_{\!} \widecheck{H}_1^d\!=\wpa_c \widecheck{H}_1^{cd}\!\!,\quad
\widecheck{H}_1^{ac}\!=\!\widetilde{H}_1^{ac}\!\shortminus \tfrac{1}{2}\widetilde{m}^{ac}  \widetilde{m}_{bd}
 \widetilde{H}_1^{bd}\!\! .
\label{eq:higherwavecoordinateLiederivativeH1}
\eq
By \eqref{eq:higherwavecoordinateLiederivativeH1} applied to \eqref{eq:wavecoordhattilde}
(with notation as in \eqref{eq:newcoordwavecoord})
we have
\begin{lemma} Let $\widetilde{H}_1^J=\widehat{\mathcal L}_{\widetilde{Z}}^I
\widetilde{H}_1$. We have
\begin{multline}
|\pa_{q^*} \widehat{\mathcal L}_{\widetilde{Z}}^I
\widetilde{H}_1|_{\widetilde{L}\widetilde{\mathcal T}}
+|\pa_{q^*}  \trs\widehat{\mathcal L}_{\widetilde{Z}}^I \widetilde{H}_1 |
\les |\overline{\wpa} \widehat{\mathcal L}_{\widetilde{Z}}^I \widetilde{H}_1|
+ {\sum}_{|J|+|K|\leq |I|} \big|  \widehat{\mathcal L}_{\widetilde{Z}}^{J} \widetilde{h}^1\big|
    \big|\pa  \widehat{\mathcal L}_{\widetilde{Z}}^{K} \widetilde{h}^1\big|\\
    + \frac{M\big|\chi^\prime\big(\tfrac{r}{t+1}\big)\big|}{\tplusr^2}
    + \frac{M}{\tplusr^3}+ \frac{M\big(1+\ln{\tplusr}\big)}{\tplusr^2}
    {\sum}_{|J|\leq |I|} \big|  \widehat{\mathcal L}_{\widetilde{Z}}^{J}\widetilde{H}_{1}\big|
    + \frac{M}{\langle r\!+\!t\rangle}
    {\sum}_{|J|\leq |I|} \big| \wpa \widehat{\mathcal L}_{\widetilde{Z}}^{J}\widetilde{H}_{1}\big|,
\label{eq:higherwavecoordinateLiederivativeH1Wsec7}
\end{multline}
where $\chi^\prime(s)$, is a function supported when $1/4\leq s\leq 1/2$.
\end{lemma}
\section{The $L^2$ estimates for the wave equation}\label{sec:EnergyEstimates}
\subsection{The energy estimate with asymptotically Schwarzschild coordinates and weights}
In the energy estimate on Minkowski space, one can introduce a spacetime integral by multiplying by a weight $w(q)$, with $w'>0$, since $\nabla w$ is a future directed null vector field. For the metric $g$, the positive mass theorem implies $\lVert\nabla w\rVert_g\sim 2Mr^{-1}$, so we instead adapt the weight to the approximate optical function $q^*$:
\begin{equation}\label{def:wdef}
w = w(q^*) =\bigg\{ \begin{aligned}
&1 + (1+q^*)^{1+2\gamma},\qquad & q^* \geq 0, \\
&1 + (1-q^*)^{-2\mu},\qquad & q^* < 0.
\end{aligned}
\end{equation}
Then, $\lVert \nabla w\rVert_g = o(r^{-1})$ along the light cone. If $\mu < 1/2$, we have the basic inequality
\[
(1+2\gamma)^{-1}(1+|q^*|)\,\pa_{q^*}w \leq w \leq 2(2\mu)^{-1}(1+|q^*|)^{1+2\mu}\pa_{q^*}w.
\]
The energy momentum tensor of the wave equation in the modified coordinates is
\begin{equation}
\widetilde{T}_{ab}[\phi] = \wpa_a\phi \,\wpa_b\phi - \frac12 \wt{g}_{ab}\wt{g}^{\,cd}\wpa_c\phi\,\wpa_d\phi.
\end{equation}
This is consistent with the energy momentum tensor $T[\bphi, \bf{0}]$ for the Maxwell-Klein Gordon equations in modified coordinates. We can define the weighted scalar current
\[
\widetilde{J}_w^a[\phi] = -\widetilde{g}^{\,ab}\widetilde{T}_{b0}[\phi]w = -\bigtwo(\widetilde{g}^{\,ab}\wpa_b\phi\,\wpa_t\phi - \delta_0^a\widetilde{g}^{\,bc}\wpa_b\phi\,\wpa_c\phi/2\bigtwo)w,
\]
the weighted energy
\beq
\mathcal{E}[\phi](T) = \sup_{t\in[0,T]}\int_{\Sigma_t}|\wpa\phi|^2 w \, d\wt{x},
\eq
and the weighted spacetime energy
\beq
\mathcal{S}[\phi](T)
= \int_0^T\int_{\Sigma_t}|\wpao\phi|^2w'\, d\wt{x}\, dt,
\eq
where $|\wpao\phi|^2=|\wt{S}_1\phi|^2 +|\wt{S}_2\phi|^2 +|\wL \phi|^2$ is the norm of the derivatives tangential to the outgoing curved light cones
$r^*-t=q^*$, and $|\swpa\phi|^2=|\wt{S}_1\phi|^2 + |\wt{S}_2\phi|^2$ is the norm of the derivatives tangential to the sphere with constant $r^*$ and $t$.
The following is a generalization of the theorem in \cite{LR10,LT18}.
\begin{theorem}\label{thm:EEst}
Take $\gamma, \mu > 0$. There exists an $\varepsilon_0 > 0$ (which depends on $\gamma, \mu$) such that, if $g$ satisfies the metric assumption
\begin{subequations}\label{est:H1}
\begin{align}
M & \leq \varepsilon_0,\label{est:H1M} \\
|\wt{H}_1|+\langle t-r^*\rangle|\wpa \wt{H}_1|+ \langle t+r^*\rangle
(|\wpao\wt{H}_1|)
&\leq \varepsilon_0\langle t-r^*\rangle^{1/2-\mu}\langle t+r^*\rangle^{-1/2-\mu},\label{est:H11} \\
|\wt{H}_1|_{\wt{L}\wt{L}}+\langle t-r^*\rangle|\wpa \wt{H}_1|_{\wt{L}\wt{L}}+
\langle t+r^*\rangle(|\wpao\wt{H}_1|_{\wt{L}\wt{L}} )&\leq \varepsilon_0\langle t-r^*\rangle\langle t+r^*\rangle^{-1-2\mu},\label{est:H12}
\end{align}
\end{subequations}
with
\[
|\wpa \wt{H}_1|_{\wt{L}\wt{L}} = {\sum}_{\wt{U}\in\{\wt{L}, \wt{\underline{L}}, \wt{S}_1, \wt{S}_2
\}}|\wt{U}^a\wt{L}^b\wt{L}^c\wpa_a\wt{H}_{bc}|,\qquad
|\wpao\wt{H}_1|^2_{\wt{L}\wt{L}}
= {\sum}_{\wt{U}\in\{\wt{L}, \wt{S}_1, \wt{S}_2\}}|\wt{U}^a\wt{L}^b\wt{L}^c\wpa_a\wt{H}_{bc}|^2
\]
then with
\[
\widetilde{\Box}={\Box}^{\wt{g}}=\wt{g}^{\,ab}\wpa_a\wpa_b
\]
we have
\beq
\sup_{t\in[0,T]}\int_{\Sigma_t}|\wpa\phi|^2 w \, d\wt{x} +\int_0^T\int_{\Sigma_t}|\wpao\phi|^2w'\, d\wt{x}\, dt\leq 8\int_{\Sigma_0}|\wpa\phi|^2 w \, d\wt{x} + 12\int_0^T\int_{\Sigma_t}|\widetilde{\Box}\phi||\wpa\phi|w\, dx \, dt.
\eq
\end{theorem}
\begin{remark}
If we had used pointwise bounds on components of $H$ with respect to the null frame in Minkowski space, we would have only been able to prove slowly growing energy,  even in the case $\widetilde{\Box}\phi = 0$ (cf. \cite{LR10}).
\end{remark}
Before proving this, we state some consequences of the metric assumptions.
\begin{proposition}\label{prop:MetCons}
If we replace $\wt{H}_1^{ab}$ with $(\wt{g}-\widehat{m})^{ab}$, then for $\varepsilon_0 < 1$ the metric assumptions \eqref{est:H11} and \eqref{est:H12} hold up to a constant (replacing $\varepsilon_0$ with $C\varepsilon_0$ on the right hand side). We additionally have the following estimates:
\begin{subequations}\label{est:derH}
\begin{align}
|\wpa_t \wt{H}_1^{ab}\wpa_a\phi\,\wpa_b\phi|
&\leq C\varepsilon_0\big( \langle{t+r^*}\rangle^{-1-2\mu}|\wpa\phi|^2
 + \langle{t-r^*}\rangle^{-1/2-\mu}\langle{t+r^*}\rangle^{-1/2-\mu}
|\wpao\phi||\wpa\phi|\big),\label{est:derH1}\\
|\wpa_a \wt{H}_1^{ab} \wpa_b\phi|
&\leq C\varepsilon_0\big(\langle{t+r^*}\rangle^{-1-2\mu}|\wpa\phi|
+ \langle{t-r^*}\rangle^{-1/2-\mu}\langle{t+r^*}\rangle^{-1/2-\mu}
|\wpao\phi|\big).\label{est:derH2}
\end{align}
\end{subequations}
\end{proposition}

\begin{proof}
The first statement follows directly from Lemma \ref{lem:m0approx}. The estimates \eqref{est:derH1} and \eqref{est:derH2} are trivial when $r^* < (t+1)/2$, as we have the approximation $\langle t-r^*\rangle \approx \langle t+r^*\rangle$. In the exterior region, this is slightly more complicated. We prove the estimate \eqref{est:derH2} here; \eqref{est:derH1} follows from similar reasoning. We can expand
\[
\wpa_a\wt{H}_1^{ab}\wpa_b\phi = \widehat{m}^{ac}\widehat{m}^{bd}\wpa_a (\wt{H}_{1cd})\wpa_b\phi + O(|g-\widehat{m}||\wpa H_1|)\wpa\phi,
\]
and further expand $\widehat{m}^{ab}, \widehat{m}^{cd}$ in our null frame, which gives
\[
|\wpa_a\wt{H}_1^{ab}\wpa_b\phi|\lesssim  \varepsilon\langle t-r^*\rangle^{-2\mu}\langle{t+r^*}\rangle^{-1-2\mu}|\wpa\phi|
+ |\wpa \wt{H}_1|(|\wpao\phi|) + (|\wpao{\wt{H}}_1| +|\wpa \wt{H}_1|_{\wt{L}\wt{L}})|\wpa\phi|
\]
 Our estimate follows.
\end{proof}

\begin{proof}[Proof of Theorem \ref{thm:EEst}]
We apply the divergence theorem in the $(\wt{t},\wt{x})$ coordinates to $\widetilde{J}_w^a[\phi]$ on $[0,T]\times \mathbb{R}^3$ .
\beq\label{est:divergence}
\int_{\Sigma_T}\widetilde{J}_w^{\,0}[\phi]\, d\widetilde{x}-  \int_{\Sigma_0}\widetilde{J}_w^{\,0}[\phi] \, d\widetilde{x} =\int_0^T\int_{\Sigma_t}\wpa_a(\widetilde{J}_w^a[\phi])\, d\widetilde{x} \, d\widetilde{t}.
\eq
Defining
\[
E[\phi](t) = \int_{\Sigma_t}|_{\,}\wpa\phi|^2 w \, d\widetilde{x},
\]
A pointwise calculation gives
\beq
\int_{\Sigma_t}\left|\tfrac12\big|\wpa\phi\big|^2 - \wt{T}_{00}\right|\, d\wt{x}\lesssim \varepsilon_0 E[\phi](t).
\eq
Additionally,
\beq
|(-\wt{g}^{\,0b}+\widehat{m}^{0b})\widetilde{T}[\phi]_{b0}w|\lesssim \varepsilon_0 E[\phi](t).
\eq
Therefore, for sufficiently small $\varepsilon_0$, the inequality
\begin{equation}\label{est:timeslice}
\tfrac13E[\phi](t) \leq \int_{\Sigma_t}\widetilde{J}_w^{\,0}[\phi]\, d\wt{x} \leq  \tfrac23E[\phi](t).
\end{equation}
holds for all $t$. Now we estimate the right hand side. We write
\beq
-\wpa_a(\widetilde{g}^{\,ab}\widetilde{T}[\phi]_{b0}\,w) = (-\wpa_a\widetilde{g}^{\,ab})\wpa_b\phi\,\wpa_t\phi\, w + \frac12(\wpa_t\widetilde{g}^{\,bc})\wpa_b\phi\,\wpa_c\phi \, w-\widetilde{g}^{\,ab}\wpa_a\wpa_b\phi\,\wpa_0\phi\, w - \widetilde{g}^{\,ab}\widetilde{T}_{b0}\wpa_a\, w
\eq
In the far interior, $\frac{r}{t+1}<\frac34$, the bounds \eqref{est:H1} imply
\[
|\widetilde{g} - \widehat{m}|\lesssim \varepsilon_0\langle t \rangle^{-2\mu} , \qquad |\wpa\widetilde{g}|\lesssim \varepsilon_0 \langle t \rangle^{-1-2\mu},\qquad|w|\lesssim 1, \qquad \qquad |\wpa w| \lesssim \langle t \rangle^{-1-2\mu},
\]
and consequently
\[
|(-\wpa_a\widetilde{g}^{\,ab})\wpa_b\phi\,\wpa_t\phi\, w| + |(\wpa_t\widetilde{g}^{\,bc})\wpa_b\phi\,\wpa_c\phi\, w|  + |g^{\,ab}\widetilde{T}_{b0}\wpa_a \,w + \tfrac12\widetilde{T}_{\wt{L}0}\wt{\underline{L}}w|\lesssim \varepsilon_0 \langle t\rangle^{-1-2\mu}|\wpa\phi|^2 w,
\]
follows directly. Outside this region, we use the null decomposition. Lemma \ref{lem:m0approx} implies
\[
|\wpa(\widetilde{g}^{\,ab} - \widetilde{H}_1^{ab})| \lesssim \frac{M\ln(1+\langle t+r^*\rangle)}{\langle t+r^*\rangle^2},
\]
and therefore
\[
|\wpa_a(\widetilde{g}^{\,ab} -\widetilde{H}_1^{ab})\wpa_b\phi\,\wpa_t\phi \,w| + |\wpa_t(\widetilde{g}^{\,bc} - \wt{H}_1^{bc})\wpa_b\phi\,\wpa_c\phi \,w| \lesssim \frac{M\ln(1+\langle t+r^*\rangle)}{\langle t+r^*\rangle^2}|\wpa\phi|^2 w.
\]

The inequalities \eqref{est:derH} imply
\begin{align*}
|\wpa_t \wt{H}_1^{ab}\wpa_a\phi\,\wpa_b\phi\, w|&\lesssim \varepsilon_0
\big(\langle{t+r^*}\rangle^{-1-2\mu}|\wpa\phi|^2w
+ \langle{t-r^*}\rangle^{-1-2\mu}|\wpao\phi|^2 w\big),\\
&\lesssim \varepsilon_0\big(\langle{t+r^*}\rangle^{-1-2\mu}|\wpa\phi|^2 +
|\wpao\phi|^2 w'\big),\\
|\wpa_a \wt{H}_1^{ab} \wpa_b\phi\, \wpa_t\phi|
&\lesssim \varepsilon_0\big(\langle{t+r^*}\rangle^{-1-2\mu}|\wpa\phi|^2
+ \langle{t-r^*}\rangle^{-1/2-\mu}\langle{t+r^*}\rangle^{-1/2-\mu}
|\wpao\phi||\wpa\phi|\big),\\
&\lesssim \varepsilon_0\big(\langle{t+r^*}\rangle^{-1-2\mu}|\wpa\phi|^2 +
 |\wpao\phi|^2 w'\big).
\end{align*}
Therefore, we have the pointwise inequality
\begin{equation}\label{est:DerMetric}
|\wpa_a \wt{g}^{\,ab}\wpa_b\phi\,\wpa_t\phi\,w|
 + |\wpa_t \wt{g}^{\,bc}\wpa_b\phi\,\wpa_c\phi\, w|
 \lesssim \varepsilon_0\big(\langle{t+r^*}\rangle^{-1-2\mu}|\wpa\phi|^2 +
 |\wpao\phi|^2 w'\big)
\end{equation}
Since the integral of $\langle t \rangle^{-1-2\mu}$ is finite in time, this gives the bound
\begin{align}
\int_0^T\!\!\!\!\!\int_{\Sigma_t}\!\!\!\!|
\wpa_a \wt{g}^{\,ab}\wpa_b\phi\,\wpa_t\phi\, w|\!
 + |\wpa_t \wt{g}^{\,bc}\wpa_b\phi\,\wpa_c\phi\, w|\, d\wt{x}\, d\wt{t}
 \lesssim \varepsilon_0(\mathcal{E}[\phi](T) + \mathcal{S}[\phi](T))
\end{align}

When $\wpa$ falls on $w$ give the spacetime energy $\mathcal{S}[\phi](T)$ plus error terms. We first bound
\beq
|\wt{g}^{\,ab}\wt{T}_{0a}\wpa_b w +\frac12\wt{T}_{0\wt{L}}\underline{\wt{L}} w|
\lesssim \varepsilon_0\big(\langle t-r^*\rangle^{1/2-\mu}
\langle t+r^*\rangle^{-1/2-\mu}|\wpa\phi||\wpao\phi|
+ \langle t-r^*\rangle\langle t+r^*\rangle^{-1-2\mu}|\wpa\phi|^2\big)w',
\eq
which follows from the null decomposition and the first statement of Proposition \ref{prop:MetCons}.
The inequality $\langle t-r^*\rangle w' \lesssim w$ and Hölder's inequality allow us to bound the integral of the right hand side by a constant times $\varepsilon_0(\mathcal{E}[\phi](T) + \mathcal{S}[\phi](T))$. Additionally,
\[
|(\widetilde{T}_{\wt{L}0} - \frac12|\overline{\wt{\partial}}\phi|^2)
\wt{\underline{L}}w| \lesssim \varepsilon_0\big(\langle t-r^*\rangle
\langle t+r^*\rangle^{-1-2\mu}|\wpa\phi|^2
+ \varepsilon_0\langle t-r^*\rangle^{1/2-\mu} \langle t+r^*\rangle^{-1/2-\mu}
|\wpa\phi||\wpao\phi|\big)w'.
\]
We can bound this in the same way. Therefore,
\begin{equation}\label{est:spacetime}
\int_0^T\int_{\Sigma_t}|\wt{g}^{\,ab}\wt{T}_{0a}\wpa_b w +\frac14|\overline{\wt{\partial}}\phi|^2\wt{\underline{L}}w|\lesssim \varepsilon_0(\mathcal{E}[\phi](T) + \mathcal{S}[\phi](T)).
\end{equation}
We recall that
\begin{equation}
\mathcal{S}[\phi](T) = -\int_0^T\int_{\Sigma_t}\tfrac12|\overline{\wt{\partial}}\phi|^2\wt{\underline{L}}w.
\end{equation}
Then,
\beq
\int_0^T\int_{\Sigma_t}\left|\wpa_a(\widetilde{J}_w^a[\phi]) + \widetilde\Box\phi\wpa_t\phi-\tfrac14|\overline{\wt{\partial}}\phi|^2\wt{\underline{L}}w\right|\, d\widetilde{x} \, d\widetilde{t} \lesssim \varepsilon_0(\mathcal{E}[\phi](T) + \mathcal{S}[\phi](T))
\eq
To close the proof, we first rewrite \eqref{est:divergence} as
\begin{align}
\label{eqn:EnergyIdentity}
\int_{\Sigma_T}\!\!\!\widetilde{J}_w^{\,0}[\phi]\, d\widetilde{x}= \int_{\Sigma_0}\!\!\!\widetilde{J}_w^{\,0}[\phi] \, d\widetilde{x} + \int_0^T\!\!\!\int_{\Sigma_t}\!\!\!\left(\wpa_a(\widetilde{J}_w^a[\phi]) + \widetilde\Box\phi\wpa_t\phi -\tfrac14|\overline{\wt{\partial}}\phi|^2\wt{\underline{L}}w\right) -\widetilde\Box\phi\wpa_t\phi +\tfrac14|\overline{\wt{\partial}}\phi|^2\wt{\underline{L}}w\, d\widetilde{x} \, d\widetilde{t}
\end{align}
Therefore, there exists a $C$ depending on $\gamma, \mu$ such that
\begin{equation}
\tfrac13E[\phi](T) +\tfrac12 \mathcal{S}[\phi](T) \leq \tfrac23E[\phi](0)+ C\varepsilon_0(\mathcal{E}[\phi](T) + \mathcal{S}[\phi](T)) + \int_0^T\int_{\Sigma_t}|\wpa_t\phi||\widetilde{\Box}\phi|w\, d\widetilde{x} \, d\widetilde{t}
\end{equation}
We repeat this estimate with $T$ replaced with $T'\in [0,T]$ where $E[\phi]$ attains its maximum; i.e., $E[\phi](T') = \mathcal{E}[\phi](T)$, and add the two estimates to get:
\begin{equation}
\tfrac13\mathcal{E}[\phi](T) +\tfrac12 \mathcal{S}[\phi](T) \leq \tfrac43E[\phi](0)+ 2C\varepsilon_0(\mathcal{E}[\phi](T) + \mathcal{S}[\phi](T)) + 2\int_0^T\int_{\Sigma_t}|\wpa_t\phi||\widetilde{\Box}\phi|w\, d\widetilde{x} \, d\widetilde{t}.
\end{equation}
Multiplying this by 6 and assuming $2C\varepsilon_0 < 1$ gives
\begin{equation}
\mathcal{E}[\phi](T) +\mathcal{S}[\phi](T) \leq 8E[\phi](0)+ 12\int_0^T\int_{\Sigma_t}|\wpa_t\phi||\widetilde{\Box}\phi|w\, d\widetilde{x} \, d\widetilde{t}.{}\tag*{\qedhere}
\end{equation}
\end{proof}

\subsection{Poincare lemmas with weights}\label{sec:poincare}

We restate results from \cite{LR10}, which will be particularly useful in bounding $L^2$ norms of Lie derivatives of $H$.
\begin{lemma}\label{est:1dpoincare}
Let $0\leq a\leq 2$, $\mu > -1/2$, and $\gamma > 0$, and let $\phi$ be a differentiable function on $[0, \infty)$ such that $r^{\gamma}\phi$ vanishes at $\infty$. Then, for all $t\geq 0$, we have the inequality
\begin{equation}
\int_0^\infty \frac{|\phi|^2\overline{w}_{\gamma, \, \mu}(r^*-t)}{(1+|r^*-t|)^2(1+t+r^*)^a}\, r^{*2}\, dr^* \leq C_{\gamma, \mu} \int_0^\infty \frac{|\pa_{r^*}\phi|^2\overline{w}_{\gamma, \,\mu}(r^*-t)}{(1+t+r^*)^a}\, r^{*2}\, dr^*,
\end{equation}
where
\begin{equation}
\overline{w}_{\gamma, \mu}(y) = \bigg\{\begin{aligned}
&(1+|y|)^{1+2\gamma},\qquad & y \geq 0, \\
&(1+|y|)^{-2\mu},\qquad & y < 0.
\end{aligned}
\end{equation}
\end{lemma}
\begin{proof}
Let $f$, $g$, and $\phi$ be functions of $\rs$. By a density argument, we can assume $\phi$ is smooth and compactly supported in $[0, \infty)$. Then,
\[
\int_0^\infty(Cf\pa_{\rs}\phi + g\phi)^2\, dr^* - \int_0^\infty \pa_{r^*}(Cfg\phi^2)\, dr^* \geq 0.
\]
as long as $fg\phi^2$ vanishes at 0 and $\infty$. Therefore,
\[
\int_0^\infty C^2|f\pa_{r^*}\phi|^2\, dr^* \geq \int_0^\infty \big[C\pa_{r^*}(fg) - g^2\big]|\phi|^2\, d\rs.
\]

Now set
\begin{align*}
f(r^*) &= r^{*}\overline{w}_{\gamma, \mu}^{1/2}(1+t+r^*)^{-a/2}, \\
g(r^*) &= r^{*}(1+|q^*|)^{-1}\overline{w}_{\gamma,\mu}^{1/2}(1+t+r^*)^{-a/2}.
\end{align*}
Then,
\[
\partial_{r^*}(fg) = \left(\frac{2}{r^*} - \frac{a}{1+t+r^*} - \frac{\sgn q^*}{(1+|q^*|)} + \frac{\overline{w}_{\gamma, \mu}'(q^*)}{\overline{w}_{\gamma, \mu}(q^*)} \right)fg.
\]
Since $a < 2$,  $\frac{2}{r^*} - \frac{a}{1+t+r^*} > 0$, and
\[
-\frac{\sgn q^*}{(1+|q^*|)} + \frac{\overline{w}_{\gamma, \mu}'(q^*)}{\overline{w}_{\gamma, \mu}(q^*)} = \begin{cases}
\frac{1+2\mu}{(1+|q^*|)}, & q^* < 0, \\
\frac{2\gamma}{(1+|q^*|),} & q^* > 0.
\end{cases}
\]
Setting $C = 2\max((2\gamma)^{-1}, (1+2\mu)^{-1})$ gives our result.
\end{proof}
This lemma has two immediate consequences, which will be useful in the fixed-time and spacetime energies.
\begin{corollary}\label{cor:poincare}
Recalling the definition \eqref{def:wdef}, for fixed $t >0$, $-1 \leq b \leq 1$, and $\phi \in C_0^\infty(\mathbb{R}^3)$, the following estimate holds:
\begin{equation}\label{cor1331}
\int_{\mathbb{R}^3}\frac{|\phi|^2}{(1+|q^*|^2)(1+t+r^*)^{1-b}} \, w\, d\widetilde{x}  \leq C_{\gamma, \mu} \int_{\mathbb{R}^3}\frac{|\pa_{r^*}\phi|^2}{(1+t+r^*)^{1-b}} \, w\, d\widetilde{x}.
\end{equation}
If additionally $b < 2\gamma$, $\mu> 0$, and $q^*_- = \max(-q^*, 0)$, we additionally have the estimate
\begin{equation}\label{cor1332}
\int_{\mathbb{R}^3} \frac{|\phi|^2(1+|q^*|)^{-b}}{(1+|q^*|^2)(1+|t+r^*|^{1-b})} \, \frac{w}{(1+q^*_-)^{2\mu}}\, d\widetilde{x} \leq C_{b, \gamma, \mu} \int_{\mathbb{R}^3}|\pa_{r^*}\phi|^2 \, w'\, d\widetilde{x}.
\end{equation}

\end{corollary}
\begin{proof}
The inequality \eqref{cor1331} follows directly from Lemma \ref{est:1dpoincare} along with the approximation $w \leq \overline{w}_{\gamma, \mu} + \overline{w}_{\gamma, 0} \leq 2w$, and
\eqref{cor1332} follows from Lemma \ref{est:1dpoincare} using $\overline{w}_{\gamma - b/2, \mu + b/2}$ and the pointwise inequality
\[
 \frac{(1+|q^*|)^{-b}}{(1+|t+r^*|^{1-b})} \, \frac{w}{(1+q^*_-)^{2\mu}} \leq w'.
 {}\tag*{\qedhere}
\]
\end{proof}

\section{The Decay estimates for the wave equation}\label{sec:DecayEstimates}
We consider energy norms with the following weight function
$$
w_{p_{\,},\gamma}(t,x)=\langle q^*_+\rangle^{1-1/\!p\,+\gamma},\qquad 0<\gamma<1.
$$
In this section we will work with the flat wave operator in the curved coordinates
\beq
\Box^*=\widehat{m}^{ab}\wpa_a\wpa_b=\widetilde{\Box}^{\,\widehat{m}} ,
\eq
taking advantage of already established formulas using the fundamental solution in flat coordinates.

\subsection{Weighted Klainerman-Sobolev estimates}
In this section we provide a straightforward generalization of the
Klainerman-Sobolev inequalities, expressing pointwise decay in terms
of the bounds on $L^2$ norms involving vector field $Z\in {\mathcal Z}$.

We have the following global Sobolev inequality, see \cite{LR10}
\begin{proposition} \label{prop:K-S}
For any function $\phi\in C^\infty_0(\R^3)$ and an arbitrary
$(t,x)$,
$$
|\phi(t,x)|(1+t+|q^*|) (1+|q^*|)^{1/2} w_{2,\gamma}(t,x) \leq
C{\sum}_{|I|\leq 3} \|w_{2,\gamma} \wZ^I \phi(t,\cdot)\|_{L^2}.
$$
\end{proposition}

\subsection{The weighted $L^1$-$L^\infty$ estimates}
To get improve decay estimates in the interior we will use H\"ormander's $L^1$--$L^\infty$
estimates for the fundamental solution of $\Box$, see \cite{H87, L90}:
 \begin{proposition}\label{prop:hormander} Suppose that $w$ is a solution of
 $\Box^* u=F$ (i,e. the flat Minkowski wave operator) with vanishing data $ u\big|_{t=0}=\pa_t u\big|_{t=0}=0$. Then
 \beq |u(t,x)|(1+t+|x|)w_{1,\gamma}(t,x)  \leq
C{\sum}_{|I|\leq
2}\int_0^t\int_{\mathbf{R}^3}\frac{|Z^I F(s,y)|}{1+s+|y|}\,w_{1,\gamma} (s,y)\, dy\,ds.
 \eq
\end{proposition}
In \cite{L90} the proof was without the weight, but by a domain of dependency argument
the weight is larger in the support of the inhomogeneous term.

Also for the linear homogeneous solution we have from \cite{L90}:
\begin{lemma}\label{lemma:linearhomogenousdecayweighted} If $v$ is the solution of
 $\Box^* v=0$, with data $v\big|_{t=0}=v_0$ and $ \pa_t v\big|_{t=0}=v_1 $
 then for any $\gamma>0$;
 \beq
 (1+t+r)|v(t,x)|w_{1,\gamma}(t,x) \leq C{\sup}_x \big( (1+|x|)^{2+\gamma}
 ( |v_1(x)|+|\partial v_0(x)|) +(1+ |x|)^{1+\gamma}|v_0(x)|\big).
 \eq
\end{lemma}
\begin{proof} The proof is an immediate consequence of
Kirchoff's formula
 $$
v(t,x)=t\int_{|{\omega}|=1}{\big(v_1(x+t{\omega})+\langle
v^{\prime}_0(x+t{\omega}),{\omega}\rangle \big)\,dS({\omega})} +
 \int_{|{\omega}|=1}{v_0(x+t{\omega})\,dS({\omega})},
 $$
where $dS({\omega})$ is the normalized surface measure on $S^2$.
 Suppose that $x=r\mathbf{e}_1$, where $\mathbf{e}_1=(1,0,0)$.
 Then for $k=1,2$ we must estimate
 \beq
 \int\!\! \frac{dS(\omega)}{1\!+\!|{}_{\,}r\mathbf{e}_1\!+t\omega|^{k+\gamma}}
 =\int_{-1}^1\frac{C d\omega_1}{1+\big(
 (r\!-\!t\omega_1)^2+t^2(1\!-\!\omega_1^2)\big)^{(k+\gamma)/2}}
 \leq \int_0^2\!\! \frac{C d s}{1+\big(
 (r\!-\!t\!+\!t s)^2+t^2 s\big)^{(k+\gamma)/2}}.
 \eq
 If $k=2$ we make the change of variables $t^2 s=\tau$ to get
 an integral bounded by $Ct^{-2}\langle (r^*-t)_+\rangle^{-\gamma} $ and if $k=1$, we make the change
 of variables $t s=\tau$ to get an integral bounded by $t^{-1}\langle (r^*-t)_+\rangle^{-\gamma}$.
 This proves the result for $r<2t$, say, but for $r>2t $ it follows by inspection.
 \end{proof}

\subsection{The weighted $L^\infty$-$L^\infty$ estimates}

We will now derive sharp estimates for the first order derivatives, following
\cite{L90} we have
\begin{lemma}\label{lem:transversalder} Let $D_t=\{x;\, \big|t-|x|\big|\leq t/4 \}$, and let $\overline{w}(q^*)$ be any positive continuous function,
where $q^*=r^*-t$, $r^*=r+M\ln{r}$. Suppose that $\Box^* \phi_{\mu\nu}=F_{\mu\nu}$. Let
$U,V\in \{\widetilde{L},\widetilde{\underline{L}},S_1,S_2\}$ and $\phi_{UV}= \phi_{\mu\nu} U^\mu V^\nu$.
 Then
\begin{multline}\label{eq:transversalder}
(1+t+|x|) \,|\pa \phi_{UV}(t,x)\, \overline{w}(q^*)| \les\!\sup_{|q^*|/4\leq
\tau\leq t}
{\sum}_{|I|\leq 1}\|\,Z^I\! \phi(\tau,\cdot)\, \overline{w}\|_{L^\infty}\\
+ \int_{|q^*|/4}^t\Big( (1+\tau)\|
\,F_{UV}(\tau,\cdot)\, \overline{w}\|_{L^\infty(D_\tau)} +{\sum}_{|I|\leq 2} (1+\tau)^{-1}
\| Z^I \phi(\tau,\cdot)\, \overline{w}\|_{L^\infty(D_\tau)}\Big)\, d\tau.
\end{multline}
\end{lemma}
\begin{proof}
 Since $\Box\phi=-r^{-1}(\pa_t^2-\pa_r^2)(r\phi)+r^{-2}\triangle_\omega \phi $, where $\triangle_\omega=\sum\Omega_{ij}^2$ and $
 |Z U|\leq C$, for $U\in \{ A,B, L,\underline{L}\}$,
it follows that
 \beq
 \big| \Box^*\big(  {\phi}_{UV}\big)
 -U^\mu V^\nu\Box^* \phi_{\mu\nu}\big|
 \leq{r}^{-2}\,\, {\sum}_{|J|\leq 1} \,\,|Z^{J} \phi|.
 \eq
We have
\begin{equation*}
\Box^* \phi
=\frac{1}{r^*}4\pa_{s^*}\pa_{q^*}(r^*\phi)
+\frac{1}{(r^*)^2}\triangle_\omega \phi,
\end{equation*}
where  $\pa_{q^*}=(\pa_{r^*}\!-\pa_t)/2 $ and $\pa_{s^*}=(\pa_{r^*}\!+\pa_t)/2$.
Hence
\beq
\Big|4\pa_{s^*}\pa_{q^*}(r^*\phi) -r^* \Box^*\phi  \Big|
\les {r}^{-1}\,\, {\sum}_{|J|\leq 2} \,\,|Z^{J} \phi|,
\eq
so with $s=t+r^*$
\beq
\Big| \pa_{s^*}\pa_{q^*}(r^* \phi_{UV})\Big|
\les r\big|(\Box^* \phi)_{UV}\big|+ (t+r)^{-1}\,\, {\sum}_{|J|\leq 2} \,\,|Z^{J} \phi|,
\qquad |{}_{\,} t-r^*|\leq  t/4 .
\eq
Integrating this along the flow lines of the vector field $\pa_{s^*}$ from the boundary of $D=\cup_{\tau\geq 0}D_\tau$
to any point inside $D$, using that $\overline{w}$ is essentially constant along the flow lines, gives that for any $(t,x)\in D$
\begin{multline}
|\pa_q(r \phi_{UV}(t,x))\, \overline{w}| \les\!\sup_{0\leq
\tau\leq t}
{\sum}_{|I|\leq 1}\|\,Z^I \phi(\tau,\cdot)\, \overline{w}\|\\
+ \int_{|q^*|/4}^t\Big( (1+\tau)\|
\,F_{UV}(\tau,\cdot)\, \overline{w}\|_{L^\infty(D_\tau)} +{\sum}_{|I|\leq 2} (1+\tau)^{-1}
\| Z^I \phi(\tau,\cdot)\, \overline{w}\|_{L^\infty(D_\tau)}\Big)\, d\tau.
\end{multline}
The lemma follows from also using \eqref{eq:tangentialbyvectorfield}-\eqref{eq:transversalderivativeprojectionsec2} and  that the estimate is trivially true  when
$|{}_{\,}r\!-\!t|\!\geq \!  t/4$.
\end{proof}

\section{Energy bounds and decay estimates for Maxwell-Klein-Gordon}\label{sec:ChrisResults}
We restate Theorem 1.1 from \cite{Ka18}, along with results from Theorem 7.2 in that paper, which will form a portion of our bootstrap argument. This takes the form of a set of energy and decay estimates which will follow from the harmonic coordinate condition and the bootstrap assumption on the metric. In order to properly set up the argument, we must first define certain quantities arising from the MKG system which appear in decay rates and in the initial conditions. Next, we prove a result which expands on the identification $h^1 \sim H_1$, and which allows us to simplify the required metric bounds. Finally, we restate the main theorem and consequent decay estimates. In Section \ref{sec:sharpdecayfields} we will confirm that the required bounds on the energy momentum tensor follow from the bootstrap assumption.
\subsection{Setting up the initial value problem}\label{sec:MKGIVPdefs}
We recall from Section \ref{subsubsec:Charge} that in order to best capture the decay of $\bF$ we decompose it into $\bF^0\!\! + \bF^1\!\!$, where $\bF^0$ is a fixed field depending on the charge $q$ which picks up the symmetric part of the decay in the exterior, and $\bF^1$ is the remainder. More precisely, we define
\beq
q(t) = \int_{\Sigma_t}-\sqrt{|g|}J^0\, dx,
\eq
where the volume form is taken with respect to $m$. Since $J$ is divergence free with respect to $\nabla$, as long as $\phi$ has good decay in space, we can take $q$ to be constant in time. Then, for a smooth increasing function $\chi_{\bF}(x)$ which is identically 0 for $x\leq 0$ and 1 for $x\geq 1$, we define the 2-forms $\bF^0, \bF^1$ by
\begin{equation}
\bF^0_{0i} = \frac{\omega_iq}{4\pi}\frac{\chi_{\bF}(r^*-t-2)\partial_r(r^*)}{r^{*2}}, \qquad \bF^0_{ij} = \bF^0_{00} =0, \qquad \bF^1 = \bF - \bF^0.
\end{equation}
Although the $\bF^0, \bF^1$ decomposition is useful in our calculations, when writing the initial conditions there is another useful decomposition. We may decompose $\bF$ into its electric and magnetic field components, respectively defined by
\beq
 {\bf E}_i = \bF_{0i}, \qquad {\bf B}_i =\varepsilon_{0ijk}\bF^{jk}.
\eq
This decomposition is of course not preserved after Lie differentiation with respect to the Lorentz boosts. However, it has the nice property that the poor decay of $\bF$ in space can be localized to the curl-free part of ${\bf E}$. If we define ${\bf E}^{df}$ to be the divergence-free part of ${\bf E}$ (in the Helmholtz decomposition), then ${\bf E}^{df}_{i} \sim \bF^1_{0i}$.

For a  $(0,k)$ tensor ${\bf T}$ and a complex function $\bphi$, we define the initial data norms
\begin{align}
\lVert {\bf T} \rVert^2_{H^{N, s_0}} &= {\sum}_{|\alpha|\leq N}{\sum}_{a_i\in (1,2,3)}\int_{\mathbb{R}^3}(1+r^2)^{s_0+|\alpha|}|\underline{\wpa}{}^\alpha{\bf T}(\wpa_{a_1}, \hdots, \wpa_{a_k})|^2\, dx, \\
\lVert \bphi \rVert^2_{H^{N, s_0}} &= {\sum}_{|\alpha|\leq N}\int_{\mathbb{R}^3}(1+r^2)^{s_0+|\alpha|}|\underline{\widetilde{D}}{}^\alpha\bphi|^2\, dx,
\end{align}
where $\underline{\wpa}, \underline{\widetilde{D}}$ are spatial derivatives in modified coordinates, and $\alpha$ is a multiindex.
\subsection{The metric bounds} When dealing with derivatives of the energy momentum tensor, we must deal with terms containing both the metric and the inverse metric, so it is useful to to establish similar energy and decay estimates for $h^1$ and $H_1$.
\begin{proposition}\label{prop:Hhequiv1}
Let $g_{\alpha\beta} = m_{\alpha\beta} + h^0_{\alpha\beta} + h^1_{\alpha\beta}$, and let $g^{\alpha\beta} = m^{\alpha\beta} + H_0^{\alpha\beta} + H_1^{\alpha\beta}$, as in Section \ref{exteriordecaymass}. Suppose further that $|Z^I(h^0_{\alpha\beta})|+|Z^I(h^1_{\alpha\beta})|\leq \varepsilon\langle t+r\rangle^{-1+\delta}$ for some collection of vector fields $Z$, a multiindex $I$ with $|I| \leq N$, and a small constant $\delta \in (0,1/2)$. Then, for vector fields $X, Y$ with bounded components, and for sufficiently small $\varepsilon$,
\beq
|X^\alpha Y^\beta Z^I(H_{1\alpha\beta})|\leq |X^\alpha Y^\beta Z^I(h^1_{\alpha\beta})| + C\varepsilon^2 \langle t+r\rangle^{-2+2\delta}.
\eq
\end{proposition}
\begin{proof}
We first apply the fields $Z$ to the identity
\beq
g^{\alpha\beta} = g^{\alpha\gamma}g_{\gamma\delta}g^{\delta\beta},
\eq
which, via induction, gives the preliminary estimate
\beq
|Z^I(H_{\alpha\beta})|\leq 2\varepsilon\langle t+r\rangle^{-1+\delta},
\eq
for sufficiently small $\varepsilon$. We can now apply $Z^I$ to the identity
\beq
g_{\alpha\beta} = g_{\alpha\gamma}g^{\gamma\delta}g_{\delta\beta},
\eq
which gives
\[
|Z^I(H_{\alpha\beta}) + Z^I(h_{\alpha\beta})| \leq C\varepsilon^2\langle t+r\rangle^{-2+2\delta}.
\]
Taking the identity $H_{0\alpha\beta} = -h^0_{\alpha\beta}$ and contracting with the fields $X, Y$ gives our result.
\end{proof}
A similar energy result holds, which follows from a pointwise calculation, again taking derivatives of the identity $g_{\alpha\beta} = g_{\alpha\gamma}g^{\gamma\delta}g_{\delta\beta}$, then taking an $L^\infty$ bound on $H$ or $h$ when applicable:
\begin{proposition}\label{prop:Hhequiv2}
Let $\overline{w}$ be a positive function, and take $N$ such that $|Z^I(h^0_{\alpha\beta})|+|Z^I(h^1_{\alpha\beta})|\leq \varepsilon\langle t+r\rangle^{-1+\delta'}$ for $|I|\leq N/2$. Then, for a given region $\Omega$ with volume element $dV$,
\beq
\int_{\Omega} |Z^I(H_{1\alpha\beta})|^2 \overline{w} \, dV\leq \int_{\Omega} |Z^I(h^1_{\alpha\beta})|^2 \overline{w}dV + C\varepsilon{\sum}_{|J|\leq|I|-1}\int_{\Omega} \langle t+r\rangle^{-2+2\delta}(|Z^Jh|^2+|Z^JH|^2) \overline{w}dV.
\eq
\end{proposition}
\subsection{Energy bounds and decay estimates for Maxwell-Klein-Gordon on a fixed background}
We can now state the relevant result from \cite{Ka18}, using the initial norms and decompositions defined in Section \ref{sec:MKGIVPdefs}.
\begin{theorem}\label{thm:ChrisResults}
Let $(\mathcal{M}, g)$ be an asymptotically flat Lorentzian manifold which admits a  global system of coordinates $(t, x_1, x_2, x_3)$ (not necessarily satisfying Einstein's equations), which are harmonic with respect to $g$ and satisfy the initial splitting condition $g_{0i} = 0$. Define the consequent quantities $M, \chi, h^0, h^1, H_0, H_1$ as in Section \ref{exteriordecaymass}. Additionally, take real constants $s, s_0, \gamma, \mu, \delta$ satisfying $\mu < 1/2$, $\gamma > 1/2$, $\frac12 < s<1<s_0<\frac32$, and $4\delta < \min(1-2\mu, s-\frac12, 1-s, \gamma - \frac12, 1-\gamma, 1+\gamma-2s)$, and take an integer $N \geq 11$. For $0<T\leq\infty$, define the intervals
$
I_T =
[0,T] $.
There exists a constant $\varepsilon_g>0$ depending only on $(s, s_0, \gamma, \mu, \delta, N)$ such that, for all $T$, if  $M, h^1$ satisfies the decay estimates
\begin{subequations}
\label{MetLIint}
\begin{align}
 M &< \varepsilon_g, \label{MetLIint1}\\
{\sup}_{t\in I_T}|\Lie_{{\wt{X}}}^Ih^1| &< \varepsilon_g\langle t+r \rangle^{-1+\delta}, \label{MetLIint2}\\
{\sup}_{t\in I_T}|\Lie_{{\wt{X}}}^I h^1|_{\mathcal{L}\mathcal{T}} &< \varepsilon_g\left(\tfrac{\langle t-r^*\rangle}{\langle t+r\rangle}\right)^{\gamma}\langle t+r \rangle^{-1+\delta}, \label{MetLIint3}
\end{align}
\end{subequations}
\noeqref{MetLIint}\noeqref{MetLIint1}\noeqref{MetLIint2}\noeqref{MetLIint3} for $|I|\leq N-6$, $t \in I_T$, and the energy estimates
\begin{subequations}
\label{MetL2int}
\begin{align}
\label{L21int}
\!\!\!\left\lVert|\wpa\Lie_{{\wt{X}}}^Ih^{1\!}|w_\gamma^{1/2}
\right\rVert_{L^2(\mathbb{R}^3)}\!\! + \left\lVert{\langle r^*\!\!\!-\!t\rangle}^{-1}|\Lie_{{\wt{X}}}^I{h}^{1\!}|w_\gamma^{1/2}
\right\rVert_{L^2(\mathbb{R}^3)}&\!\!\leq \!\varepsilon_g (1\!\!+\!t)^{\delta/2}\!\!,\,\,\,\,\,\, \\
\label{L22int}
\!\!\!\!\!\left\lVert\big(|\wpa\Lie_{{\wt{X}}}^I{h}^{1\!}|_{\mathcal{L}\mathcal{L}}\!+ \! |\overline\wpa\Lie_{{\wt{X}}}^I{h}^{1\!}|\big) (w_\gamma')^{1/2}\right\rVert_{L^2(I_{T}\!\times\mathbb{R}^3)}&\!\!\leq \!\varepsilon_g (1\!\!+\!T)^{\delta/2}\!\!,\,\,\,\,\,\, \\
\label{L23int}
\!\left\lVert{\!\langle r^*\!\!\!-\!t\rangle}^{-1-s}{\!\langle r^*\!\!\!+\!t\rangle}^{-1+s}|\Lie_{{\wt{X}}}^I{h}^{1\!}|_{\mathcal{L}\mathcal{L}}
(w_\gamma')^{1/2}
\right\rVert_{L^2(I_{T}\!\times\mathbb{R}^3)}&\!\!\leq \!\varepsilon_g,
\end{align}
\end{subequations}\noeqref{L21int}\noeqref{L22int}
for $|I|\leq N$ and $t\in I_T$, and
\beq
w_\gamma = \begin{cases}
1 + (1 + t - r^*)^{-2\mu} & r^* \leq t, \\
1+(1 + (r^* - t))^{1+2\gamma} & r^* \geq t,
\end{cases}
\eq
 then the system \eqref{eq:MKG} is well-posed for small initial data.
Specifically, there exists an $\varepsilon_0 > 0$ such that for all $\varepsilon < \varepsilon_0$ and for all one-forms ${\bf E}_0, {\bf B}_0$ and scalar functions $\bphi_0, \,\,\dot{\!\!\bphi}_0$ on $\mathbb{R}^3$ satisfying
\beq\label{def:ID}
\lVert{\bf E}_{0}^{df}\rVert_{H^{k, s_0}} + \lVert{\bf B}_{0}\rVert_{H^{k, s_0}} + \lVert D\bphi_0 \rVert_{H^{k, s_0}} + \lVert\,\,\dot{\!\!\bphi}_0\rVert_{H^{k, s_0}}< \varepsilon,
\eq
for a given $\bf{A}$, there exists a solution to the system \eqref{eq:MKG} on $(t, x)\in I_T\times \mathbb{R}^3$, with
\begin{alignat}{4}
\bF_{0i}(0,x) &= ({\bf E}_0)_i(x) \qquad & \epsilon_{0ijk}\bF^{jk}(0,x) &= ({\bf B}_0)_i, \\
\bphi(0,x) &= \bphi_0(x), &\wpa_t\bphi(0,x) &= \,\,\dot{\!\!\bphi}_0 .
\end{alignat}
All constants, including $\varepsilon_0$ and $\varepsilon_g$, can be taken independently of $T$. We also have decay estimates, which we adapt here from Theorems 5.6 and 6.6 of \cite{Ka18}, as well as energy estimates, which we adapt from Theorems 3.1 and 4.2.
 For $|I|\leq N-4$, we have
\begin{subequations}\label{est:DPhiDF1}
\begin{align}
\!\!\!\!\!\!|D_{\widetilde{L}} D_{\wt{X}}^I\bphi| + |\alpha[\Lie_{\wt{X}}^I \bF^1]|
&\leq C \varepsilon \langle t+r^*\rangle^{-2}\langle t-r^*\rangle^{1/2-s}
\langle (r^*-t)_+\rangle^{s-s_0},\quad \label{est:F1Phi1}\\
\!\!\!\!\!\!|\slashed{D}D_{\wt{X}}^I \bphi| + |\rho[\Lie_{\wt{X}}^I \bF^1]| + |\sigma[\Lie_{\wt{X}}^I \bF^1]|&\leq C\varepsilon\langle t+r^*\rangle^{-1-s}\langle t-r^*\rangle^{-1/2}\langle (r^*-t)_+\rangle^{s-s_0},\quad\label{est:F1Phi2} \\
\!\!\!\!\!\!\langle t-r^* \rangle^{-1} |D_{\wt{X}}^I\bphi| + |D_{\widetilde{\underline{L}}}D_{\wt{X}}^I \bphi| + |\underline{\alpha}[\Lie_{\wt{X}}^I \bF^1]|&\leq C \varepsilon\langle t+r^*\rangle^{-1}\langle t-r^*\rangle^{-1/2-s}\langle (r^*-t)_+\rangle^{s-s_0}.\quad\label{est:F1Phi3}
\end{align}\noeqref{est:DPhiDF1} \noeqref{est:F1Phi1} \noeqref{est:F1Phi2} \noeqref{est:F1Phi3}
\end{subequations}
By the estimate (2.49) in \cite{Ka18}, $\bF^0$ is supported in $r^*\! \geq \! t\!+\!2$ and satisfies the following bounds for  $|I|\!\leq\! N$:
\begin{subequations}\label{est:DF2}
\begin{align}
|\alpha[\Lie_{\wt{X}}^I \bF^0]| &\leq C\varepsilon \langle t+r^*\rangle^{-3}\langle t-r^*\rangle, \label{est:DF21}\\
|\rho[\Lie_{\wt{X}}^I \bF^0]| +|\sigma[\Lie_{\wt{X}}^I \bF^0]|+|\underline{\alpha}[\Lie_{\wt{X}}^I \bF^0]|&\leq C \varepsilon \langle t+r^*\rangle^{-2}\label{est:DF21}.\\
\end{align}\noeqref{est:DF2}\noeqref{est:DF21}\noeqref{est:DF22}
\end{subequations}
Furthermore, the following energy bounds hold for $|I|\leq N$ and $t\in I_T$:
\begin{subequations}\label{est:ChrisEnergyBounds}
\begin{align}
\int_{\Sigma_t}\!\!\!\!\big(\!\langle t+r^*\rangle^{2s}(|\rho[\Lie_{\wt{X}}^I \bF^1]|^2\! \!+ \!|\sigma[\Lie_{\wt{X}}^I \bF^1]|^2 \!\!+\! | \alpha[\Lie_{\wt{X}}^I \bF^1]|^2) \!+\! \langle t-r^*\rangle^{2s}|\underline{\alpha}[\Lie_{\wt{X}}^I \bF^1]|^2\big)\langle (t-r^*)_-\rangle^{2s_0 - 2s} &\leq C\varepsilon^2, \qquad\label{est:ChrisEnergyBounds1}\\
\int_{\Sigma_t}\!\!\!\!\big(\langle t+r^*\rangle^{2s}(|D_{\widetilde{L}}D_{{\wt{X}}}^I\bphi|^2 + |\slashed{D}D_{\wt{X}}^I\bphi|^2 + |\tfrac{\bphi}{r}|^2) + \langle t-r^*\rangle^{2s}|D_{\widetilde{\underline{L}}}D_{\wt{X}}^I\bphi|^2\big)\langle (t-r^*)_-\rangle^{2s_0 - 2s} &\leq C\varepsilon^2.\label{est:ChrisEnergyBounds2}
\end{align}
\end{subequations}
\noeqref{est:ChrisEnergyBounds}\noeqref{est:ChrisEnergyBounds1}\noeqref{est:ChrisEnergyBounds2}
\end{theorem}
\begin{remark}
The version of this theorem which appears in \cite{Ka18} requires the estimates \eqref{MetLIint} and \eqref{MetL2int} for both $h$ and $H$; however, we can reduce this using Propositions \ref{prop:Hhequiv1} and \ref{prop:Hhequiv2}. Additionally, it follows from direct computation and Corollary \ref{cor:Aderivatives} that the estimates \eqref{MetLIint}, \eqref{MetL2int}, are equivalent to the corresponding estimates with $h, H$ replaced by $\wt{h}, \wt{H}$, as Lie derivatives are unaffected by the change in coordinates.
\end{remark}
\begin{remark}
In the application to the full system, we will require $s_0 = \gamma+\frac12$ to close the argument. However, afterwards, we may repeat this estimate with higher $s_0$ to show better decay estimates for $\{\bphi, \bF\}$ in the far exterior. This is a consequence of the fact that terms that are quadratic in $\bphi, \bF$ appear by themselves in the right hand side of the equation for $h^1$, but in the MKG system $h^1$ only appears coupled with $\bphi$ and $\bF$ terms.
\end{remark}
Since $[D_a, D_b]\bphi = i\bF_{ab}\bphi$, we must deal with commutator terms which arise even in the Minkowski background.
\begin{corollary}\label{cor:DDComm}
For $\widetilde{U}\in \{\widetilde{L}, \widetilde{\underline{L}}, \widetilde{S}_1, \widetilde{S}_2\}$, the decay bounds for $\bphi$ in the estimates \eqref{est:DPhiDF1} hold up to a constant if $D_{\widetilde{U}}D_{\widetilde{X}}^I\phi$ is replaced with $\widetilde{U}^a\widetilde{X}^b(\Lie_{\widetilde{X}}^{I_1}\bF)_{ab}D_{\widetilde{X}}^{I_2}\bphi$, with $|I_1| + |I_2| + 1 \leq |I|$. Additionally,
\beq
\sup_{t\in I_t}\int_{\Sigma_t} \langle t-r^*\rangle^{2s}||\widetilde{U}^a\widetilde{X}^b(\Lie_{\widetilde{X}}^{I_1}\bF)_{ab}D_{\widetilde{X}}^{I_2}\bphi|^2\langle (t-r^*)_-\rangle^{2s_0 - 2s} \leq C\varepsilon^2.
\eq
\end{corollary}
\begin{proof}
We first take the following bounds for a 2-form $\mathbf{G}$, which follow from the null decompositions of the fields $\widetilde{X}$:
\begin{align}
|\mathbf{G}(\widetilde{L}, \widetilde{X})| &\lesssim \langle t+r^*\rangle |\alpha[\mathbf{G}]| + \langle t-r^*\rangle|\rho[\mathbf{G}]| + |\mathbf{G}| ,\\
|\mathbf{G}(\underline{\widetilde{L}}, \widetilde{X})| &\lesssim \langle t+r^*\rangle |\mathbf{G}|, \\
 |\mathbf{G}(\widetilde{S}_i, \widetilde{X})| &\lesssim \langle t+r^*\rangle (|\alpha[\mathbf{G}]| + |\sigma[\mathbf{G}]|) + \langle t-r^*\rangle|\underline{\alpha}[\mathbf{G}]|.
\end{align}
The decay bound then follows from \eqref{est:DPhiDF1} and \eqref{est:DF2}.
For the energy bound, we decompose $\bF = \bF^0 + \bF^1$. For the terms containing $\bF^0$, our result follows from \eqref{est:DF2}, and \eqref{est:ChrisEnergyBounds}. For terms containing derivatives of  $\bF^1$ and $\bphi$, for $|I|\geq 6$, we can apply \eqref{est:DPhiDF1} to $\bF^1$ or $\bphi$, and bound the remainder with \eqref{est:ChrisEnergyBounds}.
\end{proof}
\subsection{Energy and Decay bounds for the energy momentum tensor}\label{sec:EDBounds}
We can use the decay estimates stated in the previous section to prove the
sharp decay of the energy momentum tensor.
\begin{corollary}\label{cor:DecayBoundsOnT}
If $H$ satisfies the estimates in \eqref{est:H1} and $\bF, \bphi$ satisfy \eqref{est:DPhiDF1}, \eqref{est:DF2}, then the following estimates hold for $|I|\leq k-6, \mathcal{T} = \{\widetilde{L}, \widetilde{S_1}, \widetilde{S_2}\}, \mathcal{U}, \mathcal{V} = \{\widetilde{L}, \widetilde{\underline{L}}, \widetilde{S_1}, \widetilde{S_2}\}$:
\begin{subequations}
\begin{align}
|\Lie_{\wt{X}}^I\widecheck{T}|_{\mathcal{U}\mathcal{V}}&\lesssim \varepsilon^2\langle t+r^*\rangle^{-2}\langle{t-r^*}\rangle^{-2}\langle
(t-r^*)_+\rangle^{1-2s}, \\
|\Lie_{\wt{X}}^I\widecheck{T}|_{\mathcal{T}\mathcal{U}}&\lesssim \varepsilon^2\langle t+r^*\rangle^{-2-s}\langle{t-r^*}\rangle^{-2+s}\langle (t-r^*)_+\rangle^{1-2s}.
\end{align}
\end{subequations}
\end{corollary}
\begin{remark}
The sharp decay rate in the exterior comes from $\bF^0$, which cannot be expected to decay faster than $r^{-2}$.
\end{remark}
\begin{proof}
We expand
\begin{equation}
\widecheck{T}[\bphi, \bF]_{ab} = \mathfrak{R}\bigtwo(D_a\bphi\overline{D_b\bphi}\bigtwo) + \widetilde{g}^{cd}\bF_{ac}\bF_{bd}- \widetilde{g}_{ab}\widetilde{g}^{ce}\widetilde{g}^{df}\bF_{cd}\bF_{df}/4 .
\end{equation}
We decompose $\widecheck{T}[\phi, \bF] = \widecheck{T}_{\whm}[\bphi, \bF]+\widecheck{T}_{\whh}[\bphi, \bF]$ such that
\begin{align}
\widecheck{T}_{\whm}[\bphi, \bF]_{ab} &= \mathfrak{R}\bigtwo(D_a\bphi\overline{D_b\bphi}\bigtwo) + \whm^{cd}\bF_{ac}\bF_{bd}- \whm_{ab}\whm^{ce}\whm^{df}\bF_{cd}\bF_{df}/4,\\
\widecheck{T}_{\whh}[\bphi, \bF]_{ab} &= (\widetilde{g}^{cd} - \widehat{m}^{cd})\bF_{ac}\bF_{bd} - (\widetilde{g}_{ab}\widetilde{g}^{ce}\widetilde{g}^{df} - \whm_{ab}\whm^{ce}\whm^{df})\bF_{cd}\bF_{df}/4.
\end{align}
As with the decomposition  $F$ in the Einstein portion of the system, we note that $\widecheck{T}_{\whm}[\bphi, \bF]_{ab}$ behaves nicely with respect to the null decomposition, and $\widecheck{T}_{\whh}[\bphi, \bF]_{ab}$ consists of cubic terms. Additionally, since $[D_a, D_b]\phi = i\bF_{ab}\phi$,
\beq \label{id:LieT}
\Lie_{\wt{X}}(\mathfrak{R}(D_a\phi \overline{D_b\psi})) = \mathfrak{R}(D_aD_{\wt{X}}\phi \overline{D_b\psi} +  D_a\phi \overline{D_bD_{\wt{X}}\psi} - i\wt{X}^c \bF_{ac}\phi\overline{D_b\phi}- i\wt{X}^c \bF_{bc}\overline{\phi}D_a\phi),
\eq
and
\beq
\Lie_{\wt{X}}(\mathfrak{I}(\phi \overline{D_b\psi})) = \mathfrak{I}(D_{\wt{X}}\phi \overline{D_b\psi} +  \phi \overline{D_bD_{\wt{X}}\psi} - i\wt{X}^c F_{bc}\phi\overline{\phi}).
\eq
Then,
\beq
{\sum}_{|I|\leq N-6}| (\Lie_{\wt{X}}^IT_{\whh}[\bphi, \bF])_{ab}| \lesssim \varepsilon_g\langle t+r^*\rangle^{-1+\delta} {\sum}_{|I|\leq N-6}||\Lie_{\wt{X}}^I\bF|^2,
\eq
using \eqref{MetLIint}, along with the estimate $-1+\delta < -s$. The estimate on $T_{\whm}[\bphi, \bF]_{ab}$ follows directly from the null decomposition, as well as Corollary \ref{cor:DDComm} for the commutator terms.
\end{proof}
 For higher derivatives of $\widecheck{T}$ we can also establish energy bounds which will be useful in propagating the bootstrap assumption.
\begin{corollary}\label{cor:EnergyBoundsOnT}
For sufficiently small $\varepsilon_g$, and $\gamma = s_0 - \tfrac12$, the following inequality holds:
\begin{equation}
\lVert\Lie_{\wt{X}}^I\widecheck{T}[\bphi, \bF](t, \cdot) w^{1/2}\rVert_{L^2(\mathbb{R}^3)} \lesssim \frac{\varepsilon^2}{1+t}.
\end{equation}
\end{corollary}
\begin{proof}
We again split $\widecheck{T} = \widecheck{T}_{\whm} + \widecheck{T}_{\whh}$.
An $L^2-L^\infty$ estimate along with Corollary \ref{cor:DDComm}  gives
\begin{multline}
\lVert\Lie_{\wt{X}}^I\widecheck{T}_{\whm}[\bphi, \bF](t, \cdot) w^{1/2}\rVert_{L^2(\mathbb{R}^3)}\\
 \lesssim \varepsilon\langle t \rangle^{-1}\lVert\langle t-r^*\rangle^{-1}(|D\bphi| + |\bF^1| + |i_{\widetilde{X}}\bF||\bphi|)w^{1/2}\rVert_{L^2(\mathbb{R}^3)} +\lVert\varepsilon^2\langle t+r\rangle^{-4} w^{1/2}\rVert_{L^2(\mathbb{R}^3)}.
\end{multline}
The second term on the right comes from the $|\bF^0|^2$ terms appearing in $T_{\whm}$. The final bound follows from Corollary \ref{cor:DDComm} for the first term and direct integration for the second, using the inequality $1+2\gamma < 3$. In the case of $T_{\whh}$, methods used for $T_{\whm}$ as well as Lemma \ref{lem:m0decomp} and Corollary \ref{cor:DDComm} can be used to bound all terms except those where the most derivatives fall on $h^1$ and $H_1$. We can therefore reduce this to
\beq
\lVert\Lie_{\widetilde{X}}^IT_{\whh}[\bphi, \bF](t, \cdot) w^{1/2}\rVert_{L^2(\mathbb{R}^3)} \lesssim \varepsilon^2\lVert \langle t-r^*\rangle^{-2}\langle t+r^*\rangle^{-2} (|\Lie_{\widetilde{X}}^I h^1| + |\Lie_{\widetilde{X}}^I H_1|)w^{1/2}\rVert_{L^2(\mathbb{R}^3)} \lesssim \varepsilon^2\varepsilon_g\langle t\rangle^{-1}.
\eq
This follows from our bootstrap assumption.
\end{proof}

\section{Precise statement of the theorem and the structure of the proof}\label{sec:StatementandBootstrap}
In this section we give the precise statement of the theorem.
We also start the proof that will be given over the next few sections as well as concluding the proof using the results from the next few sections.

\subsection{The energies and statement of the theorem}
For a given $0< \gamma < 1$, $0<\mu < 1-\gamma$, we define the energy at time $t$,
\begin{equation} \label{eq:energydefE}
	E_N(T)={\sup}_{0\leq t\leq T}{\sum}_{\vert I \vert \leq N\,\,}
	\int
	|\wpa \widetilde{Z}^I\! h^1 (t,x)|^2 w\, dx, \quad
	\text{where}
	\quad
	w(t,x)
	=
	\bigg\{\begin{aligned}
		&(1+|r^*\!\!-t|)^{1+2\gamma}\!\!,\quad\!\!
		&
		r^*\!>t,
		\\
		&\sim 1
,
		\quad
		&
		r^*\!\leq t.
	\end{aligned}
\end{equation}
and space time norm
\begin{equation} \label{eq:energydefS}
	S_N(T)
	\!=
	\!\!{\sum}_{\vert I \vert \leq N\,\,}\int_0^T\!\!\!\int 	
	|\wpao
 \widetilde{Z}^I\! h^{\!1} (t,x)|^2 w'dx dt,
	\quad
	\text{where}
	\quad
w^\prime(t,x)\!=\bigg\{\begin{aligned} &(1\!+2\gamma)(1\!+|r^*\!\!-t|)^{2\gamma}\!\!,\quad\!\! &r^*\!>t,\\
&2\mu (1\!+|r^*\!\!-t|)^{-1-2\mu}\!,\quad &r^*\!\leq t,\end{aligned}
\end{equation}
where $|\wpao h |^2=|\swpa h|^2 \!+|\wL h|^2$ is the norm of the derivatives tangential to the outgoing curved light cones
$r^*\!\!-t=q^*$, where $|\swpa h|^2\!=\!\sum_{\,i=1,2}|S_i h|^2$ is the norm of the derivatives tangential to the sphere.
Moreover let
\begin{equation} \label{eq:energydefQ}
	Q_N(T)={\sup}_{0\leq t\leq T}{\sum}_{\vert I \vert \leq N\,\,}
	\int
	|\!D_{\widetilde{Z}}^I\widetilde{D} \bphi(t,x)|^2\!\!
+|\widetilde{Z}^{I} \bF^1 (t,x)|^2\,\,  w\, dx,
\end{equation}
where $D_{\widetilde{Z}}\bpsi=\widetilde{Z}^a \widetilde{D}_a$,
and $\widetilde{D}_a=(\wpa_a+i\widetilde{A}_a)\bpsi$.
We now restate a precise version of Theorem \ref{thm:EMKGIntro}:

\begin{theorem}\label{thm:mainrestated} Suppose that $1/2<\gamma<1$, and $N \geq 11$. There are constants $C_N, C_N^{\,\prime}<\infty$, $\mu>0$, and $\varepsilon_N>0$ such that if for some $\varepsilon\leq \varepsilon_N$
\beq\label{eq:initialassumption}
E_N(0)+Q_N(0)+M^2+q^2 \leq \varepsilon^2,
\eq
then Einstein-MKG  have a global solution satisfying
\beq\label{eq:finalestimate}
E_N(t)+S_N(t)\leq C_N\varepsilon^2(1\!+t)^{C_N^{\,\prime}\varepsilon},\quad\text{and}\quad  Q_N(t)\leq C_N\varepsilon^2.
\eq
Additionally, the decay estimates \eqref{eq:h1improcedhormanderweakbound}, \eqref{eq:h1improcedhormanderweakboundder}, \eqref{eq:wavecoordinatederivativeLie}, and \eqref{eq:wavecoordinatefunctionLie} hold with $\delta$ replaced by $C_N^{\,\prime}\varepsilon$, as do \eqref{est:DPhiDF1} and \eqref{est:DF2} for $s > 1/2$ satisfying $2s < 1+\gamma$.

Moreover, in the case $Q_N(0)=0$, i.e. the vacuum case, then the same result holds if $0<\gamma<1$.
\end{theorem}

\subsection{The bootstrap energy assumptions and the structure of the proof}
Given $1/2\!<\!\gamma\!<\!1$ (or $0\!<\!\gamma\!<\!1$ in the vacuum case)
pick a $\delta$ so small that $1/2\!+\!8\delta \!< \gamma\!<\!1\!-\!8\delta$
(or $8\delta \!< \!\gamma\!<\!1\!-\!8\delta$ in the vacuum case).\!
We will start by making the bootstrap assumptions that for some $C_b\!<\!\infty\!$
\beq\label{eq:bootstrapassumptionmetric}
E_N(t)+S_N(t)\leq C_{b}\, \varepsilon^2 (1\!+t)^{\delta}, \qquad\text{for}\quad 0\leq t\leq T,
\eq
and
\beq\label{eq:bootstrapassumptionfields}
 Q_N(t)\leq C_{b}\, \varepsilon^2 (1\!+t)^{\delta},  \qquad\text{for}\quad 0\leq t\leq T.
\eq
All energies are continuous, so it suffices to show that for some $C_b$ the estimates \eqref{eq:bootstrapassumptionmetric}, \eqref{eq:bootstrapassumptionfields} imply the same estimates with $C_b$ replaced by $C_b/2$.

\subsubsection{Getting back the bootstrap for the fields}
In the first part of sections \ref{sec:ProofMetricDecay} and \ref{sec:wavecoordbounds} we will show that the energy bounds required for Theorem \ref{thm:ChrisResults} follow from the bootstrap assumption \eqref{eq:bootstrapassumptionmetric} combined with the initial assumption \eqref{eq:initialassumption} for sufficiently small $\varepsilon$ depending on $C_b$. Specifically, we select $\varepsilon_b' = \varepsilon_b'(C_b) > 0$ small enough that $CC_b\varepsilon_b'^2 < \varepsilon_g^2$, where $\varepsilon_g$ is as in Theorem \ref{thm:ChrisResults} and $C$ refers to the constants which we use to retrieve the energy bounds \eqref{MetLIint}, \eqref{MetL2int} from \eqref{eq:bootstrapassumptionmetric}. Then there exist quantities $C_Q, \varepsilon_Q$ such that, if $\varepsilon < \varepsilon_Q$ and $\varepsilon < \varepsilon_b'$,
\beq\label{est:BootstrapRecovered}
Q_N(t)\leq C_Q\varepsilon^2.
\eq
In other words, we recover the bootstrap \eqref{eq:bootstrapassumptionfields} for $C_b \geq C_Q$.
It is important to note that, provided $CC_b\varepsilon_b'^2 < \varepsilon_g^2$, we may select $C_Q$ and $\varepsilon_Q$ independent of $C_b, \varepsilon_b'$ such that \eqref{est:BootstrapRecovered} holds. This bound on $Q_N$, combined with decay estimates for $\bF$ and $\bphi$, will give an $L^2$ bound for the energy momentum tensor for the Maxwell Klein Gordon system.

\subsubsection{Getting back the metric bootstrap}
We will finally prove in Section \ref{sec:ProofMetricEnergy}, see \eqref{eq:finalbootstrapestimate},
that the bootstrap assumptions imply that for some constant $C_N$ (with $C_N \geq C_Q$), that
is independent of the  constants $C_b$, $T,$ $\delta$,  and constants $C_b^{\,\prime}=C_b^{\,\prime}(C_b)<\infty$ and $\varepsilon_b=\varepsilon_b(C_b)>0$, that depend continuously on $C_b$ but are independent of $T,$ $\delta$, we have
\beq
E_N(t)+S_N(t)\leq C_{N}\, \varepsilon^2 (1\!+t)^{C_b^\prime \varepsilon}, \qquad\text{for}\quad 0\leq t\leq T, \quad \text{if}\quad \varepsilon\leq \varepsilon_b .
\eq
However,we can pick $C_b^\prime$ and $\varepsilon_b$ just depending on $N$ as well. In fact, this is achieved by picking $C_b=2C_N$ assuming that $C_N\geq 1$ as we may, and the further picking $0<\varepsilon_N\leq \min(\varepsilon_Q, \varepsilon_b(2C_N), \varepsilon'_b(2C_N))$ so small that $\varepsilon_N C_b^\prime(2C_N)\leq \delta$. Once we found the universal constant $C_N$, through the estimates, we can now go back and make the more intelligent bootstrap assumption that
\eqref{eq:bootstrapassumptionmetric} hold with $C_b=2C_N$. It follows that \eqref{eq:bootstrapassumptionfields} holds also. Therefore, we may extend the bootstrap argument to all time, which implies that \eqref{eq:finalestimate} holds globally.

\subsubsection{The structure of the proof}\label{sec:ProofStructure}
We will appeal to the results stated in Section \ref{sec:ChrisResults} in a nested way in the proof. First we will show that the apriori bounds stated above imply weak decay estimates for the metric, and also $L^2$ estimates for the components controlled by the wave coordinate condition. This in turn is sufficient to satisfy the assumptions on the metric needed to prove the results in Section  \ref{sec:ChrisResults}. Then using the decay estimates for the fields from Section \ref{sec:ChrisResults} we will be able to prove the sharp decay estimates for the metric which are needed to prove the sharper $L^2$ estimates for the metric which gives back the apriori bounds for the metric. Moreover the sharper $L^2$ bounds for the field also follows from  Section \ref{sec:ChrisResults}.
\subsection{The general statement}\label{sec:TheoremGenericProof}
When proving this for a more generic field, we will need to bound (derivatives of) the energy-momentum tensor at two points. First, in Proposition \ref{prop:theL1bound}, we prove that the bootstrap assumptions \eqref{eq:bootstrapassumptionmetric}-\eqref{eq:bootstrapassumptionfields} imply slow growth for the energy-momentum tensor in $L^1$, which in turn implies the sharper decay and energy results in Sections \ref{subsec:WCCDecay} and \ref{subsec:WCCEnergy}. In sections \ref{sec:sharpdecayfields} and \ref{sec:MKGBootstrap}, we use these bounds to propagate the field assumption \eqref{eq:bootstrapassumptionfields}, as well as improve bounds on the energy-momentum tensor. We can hence state a more general version of Theorem \ref{thm:EExistence} abstractly in the following form:
\begin{theorem}\label{thm:EFEgenericstability}
Let $\widecheck{T}$ be the reduced energy momentum tensor for some system of equations, and let $Q$ be a bootstrap energy for this system such that the following hold:
\begin{enumerate}
\item The bootstrap assumptions \eqref{eq:bootstrapassumptionmetric}-\eqref{eq:bootstrapassumptionfields} imply the inequality
\beq \label{eq:L1boundT}
\|\widetilde{Z}^J \widecheck{T}(t,\cdot)\|_{L^1(w_{1,\gamma})}
\lesssim \varepsilon^2 (1+t)^{\delta},\qquad |J|\leq N-3,
\eq
\item The improved decay bounds given in Lemma \ref{lem:h1improcedhormanderweakbound} and Proposition \ref{prop:wavecoorddecayhighlow}, as well as the improved energy bound \eqref{est:L2HLL} given by the wave coordinate condition, imply the improved bootstrap
\beq\label{QBoundsgeneric}
Q_N(T)\leq C\varepsilon^2,
\eq
as well as the decay bounds
\begin{subequations}
\begin{align}
|(\Lie_{\widetilde{X}}^I\widecheck{T})| &\lesssim \varepsilon^2\langle t+r^* \rangle^{-2}\langle t-r^*\rangle^{-2}\langle (t-r^*)_+\rangle^{1-2s},\qquad |I|\leq N-6\label{TBoundsLinftygeneric1},\\
|(\Lie_{\widetilde{X}}^I\widecheck{T})|_{\mathcal{U}\mathcal{T}} &\lesssim \varepsilon^2\langle t +r^*\rangle^{-2-s}\langle t-r^*\rangle^{-2+s}\langle (t-r^*)_+\rangle^{1-2s},\qquad |I|\leq N-6\label{TBoundsLinftygeneric2}
\end{align}
\end{subequations}
and the energy bound
\beq\label{TBoundsL2generic}
\int_{\Sigma_t}|(\Lie_{\widetilde{X}}^I\widecheck{T})|^2w\, dx \leq \frac{C\varepsilon^2}{1+t},\qquad |I|\leq N.
\eq
\end{enumerate}
Then the statement of Theorem \ref{thm:mainrestated} holds for the general Einstein-field system.
\end{theorem}
\begin{proof}
We prove this in the same way as Theorem \ref{thm:mainrestated}, first proving Proposition \ref{prop:theL1bound} directly using the bound \eqref{eq:L1boundT}. Consequently, the propagation of the bootstrap in Theorem \ref{thm:FieldBootstrap} follows from \eqref{QBoundsgeneric}. Finally, the bounds \eqref{TBoundsLinftygeneric1}, \eqref{TBoundsLinftygeneric1}, and \eqref{TBoundsL2generic} replace those given in section \ref{sec:MKGBootstrap}.
\end{proof}

\section{The weak decay of the metric assuming weak energy bounds}\label{sec:ProofMetricDecay}

In this section we will prove the weak decay estimates of the metric assuming the apriori energy bounds in the previous section.

\subsubsection{The weak decay of the metric and the fields using Klainerman-Sobolev}\label{sec:ProofMetricDecayWeak}

It follows from the weighted Klainerman-Sobolev and the assumed bounds
that
\beq
|\pa \widetilde{Z}^J {h}^1(t,x)|+|\widetilde{Z}^J \mathbf{F}^1|
\leq C\varepsilon (1+t)^{\delta/2}  (1+t+|q^*|)^{-1} (1+|q^*|)^{-1/2}
\langle q_+^*\rangle^{-1/2-\gamma},\qquad |J|\leq N-3.
\eq
Let us for simplicity now assume $\gamma>1/2$.
Integrating this from initial data in the $t-r$ direction also gives
\beq\label{eq:h1weakbound}
|\widetilde{Z}^J h^1(t,x)|
\leq C\varepsilon (1+t)^{\delta/2}  (1+t+|q^*|)^{-1} (1+|q^*_-|)^{1/2}
 \langle q_+^*\rangle^{1/2-\gamma},\qquad |J|\leq N-3.
\eq

\subsubsection{The $L^1$ bound of the inhomogeneous terms}\label{sec:theL1bound}
We now prove that weak $L^2$ bounds for the metric and fields imply weak $L^1$ bounds for the semilinear terms and energy-momentum tensor. This will be used in the following section to get an improved weak decay estimate.

\begin{proposition}\label{prop:theL1bound}
The bootstrap assumptions \eqref{eq:bootstrapassumptionmetric}-\eqref{eq:bootstrapassumptionfields} imply the bound
\beq \label{eq:L1bound}
\|\widetilde{Z}^J {F}(g)[\pa h,\pa h](t,\cdot)\|_{L^1(w_{1,\gamma})}
+ \|\widetilde{Z}^J {T}(t,\cdot)\|_{L^1(w_{1,\gamma})}
\lesssim \varepsilon^2 (1+t)^{\delta},\qquad |J|\leq N-3.
\eq
\end{proposition}
\begin{proof}
The first part of \eqref{eq:L1bound} is easy so let us first do the second part. Again the part for
$\mathbf{F}^1$ follows directly from the $L^2$ bounds although there are some issues
with subtracting of the charge in order to get the bound for $\mathbf{F}$.
These are similar to subtracting of the mass for $h$. The main difficulty is
therefore the estimate for $\bphi$.

We have
\beq
Q_{\alpha\beta}[D\bphi] = \mathfrak{R}\bigtwo(D_\alpha\bphi\cdot\overline{D_\beta\bphi}
 -  g_{\alpha\beta} g^{\mu\nu}D_\mu\bphi\cdot\overline{D_\nu\bphi}/2\bigtwo),
\eq
where $D_a\bphi=\wpa_a\bphi+i A_a \bphi$ and
\beq
\widetilde{Z} \, \mathfrak{R}\bigtwo(D_\alpha\bphi\cdot\overline{D_\beta\bpsi}\bigtwo)
=\mathfrak{R}\bigtwo(D_{\widetilde{Z}} D_\alpha  \bphi\cdot\overline{D_\beta\bpsi}\bigtwo)
+\mathfrak{R}\bigtwo(D_\alpha\bphi\cdot
\overline{D_{\widetilde{Z}}D_\beta \bpsi}\bigtwo).
\eq
In view of \eqref{eq:h1weakbound} it follows that for $|J|\leq N-3$
\beq
|\widetilde{Z}^J Q_{\alpha\beta}[D\bphi] |
\lesssim {\sum}_{|K|\leq |J|} |\!D_{\widetilde{Z}}^K\widetilde{D} \bphi|^2 .
\eq

We have
\beq
Q_{\alpha\beta}[\bF]
= g^{\mu\nu}\bF_{\alpha \mu} \bF_{\beta\nu}
-\frac{1}{4} g_{\alpha\beta} g^{\mu\nu} g^{\gamma\delta}
\bF_{\mu\gamma} \bF_{\nu\delta}  ,
\eq
and hence in view of \eqref{eq:h1weakbound} it follows that for $|J|\leq N-3$
\beq
|\widetilde{Z}^J Q_{\alpha\beta}[\bF] |
\lesssim {\sum}_{|K|\leq |J|} |{\widetilde{Z}}^K\bF|^2
\lesssim {\sum}_{|K|\leq |J|} |{\widetilde{Z}}^K\bF^1|^2
+{\sum}_{|K|\leq |J|} |{\widetilde{Z}}^K\bF^0|^2 .
\eq
We also have
\beq
|\widetilde{Z}^J F(g)[\pa h,\pa h]|
\lesssim {\sum}_{|K|\leq |J|} |{\widetilde{Z}}^K \pa h|^2
\lesssim {\sum}_{|K|\leq |J|} |{\widetilde{Z}}^K \pa h^1|^2
+{\sum}_{|K|\leq |J|} |{\widetilde{Z}}^K  \pa h^0|^2 .
\eq
We have
\beq
{\sum}_{|K|\leq |J|} |{\widetilde{Z}}^K\bF^0|^2
+{\sum}_{|K|\leq |J|} |{\widetilde{Z}}^K  \pa h^0|^2
\lesssim  \frac{\varepsilon^2}{(1+t+r)^{4}} .
\eq
It follows that
\beq
{\sum}_{|K|\leq |J|} \int  |{\widetilde{Z}}^K\bF^0|^2\!
+ |{\widetilde{Z}}^K  \pa h^0|^2 \, \, w_{1,\gamma}  dx
\lesssim \frac{\varepsilon^2}{(1+t)^{1-\gamma}},\qquad \gamma<1.
\eq
This together with the bootstrap assumptions
\eqref{eq:bootstrapassumptionmetric}-\eqref{eq:bootstrapassumptionfields} proves
\eqref{eq:L1bound}, since $w_{1, \gamma}\lesssim w$.
\end{proof}

\subsubsection{The improved weak decay of the metric and the fields using H\"ormander's $L^1-L^\infty$ estimate}\label{sec:ImprovedWeakDecay}
We will now show that
\begin{lemma}\label{lem:h1improcedhormanderweakbound}
\beq\label{eq:h1improcedhormanderweakbound}
|\widetilde{Z}^J h^1(t,x)|
\leq C\varepsilon (1+t)^{\delta}  (1+t+|q^*|)^{-1}
 \langle q_+^*\rangle^{-\gamma},\qquad |J|\leq N-5,
\eq
and
\beq\label{eq:h1improcedhormanderweakboundder}
(1\!+t\!+|q^*|)|\overline{\wpa} \widetilde{Z}^J h^1(t,x)|+\langle q^*\rangle |\wpa \widetilde{Z}^J h^1(t,x)|
\leq C\varepsilon (1\!+t)^{\delta}  (1\!+t\!+|q^*|)^{-1}
 \langle q_+^*\rangle^{-\gamma}\!\!,\qquad |J|\leq N\!-6.
\eq
\end{lemma}
\begin{proof} \eqref{eq:h1improcedhormanderweakboundder} follows directly from \eqref{eq:h1improcedhormanderweakbound}.
To prove  \eqref{eq:h1improcedhormanderweakbound} we will apply Proposition \ref{prop:hormander} and Lemma
\ref{lemma:linearhomogenousdecayweighted}
 to $\Box^* \widehat{\mathcal{L}}_{\widetilde{X}}^I
\widetilde{h}^{1}$. Applying ${\mathcal{L}}_{\widetilde{X}}^I$
to \eqref{eq:newcoordeinsteinwithRerrors}
 using Lemma \ref{lem:TheCommutatorLemma} and Lemma \ref{lem:errorlemma}
\begin{multline}\label{eq:LiecommutatortildekapparepeatstarFineq}
	|\Box^* \widehat{\mathcal{L}}_{\widetilde{X}}^I
\widetilde{h}^{1}|
	\lesssim \bigtwo|\Box^* \widehat{\mathcal{L}}_{\widetilde{X}}^I
\widetilde{h}_{cd}^1
-\Lie_{\widetilde{X}}^I \big( \kappa\widetilde{\Box} \widetilde{h}^1_{cd}\big)\bigtwo|
+\big|{\mathcal{L}}_{\widetilde{X}}^I\widetilde{F}(\tg)[\wpa \th,\wpa \th]\big|
+ |{\mathcal{L}}_{\widetilde{X}}^I\widetilde{\widecheck{T}}|
+ \big|\Lie_{\widetilde{X}}^I\widetilde{R}^{\, mass}\big|
+\big|\Lie_{\widetilde{X}}^I\widetilde{R}^{\, cov}\big|\\
\lesssim {\sum}_{|J|+|K|\leq |I|+1,\, |J|\leq  |I|}
\bigtwo(\frac{M\ln{\langle t+r\rangle}\!\!}{\langle t+r\rangle^2}
+\frac{|\widehat{\mathcal{L}}_{\widetilde{X}}^J\widetilde{H}_1|}{\langle t-r^*\rangle}\bigtwo)
\big|\wpa\widehat{\Lie}^{K}_{\widetilde{X}}\widetilde{h}^{1}\big|
+\big|{\mathcal{L}}_{\widetilde{X}}^I\widetilde{F}(\tg)[\wpa \th,\wpa \th]\big|
+ |{\mathcal{L}}_{\widetilde{X}}^I\widetilde{\widecheck{T}}|\\
+ \frac{M H(r \!<\! t/2)}{\langle t+r\rangle^3}+\frac{M^2 }{\langle t+r\rangle^4}
+\frac{M\big(\ln{\langle t\!+r\rangle}+1\big)}{\langle t\!+r\rangle^3}
{\sum}_{|J|\leq |I|}\big|\widetilde{Z}^J\widetilde{H}_1\big|
+\frac{M\big(\ln{\langle t\!+r\rangle}+1\big)}{\langle t\!+r\rangle^2}
{\sum}_{|J|\leq |I|}\big|\wpa\widetilde{Z}^J\widetilde{H}_1\big|.
\end{multline}
We have already estimated $\big|{\mathcal{L}}_{\widetilde{X}}^I\widetilde{F}(\tg)[\wpa \th,\wpa \th]\big|
+ |{\mathcal{L}}_{\widetilde{X}}^I\widetilde{\widecheck{T}}|$
in the weighted $L^1$ norm,
and using Cauchy Schwarz the other terms are bounded by
\beq
\frac{M}{\langle t+r\rangle^3}H(r\! <\! t/2)+\frac{M^2(1+\ln{\langle t+r\rangle})^2}{\langle t+r\rangle^4}
+{\sum}_{|K|\leq |I|+1}\big|\wpa
\widehat{\Lie}^{K}_{\widetilde{X}}\widetilde{h}^{1}\big|^2 +
{\sum}_{|J|\leq |I|}\Big(
\frac{|\widehat{\mathcal{L}}_{\widetilde{X}}^J\widetilde{H}_1|}{\langle t-r^*\rangle}\Big)^2.
\eq
We already dealt with terms similar to the first two terms and the first sum in Section \ref{sec:ImprovedWeakDecay}.
For the first term one uses that it vanishes when $r>t$ so that it can absorb any exterior weight.
It therefore only remains to take the weighted $L^1$ norm of the last sum. However,
this follows from that $H_1=-h^1+O(h^2)$ and Corollary \ref{cor:poincare}, since
$w_{1, \gamma}\lesssim w$. In conclusion we have
\beq
\int |\Box^* \widehat{\mathcal{L}}_{\widetilde{X}}^I
\widetilde{h}^{1}|\, w_{1, \gamma}\, dx \lesssim  \varepsilon (1+t)^{\delta}.
\eq
This together with Proposition \ref{prop:hormander} and Lemma \ref{lemma:linearhomogenousdecayweighted}
gives \eqref{eq:h1improcedhormanderweakbound}.
\end{proof}

\section{The decay and $L^2$ bounds from the wave coordinate condition}\label{sec:wavecoordbounds}
In this section we prove improved decay and $L^2$ estimates for the components of the metric controlled by the wave coordinate condition. These are needed both the sharp energy estimate for the metric in section \ref{sec:ProofMetricEnergy} and for the energy estimates for the field in section \ref{sec:ChrisResults}.

\subsection{The strong decay estimates from the wave coordinate condition}\label{subsec:WCCDecay}
By \eqref{eq:higherwavecoordinateLiederivativeH1Wsec7}
we have
\begin{multline}
|\pa_{q^*} \widehat{\mathcal L}_{\widetilde{Z}}^I
\widetilde{H}_1|_{\widetilde{L}\widetilde{\mathcal T}}
+|\pa_{q^*}  \trs\widehat{\mathcal L}_{\widetilde{Z}}^I \widetilde{H}_1 |
\les |\overline{\wpa} \widehat{\mathcal L}_{\widetilde{Z}}^I \widetilde{H}_1|
+ {\sum}_{|J|+|K|\leq |I|} \big|  \widehat{\mathcal L}_{\widetilde{Z}}^{J} \widetilde{h}^1\big|
    \big|\wpa  \widehat{\mathcal L}_{\widetilde{Z}}^{K} \widetilde{h}^1\big|\\
    + \frac{M\big|\chi^\prime\big(\tfrac{r}{t+1}\big)\big|}{(1+t+r)^2}
    + \frac{M}{(1+t+r)^3}+ \frac{M(1+\ln{\langle t+r\rangle})}{(1+t+r)^2}
    {\sum}_{|J|\leq |I|} \big|  \widehat{\mathcal L}_{\widetilde{Z}}^{J}\widetilde{H}_{1}\big|
    + \frac{M}{1+t+r}
    {\sum}_{|J|\leq |I|} \big| \wpa \widehat{\mathcal L}_{\widetilde{Z}}^{J}\widetilde{H}_{1}\big|,
\label{eq:higherwavecoordinateLiederivativeH1W}
\end{multline}
where $\chi^\prime(s)$, is a function supported when $1/4\leq s\leq 1/2$.

\subsubsection{The additional decay estimates from the wave coordinate condition}

\begin{proposition}\label{prop:wavecoorddecayhighlow} For $|I|\leq N-6$ we have
\begin{equation}
|\pa_{q^*} \widehat{\mathcal L}_Z^I{H}_1|_{L\mathcal T}
+|\pa_{q^*}  \trs\widehat{\mathcal L}_Z^I {H}_1 |\les
\varepsilon(1+t+r)^{-2+2\delta}\langle q_+^*\rangle^{-\gamma}\langle q^*\rangle^{-\delta},
\label{eq:wavecoordinatederivativeLie}
\end{equation}
\begin{equation}
 |\widehat{\mathcal L}_Z^I{H}_1|_{L\mathcal T}\!
+|\trs\widehat{\mathcal L}_Z^I {H}_1 |\!
\les \varepsilon
(1\!+t\!+r)^{-1-\gamma+\delta}\!\! +
\varepsilon(1\!+t)^{-2+2\delta}\langle q^*_-\rangle^{1-\delta}\!
\les \varepsilon(1\!+t\!+r)^{-1-\gamma+\delta} \langle q^*_-\rangle^{\gamma}\! .
\label{eq:wavecoordinatefunctionLie}
\end{equation}
\end{proposition}
\begin{proof} By \eqref{eq:higherwavecoordinateLiederivativeH1W} and
\eqref{eq:h1improcedhormanderweakbound}-\eqref{eq:h1improcedhormanderweakboundder}
we have for $|I|\leq 6$
\begin{multline}
|\pa_{q^*} \widehat{\mathcal L}_{\widetilde{Z}}^I
\widetilde{H}_1|_{\widetilde{L}\widetilde{\mathcal T}}
+|\pa_{q^*}  \trs\widehat{\mathcal L}_{\widetilde{Z}}^I \widetilde{H}_1 |
\lesssim \frac{\varepsilon (1+t)^{\delta}}{ (1+t+|q^*|)^2}
 \langle q_+^*\rangle^{-\gamma}
+\frac{\varepsilon^2 (1+t)^{2\delta} }{ (1+t+|q^*|)^{2}\langle q^*\rangle \langle q_+^*\rangle^{2\gamma}}
\\
    + \frac{M\big|\chi^\prime\big(\tfrac{r}{t+1}\big)\big|}{(1+t+r)^2}
    + \frac{M}{(1+t+r)^3}+ \frac{M\varepsilon(1+\ln{\langle t+r\rangle})(1+t)^{\delta}}{(1+t+r)^3\langle q_+^*\rangle^{\gamma}}
    + \frac{M\varepsilon(1+t)^{\delta}}{(1+t+r)^2\langle q^*\rangle\langle q_+^*\rangle^{\gamma}},
\label{eq:higherwavecoordinateLiederivativeH1W2}
\end{multline}
from which the first part follows. The second follows from integrating the first
in the $q^*$ direction from initial data where it is bounded by
$ \varepsilon(1\!+r)^{-1-\gamma+\delta}$,
because of Sobolev's lemma and the fact the energy is bounded initially.
\end{proof}

\subsection{The additional $L^2$ bounds from the wave coordinate condition}\label{subsec:WCCEnergy}
 Using \eqref{eq:h1improcedhormanderweakbound}-\eqref{eq:h1improcedhormanderweakboundder}
in \eqref{eq:higherwavecoordinateLiederivativeH1W} and the fact that $H_1=-h^1+O(h^2)$ we have
\begin{multline}
|\pa_{q^*} \widehat{\mathcal L}_{\widetilde{Z}}^I
\widetilde{H}_1|_{\widetilde{L}\widetilde{\mathcal T}}
+|\pa_{q^*}  \trs\widehat{\mathcal L}_{\widetilde{Z}}^I \widetilde{H}_1 |
\les |\overline{\wpa} \widehat{\mathcal L}_{\widetilde{Z}}^I \widetilde{h}^1|\\
+ \frac{M(1\!+\ln{\,\langle t\!+r\rangle})\!+
\varepsilon(1\!+t)^{\delta}\langle q_+^*\rangle^{-\delta}\!\!\!\!}
    {(1+t+r)}
    \sum_{|J|\leq |I|}\!\!\!\Big(
  \frac{1}{\langle q^*\rangle}
   \big|  \widehat{\mathcal L}_{\widetilde{Z}}^{J}\widetilde{h}^{1}\big|
  +\big| \wpa \widehat{\mathcal L}_{\widetilde{Z}}^{J}\widetilde{h}^{1}\big|\Big)
    + \frac{M\big|\chi^\prime\big(\tfrac{r}{t+1}\big)\big|}{(1+t+r)^2}
    + \frac{M}{(1+t+r)^3} .
\label{eq:higherwavecoordinateLiederivativeH1Wsimplified}
\end{multline}
The following $L^2$ will be used both in the energy estimate for the fields and for the metric:
\begin{lemma}\label{lem:wavecoordL2est0} If  $|I|\geq 6$ then
\begin{multline}
\int_0^T\!\!\!\int \big(	|\pa_{q^*} \widehat{\mathcal L}_{\widetilde{Z}}^I{H}_1|_{L\mathcal T}^2
+|\pa_{q^*}  \trs\widehat{\mathcal L}_{\widetilde{Z}}^I {H}_1 |^2\big) w'dx dt
\\
\lesssim
 \int_0^T\!\!\!\int |\wpao
 \widetilde{Z}^I\! h^{\!1} (t,x)|^2 w'dx dt +(M^2\!+\varepsilon^2)
 \Big(1+  \!\!{\sum}_{|J|\leq |I|}\int_0^T\!\!\!\!\!\frac{dt}{(1+t)^{2-2\delta}}\int |\wpa
 \widetilde{Z}^J\! h^{\!1} (t,x)|^2 w dx\Big) .
\end{multline}
\end{lemma}
\begin{proof}
Using the first part of Corollary \ref{cor:poincare} and the fact that $w'\sim w\langle q^*\rangle^{-1}\langle q_-^*\rangle^{\,-2\mu}$ in \eqref{eq:higherwavecoordinateLiederivativeH1Wsimplified}
gives the result.
\end{proof}
The following $L^2$ bounds are needed estimates for the fields.
\begin{lemma}\label{lem:wavecoordL2est} If $0\leq b\leq \gamma $ and $|I|\geq 6$ then
\begin{multline}
\int_0^T\!\!\!\int \Big(	\frac{| \widehat{\mathcal L}_{\widetilde{Z}}^I{H}_1|_{L\mathcal T}^2}{\langle q^*\rangle^2}
+\frac{|  \trs\widehat{\mathcal L}_{\widetilde{Z}}^I {H}_1 |^2}{\langle q^*\rangle^2}\Big) \Big(\frac{\langle q^*\rangle}{\langle\, t\!+r^*\rangle}\Big)^{1-b} w^\prime  \, dx dt
\\
\lesssim
 \int_0^T\!\!\!\int |\wpao
 \widetilde{Z}^I\! h^{\!1} (t,x)|^2 w'dx dt +(M^2\!+\varepsilon^2)
 \Big(1+  \!\!{\sum}_{|J|\leq |I|}\int_0^T\!\!\!\!\!\frac{dt}{(1+t)^{2-2\delta}}\int |\wpa
 \widetilde{Z}^J\! h^{\!1} (t,x)|^2 w dx\Big) .
\end{multline}
\end{lemma}
\begin{proof}
Using the second part of Corollary \ref{cor:poincare} and the fact that $w'\sim w\langle q^*\rangle^{-1}\langle q_-^*\rangle^{\,-2\mu}$ and the previous lemma
gives the result.
\end{proof}
Hence we have the additional apriori assumption
\beq\label{est:L2HLL}
\int_0^T\!\!\!\int 	|\wpa \widehat{\mathcal L}_{\widetilde{Z}}^I{H}_1|_{L\mathcal T}^2
+|\wpa\trs\widehat{\mathcal L}_{\widetilde{Z}}^I {H}_1 |^2 w'dx dt\les C \varepsilon^2 (1+T)^{\delta}.
\eq

The following $L^2$ estimates are needed for the energy estimates of the metric.

\beq
 \int_0^T\!\!\int_{\Sigma_t}\!
R_{tan}^{\,2}\frac{(1\!+t)^{1+2\nu}\!\!\!\!\!}{\varepsilon}\,\,w\, dx \,dt
\les {\sum}_{|J|\leq |I|} C\varepsilon\int_0^T\!\!\int_{\Sigma_t}\!\frac{ |\overline{\wpa} h^{1 J}|^2}{ \langle q^*\rangle^{1+2\nu}} \,\,w\, dx \,dt
\les  C\varepsilon S_N(T).
\eq

\section{The sharp energy and decay estimates for the fields}\label{sec:sharpdecayfields}
We can now appeal to the results of Section \ref{sec:ChrisResults} to further bound the metric estimates in the bootstrap assumption. We first show that the bootstrap assumption implies the required bound on metric components.

\begin{proposition}
For sufficiently small $\varepsilon$, the bounds \eqref{MetLIint} and \eqref{MetL2int} follow from \eqref{eq:bootstrapassumptionmetric}.
\end{proposition}
\begin{proof}
The bound \eqref{eq:bootstrapassumptionmetric} immediately implies
\beq
{\sup}_{t\in[0,T]}\lVert\wpa \Lie_{\wt{Z}}^Ih^1w^{1/2}\rVert_{L^2(\mathbb{R}^3)} + \lVert\slashed{\wpa}
\Lie_{\wt{Z}}^Ih^1w'_\gamma{}^{1/2}\rVert_{L^2{([0,T]\times\mathbb{R}^3)}}\leq C\varepsilon C_b^{1/2\!}(1\!+T)^{{\delta}\!/{2}}\!\!,
\eq
The estimate
\beq
\left\lVert{\tminusrstar}^{-1}|\Lie_{\widetilde{Z}}^Ih^{1\!}|w^{1/2}
\right\rVert_{L^2(\mathbb{R}^3)}\!\!\leq C\varepsilon C_b^{1/2\!}(1\!+T)^{{\delta}\!/{2}}\!\!,
\eq
follows directly from Corollary \ref{cor:poincare} in spherical coordinates.
Likewise, \eqref{est:L2HLL} gives
\beq
\!\left\lVert |\wpa\Lie_{\widetilde{Z}}^Ih^{1\!}|_{\mathcal{L}\mathcal{L}}
w'_\gamma{}^{1/2}
\right\rVert_{L^2(I_{T'}\!\times\mathbb{R}^3)}\!\!\leq C\varepsilon (1+T')^{\delta/2},
\eq
noting that the difference consists of similar components containing lower order derivatives. To prove \eqref{L23int} we convert to spherical coordinates and appeal to Corollary \ref{cor:poincare} more directly, specifically the estimate \eqref{cor1332}. For any $s \in (\tfrac12 + 4\delta, 1-4\delta)$ we take $b = 2s - 1 + 2\delta$, so $b\in(10\delta, 1-6\delta)$.
\begin{multline}
\lVert\langle t-\rs\rangle^{-1-s}\langle t+\rs\rangle^{s-1}|\Lie_{X}^Ih^{1\!}|_{\mathcal{L}\mathcal{L}}\rVert_{L^2([0,T'] \times\mathbb{R}^3)} \lesssim {T'}^{-\delta}\lVert\langle t-\rs\rangle^{-1-s}\langle t+\rs\rangle^{s-1+\delta}|\Lie_{\widetilde{Z}}^Ih^{1\!}|_{\mathcal{L}\mathcal{L}}\rVert_{L^2([0,T'] \times\mathbb{R}^3)}\\
\lesssim {T'}^{-\delta}\lVert\langle t-\rs\rangle^{-\tfrac32-\tfrac{b}{2} +\delta}\langle t+\rs\rangle^{\tfrac{b-1}{2}}|\Lie_{\widetilde{Z}}^Ih^{1\!}|_{\mathcal{L}\mathcal{L}}\rVert_{L^2([0,T'] \times\mathbb{R}^3)}.
\end{multline}
Since $w \geq 1$ and $\mu < 1-\gamma < \tfrac12 - 4\delta$, it follows that
\beq
\langle t-\rs\rangle^{-\tfrac12 +\delta} \lesssim \frac{w}{(1+q^*_-)^{\mu}},
\eq
so we may apply \eqref{cor1332} and \eqref{est:L2HLL}, then sum over dyadic regions in time $[2^k, 2^{k+1}]$.
\end{proof}

 It remains to show that the initial conditions \eqref{def:ID} follow from the initial assumption \eqref{eq:initialassumption} for sufficiently small $\varepsilon$.
\begin{proposition}\label{prop:BootstrapIC}
There exists a fixed constant $C$ such that
\beq
\lVert{\bf E}^{df}_0\rVert^2_{H^{N, 1/2+\gamma}} + \lVert{\bf B}_{0}\rVert^2_{H^{N, 1/2+\gamma}} + \lVert D\bphi_0 \rVert^2_{H^{N, 1/2+\gamma}} + \lVert\,\,\dot{\!\!\bphi}_0\rVert^2_{H^{N, 1/2+\gamma}}< C(1+C_b\varepsilon^2)Q_N(0).
\eq
\end{proposition}
\begin{proof}
The quantities $\lVert D\bphi_0 \rVert^2_{H^{N,1/2+\gamma}} + \lVert\,\,\dot{\!\!\bphi}_0\rVert^2_{H^{N, 1/2+\gamma}}$ follow immediately from the definitions. For the estimates on the fields, we first use the standard frame in modified coordinates. This allows us to convert between Lie derivatives $\Lie_{\widetilde{Z}}$ and standard derivatives $\widetilde{Z}$ without modifying the estimates, as the difference consists of lower order derivatives of the same weight. The estimate on $\lVert{\bf B}_{0}\rVert^2_{H^{N, 1/2+\gamma}}$ therefore follows from the definition as well as \eqref{MetLIint} in the case where most derivatives fall on $\bF^1$, or a Klainerman-Sobolev estimate similar to the worst components of $\bF^1$ in \eqref{est:DPhiDF1}, combined with \eqref{MetL2int}, in the case where most derivatives fall on $\varepsilon_{0ijk}$. The bound on $\lVert{\bf E}_{0}^{df}\rVert^2_{H^{N, 1/2+\gamma}}$ follows from an elliptic estimate as well as a Hardy inequality on the current, and follows from Lemma 10.1 in \cite{LS} and Hölder's Inequality.
\end{proof}
\begin{theorem}\label{thm:FieldBootstrap}
Given constants $\gamma, \delta, \mu, N$ satisfying $N \geq 11$, $\gamma > 1/2$, $\mu < 1/2$, and $8\delta < \min(\gamma-1/2, 1-\gamma)$, the bounds \eqref{MetLIint} and \eqref{MetL2int} follow from the bootstrap assumption \eqref{eq:bootstrapassumptionmetric} for sufficiently small $\varepsilon_b$, and therefore the system \eqref{eq:MKG} is well-posed while \eqref{eq:bootstrapassumptionmetric} holds. Additionally,
\beq
Q_N(T)\leq C\varepsilon^2,
\eq
and the bounds \eqref{est:DPhiDF1} hold for $s_0 = 1/2+\gamma$, $s = 1/2+\gamma/2 - 3\delta$.
\end{theorem}
\begin{proof} We first note that $1-s = (1-\gamma)/2+3\delta \geq 6\delta$, $s-1/2\geq 1/4 \geq 6\delta$, $1-2\mu \geq 2(\gamma-1/2)$ $1+\gamma-2s = 6\delta$, $s_0 \geq 1$. Therefore, the restrictions on $s,s_0,$ and $\delta$ in Theorem \ref{thm:ChrisResults} are satisfied.
This follows from Equations \eqref{est:ChrisEnergyBounds1} and \eqref{est:ChrisEnergyBounds2}, along with the identity \eqref{id:LieT} and Corollary \ref{cor:DDComm}, which handle the commutator terms.
\end{proof}
\subsubsection{The energy and decay estimates for the energy momentum tensor}\label{sec:MKGBootstrap}
Corollaries \ref{cor:DecayBoundsOnT} and \ref{cor:EnergyBoundsOnT} follow directly from Theorem \ref{thm:FieldBootstrap}.
\begin{proposition}\label{prop:energymomentumLinfty}
Given constants $\gamma, \delta, \mu, N$ satisfying $N \geq 11$, $1/2 < \gamma < 1$, $0 < \mu < 1-\gamma$, and $0 < 8\delta < \min(1-2\mu, \gamma-1/2, 1-\gamma)$, as well as the initial bounds \eqref{eq:initialassumption}, and the bootstrap assumption \eqref{eq:bootstrapassumptionmetric}, then for sufficiently small $\varepsilon$, the  following bounds hold for $|I|\! \leq \!N\!-\!6$ and $s = (1+\gamma)/2 - 3\delta$:
\begin{subequations}
\begin{align}
|(\Lie_{\widetilde{X}}^I\widecheck{T})| &\lesssim \varepsilon^2\langle t+r^* \rangle^{-2}\langle t-r^*\rangle^{-2}\langle (t-r^*)_+\rangle^{1-2s},\\
|(\Lie_{\widetilde{X}}^I\widecheck{T})|_{\mathcal{U}\mathcal{T}} &\lesssim \varepsilon^2\langle t +r^*\rangle^{-2-s}\langle t-r^*\rangle^{-2+s}\langle (t-r^*)_+\rangle^{1-2s}.
\end{align}
\end{subequations}
\end{proposition}
We state a revised form of Corollary \ref{cor:EnergyBoundsOnT}.
\begin{lemma}\label{prop:energymomentumL2}
Given constants $\gamma, \delta, \mu, N$ satisfying $N \geq 11$, $1/2 < \gamma < 1$, $0 < \mu < 1-\gamma$, and $0 < 8\delta < \min(1-2\mu, \gamma-1/2, 1-\gamma)$, as well as the initial bounds \eqref{eq:initialassumption}, and the bootstrap assumption \eqref{eq:bootstrapassumptionmetric}, then for sufficiently small $\varepsilon$,  the following bound holds:
\beq
\int_{\Sigma_t}|(\Lie_{\widetilde{X}}^I\widecheck{T})|^2w\, dx \leq \frac{C\varepsilon^2}{1+t}.
\eq
\end{lemma}

\section{The sharp decay estimates for the metric}\label{sec:metricdecaysharp}
We now use the weak decay estimates for the metric in Section \ref{sec:ProofMetricDecay}
and the decay estimates for the fields in section \ref{sec:sharpdecayfields} in the wave equation for the metric to obtain the sharp decay estimates for the metric.

From Lemma \ref{lem:transversalder},
the decay estimates we have proven so far and Lemma \ref{lem:TheCommutatorLemma}
to estimate the commutator
$\Box^* \widehat{\mathcal{L}}_{\widetilde{X}}^I
\widetilde{h}_{cd}^1
	-\Lie_{\widetilde{X}}^I \big( \kappa\widetilde{\Box}
 \widetilde{h}^1_{cd}\big)$ we get
\begin{lemma}\label{lem:transversaldersimplifiedLiecommuted}
Let $D_t\!=\!\{(t,x);\, \big|t\!-\!|x|\big|\!\leq\! c_0 t \}$, for some
$0\!<\!c_0\!<\!1$. With
$\overline{w}(q^*)\!=\!\langle q^*\rangle^{1\shortminus \delta}\langle q_+^*\rangle^{\!\gamma}\!$ we have
\begin{equation}\label{eq:transversaldersimplifiedLie}
(1+t+|x|) \,|(\wpa \widehat{\mathcal L}_{\widetilde{Z}}^I
\wt{h}^1)_{\mathcal{U}\mathcal{V}}(t,x)\, \overline{w}(q^*)| \les\varepsilon +\!\int_{|q^*|/4}^t \!\!\!\! (1+\tau)\|
\,\big(\Lie_{{\widetilde{Z}}}^I ( \kappa\widetilde{\Box} \widetilde{h}^1)\big)_{\mathcal{U}\mathcal{V}}(\tau,\cdot)\, \overline{w}\|_{L^\infty(D_\tau)} d\tau,
\quad|I|\!\leq\! N\!-7.
\end{equation}
\end{lemma}
\begin{proof} By Lemma \ref{lem:transversalder}
\begin{multline}\label{eq:transversalder2}
(1+t+|x|) \,|(\wpa \widehat{\mathcal L}_{\widetilde{Z}}^I
\wt{h}^1)_{\mathcal{U}\mathcal{V}}(t,x)\, \overline{w}(q^*)| \les\!\sup_{|q^*|/4\leq
\tau\leq t}
{\sum}_{|J|\leq |I|+1}\|\,\widehat{\mathcal L}_{\widetilde{Z}}^J
\wt{h}^1(\tau,\cdot)\, \overline{w}\|_{L^\infty}\\
+ \int_{|q^*|/4}^t\Big( (1+\tau)\|
\,(\Box^* \widehat{\mathcal L}_{\widetilde{Z}}^I \wt{h}^1)_{\mathcal{U}\mathcal{V}}(\tau,\cdot)\, \overline{w}\|_{L^\infty(D_\tau)} +{\sum}_{|J|\leq |I|+2} (1+\tau)^{-1}
\| \widehat{\mathcal L}_{\widetilde{Z}}^J
\wt{h}^1(\tau,\cdot)\, \overline{w}\|_{L^\infty(D_\tau)}\Big)\, d\tau.
\end{multline}
Using the estimate \eqref{eq:h1improcedhormanderweakbound}:
\beq
|\widehat{\mathcal L}_{\widetilde{Z}}^J
\wt{h}^1(t,x)|\lesssim \varepsilon (1+t)^{\delta}  (1+t+|q^*|)^{-1}
 \langle q_+^*\rangle^{-\gamma},\qquad |J|\leq N-5 {}\tag*{\qedhere}
\eq
it follows that
\begin{equation}\label{eq:transversaldersimplifiedLie0}
(1\!+t\!+\!|x|) \,|(\wpa \widehat{\mathcal L}_{\widetilde{Z}}^I
\wt{h}^1)_{\mathcal{U}\mathcal{V}}(t,x)\, \overline{w}(q^*)| \les\varepsilon + \!\!\int_{|q^*|/4}^t \!\!\!\!\! \!\!(1\!+\!\tau)\|
\big( \kappa\widetilde{\Box}\Lie_{\widetilde{Z}}^I  \widetilde{h}^1\big)_{\mathcal{U}\mathcal{V}}(\tau,\cdot)\, \overline{w}\|_{L^\infty(D_\tau)} d\tau,
\qquad|I|\!\leq \!N\!\!-\!7.
\end{equation}
By Lemma \ref{lem:TheCommutatorLemma}
\beq
\bigtwo|\Box^* \widehat{\mathcal{L}}_{\widetilde{Z}}^I
\widetilde{h}_{cd}^1
	-\Lie_{\widetilde{Z}}^I \big( \kappa\widetilde{\Box} \widetilde{h}^1_{cd}\big)\bigtwo|
\lesssim
\sum_{|J|+|K|\leq |I|+1,\, |J|\leq |I|}\Big(\frac{M\ln{\langle t+r\rangle}}{\langle t+r\rangle^2}
+\frac{\big|\big(\widehat{\mathcal{L}}_{\widetilde{X}}^J
\widetilde{H}_1\big)_{\widetilde{L}\widetilde{L}}
\big|}{\langle t-r^*\rangle}
+\frac{|\widehat{\mathcal{L}}_{\widetilde{X}}^J\widetilde{H}_1|}
{\langle t+r\rangle}\Big)
\big|\wpa
\widehat{\Lie}^{K}_{\widetilde{X}}\widetilde{h}^{1}\big|.
\eq
Using the estimates \eqref{eq:h1improcedhormanderweakbound}, \eqref{eq:h1improcedhormanderweakboundder} and
\eqref{eq:wavecoordinatefunctionLie} we obtain
\beq
\bigtwo|\Box^* \!\widehat{\mathcal{L}}_{\widetilde{Z}}^I
\widetilde{h}_{cd}^1
	-\Lie_{\widetilde{Z}}^I \big( \kappa\widetilde{\Box} \widetilde{h}^1_{cd}\big)\bigtwo|
\!\lesssim\!
\Big(\!\frac{M\!\ln{\tplusr}\!\!}{\tplusr^2}
+\frac{\!\varepsilon(1\!+\!t\!+\!r)^{\shortminus 1\shortminus \gamma+\delta}\!\!\! \!\!\!}{\langle q^*\rangle^{1-\gamma}\langle q_+\rangle^{\gamma}}
+\frac{\varepsilon (1+t)^\delta}{\!\!\tplusr^2 \langle q_+^*\rangle^{\gamma}}\!\Big)
\frac{\varepsilon (1\!+\!t)^\delta \langle q^*\rangle^{\shortminus 1}\!\!\!\!\!\!}
{\!\!\!\!\tplusr^2 \langle q_+^*\rangle^{\gamma}\!\!\!\!}
\lesssim\! \frac{\varepsilon\langle q_+^*\rangle^{\shortminus 1\shortminus \gamma +\delta}\!\!\!\!}{\langle t\!+\!r^*\rangle^3}.
\eq
Hence
\beq
\int_{|q^*|/4}^t (1+\tau)\|\Box^* \widehat{\mathcal{L}}_{\widetilde{Z}}^I
\widetilde{h}_{cd}^1
	-\Lie_{\widetilde{Z}}^I \big( \kappa\widetilde{\Box} \widetilde{h}^1_{cd}\big)(\tau,\cdot)
\overline{w}\|_{L^\infty(D_\tau)} d\tau\lesssim \varepsilon.
\eq
The lemma follows from this and \eqref{eq:transversaldersimplifiedLie0}.
\end{proof}
Also using Lemma \ref{lem:errorlemma} we get
\begin{lemma}\label{lem:transversaldersimplifiedLiecommutedF}
Let $D_t\!=\!\{(t,x);\, \big|t\!-\!|x|\big|\!\leq\! c_0 t \}$, for some constant
$0\!<\!c_0\!<\!1$, and
$\overline{w}(q^*)\!=\!\langle q^*\rangle^{1-\delta}\langle q_+^*\rangle^{\!\gamma}\!$. Then for $|I|\!\leq\! N\!-7$
\begin{equation}\label{eq:transversaldersimplifiedLie}
(1+t+|x|) \,|\big(\wpa \widehat{\mathcal L}_{\widetilde{Z}}^I\widetilde{h}^{\!1}\big)_{\mathcal{U}\mathcal{V}}(t,x)\, \overline{w}(q^*)|
\les\varepsilon +\! \int_{|q^*|/4}^t \!\!\!\!\!\!\!\!\!
(1+\tau)\|
\big(\Lie_{\widetilde{Z}}^I \big( \kappa^2(
\widetilde{F}(\wpa \widetilde{h}^{\!1}\!\!,\wpa
\widetilde{h}^{\!1})+\widetilde{\widecheck{T}}\,)\big)\big)_{UV}(\tau,\cdot)\,
\overline{w}\|_{L^{\!\infty}(\!D_\tau\!)} d\tau.
\end{equation}
\end{lemma}
\begin{proof} Using \eqref{eq:newcoordeinsteinwithRerrorshigherorder} and Lemma \ref{lem:transversaldersimplifiedLiecommuted} it suffices to show
\begin{equation}
\int_{|q^*|/4}^t \!\!\!\!\!\!\!\!\!
(1+\tau)\|
\big(\big(|\widetilde{R}_{cd}^{\,mass\, I}|
+|\widetilde{R}_{cd}^{\,cov\, I}| + |\widetilde{R}_{cd, 0}^{\,cube\, I}|\Big)(\tau,\cdot)\,
\overline{w}\|_{L^{\!\infty}(\!D_\tau\!)} d\tau \lesssim \varepsilon.
\end{equation}
This follows from \eqref{eq:h1improcedhormanderweakbound}, \eqref{eq:h1improcedhormanderweakboundder}, and Lemmas \ref{lem:differentiatedinhomogen} and \ref{lem:errorlemma}, which imply the pointwise bound
\beq
|\widetilde{R}_{cd}^{\,mass\, I}|
+|\widetilde{R}_{cd}^{\,cov\, I}| + |\widetilde{R}_{cd, 0}^{\,cube\, I}|\lesssim \varepsilon(\langle t+r^*\rangle^{-3+3\delta}\langle q^*\rangle^{-1} + \langle t+r^*\rangle^{-3}\langle q_+^*\rangle^{-1}).
\eq
The result follows from direct integration.
\end{proof}
\begin{proposition}
Let $D_t=\{(t,x);\, \big|t\!-\!|x|\big|\!\leq\! c_0 t \}$ for some constant
$0\!<\!c_0\!<\!1$ and
$\overline{w}(q^*)=\langle q^*\rangle^{1-\delta}\langle q_+^*\rangle^{\!\gamma}\!$.
Let $|I|\!\leq\! N-7$. Then
if $\mathcal{T} = \{\widetilde{L},\mathcal{S}_1,\mathcal{S}_2\}$ we have
\begin{equation}\label{eq:transversaldersimplifiedLieTangential}
(1+t+|x|) \,|\big(\wpa \widehat{\mathcal L}_{\widetilde{Z}}^I\widetilde{h}^{\!1}\big)_{\mathcal{T}\mathcal{U}}(t,x)\, \overline{w}(q^*)|
\les\varepsilon + {\sum}_{|J|\leq |I|}\int_{|q^*|/4}^t \!\!\!
(1+\tau)
\,\|\big(\Lie_{\widetilde{Z}}^J \widetilde{\widecheck{T}}\big)_{\mathcal{T}\mathcal{U}}(\tau,\cdot)\,
\overline{w}\|_{L^{\!\infty}(\!D_\tau\!)} d\tau.
\end{equation}
On the other hand
\begin{multline}\label{eq:transversaldersimplifiedLieLtilde}
(1+t+|x|) \,|\wpa \widehat{\mathcal L}_{\widetilde{Z}}^I\widetilde{h}^{\!1}(t,x)\, \overline{w}(q^*)|
\les\varepsilon + {\sum}_{|J|\leq |I|}\int_{|q^*|/4}^t \!\!\!
(1+\tau)
\,\|\Lie_{\widetilde{Z}}^J \widetilde{\widecheck{T}}(\tau,\cdot)\,
\overline{w}\|_{L^{\!\infty}(\!D_\tau\!)}d\tau\\
+{\sum}_{J+K=I}\int_{|q^*|/4}^t \!\!\!
(1+\tau)\||\wpa\, \wh^{1J\!}|_{\widetilde{\mathcal{T}}\widetilde{\mathcal{T}}}|\wpa\,
 \widetilde{h}^{1K\!}|_{\widetilde{\mathcal{T}}\widetilde{\mathcal{T}}}
 (\tau,\cdot)\,
\overline{w}\|_{L^{\!\infty}(\!D_\tau\!)} d\tau.
\end{multline}
\end{proposition}
\begin{proof} By Lemma \ref{lem:differentiatedinhomogen}, Lemma \ref{lem:PPS} and
\eqref{eq:h1improcedhormanderweakbound}-\eqref{eq:h1improcedhormanderweakboundder}
we have
\begin{multline}
\bigtwo| \Lie_{\widetilde{Z}}^I
\bigtwo(\!\kappa^{2\!} \widetilde{F}_{ab}(\tg)
[\wpa \widetilde{h}^1\!\!,\wpa\widetilde{h}^1]\bigtwo)
-\!\!\!\!\sum_{J+K=I}\!\!\!\!\wL_{a}\wL_b
\widehat{P}(\pa_{q^*} \wh^{1J}\!\!,\pa_{q^*} \widetilde{h}^{1K})\bigtwo| \\
\lesssim \!\!\!\!\sum_{J+K=I}\!\!\!\! |\overline{\wpa}\, \wh^{1J}| \,|\wpa \widetilde{h}^{1K}|
+\frac{M\ln{\tplusr}}{\langle t+r\rangle}
\!\!\!\!\sum_{J+K=I}\!\!\!\!|\wpa \widetilde{h}^{1J\!}|\,|\wpa \widetilde{h}^{1K\!}|
+\!\!\!\!\!\!\!\!\!\!\!\!\!\sum_{\,\,\,\,\,\,\,\,\,\,\,|J|+|K|+|L| \leq |I|\!\!\!\!\!\!\!\!\!\!}
\!\!\!\!\!\!\!\!\!\!\!
|\widehat{\Lie}^{L}_{\widetilde{Z}} \widetilde{H}_1| \,
|\wpa \widetilde{h}^{1J\!}|\,|\wpa\widetilde{h}^{1K\!}|\\
\lesssim  \varepsilon^2   (1\!+t\!+|q^*|)^{-3+2\delta}
\langle q^*\rangle^{-1}\langle q_+^*\rangle^{-\gamma}
+\varepsilon^2   (1\!+t\!+|q^*|)^{-3+3\delta}
\langle q^*\rangle^{-2}\langle q_+^*\rangle^{-\gamma}.
\end{multline}
It follows that
\beq
\int_{|q^*|/4}^t \!\!\!\!\!\!\!\!\!
(1+\tau)\bigtwo|\bigtwo| \bigtwo(\Lie_{\widetilde{Z}}^I
\bigtwo(\!\kappa^{2\!} \widetilde{F}_{ab}(\tg)
[\wpa \widetilde{h}^1\!\!,\wpa\widetilde{h}^1]\bigtwo)
-\!\!\!\!\sum_{J+K=I}\!\!\!\!\wL_{a}\wL_b
\widehat{P}(\pa_{q^*} \wh^{1J}\!\!,\pa_{q^*} \widetilde{h}^{1K})\bigtwo) (\tau,\cdot)\,
\overline{w}\bigtwo|\bigtwo|_{L^{\!\infty}(\!D_\tau\!)} d\tau\lesssim \varepsilon.
\eq
The first estimate follows directly from this using Lemma
\ref{lem:transversaldersimplifiedLiecommutedF}.
To prove the second estimate we use Lemma \ref{lem:PPS} and \eqref{eq:wavecoordinatederivativeLie}
and \eqref{eq:h1improcedhormanderweakboundder}:
\begin{multline}
\big| \widehat{P}(\pa_{q^*} \wh^{1J}\!\!,\pa_{q^*} \widetilde{h}^{1K})\big|\\
\les \big(|\pa_{q^*}\wh^{1J\!}|_{\wL\widetilde{\mathcal
T}}+|\pa_{q^*}\trs \wh^{1J\!}|\big)|\wpa\,\widetilde{h}^{1K\!}| +|\wpa\,\wh^{1J\!}
|\big(|\pa_{q^*}\widetilde{h}^{1K\!}|_{\wL\widetilde{\mathcal
T}}+|\pa_{q^*}\trs \widetilde{h}^{1K\!}|\big)
+|\wpa\, \wh^{1J\!}|_{\widetilde{\mathcal{T}}\widetilde{\mathcal{T}}}|\wpa\,
 \widetilde{h}^{1K\!}|_{\widetilde{\mathcal{T}}\widetilde{\mathcal{T}}}\\
 \lesssim  \varepsilon^2   (1\!+t\!+|q^*|)^{-3+3\delta}
\langle q^*\rangle^{-1-\delta}\langle q_+^*\rangle^{-2\gamma}
+|\wpa\, \wh^{1J\!}|_{\widetilde{\mathcal{T}}\widetilde{\mathcal{T}}}|\wpa\,
 \widetilde{h}^{1K\!}|_{\widetilde{\mathcal{T}}\widetilde{\mathcal{T}}},
\end{multline}
from which the second estimate follows.
\end{proof}

\begin{proposition}\label{prop:sharpfirstorderwave} For $|I|\leq N-7$ we have
\beq\label{eq:sharphTUder}
(1+t+|x|) \,|(\wpa \widehat{\mathcal L}_{\widetilde{Z}}^I
\wt{h}^1)_{\mathcal{T}\mathcal{U}}(t,x)\, \overline{w}(q^*)|\les \varepsilon .
\eq
Moreover,
\beq\label{eq:sharphUVder}
(1+t+|x|) \,|(\wpa \widehat{\mathcal L}_{\widetilde{Z}}^I
\wt{h}^1)_{\mathcal{U}\mathcal{V}}(t,x)\, \overline{w}(q^*)|
\les \varepsilon \Big(1+ \ln\frac{t\!+\!r^*}{\langle t\!-\!r^*\rangle}\Big).
\eq
\end{proposition}
\begin{proof} By Proposition \ref{prop:energymomentumLinfty} we have
\begin{multline}
\int_{|q^*|/4}^t \!\!\!\!\!\!\!\!\!
\langle\tau\rangle
\,\|\big(\Lie_{\widetilde{Z}}^J \widetilde{\widecheck{T}}\big)_{TU}(\tau,\cdot)\,
\overline{w}\|_{L^{\!\infty}(\!D_\tau\!)} d\tau\!
\les \epsilon\!\!\int_{|q^*|/4}^t \!\!\!\!\!\!\!\!\!\langle\tau\rangle^{\!-1-s}
\!\!\!\sup_{|q|\leq c_0 \tau}\!\langle q\rangle^{\!-2+s}\langle q_-\rangle^{\!1-2s}
\langle q\rangle^{1-\delta}\langle q_+\rangle^{\!\gamma} d\tau\!
\les \epsilon\!\!\int_{|q^*|/4}^t \!\!\!\!\!\!\!\!\!\langle\tau\rangle^{\!\gamma-\delta-2}
 d\tau\!\\
 \les\epsilon\langle q^*\rangle^{\!\gamma-\delta-1}\!\!,
\end{multline}
which together with \eqref{eq:transversaldersimplifiedLieTangential}
implies \eqref{eq:sharphTUder}, and
\begin{equation}\int_{|q^*|/4}^t \!\!\!\!\!\!\!\!\!
\langle\tau\rangle
\,\|\Lie_{\widetilde{Z}}^J \widetilde{\widecheck{T}}(\tau,\cdot)\,
\overline{w}\|_{L^{\!\infty}(\!D_\tau\!)}d\tau
\lesssim  \epsilon\!\!\int_{|q^*|/4}^t \!\!\!\!\!\!\!\!\!\langle\tau\rangle^{\!-1}
\!\!\!\sup_{|q|\leq c_0 \tau}\!\langle q\rangle^{\!-2}\langle q_-\rangle^{1-2s}
\langle q\rangle^{1-\delta}\langle q_+\rangle^{\gamma} d\tau\!
\les \epsilon\!\!\int_{|q^*|/4}^t \!\!\!\!\!\!\!\!\!\langle\tau\rangle^{\!-1}
 d\tau\!\les\epsilon\ln\bigthree(\!\frac{\langle t\rangle}{\langle q^{*\!}\rangle}\!\bigthree).
\end{equation}
Moreover using \eqref{eq:sharphTUder} we get
\begin{equation}
\int_{|q^*|/4}^t \!\!\!\!\!\!\!\!\!
\langle\tau\rangle\||\wpa\, \wh^{1J\!}|_{\widetilde{\mathcal{T}}\widetilde{\mathcal{T}}}|\wpa\,
 \widetilde{h}^{1K\!}|_{\widetilde{\mathcal{T}}\widetilde{\mathcal{T}}}(\tau,\cdot)\,
\overline{w}\|_{L^{\!\infty}(\!D_\tau\!)} d\tau
\lesssim \epsilon\!\!\int_{|q^*|/4}^t \!\!\!\!\!\!\!\!\!\langle\tau\rangle^{\!-1}
\!\!\sup_{|q|\leq c_0 \tau}\!\!
\langle q\rangle^{\delta-1}\langle q_+\rangle^{-\gamma}
 d\tau\!\les\epsilon\ln\bigthree(\!\frac{\langle t\rangle}{\langle q^{*\!}\rangle}\!\bigthree).
\end{equation}
\eqref{eq:sharphUVder} follows from using the last two estimates in
\eqref{eq:transversaldersimplifiedLieLtilde}.
\end{proof}

\section{The Sharp Energy bounds of the metric assuming the decay estimates}\label{sec:ProofMetricEnergy}
In the previous section we got improved decay estimates for certain components of the metric.
With the energy estimate we  can not differentiate between
different components at the highest order. However in the proof we need to use
that we have better decay estimates for tangential components.
In particular, the assumptions of the basic energy estimate in Theorem \ref{thm:EEst} hold.
We will now apply this estimate to the differentiated equations \eqref{eq:newcoordeinsteinwithRerrorshigherorder}:
\beq
\widetilde{\Box}\, \widetilde{h}^{1I}_{cd}
+\widetilde{R}_{cd}^{\,com\, I}\!
=\!\!\!{\sum}_{I^\prime+I^{\prime\prime}=I} \kappa_{I^{\prime\prime}} \Big({\sum}_{J+K=I^\prime} \widetilde{F}_{cd}(\tg)[\wpa \widetilde{h}^{1J}\!\!,\wpa\widetilde{h}^{1K}]+\widetilde{R}_{cd,\, 0}^{\,cube\, I^\prime}\!\!
+\widetilde{\widecheck{T}}_{\!cd}{\!\!}^{\!I^\prime}
+\widetilde{R}_{cd}^{\,mass\,I^\prime}\!\!+\widetilde{R}_{cd}^{\,cov\,I^\prime}\Big),
\eq
where
\beq
\kappa^{2\!} \widetilde{F}_{\!cd}(\tg)
[\wpa \widetilde{h}^{1J}\!\!,\wpa\widetilde{h}^{1K}]
=\!\wL_{c}\wL_d
\widehat{P}(\pa_{q^*} \wh^{1J}\!\!,\pa_{q^*} \widetilde{h}^{1K})
+\kappa^2\widetilde{R}_{cd}^{\,tan\, JK}
+\kappa^2\widetilde{R}_{cd,\, 1}^{\,cube\, JK}\! .
\eq

We will now use the decay estimates for lower derivatives we obtained in the previous section in the estimates in Section \ref{sec:higherordereinstein} for the various remainder terms
$\widetilde{R}^{\,\bullet}$\!. We first note that since $H_1=-h^1+O(h^2)$ using the decay estimates in the previous section we see that higher norms of $h^1$ and $H_1$ are equivalent, with constants depending on the lower norms, and below we will use this frequently without mentioning it.

\subsubsection{The main quadratic semilinear error terms}
By Lemma \ref{lem:PPS} and \eqref{eq:h1improcedhormanderweakboundder}
\beq
\bigtwo|
\widetilde{R}_{cd}^{\,tan\, JK}\!
\bigtwo|
\!\lesssim\! |\overline{\wpa}\, \wh^{1J}\!| \,|\wpa \widetilde{h}^{1K}\!|+ |\wpa\wh^{1J}\!| \,|\overline{\wpa}\, \widetilde{h}^{1K}\!|
\!\lesssim\! \frac{\varepsilon \langle q_+^*\rangle^{-\gamma}}{\tplusr^{1-\delta} \langle q^*\!\rangle}{\sum}_{|J|\leq |I|}\big|\overline{\wpa}\,
\widetilde{h}^{1 J}\big|
+\frac{\varepsilon \langle q_+^*\rangle^{-\gamma}}{\tplusr^{2-\delta}}{\sum}_{|J|\leq |I|}\big|\wpa
\widetilde{h}^{1 J}\big|.
 \eq
 Moreover, for any $\delta^\prime>\delta$
 \begin{multline}
\big| \widehat{P}(\pa_{q^*} \wh^{1J}\!\!,\pa_{q^*} \widetilde{h}^{1K})\big|\\
\les \big(|\pa_{q^*}\wh^{1J\!}|_{\wL\widetilde{\mathcal
T}}+|\pa_{q^*}\trs \wh^{1J\!}|\big)|\wpa\,\widetilde{h}^{1K\!}| +|\wpa\,\wh^{1J\!}
|\big(|\pa_{q^*}\widetilde{h}^{1K\!}|_{\wL\widetilde{\mathcal
T}}+|\pa_{q^*}\trs \widetilde{h}^{1K\!}|\big)
+|\wpa\, \wh^{1J\!}|_{\widetilde{\mathcal{T}}\widetilde{\mathcal{T}}}|\wpa\,
 \widetilde{h}^{1K\!}|_{\widetilde{\mathcal{T}}\widetilde{\mathcal{T}}}\\
 \les \frac{\varepsilon \langle q_+^*\rangle^{-\gamma}}{\tplusr^{1-\delta^\prime}\langle q^*\rangle} \big(|\pa_{q^*}\wh^{1J\!}|_{\wL\widetilde{\mathcal
T}}+|\pa_{q^*}\trs \wh^{1J\!}|\big)
+\frac{\varepsilon \langle q_+^*\rangle^{-\gamma}}{\tplusr\langle q^*\rangle^{1-\delta}}{\sum}_{|J|\leq |I|}\big|\wpa
\widetilde{h}^{1 J}\big|.
\end{multline}

\subsubsection{The quasilinear commutator error terms}
By Lemma \ref{lem:TheCommutatorLemma} and the estimates in the previous section  and using that $M\les \varepsilon$ we have
\beq
\bigtwo| \widetilde{R}_{cd}^{\,com\, I}\bigtwo|
\lesssim
\frac{\varepsilon\langle q_+^*\rangle^{-\gamma}}{\langle t\!-\!r^*\rangle\tplusr^{1-\delta}}{\sum}_{ |J|\leq |I|}\Big(\frac{\big|
{\widetilde{H}^J}_{1\, \widetilde{L}\widetilde{L}}
\big|}{\langle t\!-\!r^*\rangle}
+\frac{|{\widetilde{H}_1}^J|}{\tplusr}\Big)
+\frac{\varepsilon}{\tplusr^{1+\gamma-\delta}}{\sum}_{|J|\leq |I|}\big|\wpa
\widetilde{h}^{1 J}\big|.
\eq
The commutator terms with the Minkowski vector fields would have been highest order error terms but because we use modified coordinates and vector fields these are lower order.

\subsubsection{The lower order error terms}
 There is much more room in the mass, covariant and cubic error terms.
By Lemma \ref{lem:errorlemma}
\begin{equation}
|\widetilde{R}_{cd}^{\,mass\, I}|
+|\widetilde{R}_{cd}^{\,cov\, I}|
\lesssim \frac{\!M H(r \!<\! t/2)\!}{\tplusr^3}+\frac{M^2 }{\!\tplusr^4\!}
+\frac{\!M\big(\ln{\tplusr}\!+\!1\big)\!}{\tplusr^3}
\!\!\sum_{|J|\leq |I|}\!\!\!\big|\widetilde{H}_1^J\big|
+\frac{\!M\big(\ln{\tplusr}\!+\!1\big)\!}{\tplusr^2}
\!\!\sum_{|J|\leq |I|}\!\!\!\big|\wpa\widetilde{h}_1^J\big|.
\end{equation}
By Lemma \ref{lem:PPS} and the estimates in the previous section
where
\beq
\bigtwo|
\widetilde{R}_{cd,\, 1}^{\,cube\, JK}
\bigtwo|
\lesssim\Big(|\widetilde{H}_1|\!+\!\frac{M\ln{\tplusr}}{\tplusr}\Big)|\wpa \widetilde{h}^{1J}|\,|\wpa \widetilde{h}^{1K}|
\les \frac{\varepsilon^2 \langle q_+^*\rangle^{-\gamma}}{\tplusr^{2-\delta^\prime}\langle q^*\rangle}{\sum}_{|J|\leq |I|}\big|\wpa
\widetilde{h}^{1 J}\big|
 \eq
 and by Lemma \ref{lem:differentiatedinhomogen}
\begin{equation}
\big|\widetilde{R}_{cd,\, 0}^{\,cube\, I}\big|\les\!\!\!\!\!\!\!\!\!\!\!\!\!\!\!
\sum_{\,\,\,\,\,\,\,\,\,\,\,|J|+|K|+|L| \leq |I|,\,\, |L|\geq 1\!\!\!\!\!\!\!\!\!\!}\!\!\!\!\!\!
\!\!\!\!\!\!\!\!\!\!\!
| \widetilde{H}_1^L| \,
|\wpa \widetilde{h}^{1J\!}|\,|\wpa\widetilde{h}^{1K\!}|
\les \frac{\varepsilon^2 \langle q_+^*\rangle^{-\gamma}}{\tplusr^{2-\delta^\prime}\langle q^*\rangle}\!\!\sum_{|J|\leq |I|-1}\!\!\!\!\!\big|\wpa
\widetilde{h}^{1 J}\big|
+\frac{\varepsilon^2 \langle q_+^*\rangle^{-2\gamma}}{\tplusr^{2-\delta^\prime}\langle q^*\rangle^2}\!\!\sum_{|J|\leq |I|}\!\!\big|
\widetilde{H}_1^{1 J}\big|.
\end{equation}

\subsection{The energy estimates}
With the weighted energies ,
\begin{equation} \label{eq:energydefEsec13}
	E_N(T)=\!\!\!\!\sup_{0\leq t\leq T}{\sum}_{\!\vert I \vert \leq N\,\,}
	\int
	|\wpa \widetilde{Z}^I\! h^1 (t,x)|^2 w\, dx, \quad
	\text{where}
	\quad
	w(t,x)\!
	=\!
	\bigg\{\begin{aligned}
		&(1+|r^*\!\!-t|)^{1+2\gamma}\!\!,\quad
		&
		r^*\!>t,
		\\
		&1\!+(1\!+|r^*\!\!-t|)^{-2\mu}\!,
		\quad
		&
		r^*\!\leq t.
	\end{aligned}
\end{equation}
and
\begin{equation} \label{eq:energydefSsec13}
	S_N(T)
	\!=
	\!\!{\sum}_{\!\vert I \vert \leq N\,\,}\!\!\int_0^T\!\!\!\int 	
	|\wpao
 \widetilde{Z}^I\! h^{\!1} (t,x)|^2 w'dx dt,
	\quad
	\text{where}
	\quad
w^\prime(t,x)\!=\!\bigg\{\begin{aligned} &(1\!+2\gamma)(1\!+|r^*\!\!-t|)^{2\gamma}\!\!,\quad &r^*\!>t,\\
&2\mu (1\!+|r^*\!\!-t|)^{-1-2\mu}\!\!,\quad &r^*\!\leq t,\end{aligned}
\end{equation}
we have by Theorem \ref{thm:EEst}
\beq
E_N(T) + {S}_N(T)\leq 8 E_N(T)(0) + 16 {\sum}_{|I|\leq N} \int_0^T\Big(\int_{\Sigma_t}
|\widetilde{\Box}\widetilde{h}^{1I}|^2 w\, dx\Big)^{1/2}  \, E_N(t)^{1/2} dt.
\eq
Summing up we can identify essentially five types of error terms that we will estimate
\beq
|\widetilde{\Box}\widetilde{h}^{1I}|\leq
R_{lead}+R_{low}+R_{tan}+R_{coord}+R_{mass}+|\widecheck{T}^I|.
\eq
Here the leading order terms are bounded by
\beq
R_{lead}=\frac{C\varepsilon}{1\!+t} {\sum}_{|J|\leq |I|} |\wpa h^{1 J}|,
\eq
which can be estimated by $E_N(t)$ but are the worst that will lead to exponential growth like
$t^{\, C\varepsilon}$. The lower terms $R_{low}=R_{low,1}+R_{low,0}$ are for some fixed  $0<\gamma^{\,\prime}<\gamma$
\beq
R_{low,1}=\frac{C\varepsilon}{(1\!+t)^{1+\gamma^{\,\prime}} }{\sum}_{|J|\leq |I|} |\wpa h^{1 J}|, \qquad\text{and}\qquad
R_{low,0}=\frac{C\varepsilon}{(1\!+t)^{1+\gamma^{\,\prime}} }{\sum}_{|J|\leq |I|}\frac{|h^{1 J}|}{\langle q^*\rangle}
\eq
where in the norms with the weights we are considering the second sum is bounded by the first using Hardy's inequality, Corollary \ref{cor:poincare}.
The tangential terms can be bounded by
\beq
R_{tan}=\frac{C\varepsilon\langle q_+^*\rangle^{-\gamma}}{(1\!+t)^{1/2+\nu} \langle q^*\rangle} {\sum}_{|J|\leq |I|} |\overline{\wpa} h^{1 J}|,
\eq
for any fixed $\nu$ such that $0<\nu<1/2$, which
can be estimated in terms of $S_N$. The wave coordinate terms are bounded by
$R_{coord}=R_{coord,1}+R_{coord,0}$
\beq
R_{coord,1}=\frac{C\varepsilon\langle q_+^*\rangle^{-\gamma}}{(1\!+t)^{1/2+\nu} \langle q^*\rangle}{\sum}_{|J|\leq |I|} \Big(|\pa_{q^*}\wh^{1J\!}|_{\wL\widetilde{\mathcal
T}}+|\pa_{q^*}\trs \wh^{1J\!}|\Big),
\eq
and
\beq
R_{coord,0}=\frac{C\varepsilon\langle q_+^*\rangle^{-\gamma}}{(1\!+t)^{1/2+\nu} \langle q^*\rangle}{\sum}_{|J|\leq |I|} \Big(	\frac{| \widehat{\mathcal L}_Z^I{H}_1|_{L\mathcal T}}{\langle q^*\rangle}
+\frac{|  \trs\widehat{\mathcal L}_Z^I {H}_1 |}{\langle q^*\rangle}\Big),
\eq
can be estimated in terms of the tangential terms using the wave coordinate condition as in Lemma \ref{lem:wavecoordL2est} and Hardy's inequality.
For this we also need to pick $\nu$ and $\nu^{\,\prime}>0$ so that  $1-2(\nu-\nu^{\,\prime})<\gamma$.
Finally, the mass terms are bounded by
\beq
R_{mass}= \frac{CM }{\!\tplusr^3\langle q_+^*\rangle},
\eq
which can be bounded directly in terms of $M$ that we assume is small.

\subsubsection{The leading and lower order terms} We have
\beq
\int_0^T\Big(\int_{\Sigma_t}
R_{lead}^{\,2}w\, dx\Big)^{1/2}  \, E_N(t)^{1/2} dt
\leq\int_0^T\frac{C\varepsilon}{1\!+t}  E_N(t) dt .
\eq
and using Corollary \ref{cor:poincare} to estimate $R_{low,0}$ in terms of $R_{low,1}$ we have
\beq
 \int_0^T\Big(\int_{\Sigma_t}
R_{low}^{\,2}w\, dx\Big)^{1/2}  \, E_N(t)^{1/2} dt
\leq\int_0^T\frac{C\varepsilon}{(1\!+t)^{1+\gamma^{\,\prime}}}  E_N(t) dt .
\eq

\subsubsection{The tangential and coordinate terms} By Cauchy-Schwartz in time
\beq
 \int_0^T\!\!\!\Big(\!\int_{\Sigma_t}\!
R_{tan}^{\,2}w\, dx\Big)^{\!1/2}   E_N(t)^{1/2} dt
\leq\Big( \int_0^T\!\!\int_{\Sigma_t}\!\!
R_{tan}^{\,2}\frac{(1\!+\!t)^{1+2\nu}\!\!\!\!\!}{\varepsilon}\,\,w\, dx \,dt\Big)^{\!1/2} \Big(\int_0^T\!\!\frac{\varepsilon  E_N(t) }{(1\!+\!t)^{1+2\nu}\!\!}  \,\, dt \Big)^{\!1/2}\!\!.
\eq
Here
\beq
 \int_0^T\!\!\int_{\Sigma_t}\!
R_{tan}^{\,2}\frac{(1\!+t)^{1+2\nu}\!\!\!\!\!}{\varepsilon}\,\,w\, dx \,dt
\les {\sum}_{|J|\leq |I|} C\varepsilon\int_0^T\!\!\int_{\Sigma_t}\!\frac{ |\overline{\wpa} h^{1 J}|^2}{ \langle q^*\rangle^{2}} \,\,w\, dx \,dt
\les  C\varepsilon S_N(T).
\eq
Hence, using that $2ab\leq a^2+b^2$,
\beq
 \int_0^T\!\!\Big(\int_{\Sigma_t}\!
R_{tan}^{\,2}w\, dx\Big)^{\!1/2}   E_N(t)^{1/2} dt
\leq C^\prime\varepsilon S_N(T)+C^\prime \int_0^T\!\!\frac{\varepsilon }{(1\!+t)^{1+2\nu}\!\!}  \,\, E_N(t)\,dt.
\eq
Since
\beq
 \int_0^T\!\!\int_{\Sigma_t}\!
R_{coord,1}^{\,2}\frac{(1\!+t)^{1+2\nu}\!\!\!\!\!}{\varepsilon}\,\,w\, dx \,dt
\les {\sum}_{|J|\leq |I|} C\varepsilon\int_0^T\!\!\int_{\Sigma_t}\!\frac{ |\pa_{q^*}\wh^{1J}|_{\wL\widetilde{\mathcal
T}}^{\,2}+|\pa_{q^*}\trs \wh^{1J}|^{\,2}}{ \langle q^*\rangle^{2}} \,\,w\, dx \,dt,
\eq
by Lemma \ref{lem:wavecoordL2est0} $R_{coord,1}$ satisfy the same bound
as $R_{tan}$.
By Lemma \ref{lem:wavecoordL2est} we have
\beq
 \int_0^T\!\!\int_{\Sigma_t}\!
R_{coord,0}^{\,2}\frac{(1\!+t)^{1+2\nu^{\,\prime}}\!\!\!\!\!}{\varepsilon}\,\,w\, dx \,dt
\les  C\varepsilon S_N(T)+\int_0^T\!\!\frac{\varepsilon }{(1\!+t)^{1+2\nu^{\,\prime}}\!\!}  \,\, E_N(t)\,dt,
\eq
since we picked $\nu$ and $\nu^{\,\prime}$ so that in Lemma \ref{lem:wavecoordL2est}
$b=1-2(\nu-\nu^{\,\prime})<\gamma$.
Hence, we also have
\beq
 \int_0^T\!\!\Big(\int_{\Sigma_t}\!
R_{coord}^{\,2}w\, dx\Big)^{\!1/2}   E_N(t)^{1/2} dt
\leq C^\prime\varepsilon S_N(T)+C^\prime \int_0^T\!\!\frac{\varepsilon }{(1\!+t)^{1+2\nu^{\,\prime}}\!\!}  \,\, E_N(t)\,dt.
\eq

\subsubsection{The mass terms} We have if $\gamma<1$
\beq
\int_{\Sigma_t}R_{mass}^{\,2}w\, dx=C_N M^2\int_0^\infty \frac{\langle q_+^*\rangle^{1+2\gamma} r^2 dr}{\!\tplusr^6\langle q_+^*\rangle^2}
\leq \frac{C_N M^2}{(t\!+1)^{4-2\gamma}}.
\eq
Hence
\beq
 \int_0^T\!\!\Big(\int_{\Sigma_t}\!
R_{mass}^{\,2}w\, dx\Big)^{\!1/2}   E_N(t)^{1/2} dt
\leq  \int_0^T\!\!\frac{C_N M }{(1\!+t)^{2-\gamma}\!\!}  \,\, E_N(t)^{1/2}\,dt.
\eq

\subsubsection{The fields terms} By Proposition \ref{prop:energymomentumL2}
\beq
 \int_0^T\!\!\Big(\int_{\Sigma_t}\!
|\widecheck{T}^I|^2w\, dx\Big)^{\!1/2}   E_N(t)^{1/2} dt
\leq \!\int_0^T\!\!\frac{C_N\varepsilon^2 }{1\!+t\!\!}  \,\, E_N(t)^{1/2}\,dt.
\eq

\subsubsection{The final energy estimate} Also using that $E_N(0)\leq \varepsilon^2$ we obtain
\beq
E_N(T) + {S}_N(T)\leq 8\varepsilon^2 +  C\varepsilon S_N(T) +\int_0^T\!\!\!\frac{C\varepsilon}{1\!+t}  E_N(t) dt
+ \int_0^T\!\!\frac{C_N\varepsilon^2 }{1\!+t\!\!}  \, E_N(t)^{1/2}dt
 +\int_0^T\!\!\!\!\frac{C_N M }{(1\!+\!t)^{2-\gamma}\!\!}  \,\, E_N(t)^{1/2}dt
\eq
Next we pick $\varepsilon$ so small that $C\varepsilon<1/2$ so the term with $S_N(T)$ in the right can be absorbed in the left.
Let $G(T)^2$ denote the right hand side after this term has been removed. Then
\beq
2 G(t) G^{\,\prime}(t)\leq \frac{C\varepsilon}{1+t} G(t)^2+\frac{C_N\varepsilon^2}{1+t}G(t)
+\frac{C_N M }{(1\!+t)^{2-\gamma}}G(t)
\eq
Dividing by $G(t)$ and multiplying by the integrating factor $(1+t)^{-C\varepsilon}$ gives
\beq
2\frac{d}{dt}\Big( G(t)(1+t)^{-C\varepsilon}\Big) \leq \frac{C_N\varepsilon^2}{(1+t)^{1+C\varepsilon}}
+\frac{C_N M }{(1\!+t)^{2-\gamma+C\varepsilon}}
\eq
 and integrating and dividing by the same integrating factor gives since also $M\leq \varepsilon$
\beq
G(t) \leq G(0) (1+t)^{C\varepsilon} + C_N\varepsilon (1+t)^{C\varepsilon},
\eq
Hence, for some other constant $C_N$, that
is independent of  $C_b$ and $T$ in the bootstrap assumptions \eqref{eq:bootstrapassumptionmetric}-\eqref{eq:bootstrapassumptionfields},  and constants $C_b^{\,\prime}\!<\!\infty$ and $\varepsilon_b\!>\!0$, that may depend depends on $C_b$ but not on $T$, we have
\beq\label{eq:finalbootstrapestimate}
E_N(t)\leq C_N\varepsilon^2(1+t)^{C_b^{\,\prime}\varepsilon},\qquad \text{for}\quad 0\leq t\leq T,\quad\text{and}\quad \varepsilon\!\leq \! \varepsilon_b
\eq
If we choose $C_{b}=2\,C_N$ and $\varepsilon$ so that $C_b^{\,\prime}\varepsilon<\delta$ we get back a better estimate than we started with.

\appendix

\section{The Ricci curvature in terms of generalized wave coordinates}\label{sec:GeneralizedWaveCurvature}
Here we derive the expression for Einstein's equations for the metric $g$ in terms of a quantity that is assumed to be under control by a generalized wave coordinate condition. In this paper we will only use these equations in the case this quantity vanishes. However, in Appendix \ref{sec:AppB} we show that the equations we use can alternatively be derived as expressing the metric in generalized wave coordinates.

We consider the generalized harmonic coordinate condition
\begin{subequations}
\begin{equation}
\label{GHCC}
\Gamma^\alpha = \Gamma^{\alpha}_{\beta\gamma}g^{\beta\gamma},
\end{equation}
where $\Gamma$ is a known vector.
Using this notation, the harmonic coordinate condition is equivalent to the condition $\Gamma^\alpha = 0$. We can rewrite \eqref{GHCC} as
\begin{equation}
\label{GHCC2}
g^{\alpha\beta}\partial_\alpha g_{\beta\gamma} = \Gamma_\gamma + \Gamma^{\delta}_{\gamma\delta}.
\end{equation}
\end{subequations}
The Ricci curvature tensor
\[\label{eq:RicciDef}
R_{\alpha\beta} = \partial_\gamma\Gamma^{\gamma}_{\alpha\beta} - \partial_\beta \Gamma^{\gamma}_{\alpha\gamma} + \Gamma^\gamma_{\gamma\rho}\Gamma^{\rho}_{\alpha\beta} - \Gamma^\gamma_{\alpha\rho}\Gamma^{\rho}_{\beta\gamma},
\]
satisfies the following identity
\begin{lemma}\label{lem:RicciEisntein}
\begin{equation}
R_{\alpha\beta} \!=\! -\frac12 {\Box}^{\,g} g_{\alpha\beta} +\frac12 \Gamma^\delta \partial_\delta g_{\alpha\beta} + \frac12(\partial_\alpha \Gamma_\beta + \partial_\beta \Gamma_\alpha)  -\!\frac12 \Gamma_\alpha \Gamma_\beta + \frac12\Gamma^\gamma_{\alpha\gamma}\Gamma^{\rho}_{\beta\rho}\! - \frac14 g^{\gamma\delta}g^{\rho\sigma}\partial_\alpha g_{\delta\rho}\partial_\beta g_{\sigma\gamma}\!+ Q_{\alpha\beta},\nonumber
\end{equation}
where $Q_{\alpha\beta}$ is a linear combination of classical null forms.
\end{lemma}
\begin{remark} Note that the second term on the right is exactly the difference between the geometric wave operator and the reduced wave operator. Having this term together with the other semilinear terms provide an additional cancellation, c.f. \cite{GH19}.
\end{remark}
\begin{proof}
For the proof we expand each of the four terms in \eqref{eq:RicciDef}. First,
\begin{align*}
\partial_\gamma\Gamma^{\gamma}_{\alpha\beta} &= \frac12(\partial_\gamma g^{\gamma\delta})(\partial_\alpha g_{\beta\delta} + \partial_\beta g_{\alpha\delta} - \partial_\delta g_{\alpha\beta}) + \frac12 g^{\gamma\delta}(\partial_\alpha\partial_\gamma g_{\beta\delta} + \partial_\beta\partial_\gamma g_{\alpha\delta} - \partial_\gamma\partial_\delta g_{\alpha\beta}) \\
&= \frac12(\partial_\gamma g^{\gamma\delta})(\partial_\alpha g_{\beta\delta} + \partial_\beta g_{\alpha\delta} - \partial_\delta g_{\alpha\beta}) + \frac12 (\partial_\alpha(g^{\gamma\delta}\partial_\gamma g_{\beta\delta}) + \partial_\beta(g^{\gamma\delta}\partial_\gamma g_{\alpha\delta}) - g^{\gamma\delta}\partial_\gamma\partial_\delta g_{\alpha\beta}) \\
&\quad- \frac12\partial_\alpha g^{\gamma\delta}\partial_\gamma g_{\beta\delta} - \frac12\partial_\beta g^{\gamma\delta}\partial_\gamma g_{\alpha\delta}.
\end{align*}
It follows from \eqref{GHCC2} that
\[
-\frac12\partial_\gamma g^{\gamma\delta} = \frac12 \Gamma^\delta + \frac14 g^{\rho\sigma}g^{\delta\gamma}\partial_\gamma g_{\rho\sigma},
\]
and therefore,
\[
-\frac12\partial_\gamma g^{\gamma\delta}\partial_\delta g_{\alpha\beta} = \frac12 \Gamma^\delta\partial_\delta g_{\alpha\beta} + \frac14 g^{\rho\sigma}g^{\delta\gamma}\partial_\gamma g_{\rho\sigma}\partial_\delta g_{\alpha\beta}.
\]

If we define
\[
Q^1_{\alpha\beta} =  \frac12(\partial_\gamma g^{\gamma\delta}\partial_\alpha g_{\beta\gamma} - \partial_\alpha g^{\gamma\delta}\partial_\gamma g_{\beta\gamma}) +\frac12( \partial_\gamma g^{\gamma\delta}\partial_\beta g_{\alpha\delta} - \partial_\beta g^{\gamma\delta}\partial_\gamma g_{\beta\delta}) + \frac14 g^{\rho\sigma}g^{\delta\gamma}\partial_\gamma g_{\rho\sigma}\partial_\delta g_{\alpha\beta},
\]
and apply \eqref{GHCC2} to rewrite
\begin{align*}
\frac12\partial_\alpha(g^{\gamma\delta}\partial_\gamma g_{\beta\delta}) = \frac12\partial_\alpha \Gamma_\beta + \frac14 \partial_\alpha g^{\gamma\delta}\partial_\beta g_{\gamma\delta} + \frac14g^{\gamma\delta}\partial_\alpha\partial_\beta g_{\gamma\delta}, \\
\frac12\partial_\beta(g^{\gamma\delta}\partial_\gamma g_{\alpha\delta}) = \frac12\partial_\beta \Gamma_\alpha + \frac14 \partial_\beta g^{\gamma\delta}\partial_\alpha g_{\gamma\delta} + \frac14g^{\gamma\delta}\partial_\alpha\partial_\beta g_{\gamma\delta},
\end{align*}
we have the following identity:
\begin{equation}
\partial_\gamma \Gamma^\gamma_{\alpha\beta} \!=  \!\frac12 \Gamma^\delta \partial_\delta g_{\alpha\beta} \! + \! \frac12(\partial_\alpha \Gamma_\beta + \partial_\beta \Gamma_\alpha) \! +  \!\frac14(\partial_\alpha g^{\gamma\delta} \partial_\beta g_{\gamma\delta} \! +  \!\partial_\beta g^{\gamma\delta}\partial_\alpha g_{\gamma\delta}) \! +  \!\frac12 g^{\gamma\delta}\partial_\alpha\partial_\beta g_{\gamma\delta}  \!-  \!\frac12 \widetilde{\Box}_g g_{\alpha\beta} \! + \! Q^1_{\alpha\beta}.
\end{equation}

Now we recall the identity
\begin{equation}
\Gamma_{\alpha\gamma}^\gamma = \frac12 g^{\gamma\delta}\partial_\alpha g_{\gamma\delta}.
\end{equation}
It follows that
\begin{subequations}
\begin{equation}
-\partial_\beta \Gamma^{\gamma}_{\alpha\gamma} = -\frac12\partial_\beta g^{\gamma\delta} \partial_\alpha g_{\gamma\delta} - \frac12 g^{\gamma\delta}\partial_\alpha\partial_\beta g_{\gamma\delta}.
\end{equation}

We can define
\[
Q^2_{\alpha\beta} = \frac14(\partial_\alpha g^{\gamma\delta}\partial_\beta g_{\gamma\delta} - \partial_\beta g^{\gamma\delta}\partial_\alpha g_{\gamma\delta})
\]
in order to rewrite
\begin{equation}
-\partial_\beta \Gamma^{\gamma}_{\alpha\gamma} = -\frac14(\partial_\alpha g^{\gamma\delta}\partial_\beta g_{\gamma\delta} + \partial_\beta g^{\gamma\delta}\partial_\alpha g_{\gamma\delta}) - \frac12 g^{\gamma\delta}\partial_\alpha\partial_\beta g_{\gamma\delta} + Q^2_{\alpha\beta}.
\end{equation}
\end{subequations}

Next, we expand
\begin{align*}
\Gamma^\gamma_{\gamma\rho}\Gamma^{\rho}_{\alpha\beta} &= \frac14 g^{\rho\sigma}g^{\gamma\delta}\partial_\rho g_{\gamma\delta}\left(\partial_\alpha g_{\beta\sigma} + \partial_\beta g_{\alpha\sigma} - \partial_\sigma g_{\alpha\beta}\right).
\end{align*}
We define
\[
Q^3_{\alpha\beta} = \frac14g^{\rho\sigma}g^{\gamma\delta}\left((\partial_\rho g_{\gamma\delta} \partial_\alpha g_{\beta\sigma} - \partial_\alpha g_{\gamma\delta}\partial_\rho g_{\beta\sigma}) + (\partial_\rho g_{\gamma\delta}\partial_\beta g_{\alpha\sigma} - \partial_\beta g_{\gamma\delta}\partial_\rho g_{\alpha\sigma}) - \partial_\rho g_{\gamma\delta} \partial_\sigma g_{\alpha\beta}\right)
\]
and rewrite
\begin{align*}
\frac14 g^{\gamma\delta}\partial_\alpha g_{\gamma\delta}g^{\rho\sigma}\partial_\rho g_{\beta\sigma} &= \frac14 g^{\gamma\delta}\partial_\alpha g_{\gamma\delta}\Big(\Gamma_\beta + \frac12 g^{\rho\sigma}\partial_\beta g_{\rho\sigma}\Big)\\
\frac14 g^{\gamma\delta}\partial_\beta g_{\gamma\delta}g^{\rho\sigma}\partial_\rho g_{\alpha\sigma} &=\frac14 g^{\gamma\delta}\partial_\beta g_{\gamma\delta}\Big(\Gamma_\alpha + \frac12 g^{\rho\sigma}\partial_\alpha g_{\rho\sigma}\Big).
\end{align*}
We therefore have
\begin{equation}
\Gamma^\gamma_{\gamma\rho}\Gamma^{\rho}_{\alpha\beta} = Q^3_{\alpha\beta} + \frac12(\Gamma^{\gamma}_{\alpha\gamma}\Gamma_\beta + \Gamma^{\gamma}_{\beta\gamma}\Gamma_\alpha) + \Gamma^\gamma_{\alpha\gamma}\Gamma^{\rho}_{\beta\rho}
\end{equation}
Now we take the final term,
\[
-\Gamma^\gamma_{\alpha\rho}\Gamma^\rho_{\beta\gamma} = -\frac14 g^{\gamma\delta}g^{\rho\sigma}(\partial_\alpha g_{\delta\rho} + \partial_\rho g_{\alpha\delta} - \partial_\delta g_{\alpha\rho})(\partial_\beta g_{\sigma\gamma} + \partial_\gamma g_{\beta\sigma}-\partial_\sigma g_{\beta\gamma}).
\]

We expand it and deal with it term-by-term. We first take
\begin{multline*}
-\frac14 g^{\gamma\delta}g^{\rho\sigma}\partial_\alpha g_{\delta\rho}(\partial_\gamma g_{\beta\sigma} - \partial_\sigma g_{\beta\gamma})= -\frac14g^{\gamma\delta}g^{\rho\sigma}((\partial_\alpha g_{\delta\rho}\partial_\gamma g_{\beta\sigma} - \partial_\gamma g_{\delta\rho}\partial_\alpha g_{\beta\sigma}) \\
-(\partial_\alpha g_{\delta\rho}\partial_\sigma g_{\beta\gamma} - \partial_\sigma g_{\delta\rho}\partial_\alpha g_{\beta\gamma}))
\quad-\frac14 g^{\gamma\delta}g^{\rho\sigma}(\partial_\gamma g_{\delta\rho}\partial_{\alpha}g_{\beta\sigma} - \partial_\sigma g_{\delta\rho}\partial_\alpha g_{\beta\gamma}), \\
= -\frac14g^{\gamma\delta}g^{\rho\sigma}((\partial_\alpha g_{\delta\rho}\partial_\gamma g_{\beta\sigma} - \partial_\gamma g_{\delta\rho}\partial_\alpha g_{\beta\sigma})-(\partial_\alpha g_{\delta\rho}\partial_\sigma g_{\beta\gamma} - \partial_\sigma g_{\delta\rho}\partial_\alpha g_{\beta\gamma})).
\end{multline*}
The two terms that cancel out are identical, which follows straightforwardly from renaming contracted indices and noting symmetry of $g$. Similarly,
\begin{equation*}
-\frac14 g^{\gamma\delta}g^{\rho\sigma}\partial_\beta g_{\sigma\gamma}(\partial_\rho g_{\alpha\delta}\! - \partial_\delta g_{\alpha\rho})= -\frac14g^{\gamma\delta}g^{\rho\sigma}((\partial_\beta g_{\sigma\gamma}\partial_\rho g_{\alpha\gamma} \!- \partial_\rho g_{\sigma\gamma}\partial_\beta g_{\alpha\gamma}) - (\partial_\beta g_{\sigma\gamma}\partial_\delta g_{\alpha\rho}\! - \partial_\delta g_{\sigma\gamma}\partial_\beta g_{\alpha\rho}))
\end{equation*}

Next, we take
\[
-\frac14 g^{\gamma\delta}g^{\rho\sigma}(\partial_\rho g_{\alpha\delta}\partial_\gamma g_{\beta\sigma} + \partial_\delta g_{\alpha\rho}\partial_\sigma g_{\beta\gamma}).
\]
By renaming indices we have
\begin{multline*}
-\frac14 g^{\gamma\delta}g^{\rho\sigma}(\partial_\rho g_{\alpha\delta}\partial_\gamma g_{\beta\sigma} + \partial_\delta g_{\alpha\rho}\partial_\sigma g_{\beta\gamma}) = -\frac12 g^{\gamma\delta}g^{\rho\sigma}(\partial_\rho g_{\alpha\delta}\partial_\gamma g_{\beta\sigma}) \\
= -\frac12 g^{\gamma\delta}g^{\rho\sigma} ((\partial_\rho g_{\alpha\delta}\partial_\gamma g_{\beta\sigma} - \partial_\gamma g_{\alpha\delta} \partial_\rho g_{\beta\sigma}) + \partial_\gamma g_{\alpha\delta}\partial_\rho g_{\beta\sigma}) \\
= -\frac12 g^{\gamma\delta}g^{\rho\sigma}(\partial_\rho g_{\alpha\delta}\partial_\gamma g_{\beta\sigma} - \partial_\gamma g_{\alpha\delta} \partial_\rho g_{\beta\sigma}) -\frac12(\Gamma_\alpha + \Gamma^{\gamma}_{\alpha\gamma})(\Gamma_\beta + \Gamma^\delta_{\beta\delta}).
\end{multline*}
Defining
\begin{align*}
Q^4_{\alpha\beta} &= -\frac14g^{\gamma\delta}g^{\rho\sigma}((\partial_\alpha g_{\delta\rho}\partial_\gamma g_{\beta\sigma} - \partial_\gamma g_{\delta\rho}\partial_\alpha g_{\beta\sigma})-(\partial_\alpha g_{\delta\rho}\partial_\sigma g_{\beta\gamma} - \partial_\sigma g_{\delta\rho}\partial_\alpha g_{\beta\gamma})) \\
&\quad-\frac14g^{\gamma\delta}g^{\rho\sigma}((\partial_\beta g_{\sigma\gamma}\partial_\rho g_{\alpha\gamma} - \partial_\rho g_{\sigma\gamma}\partial_\beta g_{\alpha\gamma}) - (\partial_\beta g_{\sigma\gamma}\partial_\delta g_{\alpha\rho} - \partial_\delta g_{\sigma\gamma}\partial_\beta g_{\alpha\rho})) \\
&\quad+ \frac14 g^{\gamma\delta}g^{\rho\sigma}\partial_\rho g_{\alpha\delta}\partial_\sigma g_{\beta\gamma} + \frac14 g^{\gamma\delta}g^{\rho\sigma}\partial_\delta g_{\alpha\rho}\partial_\gamma g_{\beta\sigma}-\frac12 g^{\gamma\delta}g^{\rho\sigma}(\partial_\rho g_{\alpha\delta}\partial_\gamma g_{\beta\sigma} - \partial_\gamma g_{\alpha\delta} \partial_\rho g_{\beta\sigma}),
\end{align*}
we can expand
\begin{equation}
-\Gamma^\gamma_{\alpha\rho}\Gamma^\rho_{\beta\gamma} = -\frac14 g^{\gamma\delta}g^{\rho\sigma}\partial_\alpha g_{\delta\rho}\partial_\beta g_{\sigma\gamma} - \frac12 (\Gamma_\alpha + \Gamma^{\gamma}_{\alpha\gamma})(\Gamma_\beta + \Gamma^\delta_{\beta\delta})+ Q^4_{\alpha\beta}
\end{equation}

We can combine everything to get
\begin{align}
R_{\alpha\beta} &= -\frac12 {\Box}^{\,g} g_{\alpha\beta} +\frac12 \Gamma^\delta \partial_\delta g_{\alpha\beta} + \frac12(\partial_\alpha \Gamma_\beta + \partial_\beta \Gamma_\alpha)  -\frac12 \Gamma_\alpha \Gamma_\beta + \frac12\Gamma^\gamma_{\alpha\gamma}\Gamma^{\rho}_{\beta\rho} - \frac14 g^{\gamma\delta}g^{\rho\sigma}\partial_\alpha g_{\delta\rho}\partial_\beta g_{\sigma\gamma}+ \\
&\quad + Q^1_{\alpha\beta} + Q^2_{\alpha\beta} + Q^3_{\alpha\beta} + Q^4_{\alpha\beta}.\qedhere
\end{align}
\end{proof}

\subsubsection{The expression for Ricci curvature in terms of the wave operator}
For a general metric $g$ twice the Ricci curvature can by Lemma \ref{lem:RicciEisntein} be written
\begin{equation}\label{eq:ricciingenharm}
2 R_{\mu\nu}=-\Box^{\,g} g_{\mu\nu}+\pa_\mu \Gamma_{\!\nu}+\pa_\nu\Gamma_{\!\mu}
+F_{\!\mu\nu}(g)[\pa g,\pa g]-\Gamma_{\!\mu} \Gamma_{\!\nu}+\Gamma^\delta\pa_\delta g_{\mu\nu},
\quad\text{where}
\quad \Gamma_{\!\mu}\!=g_{\mu\delta}g^{\alpha\beta}\Gamma_{\!\alpha\beta}^\delta ,
\end{equation}
and $F_{\mu\nu}$ is as in \eqref{eq:EinsteinWaveintro}.
The Einstein vacuum equations in harmonic coordinates are $R_{\mu\nu}\!=0$ and $\Gamma_{\!\mu}\!=0$.
The Minkowski metric and the Schwarzschild expressed in harmonic coordinates satisfy these.
Since $m_0$ in \eqref{eq:hup0def} is the leading term in the expansion of the Schwarzschild metric it is therefore not
surprising that it approximately satisfies these.
By a similar calculation as above using \eqref{eq:H0wavecoord}
 we have
\begin{multline}
\Gamma^{0\mu}=m_0^{\alpha\beta} \Gamma^{0\mu}_{\alpha\beta}
=\pa_\alpha m_0^{\alpha\beta}
-\tfrac{1}{2}m_0^{\alpha\beta} m^{-1}_{0\gamma\mu}\,\pa_\alpha m_0^{\gamma\mu}
=\pa_\alpha H_0^{\alpha\mu}
-\tfrac{1}{2}m^{\gamma\mu} m_{\alpha\beta}\,\pa_\gamma H_0^{\alpha\beta}
+W^\mu(H_0)[H_0,\pa H_0]\\
={\chi}^{\,\prime\!}\big(\tfrac{r}{t+1}\big)M\delta^{\mu 0} r^{-2}
+\chi^{\mu}\big(\tfrac{r}{1+t},\omega,\tfrac{M}{r}\big)M^2r^{-3} .
\end{multline}
Also using \eqref{eq:Boxrm0} it follows that
\begin{multline}
2 R_{\mu\nu}^0=-\Box^{m_0} m^0_{\mu\nu}+\pa_\mu \Gamma^0_\nu+\pa_\nu\Gamma^0_\mu
+\Gamma^{0\delta}\pa_\delta m^0_{\mu\nu}
+F_{\mu\nu}(m_0)[\pa m^0,\pa m^0]-\Gamma^0_\mu \Gamma^0_\nu\\
=\chi^\prime_{2,\,\mu\nu}\big(\tfrac{r}{1+t},\omega,\tfrac{M}{r}\big)M r^{-3}
+\chi_{2,\,\mu\nu}\big(\tfrac{r}{1+t},\omega,\tfrac{M}{r}\big)M^2 r^{-4}.
\end{multline}

\section{Einstein's equations in generalized wave coordinates}\label{sec:AppB}
There is a different way of deriving the reduced Einstein's equations in the new coordinates.
Instead of using covariant derivatives with respect to the new coordinates one can considers the
new coordinates as generalized wave coordinates and use the procedure for deriving
Einstein's equations in generalized wave coordinates.
Although we decided not to do it that way here we include the calculation
since it is of interest that it can be interpreted this way and it could be used elsewhere.

\subsubsection{The reduced Einstein's equations in generalized wave coordinates}
  By \eqref{eq:ricciingenharm} Einstein's equations $\widetilde{R}_{ab}=0$ in generalized wave coordinates are
\begin{equation}\label{eq:ReducedEisnteinGenWave}
\widetilde{\Box}^{\,\widetilde{g}}\,  \widetilde{g}_{ab} -\wpa_a\Wcal_{\!b}(\tg)-\wpa_b
\Wcal_{\!a}(\tg)-\Wcal^{c}(\tg)
\widetilde{\pa}_c\widetilde{ g}_{ab}=\widetilde{F}_{\!ab} (\tg) [\widetilde{\pa}\widetilde{ g},
\widetilde{\pa}\widetilde{ g}]-\Wcal_{\!a}(\tg)\Wcal_{\!b}(\tg)+\widetilde{T}_{ab},
\end{equation}
where $\widetilde{\Box}^{\,\tg\!}\!=
\widetilde{g}^{ab}\wpa_a\wpa_b$, $\Wcal_{\!c\!}\!=\tg_{\!cd}\Wcal^{d}$
 and $\Wcal^c$ are some given functions
of the coordinate $\widetilde{x}$ and the metric $\widetilde{g}$ not depending on its derivative.
One can show that if $\widetilde{g}$ satisfy  the reduced equations \eqref{eq:ReducedEisnteinGenWave}
and $\widetilde{\Gamma}^{c}_{ab}$ are its Christoffel symbols, then the generalized wave coordinate
condition
\begin{equation}\label{eq:generalizedwavecoordinatecondition}
\widetilde{g}^{ab}\widetilde{\Gamma}^{c}_{ab}
=\Wcal^c(\tg)
\end{equation}
 holds if it holds initially. In particular if  $g$
is the metric expressed in  wave coordinates $x$,
and we choose new coordinates $\widetilde{x}=\widetilde{x}(x)$
which are given fixed functions of the original coordinates then by \eqref{eq:wavetildetowavecoord}-\eqref{eq:christoffelm}
\begin{equation}\label{eq:Wdef}
\Wcal^c(\tg)
=\tg^{ ab}\widehat{\Gamma}^{c}_{ab} =\Wcal^c(\tg)[\widehat{\Gamma}],
\end{equation}
where $\widehat{\Gamma}$ given by \eqref{eq:christoffelm} are the Christoffel symbols of the fixed metric $\widetilde{m}$ and
the last equality indicates that it is linear in $\widehat{\Gamma}$ which by itself is
$O(M\langle t+r\rangle^{-2}\ln{\langle t+r\rangle})$.
By \eqref{eq:ReducedEisnteinGenWave} and \eqref{eq:Wdef}
\begin{equation}\label{eq:EisnteinGenWave}
\widetilde{\Box}^{\,\widetilde{g}}\,  \widetilde{g}_{ab}
-\wpa_a\Wcal_b(\tg)[\widehat{\Gamma}]
-\wpa_b\Wcal_a(\tg)[\widehat{\Gamma}]
-\Wcal^c(\tg)[\widehat{\Gamma}]\,
 \widetilde{\pa}_c\widetilde{ g}_{ab}
=\widetilde{F}_{ab} (\tg) [\widetilde{\pa}\widetilde{ g},\widetilde{\pa}\widetilde{ g}]
-\Wcal_{ab}(\tg)[\widehat{\Gamma},\widehat{\Gamma}]
+\widetilde{T}_{ab},
\end{equation}
where
\beq
\Wcal_c(\tg)[\widehat{\Gamma}]
=\widetilde{g}_{cd}\Wcal^d(\tg)[\widehat{\Gamma}],
\quad\text{and}\quad \Wcal_{ab}(\tg)[\widehat{\Gamma},\widehat{\Gamma}]
=\Wcal_a(\tg)[\widehat{\Gamma}]\,
\Wcal_b(\tg)[\widehat{\Gamma}]
\eq

Now let $h_{\alpha\beta}=g_{\alpha\beta}-m_{\alpha\beta}$ and $H^{\alpha\beta}=g^{\alpha\beta}-m^{\alpha\beta}=-h_{\alpha\beta}+O(h^2)$
\beq
\widetilde{h}_{ab}=\widetilde{g}_{ab}-\widetilde{m}_{ab}= A^{\,\,\alpha}_{ a} A^{\,\,\beta}_{ b} h_{\alpha\beta}
\quad\text{and}\quad
\widetilde{H}^{ab}=\widetilde{g}^{ab}-\widetilde{m}^{ab}= A_{\,\,\alpha}^{ a} A_{\,\,\beta}^{ b} H^{\alpha\beta}.
\eq
Since in particular the Minkowski metric is a solution of the vacuum equations, we have
by \eqref{eq:ReducedEisnteinGenWave} and \eqref{eq:Wdef} with
 $\widetilde{m}$ in place of $\tg$
\begin{equation}\label{eq:EisnteinGenWave}
\widetilde{\Box}^{\,\widetilde{m}}\,  \widetilde{m}_{ab}
-\wpa_a\Wcal_b(\widetilde{m})[\widehat{\Gamma}]
-\wpa_b\Wcal_a(\widetilde{m})[\widehat{\Gamma}]
-\Wcal^c(\widetilde{m})[\widehat{\Gamma}]\,
 \widetilde{\pa}_c\widetilde{m}_{ab}
=\widetilde{F}_{ab} (\widetilde{m}) [\widetilde{\pa}\widetilde{ m},
\widetilde{\pa}\widetilde{ m}]
-\Wcal_{ab}(\widetilde{m})[\widehat{\Gamma},\widehat{\Gamma}].
\end{equation}
Subtracting the two equations we get an equation for the difference
\begin{multline}\label{eq:EisnteinGenWave}
\!\!\!\!\!\!\widetilde{\Box}^{\,\widetilde{g}}  \widetilde{h}_{ab}-\widetilde{H}^{cd}\wpa_c\wpa_d \widetilde{m}_{ab}
-\wpa_a\widetilde{W}_{\!b}(\widetilde{h},\widetilde{m})
[\widetilde{h},\widehat{\Gamma}]
-\wpa_b\widetilde{W}_{\!a}(\widetilde{h},\widetilde{m})
[\widetilde{h},\widehat{\Gamma}]
-\Wcal^{c\!}(\tg\!)[\widehat{\Gamma}]\,
 \widetilde{\pa}_c\widetilde{h}_{ab}
 -\widetilde{W}_{2}^c[\widetilde{H},\widehat{\Gamma}]\wpa_c\widetilde{m}_{ab}\\
=\widetilde{F}_{ab} (\tg) [\widetilde{\pa}\widetilde{ g},\widetilde{\pa}\widetilde{ g}]
-\widetilde{F}_{ab} (\widetilde{m}) [\widetilde{\pa}\widetilde{ m},
\widetilde{\pa}\widetilde{ m}]
-\widetilde{W}_{ab}(\widetilde{h},\widetilde{m})
[\widetilde{h},\widehat{\Gamma},\widehat{\Gamma}]
+\widetilde{T}_{ab},
\end{multline}
where
\begin{align}
\widetilde{W}_c(\widetilde{h},\widetilde{m})[\widetilde{h},\widehat{\Gamma}]
&=\Wcal_c(\tg)[\widehat{\Gamma}]
-\Wcal_c(\widetilde{m})[\widehat{\Gamma}],\\
\widetilde{W}_{ab}(\widetilde{h},\widetilde{m})
[\widetilde{h},\widehat{\Gamma},\widehat{\Gamma}]
&=\Wcal_{ab}(\tg)[\widehat{\Gamma},\widehat{\Gamma}]
-\Wcal_{ab}(\widetilde{m})[\widehat{\Gamma},\widehat{\Gamma}]\\
\Wcal^c(\tg)[\widehat{\Gamma}]\,\widetilde{\pa}_c\widetilde{h}_{ab}
 +\widetilde{W}_{2}^c[\widetilde{H},\widehat{\Gamma}]\wpa_c\widetilde{m}_{ab}
 &=\Wcal^c(\tg)[\widehat{\Gamma}]\,
 \widetilde{\pa}_c\widetilde{ g}_{ab}
 -\Wcal^c(\widetilde{m})[\widehat{\Gamma}]\,
 \widetilde{\pa}_c\widetilde{m}_{ab}.
\end{align}
Here $W(\cdots)[u_1,\cdots,u_k]$ stands for functions that are separately linear in each of
the arguments $u_1,\dots,u_k$.

\subsubsection{The generalized wave coordinate condition}
We have
\beq
\wpa_a \tg^{ac} -\tfrac{1}{2}\tg^{ac}  \tg_{bd}\,\wpa_a \tg^{bd}
=\Wcal^c(\tg)[\widehat{\Gamma}]
\eq
and in particular
\beq
\wpa_a \widetilde{m}^{ac} -\tfrac{1}{2}\widetilde{m}^{ac}  \widetilde{m}_{bd}\,
\wpa_a \widetilde{m}^{bd}
=\Wcal^c(\widetilde{m})[\widehat{\Gamma}].
\eq
It follows that $\widetilde{H}^{ab}=\tg^{ab}-\widetilde{m}^{ab}$ satisfy
\beq
\wpa_a \big(\widetilde{H}^{ac} -\tfrac{1}{2}\widetilde{m}^{ac}  \widetilde{m}_{bd}\,
 \widetilde{H}^{bd}\big)
=\widetilde{W}_1^c(\widetilde{m},\widetilde{H})[\widetilde{H},\wpa\widetilde{m}]
+\widetilde{W}_2^c(\widetilde{m},\widetilde{H})[\widetilde{H},\wpa\widetilde{H}].
\eq

\section{Additional interior asymptotics}\label{sec:AppendixC}\label{sec:AppC}

\subsubsection{Subtracting off a better approximation from the homogeneous wave equation picking up the mass in the exterior}
We can get an improved  approximation to the solution $h$ of the homogeneous wave equation in the interior of the light cone by instead of having a cutoff which is a function of $r/t$ have a cutoff which is a function of $r^*-t$. This however
only works in the wave equation but not in the solution of the wave coordinate condition $H$.
This is not needed for the energy estimate and existence but only for more precise asymptotics in the interior, see \cite{L17} for the detailed proof.
Here we just want to show how the error terms are under control because of the wave coordinate condition. One can use this and the other results in this paper to give a more direct proof of the asymptotics in \cite{L17}.

Let $\Box^*=\widehat{m}^{ab}\wpa_a\wpa_b$ and
\beq
\widetilde{h}^{1\, e}_{ab}=\widetilde{h}_{ab}-\widetilde{h}^{0\, e}_{ab},\qquad\text{where}\quad
\widetilde{h}^{0\, e}_{ab}=\frac{M}{r^*} \widetilde{\chi}(r^*\!\!-t) \delta_{ab}
\eq
We have
\beq
\Box^* \widetilde{h}^{0\, e}_{cd}=0
\eq
Moreover, with $L^*_a=\wpa_\alpha(\rs\!-t)$, $\omega_0=0$,
 $\omega_i=x_i/|x|$, for $i\geq 1$ and $\slashed{S}_{0\beta}=\slashed{S}_{\alpha 0}=0$
 and $\slashed{S}_{ij}=\delta_{ij}-\omega_i\omega_j$
\beq
\wpa_a\wpa_a \wh^{0e}_{cd}\!
=\!\Big(\!L^*_a L^*_b \frac{ \widetilde{\chi}^{\,\prime\prime}(\rs\!\!-\!t)\!\!}{\rs}
+\!\big(\slashed{S}_{\! ab}
 -( L^*_a \omega_b\!+\omega_a L^*_b)  \big)\!
\frac{\widetilde{\chi}^{\,\prime}(\rs\!\!-\!t)\!\!}{{\rs}^2}
+ \big(2\omega_a \omega_b-\slashed{S}_{ab}\big)
\frac{\widetilde{\chi}(\rs\!\!-\!t)}{{\rs}^3}\Big)
M\delta_{cd}.
\eq
Hence with $R=(0,\omega)$
\beq
\widetilde{H}_1^{ab}\wpa_a\wpa_a \wh^{0e}_{cd}
=\Big(\widetilde{H}_{\!1\, L^* L^*}\frac{ \widetilde{\chi}^{\,\prime\prime}(\rs\!\!-\!t)\!\!}{\rs}
+(\trs\widetilde{H}_1-2\widetilde{H}_{\!1\, L^* R})\frac{\widetilde{\chi}^{\,\prime}(\rs\!\!-\!t)\!\!}{{\rs}^2}
+(2 \widetilde{H}_{\!1\, R R}-\trs\widetilde{H}_1)\frac{\widetilde{\chi}(\rs\!\!-\!t)}{{\rs}^3}\Big)
M\delta_{cd}.
\eq

\subsection*{Acknowledgments}
C. K. was supported in party by NSF Grant DMS-1500925  and ERC Consolidator Grant 77224. H.L. was supported in part by NSF grant DMS-1500925 and Simons Collaboration Grant 638955. We would also like to thank the Mittag Leffler Institute for their hospitality during the Fall 2019 program in Geometry and Relativity.

\end{document}